\documentclass[11pt,english,table]{article}
\usepackage{lmodern}

\usepackage[T1]{fontenc}
\usepackage[latin9]{inputenc}
\usepackage{geometry}
\usepackage{mathrsfs}
\usepackage{bm}
\usepackage{amsmath}
\usepackage{amsthm}
\usepackage{amssymb}
\usepackage{graphicx}
\usepackage{setspace}
\usepackage[authoryear]{natbib}
\makeatletter
\theoremstyle{definition}
\newtheorem{defn}{\protect\definitionname}
\theoremstyle{plain}
\newtheorem{thm}{\protect\theoremname}
\theoremstyle{remark}
\newtheorem*{rem*}{\protect\remarkname}
\theoremstyle{plain}
\newtheorem{prop}{\protect\propositionname}
\theoremstyle{plain}
\newtheorem{assumption}{\protect\assumptionname}
\theoremstyle{definition}
 \newtheorem{example}{\protect\examplename}
\theoremstyle{plain}
\newtheorem{cor}{\protect\corollaryname}
\theoremstyle{remark}
\newtheorem{rem}{\protect\remarkname}
\theoremstyle{plain}
\newtheorem{lem}{\protect\lemmaname}

\setcitestyle{round}
\usepackage{lscape}
\usepackage{longtable}
\usepackage{xcolor}
\usepackage{rotating}

\usepackage{array}
\usepackage{tabularx}\usepackage{multirow}\usepackage{booktabs}
\usepackage{ragged2e}\newcolumntype{C}[1]{>{\centering\arraybackslash}p{#1}}
\usepackage{rotfloat}
\newcolumntype{J}[1]{>{\justify\arraybackslash}p{#1}}
\newcolumntype{R}[1]{>{\RaggedLeft\arraybackslash}p{#1}}
\newcolumntype{Q}[1]{>{\columncolor{Gray}\RaggedLeft\arraybackslash}p{#1}}
\newcolumntype{L}[1]{>{\RaggedRight\arraybackslash}p{#1}}
\newcolumntype{G}{@{\extracolsep{0.5cm}}l@{\extracolsep{0pt}}}%
\newcolumntype{P}[1]{>{\centering\arraybackslash}p{#1}}
\newcolumntype{Y}{>{\centering\arraybackslash}X}
\newcommand{\nhphantom}[1]{\sbox0{#1}\hspace{-\the\wd0}} 

\AtBeginDocument{

}

\usepackage[colorlinks,
            linkcolor=blue,
            anchorcolor=blue,
            citecolor=blue]{hyperref}

\usepackage{babel}

\makeatletter
\renewcommand*{\fps@figure}{htb}
\makeatother

\makeatother

\usepackage{babel}
\providecommand{\assumptionname}{Assumption}
\providecommand{\corollaryname}{Corollary}
\providecommand{\definitionname}{Definition}
\providecommand{\examplename}{Example}
\providecommand{\lemmaname}{Lemma}
\providecommand{\propositionname}{Proposition}
\providecommand{\remarkname}{Remark}
\providecommand{\theoremname}{Theorem}

\begin{document}
\title{Convolution-$t$ Distributions\thanks{We thank George Tauchen and seminar participants at University of
Aarhus and North Carolina State University for valuable comments.
Chen Tong acknowledges financial support from the Youth Fund of the
National Natural Science Foundation of China (72301227), and the Ministry
of Education of China, Humanities and Social Sciences Youth Fund (22YJC790117).
This research obtained database services of TickData.com from the
School of Economics at Xiamen University. Corresponding author: Chen
Tong, Email: tongchen@xmu.edu.cn.}\emph{\normalsize{}\medskip{}
}}
\author{\textbf{Peter Reinhard Hansen}$^{\mathsection}$$\quad$\textbf{ }and$\quad$\textbf{Chen
Tong}$^{\ddagger}$\bigskip{}
 \\
\\
 {\normalsize{}$^{\mathsection}$}\emph{\normalsize{}Department of
Economics, University of North Carolina at Chapel Hill}\\
{\normalsize{}$^{\ddagger}$}\emph{\normalsize{}Department of Finance,
School of Economics, Xiamen University\medskip{}
 }}
\date{This version{\normalsize{}:}\emph{\normalsize{} \today}}
\maketitle
\begin{abstract}
We introduce a new class of multivariate heavy-tailed distributions
that are convolutions of heterogeneous multivariate $t$-distributions.
Unlike commonly used heavy-tailed distributions, the multivariate
convolution-$t$ distributions embody cluster structures with flexible
nonlinear dependencies and heterogeneous marginal distributions. Importantly,
convolution-$t$ distributions have simple density functions that
facilitate estimation and likelihood-based inference. The characteristic
features of convolution-$t$ distributions are found to be important
in an empirical analysis of realized volatility measures and help
identify their underlying factor structure.

\bigskip{}
\end{abstract}
\textit{\small{}{\noindent}Keywords:}{\small{} Multivariate heavy-tailed
distributions. Convolutions of $t$-distributions, Voigt profile.}{\small\par}

\noindent \textit{\small{}{\noindent}JEL Classification:}{\small{}
C01, C32, C46, C58}{\small\par}

\clearpage{}

\section{Introduction}

Heavy-tailed distributions play a central role in the modeling of
risk and extreme events in economics, finance, and beyond, and their
multivariate properties are key to risk management, portfolio optimization,
and the assessment of systemic risk, see e.g. \citet{EmbrechtsKluppelbergMikosch:1997},
\citet{Harvey2013}, and \citet{IbragimovIbragimovWalden:2015}.

This paper introduces a new class of multivariate heavy-tailed distributions,
\emph{convolution-$t$ distributions,} that are convolutions of mutually
independent \textit{\emph{multivariate}} $t$-distributions. The new
distributions can accommodate heterogeneous marginal distributions
and various forms of nonlinear dependencies. Moreover, the nonlinear
dependencies can be used to identify factors of economic significance.
Conveniently, the expression for the density function of convolution-$t$
distributions is simple. This facilitates straight forward estimation
and inference. The framework includes the multivariate-$t$ distribution
as a special case, and makes clear how it is restrictive in a number
of ways. For instance, the multivariate-$t$ distribution implies
all marginal distributions have identical shape (e.g. same kurtosis)
and it implies a particular type of nonlinear dependence between elements,
because they all share a single common Gamma-distributed (mixing)
variable.

A convolution-$t$ distribution, denoted $\mathrm{CT}_{\boldsymbol{n},\boldsymbol{\nu}}(\mu,\Xi)$,
is characterized by a location vector, $\mu$, a scale-rotation matrix,
$\Xi$, and two $K$-dimensional vectors, $\bm{n}$ and $\bm{\nu}$,
that specify cluster sizes and degrees of freedom, respectively. The
matrix, $\Xi$, define key features of the distribution and the resulting
non-linear dependencies identify the latent factor structure. We provide
a novel identifying representation of $\Xi$, which can accommodate
block structure and other sparse structures when such are needed.
This makes the framework amenable to high-dimensional applications
of the convolution-$t$ distributions.

The main contributions of this paper are as follows: First, we introduce
a family of convolution-$t$ distributions and obtain their marginal
density and cumulative distribution function by means of the characteristic
function. Second, we characterize the identification problem and show
that convolution-$t$ distributions can help unearth latent factor
structures, which is not possible from the covariance structure alone.
Third, we derive the score, hessian, and information matrix from the
log-likelihood function. And the proof of consistency and asymptotic
normality of maximum-likelihood estimators are well established. The
asymptotic properties are confirmed in a simulation study. Fourth,
we extend the moment-based approximation method by \citet{Patil:1965}
to convolutions of any number of $t$-distributions. Fifth, we show
that the convolution-$t$ distributions provide substantially empirical
gain in an application with ten financial volatility series. Importantly,
the convolution-$t$ framework adds insight about the latent factor
structure, where nonlinear dependencies identify a common market factor,
and a cluster structure that aligns with sector classifications.

Convolutions of $t$-distributions arise in classical problems, such
as the Behrens-Fisher problem (statistics) and the Voigt profile (spectroscopy),
where the latter is the convolution of a Gaussian distribution and
a Cauchy distribution.\footnote{This convolution is closely related to a Mills ratio that appears
in the Heckman model.} The literature has mainly focused on convolutions of two $t$-distributions,
see e.g. \citet{Chapman:1950}, \citet{Ghosh:1975}, and \citet{PrudnikovBrychkovMarichev:1986}.\footnote{\citet{Ruben:1960} expressed the density function as an integral
that, in the general case, involves hypergeometric functions. \citet{RahmanSaleh:1974}
and \citet{Dayal:1976} derived an expression for the distribution
of the Behrens-Fisher statistic using Appell series.} \citet{BergVignat:2010} derived an expression for the density, which
includes convolutions of two multivariate $t$-distributions. Their
expression comprises an infinite sum with coefficients that are given
from integrals. 

The convolution-$t$ distributions offer a parametric approach to
modeling multivariate variables with complex dependencies and heterogeneous
marginal distributions, similar to copula-based methods, see e.g.
\citet{Patton:Copula2012} and \citet{FanPatton:Copula2014}. For
instance, \citet{Fang2002} proposed the meta-elliptical distributions,
which construct the joint density function by combining elliptical
copula function with certain marginals, e.g. Student's $t$-distribution.
The Meta-$t$ distribution of \citet{Demarta2007} is often used in
high dimensional settings, see e.g. \citet[2017b, 2023]{OhPatton2017JBES}\nocite{OhPatton2017}\nocite{OhPatton2023},
\citet{OpschoorLucasBarraVanDick:2021}, and \citet{CrealTsay2015}.
Another related strand of literature is that on multivariate stochastic
scale mixture of Gaussian family distributions, see e.g. \citet{Eltoft2006},
\citet{FinegoldDrton:2011}, and \citet{Forbes2014}. However, for
the aforementioned two types of multivariate distributions, only the
scaling (or covariance) matrix matters. The structure of convolution-$t$
distributions is differ from existing distributions in many ways,
with an important one being the underlying cluster structure in convolution-$t$
distributions that, in conjunction with a novel scale-rotation matrix,
$\Xi$, define particular nonlinear dependencies that can be used
to identify clusters and latent factor structures in the data.

The rest of this paper is organized as follows. We introduce the convolution-$t$
distributions in Section \ref{sec:Convolution-t-Distributions} and
establish a number of its properties, including moments and marginal
densities. We derive the score, Hessian, and Fisher information from
the log-likelihood function in Section \ref{sec:Likelihood-Analysis},
and characterize an interesting identification problem in this model.
We consider the finite-sample properties of the maximum likelihood
estimator (MLE) in Section \ref{subsec:ScoreInfo}, and confirm that
the analytical expressions for standard errors, based on the asymptotic
expressions, are reliable. We introduce a standardized variant of
convolution-$t$ distribution with finite variance in Section \ref{sec:Simulation-Study},
which makes it easier to interpret estimators in the empirical analysis.
We derive a moment-based approximation methods of marginal convolution-$t$
distributions in Section \ref{sec:Approximating}. In Section \ref{sec:Empirical-Analysis},
we apply the convolution-$t$ framework to model a vector of realized
volatilities. We conclude in Section 8 and present all proofs in the
Appendix. Some additional theoretical results and simulation results
for non-standard situations are presented in the `Supplemental Material'.

\section{Convolution-$t$ Distributions\label{sec:Convolution-t-Distributions}}

The $n$-dimensional Student's $t$-distribution has density function,
\begin{equation}
f_{X}(x)=\tfrac{\Gamma\left(\tfrac{\nu+n}{2}\right)}{\Gamma\left(\tfrac{\nu}{2}\right)}(\nu\pi)^{-\frac{n}{2}}|\Sigma|^{-\frac{1}{2}}\left[1+\tfrac{1}{\nu}(x-\mu)^{\prime}\Sigma^{-1}(x-\mu)\right]^{-\frac{\nu+n}{2}},\quad x\in\mathbb{R}^{n},\label{eq:t-density}
\end{equation}
where $\nu>0$ is the degrees of freedom, $\mu\in\mathbb{R}^{n}$
is the location parameter, and $\Sigma\in\mathbb{R}^{n\times n}$
is the symmetric and positive definite scale matrix. We write $X\sim t_{n,\nu}(\mu,\Sigma)$
to denote a random variable with this distribution and include the
limited case as $\nu\rightarrow\infty$, such that $t_{\infty,n}(\mu,\Sigma)$
represents the Gaussian distribution, $X\sim N_{n}(\mu,\Sigma)$,
with density,
\begin{equation}
f_{X}(x)=(2\pi)^{-\frac{n}{2}}|\Sigma|^{-1/2}\exp(-\tfrac{1}{2}x^{\prime}\Sigma^{-1}x).\label{eq:N-density}
\end{equation}
If $X\sim t_{n,\nu}(\mu,\Sigma)$, then it is easy to verify that
\begin{equation}
Y=BX\sim t_{m,\nu}(B\mu,B\Sigma B^{\prime}),\label{eq:Y=00003DBX}
\end{equation}
for any $m\times n$ matrix $B$ with full row rank. Any marginal
distributions of a multivariate $t$-distribution is a univariate
$t$-distribution with the same degrees of freedom, $\nu$, as that
of the multivariate distribution. This can be too restrictive and
rules out marginal distributions with heterogeneity in the kurtosis.
Recall that a multivariate $t$-distribution has the representation:
$\sqrt{\frac{\nu}{\xi}}Z\sim t_{n,\nu}(\mu,\Sigma)$, where $Z\sim N_{n}(\mu,\Sigma)$
and $\xi\sim\mathrm{Gamma}(\tfrac{\nu}{2},2)$ are independent. This
highlights another characteristic of the multivariate $t$-distribution,
which is that common mixing variable $\xi$, induces very particular
nonlinear dependencies. The convolution-$t$ distribution, which we
introduce next, is a more versatile class of distributions.

\subsection{Convolution-$t$ Distribution}

Consider now the case where $X$ is composed of $K$ independent multivariate
$t$-distributions, 
\[
X=\left(\begin{array}{c}
X_{1}\\
\vdots\\
X_{K}
\end{array}\right),\qquad\text{where}\quad X_{k}\sim t_{n_{k},\nu_{k}}(0,I_{k}),\qquad\text{for}\quad k=1,\ldots,K,
\]
such that the dimension of $X$ is $n=\sum_{k=1}^{K}n_{k}$. Because
$X_{1},\ldots,X_{K}$ are mutually independent, it follows that the
density function for $X$ is given by
\begin{equation}
f_{X}(x)=\prod_{k=1}^{K}f_{X_{k}}(x_{k}),\qquad\text{for all}\quad x\in\mathbb{R}^{n},\label{eq:fx1}
\end{equation}
where $f_{X_{k}}(x_{k})$ has the form in (\ref{eq:t-density}) if
$\nu_{k}<\infty$ and the form (\ref{eq:N-density}) if $\nu_{k}=\infty$.
We will use the following notation for convolutions of heterogeneous
multivariate $t$-distributions (including Gaussian distributions).
\begin{defn}[Convolution-$t$ distribution]
We write $Y\sim\mathrm{CT}_{\boldsymbol{n},\boldsymbol{\nu}}(\mu,\Xi)$
for $Y=\mu+\Xi X$, where $\boldsymbol{n}=(n_{1},\ldots,n_{K})$,
$\boldsymbol{\nu}=(\nu_{1},\ldots,\nu_{K})$, $\mu\in\mathbb{R}^{m}$
and $\Xi\in\mathbb{R}^{m\times n}$, with $n=\sum_{k}n_{k}$, and
$X=(X_{1}^{\prime},\ldots,X_{K}^{\prime})^{\prime}$ with $X_{k}\sim t_{n_{k},\nu_{k}}(0,I_{k})$
independent for $k=1,\ldots,K$. 
\end{defn}
An important characteristic of the convolution-$t$ distribution is
that the multivariate density of $Y=\mu+\Xi X$ is very simple when
$\Xi$ is an invertible matrix. In this case, we have
\begin{equation}
f_{Y}(y)=|\det\Xi^{-1}|f_{X}(\Xi^{-1}(y-\mu))\qquad\text{ for all}\quad y\in\mathbb{R}^{n}.\label{eq:MultConv-t-n2n}
\end{equation}
The simple expression makes the analysis of the log-likelihood function
straight forward.

The properties of any convolution-$t$ distribution can be deduced
from the case where $\Xi$ is invertible (which implies $m=n$). The
case $m>\mathrm{rank}(\Xi)$ implies perfect collinearity in $Y$,
such that the properties can be inferred from a lower-dimensional
subvector. For the case $m<n$, we can introduce $m-n$ auxiliary
$Y$-variables, which can be integrated out to obtain the $m$-dimensional
distribution. We derived detailed results for the special case, $m=1$,
which represents a marginal distribution of any convolution-$t$ distribution.

\subsubsection{A Simple Trivariate Convolution-$t$ with Cluster Structure}

If $X\sim t_{n,\nu}(0,I_{n})$, then it is well known that $\Xi X$
and $\tilde{\Xi}X$ are observationally equivalent whenever $\Xi\Xi^{\prime}=\tilde{\Xi}\tilde{\Xi}^{\prime}$.
This situation is different for convolution-$t$ distributions, where
more structure in $\Xi$ can be identified. We illustrate this with
a simple trivariate example, which also highlights some properties
of convolution-$t$ distributions. 

Let $n_{1}=1$ and $n_{2}=2$ such that $Y=\mu+\Xi X\in\mathbb{R}^{3}$,
where $X_{1}\sim t_{\nu_{1}}(0,1)$ and $X_{2}\sim t_{\nu_{2}}(0,I_{2})$,
and suppose that $\nu_{1},\nu_{2}>2$ such that the variances, $\mathrm{var}(\sqrt{\tfrac{\nu_{1}-2}{\nu_{1}}}X_{1})=1$
and $\mathrm{var}(\sqrt{\tfrac{\nu_{2}-2}{\nu_{2}}}X_{2})=I_{2}$
are well-defined. Thus with
\[
\Xi=A\left[\begin{array}{ccc}
\sqrt{\tfrac{\nu_{1}-2}{\nu_{1}}} & 0 & 0\\
0 & \sqrt{\tfrac{\nu_{2}-2}{\nu_{2}}} & 0\\
0 & 0 & \sqrt{\tfrac{\nu_{2}-2}{\nu_{2}}}
\end{array}\right],
\]
we have $\mathrm{var}(Y)=AA^{\prime}$. It is easy to verify that
\[
A_{\text{sym}}=\left[\begin{array}{ccc}
\cellcolor{black!10}0.943 & 0.236 & 0.236\\
0.236 & \cellcolor{black!10}0.943 & \cellcolor{black!10}0.236\\
0.236 & \cellcolor{black!10}0.236 & \cellcolor{black!10}0.943
\end{array}\right],\quad\text{and}\quad A_{\text{asym}}=\left[\begin{array}{ccc}
\cellcolor{black!10}0.943 & -0.236 & -0.236\\
0.707 & \cellcolor{black!10}0.707 & \cellcolor{black!10}0.000\\
0.707 & \cellcolor{black!10}0.000 & \cellcolor{black!10}0.707
\end{array}\right],
\]
both result in the same covariance matrix,\footnote{The $A$-matrices are presented with approximate numerical values
for readability. The exact values are $1/(3\sqrt{2})\approx0.236$,
$1/\sqrt{2}\approx0.707$, and $2\sqrt{2}/3\approx0.943$.} which is given by
\[
\mathrm{var}(Y)=\left[\begin{array}{ccc}
1 & \rho & \rho\\
\rho & 1 & \rho\\
\rho & \rho & 1
\end{array}\right],\quad{\rm with}\quad\rho=\tfrac{1}{2}.
\]
The symmetric $A$-matrix, $A_{\mathrm{sym}}$, is the symmetric square
root of $\mathrm{var}(Y)$ and $A_{\mathrm{asym}}$ is an asymmetric
matrix. These two $A$-matrices would result in the exact same distribution
of $Y$, if $X$ had a multivariate $t$-distribution (including the
Gaussian distribution). For convoluted $t$-distributions the two
$A$-matrices lead to very different distributions for $Y$. In this
example, we intentionally chose $A_{\mathrm{asym}}$ to have a symmetric
$2\times2$ lower-right submatrix, because this emerges as an identifying
assumption in our likelihood analysis.

The joint density function of $Y$ is simply given by
\[
f_{Y}\left(y\right)=\frac{c_{1}c_{2}}{|\det\Xi|}\left(1+\tfrac{1}{\nu_{1}}x_{1}^{2}\right)^{-\frac{\nu_{1}+1}{2}}\left(1+\tfrac{1}{\nu_{2}}x_{2}^{\prime}x_{2}\right)^{-\frac{\nu_{2}+2}{2}},\qquad\text{for }y\in\mathbb{R}^{3},
\]
where $\left(x_{1},x_{2}\right)^{\prime}=\Xi^{-1}y$ and $c_{i}=c(\nu_{i},n_{i})$
with $c(\nu,n)=(\nu\pi)^{-\frac{n}{2}}\Gamma\left(\tfrac{\nu+n}{2}\right)/\Gamma\left(\tfrac{\nu}{2}\right)$.
While the covariance matrices for $(Y_{1},Y_{2})$, $(Y_{1},Y_{3})$,
and $(Y_{2},Y_{3})$, are identical, this is not the case for the
bivariate densities. We do have $f_{Y_{1},Y_{3}}=f_{Y_{1},Y_{2}}$,
but 
\[
f_{Y_{1},Y_{2}}\left(y_{1},y_{2}\right)=\int_{-\infty}^{+\infty}f_{Y}\left(y\right){\rm d}y_{3},\qquad f_{Y_{2},Y_{3}}\left(y_{2},y_{3}\right)=\int_{-\infty}^{+\infty}f_{Y}\left(y\right){\rm d}y_{1},
\]
are different and the bivariate distributions also depend on the two
$A$-matrices.

Contour plots of the bivariate distributions of convolution-$t$ distributions
are presented in Figure \ref{fig:TriConvolt}, panels (c)-(f), for
the case where $\nu_{1}=4$ and $\nu_{2}=8$. Panels (c) and (d) display
the distributions of $\left(Y_{1},Y_{2}\right)$ and $\left(Y_{2},Y_{3}\right)$
for the symmetric $A$-matrix and panels (e) and (f) are for the asymmetric
$A$-matrix. For comparison, we include the corresponding bivariate
distributions when $X\sim N(0,I_{3})$ in panel (a) and for $X\sim t_{6}(0,I_{3})$,
in panel (b).\footnote{We only display one bivariate distribution for each of the cases where
$X\sim N(0,I_{3})$ and $X\sim t_{6}(0,I_{3})$, because they are
identical for all pairs of variables and both $A$-matrices.} 

The contour plots reveal many interesting features of convolution-$t$
distributions. First, the convolution-$t$ distributions can generate
non-elliptical distributions, which is not possible with a multivariate
$t$-distribution. The Gaussian and multivariate $t$ lead to particular
forms of tail dependence, and the correlation, $\rho=1/2$, implies
that their probability is more concentrated along the 45$^{\circ}$
line. Second, the convolution-$t$ distribution can produce very heterogeneous
bivariate distributions, even if the $A$-matrix is symmetric, as
can be seen from panels (c) and (d). The distribution in panel (d)
has quasi-elliptical shape, however this is not always the case for
a pair of variables in the same group, even if $A$ is symmetric.
Third, the asymmetric $A$-matrix lead to contour plots with additional
type shapes and different degrees of tail dependence, such as the
very high level of tail dependence in panel (f). Many other bivariate
distributions can be generated by varying the choice of asymmetric
$A$-matrix, while preserving the covariance matrix. 
\begin{figure}[ph]
\begin{centering}
\includegraphics[height=0.8\textheight]{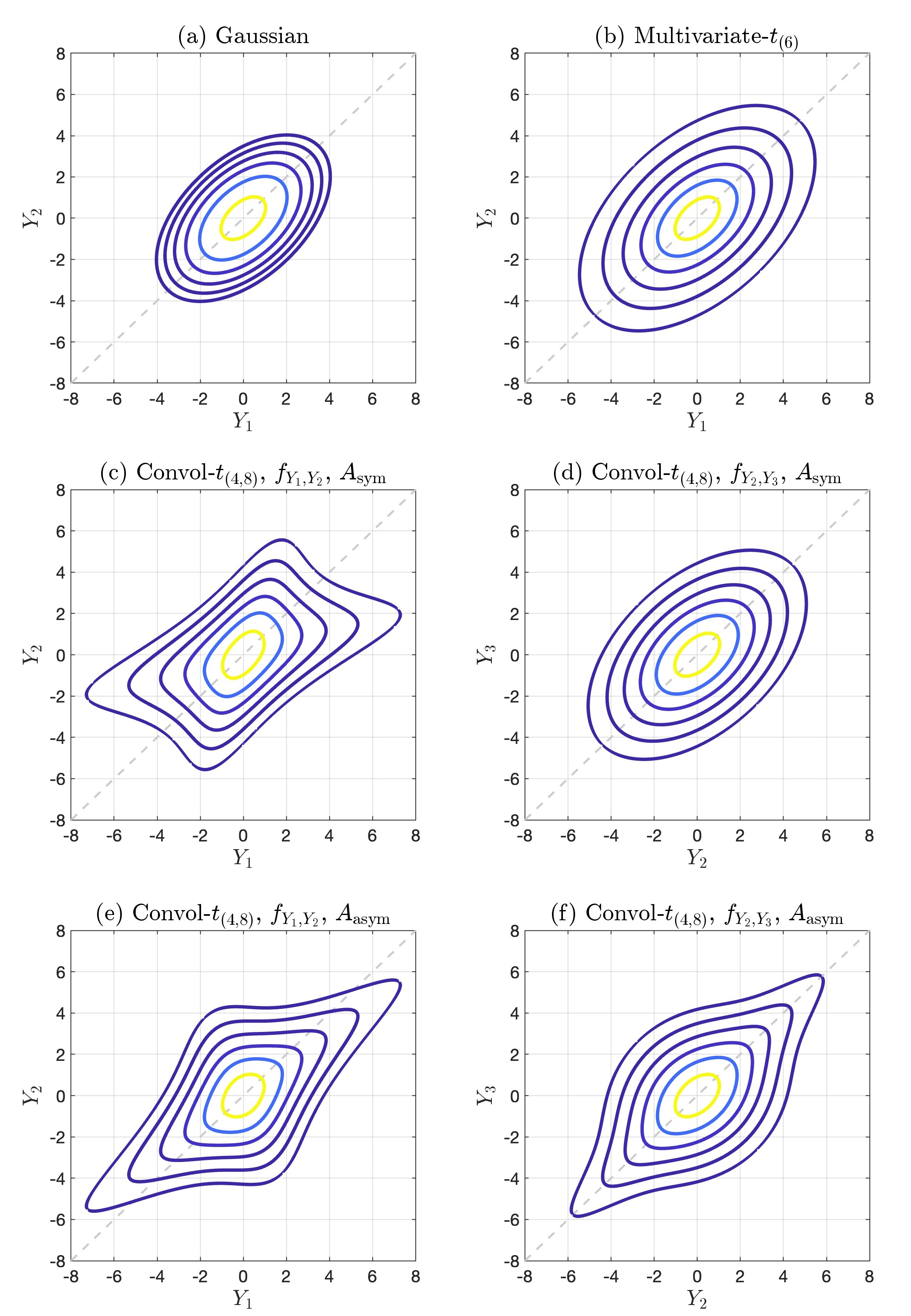}
\par\end{centering}
\caption{{\small{}Contour plots for the marginal bivariate densities from a
trivariate convolution-$t$ distributions when the covariance matrix
$\mathrm{var}(Y)$ is fixed as an equicorrelation matrix with $\rho=\frac{1}{2}$.
We consider two types of $A$ matrix, where $A_{\text{sym}}$ is from
the square root of $\mathrm{var}(Y)$, and $A_{\text{asym}}$ is obtained
by restricting (1) a block structure with positive diagonal elements
and (2) a diagonal structure for the $2\times2$ matrix in the bottom
right corner.\label{fig:TriConvolt}}}
\end{figure}

Next, we characterize the marginal distribution of a convolution-$t$
random variable.

\subsection{Marginal Distribution of Convolution-$t$}

Consider the univariate convolution of $t$-distributions, $Y_{1}=\mu+\beta^{\prime}X\in\ensuremath{\mathbb{R}}$,
where $\beta\in\mathbb{R}^{n}$. We seek to characterize the distribution
of $Y_{1}$, that could represent an element of $Y$. We achieve this
by means of the characteristic function for $Y_{1}$.\footnote{This approach as has a long history, see e.g. \citet{Gurland1948},
\citet{Gil1951}, \citet{Imhof1961}, \citet{Bohman1970}, \citet{Davies1973},
\citet[1991b]{Shephard:1991}\nocite{Shephard1991Numerical}, and
\citet{Waller1995}. This approach is also commonly used in for derivative
pricing when density function is unavailable in closed-form, see e.g.
\citet{Heston1993}, \citet{HestonNandi2000}, and \citet{BakshiMadan2000}.} 

The characteristics function for $t_{\nu}(0,1)$ is given by 
\begin{equation}
\phi_{\nu}(s)=\frac{K_{\frac{\nu}{2}}(\sqrt{\nu}|s|)(\sqrt{\nu}|s|)^{\frac{1}{2}\nu}}{\Gamma\left(\frac{\nu}{2}\right)2^{\frac{\nu}{2}-1}},\qquad\text{for}\quad s\in\mathbb{R},\label{eq:Hurst}
\end{equation}
where 
\[
K_{\nu}(x)=\tfrac{1}{2}\int_{0}^{\infty}u^{\nu-1}e^{-x\frac{u+u^{-1}}{2}}\mathrm{d}u,
\]
is the modified Bessel function of the second kind.\footnote{The modified Bessel function can be expressed in many ways and is,
confusingly, sometimes called the ``modified Bessel function of the
third kind'', e.g. \citet{Hurst:1995}. It is also known as the Basset
function, the Macdonald function, and the modified Hankel function.} The expression for $\phi_{\nu}(s)$, (\ref{eq:Hurst}), is due to
\citet{Hurst:1995} and \citet{Joarder:1995}, see \citet{Gaunt:2021}
for a very elegant proof. We also include the standard normal, $N(0,1)$,
in the analysis with the convention $\phi_{\infty}(s)\equiv e^{-s^{2}/2}$.
For $\nu=1$ (the Cauchy distribution) the expression in (\ref{eq:Hurst})
simplifies to $\phi_{1}(s)=e^{-|s|}$. Convolution-$t$ distribution
with odd degrees of freedom typically have much simpler expressions,
because their characteristic function is simpler, see the example
in Section \ref{subsec:Example2}.

Now we turn to the interesting case where $X$ is composite and derive
the characteristic function of $Y_{1}=\mu+\beta^{\prime}X=\mu+\sum_{k=1}^{K}\beta_{k}^{\prime}X_{k}$,
which we use to obtain expressions for its density and cumulative
distribution function.
\begin{thm}
\label{thm:Convo-t}Suppose that $Y_{1}=\mu+\sum_{k=1}^{K}\beta_{k}^{\prime}X_{k}\in\mathbb{R}$,
where $X_{1},\ldots,X_{K}$ are independent with $X_{k}\sim t_{n_{k},\nu_{k}}(0,I_{k})$
for $k=1,\ldots,K$. Then $Y_{1}$ has characteristic function
\[
\varphi_{Y_{1}}(s)=e^{is\mu}\prod_{k=1}^{K}\phi_{\nu_{k}}(\omega_{k}s),
\]
where $\omega_{k}=\sqrt{\beta_{k}^{\prime}\beta_{k}}$, for $k=1,\ldots,K$.
\end{thm}
By the Gil-Pelaez inversion theorem we now have the following expressions
for the marginal density and cumulative distribution function for
$Y_{1}$. These are given by
\begin{align}
f_{Y_{1}}(y) & =\frac{1}{\pi}\int_{0}^{\infty}{\rm Re}\left[e^{-isy}\varphi_{Y_{1}}(s)\right]\mathrm{d}s,\quad\text{and}\quad F_{Y_{1}}(y)=\ensuremath{\frac{1}{2}-\frac{1}{\pi}\int_{0}^{\infty}\frac{{\rm Im}\left[e^{-isy}\varphi_{Y_{1}}(s)\right]}{s}\mathrm{d}s},\label{eq:DensityCDF}
\end{align}
respective, where ${\rm Re}\left[z\right]$ and ${\rm Im}\left[z\right]$
denotes the real and imaginary part of $z\in\mathbb{C}$, respectively.

The main advantage of these expressions is that there is just a single
variable, $s$, to be integrated out, regardless of the dimensions,
$n$, and the underlying number of independent $t$-distributions,
$K$. In this case, $Y_{1}$ is a linear combination of $n=\sum_{k}n_{k}$
variables, and the conventional approach to obtain its marginal distribution
is to integrate our $n-1$ variables. This is computationally impractical
unless $n$ is small. Thus, the expressions in (\ref{eq:DensityCDF})
are likely to be computationally advantageous for $n\geq3$, whereas
the conventional approach is simpler when $n=2$, assuming that the
densities of the underlying variables are readily available, as is
the case for $t$-distributed random variables.

\subsubsection{Moments and Moments-based Density Approximation}

A key feature of the multivariate convolution-t distribution is that
it can generate heterogeneous marginal distributions with different
levels of heavy tails. The kurtosis of $Y_{1}$ is well-defined if
$\nu_{\min}=\min_{k}\nu_{k}>4$, and, as shown below, the excess kurtosis
is a linear combination of the excess kurtosis of $X_{1},\ldots,X_{K}$,
where the weights are determined by vector $\beta$ and $\boldsymbol{\nu}$. 
\begin{thm}
\label{thm:Convo-t-kurtosis}Suppose $Y_{1}=\mu+\sum_{k=1}^{K}\beta_{k}^{\prime}X_{k}$
where $X_{1},\ldots,X_{K}$ are independent with $X_{k}\sim t_{n_{k},\nu_{k}}(0,I_{k})$,
$\nu_{k}>4$, for $k=1,\ldots,K$. The variance and excess kurtosis
of $Y_{1}$ is given by
\[
\sigma_{Y_{1}}^{2}=\omega^{\prime}\omega,\qquad\kappa_{Y_{1}}=\sum_{k=1}^{K}\frac{\omega_{k}^{4}}{\left(\omega^{\prime}\omega\right)^{2}}\ \kappa_{X_{k}},
\]
where $\omega$ is a $K\times1$ vector with $\omega_{k}=\sqrt{\frac{\nu_{k}}{\nu_{k}-2}\beta_{k}^{\prime}\beta_{k}}$
for $k=1,\ldots,K$, and $\kappa_{X_{k}}$ is the excess kurtosis
of $X_{k}$ which is given by $\kappa_{X_{k}}=\frac{6}{\nu_{k}-4}$.
\end{thm}
\begin{rem*}
The sum of the weights
\[
\sum_{k=1}^{K}\omega_{k}^{4}/\left(\omega^{\prime}\omega\right)^{2}=\frac{\sum_{k=1}^{K}\omega_{k}^{4}}{\sum_{k=1}^{K}\omega_{k}^{4}+\sum_{k\neq l}^{K}\omega_{k}^{2}\omega_{l}^{2}}\leq1,
\]
implies that $\kappa_{Y_{1}}\leq\max\left(\kappa_{X_{k}}\right)$,
and in the special case where all weights are identical, $\omega_{k}=c$
we have $\kappa_{Y_{1}}=\frac{1}{K}\bar{\kappa}_{X}$ where $\bar{\kappa}_{X}=\frac{1}{K}\sum_{k=1}^{K}\kappa_{X_{k}}.$
\end{rem*}
This shows that the convolution-$t$ distribution can have heterogeneous
marginal distributions with different shapes. In Section \ref{sec:Approximating},
we use the results of Theorem \ref{thm:Convo-t-kurtosis} to approximate
the distributions of $Y_{1}$, using a univariate $t$-distribution
\begin{equation}
f_{Y_{1}}^{\star}(y)=\tfrac{\Gamma\left(\tfrac{\nu_{\star}+1}{2}\right)}{\Gamma\left(\tfrac{\nu_{\star}}{2}\right)}\left[\nu_{\star}\pi\sigma_{\star}^{2}\right]^{-\frac{1}{2}}\left[1+\frac{\left(y-\mu\right)^{2}}{\nu_{\star}\sigma_{\star}^{2}}\right]^{-\frac{\nu_{\star}+1}{2}},\label{eq:Std-t-density-fit-1}
\end{equation}
where the first four moments (assumed to be finite) are matched by
setting
\[
\mu_{\star}=\mu,\quad\nu_{\star}=4+\tfrac{6}{\kappa_{Y_{1}}},\quad\sigma_{\star}^{2}=\omega^{\prime}\omega\tfrac{\nu_{\star}-2}{\nu_{\star}}.
\]

The approximating distribution can be used to simplify calculating
of several interesting quantities, such as the Value-at-Risk (VaR)
and the Expected Shortfall (ES), given by
\[
{\rm VaR}_{\alpha}^{\star}\left(Y_{1}\right)=F_{\star}^{-1}\left(\alpha\right),\quad\mathrm{ES}_{\alpha}^{\star}\left(Y_{1}\right)=\mu_{\star}-\sigma_{\star}\frac{\nu_{\star}+\left(F_{\star}^{-1}\left(\alpha\right)\right)^{2}}{\alpha\left(\nu_{\star}-1\right)}f_{Y_{1}}^{\star}\left(F_{\star}^{-1}\left(\alpha\right)\right),
\]
where $F_{\star}^{-1}\left(\alpha\right)$ is the $\alpha$-quantile
of the univariate $t$-distribution $f^{\star}(y)$.

Because a convolution of symmetric distributions is symmetric, it
follows that all odd moments (less than $\nu_{\min}$) are zero. General
moments, including fractional moments, $r>0$, can be obtained with
the method in \citet[theorem 11.4.4]{Kawata1972}. For instance, for
$0<r<2$, we have
\[
\mathbb{E}|Y_{1}|^{r}=\delta(r)\int_{0}^{\infty}s^{-(1+r)}\left(1-{\rm Re}[\varphi_{Y_{1}}(s)]\right)\mathrm{d}s,\qquad\delta(r)=\frac{r(1-r)}{\Gamma(2-r)}\frac{1}{\sin(\tfrac{\pi}{2}(1-r))},
\]
with the convention $\delta(1)=2/\pi$. See Appendix \ref{subsec:Moments-of-Convol-t}
for expressions of general moments $2<r<\nu_{\min}$.

\subsubsection{Example 1. Convolution of Gaussian and Cauchy (Voigt Profile) }

An interesting convolution is that of a Gaussian distributed random
variable and an independent Cauchy distributed random variable.\footnote{The convolution of a Gaussian random variable and a $t$-distribution
with odd degrees of freedom is analyzed in \citet{Nason:2006}, and
\citet{Forchini:2008} generalized this analysis to $t$-distributions
with any degrees of freedom.} Suppose that $Z\sim N(0,1)$ and $X\sim\text{\ensuremath{\mathrm{Cauchy}(0,1)}}$
are independent, and we seek the distribution of their convolution,
$Y=Z+X$. The resulting distribution is known as the Voigt profile
in the field of spectroscopy, and \citet{Kendall_D:1938} derived
an expression involving the complementary error function.\footnote{Impressively, David G. Kendall published this result as a second year
undergraduate student at Oxford University, see \citet[p.164]{Bingham:1996}.} We can apply Theorem \ref{thm:Convo-t} to obtain the expression
for the density in \citet{Kendall_D:1938}.

The characteristic functions for $Z$ and $X$ are $\varphi_{Z}(s)=e^{-s^{2}/2}$
and $\varphi_{X}(s)=e^{-|s|}$, respectively, such that $\varphi_{Y}(s)=e^{-s^{2}/2-s}$,
for $s\geq0$. By Euler's formula, the expression for the density
in Theorem \ref{thm:Convo-t} simplifies to
\[
f_{Y}(y)=\frac{1}{\pi}\int_{0}^{\infty}\cos(sy)e^{-s^{2}/2-s}\mathrm{d}s,
\]
and with some additional simplifications we arrive at the following
expression for $f_{Y}(y)$. 
\begin{prop}[Kendall, 1938]
\label{Prop:Y=00003DZ+X_density}The density of $Y=X+Z$, where $Z\sim N(0,1)$
and $X\sim\text{\ensuremath{\mathrm{Cauchy}(0,1)}}$ are independent,
is given by 
\begin{equation}
f_{Y}(y)=\frac{1}{\sqrt{2\pi}}{\rm Re}\left[{\rm erfcx}\left(\tfrac{1+iy}{\sqrt{2}}\right)\right],\label{eq:YdensityAnalytical}
\end{equation}
where ${\rm erfcx}(z)=\exp(z^{2})\frac{2}{\sqrt{\pi}}\int_{z}^{\infty}e^{-t^{2}}\mathrm{d}t$
is the scaled complementary error function, and ${\rm Re}\left[\cdot\right]$
takes the real part of the complex number inside the square bracket.
\end{prop}
\begin{proof}
See Appendix \ref{sec:Proofs}.
\end{proof}
The analytical expression in (\ref{eq:YdensityAnalytical}) is very
accurate and fast to evaluate. The scaled complementary error function
can be expressed as ${\rm erfcx}(z)=\exp(z^{2})[1-{\rm erf}(z)]$,
where ${\rm erf}(z)=\frac{2}{\sqrt{\pi}}\int_{-\infty}^{z}e^{-t^{2}}\mathrm{d}t$
is the error function. These are standard non-elementary functions
that are implemented in software, such as Julia, MATLAB, and R. Interestingly,
the function, ${\rm erfcx}$, is also related to the Mills ratio that
appears in other econometric problems, such as the Heckman model,
see \citet{Heckman:1979}. This follows from 
\[
{\rm erfcx}\left(\tfrac{u}{\sqrt{2}}\right)=e^{u^{2}/2}\frac{2}{\sqrt{\pi}}\int_{-\infty}^{u/\sqrt{2}}e^{-t^{2}}\mathrm{d}t=\frac{2}{\sqrt{\pi}}\frac{\tfrac{1}{\sqrt{2}}\frac{1}{\sqrt{2\pi}}\int_{-\infty}^{u}e^{-s^{2}/2}\mathrm{d}s}{\frac{1}{\sqrt{2\pi}}e^{-u^{2}/2}}=\sqrt{\frac{2}{\pi}}\frac{1-\Phi(u)}{\varphi(u)},
\]
where we here used $\Phi(u)$ and $\varphi(u)$ to denote the cumulative
distribution function and the density, respectively, for a standard
normal random variable.

\subsubsection{Example 2. Convolutions of $t$-distributions with odd degrees of
freedom\label{subsec:Example2}}

The expression for the density of a convolution-$t$ distribution
is typically complicated. An exception is a convolution of $t$-distributions
with odd degrees of freedom, which has a much simpler expression.
The reason is that the characteristic function for a $t$-distribution
with odd degrees of freedom have the simple form, $\phi_{\nu}(s)=e^{-s_{\nu}}p(s_{\nu})$,
where $s_{\nu}=\sqrt{\nu}|s|$ and $p(s_{\nu})$ is a polynomial of
order $(\nu-1)/2\in\mathbb{N}$. The first few are given by, $\phi_{1}(s)=e^{-|s|}$,
$\phi_{3}(s)=e^{-\sqrt{3}|s|}(1+\sqrt{3}|s|)$, $\phi_{5}(s)=e^{-s_{5}}(1+s_{5}+\tfrac{1}{3}s_{5}^{2})$,
and $\phi_{7}(s)=e^{-s_{7}}(1+s_{7}+\tfrac{6}{15}s_{7}^{2}+\tfrac{1}{15}s_{7}^{3})$.\footnote{The highest-order term is $s_{\nu}^{m}/(\nu-2)!!$, where $m=(\nu-1)/2$.
E.g. for $\nu=13$ this term is $s_{13}^{6}/10395$. Also note that
$(\nu-2)!!=(2m-1)!!=2^{m}\Gamma(m+\tfrac{1}{2})/\Gamma(\tfrac{1}{2})$. } These simplifications explain that the most detailed results for
convolutions involve odd degrees of freedom, see e.g. \citet{FisherHealy:1956},
\citet{WalkerSaw:1978}, \citet{FanBerger1990}, \citet{NadarajahDey:2005},
and \citet{Nason:2006}.

For $\nu_{1}=1$ and $\nu_{2}=3$ we have $\varphi_{Y_{1}}(s)=(1+\sqrt{3}|s|)e^{-isy-(1+\sqrt{3})|s|}$
and from
\[
\int_{-\infty}^{\infty}(1+\sqrt{3}|s|)e^{-isy-(1+\sqrt{3})|s|}ds=\frac{iy+2\sqrt{3}+1}{(iy+\sqrt{3}+1)^{2}},
\]
we have that
\[
f_{t_{1}+t_{3}}(y)=\frac{1}{\pi}\mathrm{Re}\left[\frac{iy+2\sqrt{3}+1}{(iy+\sqrt{3}+1)^{2}}\right]=\frac{1}{\pi}\frac{y^{2}+16+10\sqrt{3}}{\left(y^{2}+4+2\sqrt{3}\right)^{2}}.
\]
Similarly, we can obtain
\[
f_{t_{1}+t_{5}}(y)=\frac{1}{\pi}\frac{y^{4}+2(11+3\sqrt{5})y^{2}+\tfrac{8}{3}(131+61\sqrt{5})}{\left(y^{2}+2(3+\sqrt{5})\right)^{3}},
\]
and other convolutions of $t$-distributions with odd degrees of freedom
lead to similar expressions that are ratios of polynomials in $y$.
The expressions presented above are simpler than those in \citet{NadarajahDey:2005}.\footnote{For instance, \citet{NadarajahDey:2005} has $f_{t_{1}+t_{3}}(y)=\frac{1}{\pi}(y^{2}+16+10\sqrt{3})(y^{2}+4+2\sqrt{3})^{2}(y^{2}+4-2\sqrt{3})^{4}/(4+8y^{2}+y^{4})^{2}$
and $f_{t_{1}+t_{5}}(y)=\frac{1}{\pi}(488\sqrt{5}+1048+(66+18\sqrt{5})y^{2}+3y^{4})(y^{2}+6+2\sqrt{5})^{3}(y^{2}+6-2\sqrt{5})^{6}/[5(16+12y^{2}+y^{4})^{8}]$,
which can be verified to be identical to our expressions.}
\begin{figure}[tbh]
\begin{centering}
\includegraphics[width=0.9\textwidth]{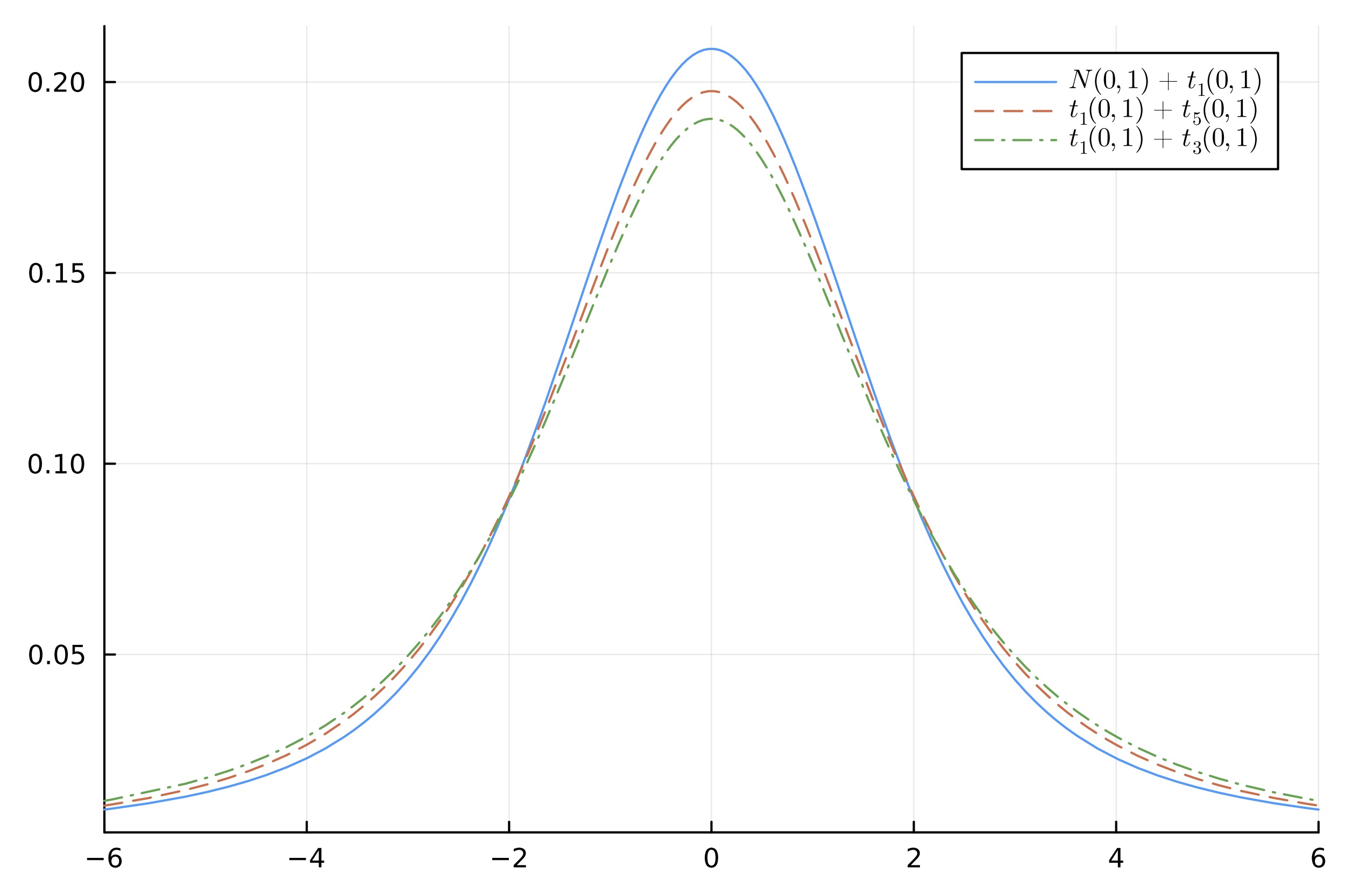}
\par\end{centering}
\caption{{\small{}Density of three simple convolution-$t$ distributions, $Y=Z+X$
where $Z\sim N(0,1)$ and $X\sim\mathrm{Cauchy(0,1)}$ independent.\label{fig:DenCauchyNormal}}}
\end{figure}

\section{Likelihood Analysis and Inference\label{sec:Likelihood-Analysis}}

We address identification, estimation, and inference of convolution-$t$
distributions in this section. Thus, consider $Y=\mu+\Xi X=\mu+\sum_{k=1}^{K}\Xi_{k}X_{k}$,
where $\Xi=(\Xi_{1},\ldots,\Xi_{K})$ with $\Xi_{k}\in\mathbb{R}^{n\times n_{k}}$,
and define $\Omega_{k}=\Xi_{k}\Xi_{k}^{\prime}$, for $k=1,\ldots,K$.

Throughout this section we make the following assumption about $\Xi\in\mathbb{R}^{n\times n}$.
\begin{assumption}
\label{assu:XiInvertible}$\det\Xi\neq0$.
\end{assumption}
It is helpful to consider the case, $K=1$, where $X$ has a multivariate
$t$-distribution. In this case $\tilde{X}=Q^{\prime}X$ has the same
distribution as $X$, for any orthonormal $n\times n$ matrix, $Q$.
Thus, if we set $\tilde{\Xi}=\Xi Q$, then $\tilde{\Xi}X=\Xi QQ^{\prime}X$
has the same distribution as $\Xi X$. Given any $\Xi$ that satisfies
Assumption \ref{assu:XiInvertible}, then $Q=\Xi^{\prime}(\Xi\Xi^{\prime})^{-\frac{1}{2}}$
satisfies $Q^{\prime}Q=I_{n}$. With this choice we have $\tilde{\Xi}=\Xi Q=(\Xi\Xi^{\prime})^{\frac{1}{2}}$,
which shows that any $\Xi$-matrix will have an observationally equivalent
symmetric $\Xi$-matrix, in the special case where $K=1$. For $K>1$
the situation is different. 

\subsection{Identification of Convolution-$t$ Distribution}

Identification is somewhat simplified if all degrees of freedom are
unique, such that they can be ordered as $\nu_{1}<\cdots<\nu_{K}$.
However, this rules out important convolutions of $t$-distributions
with the same degrees of freedom. One exception is the Gaussian case,
$\nu=\infty$, because we can always combing multiple Gaussian subvectors
into a single subvector. Thus, the following assumption is without
loss of generality.
\begin{assumption}
\label{assu:AtMostOneGaussian}If $X_{k}\sim N(0,I_{n_{k}})$ for
some $k$ then $\nu_{k^{\prime}}<\infty$ for all $k^{\prime}\neq k$.
\end{assumption}
That Assumption \ref{assu:AtMostOneGaussian} is without loss of generality
follows from the a simple argument. Suppose $X_{i}\sim N(0,I_{n_{i}})$
and $X_{j}\sim N(0,I_{n_{j}})$ then $\sum_{k}\Xi_{k}X_{k}$ and $\sum_{k\neq i,j}\Xi_{k}X_{k}+\tilde{\Xi}_{m}\tilde{X}_{m}$
have the same distribution if we set $\tilde{\Xi}_{m}=(\Xi_{i},\Xi_{j})$
and $\tilde{X}_{m}\sim N(0,I_{n_{i}+n_{j}})$. 

It is, however, not possible to combine $t$-distributed subvectors
with identical degrees of freedom into a single subvector. Even if
the subvectors are Cauchy distributed ($\nu_{k}=1)$, which is the
only $t$-distribution that is a stable distribution (as is the Gaussian
distribution). To see that there is not a way to merge two independent
Cauchy distributions into a single distribution. Suppose that $X_{1}$
and $X_{2}$ are independent and Cauchy distributed. The question
is if there is a way to express $\Xi_{1}X_{1}+\Xi_{2}X_{2}$ as $\tilde{\Xi}\tilde{X}$,
such that their distributions are identical, where $\tilde{X}$ is
an $n_{1}+n_{2}$ dimensional $t$-distribution. We have $\tilde{\Xi}\tilde{X}\sim\mathrm{Cauchy}(0,\tilde{\Xi}\tilde{\Xi}^{\prime})$
and from \citet[theorem 1]{Bian1991} it follows that $\Xi_{1}X_{1}+\Xi_{2}X_{2}$
is Cauchy distributed if and only if $\Xi_{1}\Xi_{1}^{\prime}\propto\Xi_{2}\Xi_{2}^{\prime}$.
However, this proportionality implies that $\Xi$ has reduced rank,
which is a violation of Assumption \ref{assu:XiInvertible}. 
\begin{thm}[Identification]
\label{thm:Identify}Given Assumptions \ref{assu:XiInvertible} and
\ref{assu:AtMostOneGaussian}. The location parameter, $\mu\in\mathbb{R}^{n}$
and the triplets $(n_{k},\nu_{k},\Omega_{k})$, $k=1,\ldots,K$, are
identified from the distribution of $Y=\mu+\Xi X$ (up to permutations
of the triplets), where $\Omega_{k}$ is a symmetric $n\times n$
matrix with rank $n_{k}$, $k=1,\ldots,K$ and $\Xi\Xi^{\prime}=\sum_{k}\Omega_{k}$.
\end{thm}
Theorem \ref{thm:Identify} shows that it is the $n\times n$ scaling
matrices, $\Omega_{k}=\Xi_{k}\Xi_{k}^{\prime}$, $k=1,\ldots,K$ that
hold the key to identification. 

We seek convenient identifying assumptions that will lead to a unique
$\Xi_{k}\in\mathbb{R}^{n\times n_{k}}$, which satisfies $\Omega_{k}=\Xi_{k}\Xi_{k}^{\prime}$.
The parameter-matrix, $\Xi_{k}$, is partially identified the subspace
spanned by the reduced-rank matrix, $\Omega_{k}$, and that $\mathrm{tr}\{\Xi_{k}^{\prime}\Xi_{k}\}=\mathrm{tr}\{\Omega_{k}\}$.
Suppose that $\Xi_{k}$ is a solution to $\Omega_{k}=\Xi_{k}\Xi_{k}^{\prime}$.
For any orthonormal $n_{k}\times n_{k}$ matrix $Q$, (i.e. $QQ^{\prime}=I_{n_{k}}$)
we observe that $\tilde{\Xi}_{k}=\Xi_{k}Q$ is also a solution, because
$\Omega_{k}=\tilde{\Xi}_{k}\tilde{\Xi}_{k}^{\prime}$, and $\mathrm{tr}\{\tilde{\Xi}_{k}^{\prime}\tilde{\Xi}_{k}\}=\mathrm{tr}\{Q^{\prime}\Xi_{k}^{\prime}\Xi_{k}Q\}=\mathrm{tr}\{\Xi_{k}^{\prime}\Xi_{k}QQ^{\prime}\}=\mathrm{tr}\{\Xi_{k}^{\prime}\Xi_{k}\}$.
The following theorem gives our preferred identification scheme for
$\Xi$ matrix. 
\begin{thm}
\label{Thm:XiExists}Let $\Omega_{k}$ be symmetric and positive semidefinite
with $\mathrm{rank}(\Omega_{k})=n_{k}$, $k=1,\ldots,K$. There exists
an $\Xi$ matrix with the following structure,
\begin{equation}
\Xi(\Omega_{1},\ldots,\Omega_{K})=\left[\begin{array}{cccc}
\Xi_{11} & \Xi_{12} & \cdots & \Xi_{1K}\\
\Xi_{21} & \Xi_{22}\\
\vdots &  & \ddots & \vdots\\
\Xi_{K1} & \Xi_{K2} & \cdots & \Xi_{KK}
\end{array}\right]=(\Xi_{1},\ldots,\Xi_{K}),\label{eq:Xi(Omega)}
\end{equation}
such that $\Xi_{k}\Xi_{k}^{\prime}=\Omega_{k}$ for $k=1,\ldots,K$,
where $\ensuremath{\Xi_{kl}}\in\mathbb{R}^{n_{k}\times n_{l}}$ and
$\Xi_{kk}$ is a symmetric and psd $n_{k}\times n_{k}$ matrix. Moreover,
if $\Xi_{kk}$ is also invertible for all $k=1,\ldots,K$, then $\Xi$
is unique. 
\end{thm}
From the distribution of $Y$ we can identify the triplets, $(n_{k},\nu_{k},\Omega_{k})$,
while the ordering of the $K$ triplets can arbitrary. For instance,
suppose that $\Xi(\Omega_{1},\ldots,\Omega_{K})$ has the structure
stated in Theorem \ref{Thm:XiExists}, then reordering of the triplet
will, initially, lead to a different $\Xi$ matrix, $\tilde{\Xi}=(\Xi_{j_{1}},\ldots,\Xi_{j_{K}})$,
where $j_{1},\ldots,j_{K}$ is a permutation of $1,\ldots,K$. However,
this matrix can always be manipulated into, $\Xi(\Omega_{j_{1}},\ldots,\Omega_{j_{K}})$,
which has the diagonal-symmetry structure in Theorem \ref{Thm:XiExists}. 

Each permutation of the $K$ clusters define a particular $\Xi$-matrix
with (\ref{eq:Xi(Omega)}), which are all observationally equivalent.
Fortunately, there is only a finite number of them, and we will therefore
identify $\Xi$ by the permutation that maximizes the trace, $\mathrm{tr}\{\Xi\}$.
This choice will often be meaningful, because it sorts the clusters
in a way that align with the observed variables. We illustrate this
in Example 1. This identification strategy will almost always suffice,
because permutations that produce $\Xi$-matrices with identical traces
are rare (they have Lebesgue measure zero in $\mathbb{R}^{n\times n}$).
Another simple identification scheme is available if the pairs, $(n_{k},\nu_{k})$,
$k=1,\ldots,K$, are unique. Here we may simple pick a ordering of
the clusters, which identifies $\Xi$ with the structure in Theorem
\ref{Thm:XiExists}. This is an alternative to using the trace of
$\Xi$. The two may also be combined. For instance, we may sort the
clusters by $n_{k}$, and sort clusters with the same size by maximizing
the trace. 

It is the independent components in $X$ that enables us to identify
additional parameters in $\Xi$. This is similar to independent component
analysis (ICA), which is commonly used in signal processing and machine
learning, where it is known as \textit{blind source separation}, see\emph{
}\textit{\emph{e.g.}} \citet{Stone2004}. However, our distributional
framework is more general, because we do not require $n$ mutually
independent components. To the contrary, elements of the same subvector,
$X_{k}$, are dependent, and this dependence is informative about
the underlying cluster structure.
\begin{example}
We revisit the trivariate example considered previously, where $X=(X_{1}^{\prime},X_{2}^{\prime})^{\prime}$,
with $X_{1}\sim t_{\nu_{1},1}(0,1)$ and $X_{2}\sim t_{\nu_{2},2}(0,I_{2})$.
The asymmetric $A$-matrix lead to
\[
\Xi=[\Xi_{1},\Xi_{2}]=\left[\begin{array}{ccc}
\cellcolor{black!10}\frac{4}{3}\sqrt{\frac{c_{1}}{2}} & \cellcolor{black!5}-\frac{1}{3}\sqrt{\frac{c_{2}}{2}} & \cellcolor{black!5}-\frac{1}{3}\sqrt{\frac{c_{2}}{2}}\\
\cellcolor{black!10}\sqrt{\frac{c_{1}}{2}} & \cellcolor{black!5}\sqrt{\frac{c_{2}}{2}} & \cellcolor{black!5}0\\
\cellcolor{black!10}\sqrt{\frac{c_{1}}{2}} & \cellcolor{black!5}0 & \cellcolor{black!5}\sqrt{\frac{c_{2}}{2}}
\end{array}\right],
\]
which has the structure defined in Theorem \ref{Thm:XiExists}, where
$c_{k}=\tfrac{\nu_{k}-2}{\nu_{k}}$, $k=1,2$. Here we have
\begin{eqnarray*}
\Xi\Xi^{\prime} & = & \left[\begin{array}{ccc}
\cellcolor{black!15}\tfrac{8}{9}c_{1}+\tfrac{1}{9}c_{2} & \cellcolor{black!10}\tfrac{2}{3}c_{1}-\tfrac{1}{6}c_{2} & \cellcolor{black!10}\tfrac{2}{3}c_{1}-\tfrac{1}{6}c_{2}\\
\cellcolor{black!10}\tfrac{2}{3}c_{1}-\tfrac{1}{6}c_{2} & \cellcolor{black!5}\tfrac{1}{2}c_{1}+\tfrac{1}{2}c_{2} & \cellcolor{black!5}\tfrac{1}{2}c_{1}\\
\cellcolor{black!10}\tfrac{2}{3}c_{1}-\tfrac{1}{6}c_{2} & \cellcolor{black!5}\tfrac{1}{2}c_{1} & \cellcolor{black!5}\tfrac{1}{2}c_{1}+\tfrac{1}{2}c_{2}
\end{array}\right]=\Omega_{1}+\Omega_{2}
\end{eqnarray*}
where 
\[
\Omega_{1}=\Xi_{1}\Xi_{1}^{\prime}=c_{1}\left[\begin{array}{ccc}
\tfrac{8}{9} & \tfrac{2}{3} & \tfrac{2}{3}\\
\tfrac{2}{3} & \tfrac{1}{2} & \tfrac{1}{2}\\
\tfrac{2}{3} & \tfrac{1}{2} & \tfrac{1}{2}
\end{array}\right],\qquad\Omega_{2}=\Xi_{2}\Xi_{2}^{\prime}=c_{2}\left[\begin{array}{ccc}
\phantom{{-}}\tfrac{1}{9} & -\tfrac{1}{6} & -\tfrac{1}{6}\\
-\tfrac{1}{6} & \phantom{{-}}\tfrac{1}{2} & \phantom{{-}}0\\
-\tfrac{1}{6} & \phantom{{-}}0 & \phantom{{-}}\tfrac{1}{2}
\end{array}\right].
\]
If we reverse the ordering of the triplet $(n_{k},\nu_{k},\Omega_{k})$,
such that $\check{X}=(X_{2}^{\prime},X_{1})^{\prime}$, then we obviously
have 
\[
\Xi X=\tilde{\Xi}\tilde{X},\qquad\text{with}\quad\tilde{\Xi}=[\Xi_{2},\Xi_{1}]=\left[\begin{array}{ccc}
\cellcolor{black!5}-\frac{1}{3}\sqrt{\frac{c_{2}}{2}} & \cellcolor{black!5}-\frac{1}{3}\sqrt{\frac{c_{2}}{2}} & \cellcolor{black!10}\frac{4}{3}\sqrt{\frac{c_{1}}{2}}\\
\cellcolor{black!5}\sqrt{\frac{c_{2}}{2}} & \cellcolor{black!5}0 & \cellcolor{black!10}\sqrt{\frac{c_{1}}{2}}\\
\cellcolor{black!5}0 & \cellcolor{black!5}\sqrt{\frac{c_{2}}{2}} & \cellcolor{black!10}\sqrt{\frac{c_{1}}{2}}
\end{array}\right],
\]
but this matrix does not have the required structure in Theorem \ref{Thm:XiExists},
because the upper-left $2\times2$ submatrix is not symmetric. The
desired structure is achieved with $\check{\Xi}=[\check{\Xi}_{1},\check{\Xi}_{2}]$
where $\tilde{\Xi}_{2}=\Xi_{1}$ and 
\[
\check{\Xi}_{1}=\Xi_{2}\check{P}_{11}=\sqrt{\frac{c_{2}}{2}}\left[\begin{array}{cc}
\cellcolor{black!5}\phantom{{-}}0.404 & \cellcolor{black!5}-0.242\\
\cellcolor{black!5}-0.242 & \cellcolor{black!5}\phantom{{-}}0.970\\
\cellcolor{black!5}-0.970 & \cellcolor{black!5}\phantom{{-}}0.242
\end{array}\right]
\]
with 
\[
\check{P}_{11}=\left[\begin{array}{cc}
-\frac{1}{3} & -\frac{1}{3}\\
\phantom{{-}}1 & \phantom{{-}}0
\end{array}\right]^{\prime}\left(\left[\begin{array}{cc}
\phantom{{-}}\frac{2}{9} & -\frac{1}{3}\\
-\frac{1}{3} & \phantom{{-}}1
\end{array}\right]\right)^{-1/2},
\]
inferred from the upper $2\times2$ matrix of $\Xi_{2}$. Thus $\Xi$
and $\check{\Xi}$ are observationally equivalent and we proceed to
use the maximum trace to identity the matrix. In this example we have,
$\mathrm{tr}\{\Xi\}=\frac{4}{3}\sqrt{\frac{c_{1}}{2}}+2\sqrt{\frac{c_{2}}{2}}$
and $\mathrm{tr}\{\check{\Xi}\}\simeq\sqrt{\frac{c_{1}}{2}}+1.374\sqrt{\frac{c_{2}}{2}}$.
Because the former is larger, $\Xi$ is the identify parameter.
\end{example}
There are other identification strategies than that of Theorem \ref{Thm:XiExists}.
For instance, one could impose a Cholesky structure on $\Xi_{kk}$.
We prefer the structure with symmetric diagonal blocks, because it
is well suited for the parsimonious block structures we introduce
below, as well as general symmetry restrictions that also reduce the
number of free parameters.
\begin{defn}
We say that $B\in\mathbb{R}^{n\times n}$ is a block matrix with $K$
blocks, if it can be expressed as the following form
\[
B=\left[\begin{array}{cccc}
B_{[1,1]} & B_{[1,2]} & \cdots & B_{[1,K]}\\
B_{[2,1]} & B_{[2,2]}\\
\vdots &  & \ddots\\
B_{[K,1]} &  &  & B_{[K,K]}
\end{array}\right],
\]
where $B_{[k,l]}$ is an $n_{k}\times n_{l}$ matrix such that
\begin{equation}
B_{[k,k]}=\left[\begin{array}{cccc}
d_{k} & b_{kk} & \cdots & b_{kk}\\
b_{kk} & d_{k} & \ddots\\
\vdots & \ddots & \ddots\\
b_{kk} &  &  & d_{k}
\end{array}\right]\quad\text{and }\quad B_{[k,l]}=\left[\begin{array}{ccc}
b_{kl} & \cdots & b_{kl}\\
\vdots & \ddots\\
b_{kl} &  & b_{kl}
\end{array}\right]\quad\text{if }k\neq l,\label{eq:BlockStructure}
\end{equation}
where $d_{k}$ and $b_{kl}$ for $k,l=1,\ldots,K$ are constants and
$\sum_{k}n_{k}=n$.
\end{defn}
This definition is taken from \citet{ArchakovHansen:CanonicalBlockMatrix}.
The case where the scaling matrix, $\Xi\Xi^{\prime}=\sum_{k=1}^{K}\Omega_{k}$,
has a block matrix is interesting, because $\Xi$ may be assumed to
have the same block structure, albeit not necessarily a symmetric
block structure.
\begin{cor}
If Assumptions \ref{assu:XiInvertible} and \ref{assu:AtMostOneGaussian}
hold and $\sum_{k=1}^{K}\Omega_{k}$ is a symmetric block matrix.
Then there exists a (possibly non-symmetric) block matrix, $\Xi$,
such that $\Xi\Xi^{\prime}=\sum_{k=1}^{K}\Omega_{k}$.
\end{cor}
\begin{rem}
If $\Omega_{k}$ has a symmetric block structure from group assignment
$\left(n_{1}\ldots,n_{K}\right)$, then $\Xi_{k}$ preserves the same
block structure with
\[
\Xi_{kk}=\left[\begin{array}{cccc}
d_{kk} & \beta_{kk} & \cdots & \beta_{kk}\\
\beta_{kk} & d_{kk} & \ddots\\
\vdots & \ddots & \ddots\\
\beta_{kk} &  &  & d_{kk}
\end{array}\right],\quad\Xi_{kl}=\left[\begin{array}{ccc}
\beta_{kl} & \cdots & \beta_{kl}\\
\vdots & \ddots & \vdots\\
\beta_{kl} & \cdots & \beta_{kl}
\end{array}\right]\quad\text{for }k\neq l.
\]
\end{rem}
According to Theorem \ref{Thm:XiExists}, if $\Xi_{kk}$ is symmetric
and positive definite, then the $\Xi_{k}$ matrix is unique and identified
from $\Omega_{k}$. Thus, we can directly imposing the following structure
of $\Xi$ matrix in estimation
\[
\Xi=\left[\Xi_{1},\Xi_{2},\ldots,\Xi_{K}\right]=\left(\begin{array}{cccc}
\Xi_{11} & \Xi_{12} & \cdots & \Xi_{1K}\\
\Xi_{21} & \Xi_{22} & \ddots & \vdots\\
\vdots & \ddots & \ddots & \Xi_{K-1,K}\\
\Xi_{K,1} & \cdots & \Xi_{K,K-1} & \Xi_{KK}
\end{array}\right)
\]
where $\Xi_{kk}$ can by parameterized by $\Xi_{kk}=\exp(\gamma_{k})$,
where $\gamma_{k}$ is an unrestricted symmetric matrix, and $\exp\left(\cdot\right)$
takes matrix exponential. 

Note that if $\Omega_{k}$, $k=1,2,\ldots,K$, all have a block structure
with identical block sizes, then $\Xi=\left[\Xi_{1},\Xi_{2},\ldots,\Xi_{K}\right]$
has the same block structure (but need not be symmetric). Block structures
may be inferred from estimates of $\Omega=\sum_{k}\Omega_{k}$ that
will have the same block structure, such as in our empirical analysis
where $\hat{\Omega}$ is found to have an approximate block structures.

\subsection{The Score, Hessian, and Fisher Information\label{subsec:ScoreInfo}}

Next, we proceed to establish likelihood-based results that are used
for inference about $\mu$, $\Xi$, and $\boldsymbol{\nu}$. We do
not address inference about the integer-valued vector, $\boldsymbol{n}$,
which is a non-standard problem, similar to lag-length selection in
autoregressions and determining the cointegration rank in vector autoregressions.

We define the following three key parameter vectors with
\begin{align}
\mu\in\mathbb{R}^{n\times1},\quad{\rm vec}(\tilde{\Xi})\in\mathbb{R}^{n^{2}\times1},\quad{\rm and}\quad\boldsymbol{\nu} & =(\nu_{1},\nu_{2},\ldots,\nu_{K})^{\prime}\in\mathbb{R}^{K\times1},\label{eq:parameters}
\end{align}
where the $\tilde{\Xi}$ matrix is unrestricted (i.e. it need not
have the identifying structure of Theorem \ref{Thm:XiExists}). The
corresponding score vectors are denoted by
\[
\nabla_{\mu}=\tfrac{\partial\ell}{\partial\mu},\quad\nabla_{\tilde{\Xi}}=\tfrac{\partial\ell}{\partial{\rm vec}\left(\tilde{\Xi}\right)},\quad{\rm and}\quad\nabla_{\nu_{k}}=\tfrac{\partial\ell}{\partial\nu_{k}},
\]
and, for later use, we define
\[
W_{k}=\frac{\nu_{k}+n_{k}}{\nu_{k}+X_{k}^{\prime}X_{k}},\quad X_{k}=e_{k}^{\prime}X,
\]
where $e_{k}$ is a $n\times n_{k}$ matrix from identity matrix $I_{n}=\left(e_{1},e_{1},\ldots,e_{K}\right)$.
We also introduce the notation $J_{k}\equiv e_{k}e_{k}^{\prime}$
and note that $I_{n}=\sum_{k=1}^{K}J_{k}$. 

We have the following formulas for the score vectors, hessian matrix,
and corresponding information matrix.
\begin{thm}
\label{thm:ScoreHess}Suppose that $Y=\mu+\tilde{\Xi}X\in\mathbb{R}^{n}$,
where $\det(\tilde{\Xi})\neq0$ and $X\sim\mathrm{CT}_{\boldsymbol{n},\boldsymbol{\nu}}(0,I_{n})$.
The score vectors for the parameter vectors $\mu$, $\bm{\nu}$ and
${\rm vec}(\tilde{\Xi})$, are given by
\begin{align*}
\nabla_{\mu} & =\sum_{k=1}^{K}W_{k}A^{\prime}e_{k}X_{k},\\
\nabla_{\tilde{\Xi}} & =\sum_{k=1}^{K}W_{k}{\rm vec}\left(A^{\prime}e_{k}X_{k}X^{\prime}\right)-{\rm vec}\left(A^{\prime}\right),\\
\nabla_{\nu_{k}} & =\tfrac{1}{2}\left[\psi\left(\tfrac{v_{k}+n_{k}}{2}\right)-\psi\left(\tfrac{v_{k}}{2}\right)+1-W_{k}-\log\left(1+\tfrac{X_{k}^{\prime}X_{k}}{\nu_{k}}\right)\right],
\end{align*}
where $\psi(\cdot)$ is the digamma function and $A\equiv\tilde{\Xi}^{-1}$.
The Hessian matrix has the following components
\begin{align*}
\nabla_{\mu\mu^{\prime}} & =\sum_{k=1}^{K}\tfrac{2W_{k}^{2}}{\nu_{k}+n_{k}}A^{\prime}e_{k}X_{k}X_{k}^{\prime}e_{k}^{\prime}A-W_{k}A^{\prime}J_{k}A,\\
\nabla_{\mu\tilde{\Xi}^{\prime}} & =\sum_{k=1}^{K}W_{k}\left(AX_{k}^{\prime}e_{k}^{\prime}\otimes A^{\prime}\right)K_{n}+\tfrac{2W_{k}^{2}}{\nu_{k}+n_{k}}A^{\prime}e_{k}X_{k}{\rm vec}\left(A^{\prime}e_{k}X_{k}X^{\prime}\right)^{\prime}-W_{k}A^{\prime}\left(X^{\prime}\otimes J_{k}A\right),\\
\nabla_{\mu\nu_{k}} & =\sum_{k=1}^{K}\tfrac{1}{\nu_{k}+n_{k}}A^{\prime}e_{k}X_{k}\left(W_{k}-W_{k}^{2}\right),\\
\nabla_{\tilde{\Xi}\tilde{\Xi}^{\prime}} & =-\ensuremath{\left(A\otimes A^{\prime}\right)}K_{n}+\sum_{k=1}^{K}\tfrac{2W_{k}^{2}}{\nu_{k}+n_{k}}{\rm vec}\left(A^{\prime}e_{k}X_{k}X^{\prime}\right){\rm vec}\left(A^{\prime}e_{k}X_{k}X^{\prime}\right)^{\prime},\\
 & \quad+\sum_{k=1}^{K}W_{k}\left[\left(XX_{k}^{\prime}e_{k}^{\prime}A\otimes A^{\prime}\right)K_{n}-\left(XX^{\prime}\otimes A^{\prime}J_{k}A\right)-\left(A\otimes A^{\prime}e_{k}X_{k}X^{\prime}\right)K_{n}\right],\\
\nabla_{\tilde{\Xi}\nu_{k}} & =\tfrac{1}{\nu_{k}+n_{k}}\left(W_{k}-W_{k}^{2}\right){\rm vec}\left(A^{\prime}e_{k}X_{k}X^{\prime}\right),\\
\nabla_{\nu_{k}\nu_{k}} & =\tfrac{1}{4}\psi^{\prime}\left(\tfrac{v_{k}+n_{k}}{2}\right)-\tfrac{1}{4}\psi^{\prime}\left(\tfrac{v_{k}}{2}\right)+\tfrac{1}{2\nu_{k}}+\tfrac{1}{2}\tfrac{1}{\nu_{k}+n_{k}}\left(W_{k}^{2}-2W_{k}\right),\\
\nabla_{\nu_{k}\nu_{l}}^ {} & =0,\quad{\rm if}\quad k\neq l,
\end{align*}
where $\psi^{\prime}\left(\cdot\right)$ is the trigamma function.
\end{thm}
Note that in the expressions for the score and Hessian matrix, extreme
values of $X_{k}$ are dampened by $W_{k}$. For instance, each of
the terms $\nabla_{\mu}$, $\nabla_{\mu\mu^{\prime}}$, $\nabla_{\mu\nu_{k}}$,
and $\nabla_{\nu_{k}\nu_{k}}$ are bounded, because $W_{k}\leq\frac{\nu_{k}+n_{k}}{\nu_{k}}$,
$\left\Vert W_{k}X_{k}\right\Vert \leq\frac{1}{2}\sqrt{\frac{n_{k}}{\nu_{k}}}\left(\nu_{k}+n_{k}\right)$,
and $W_{k}X_{k}^{\prime}X_{k}\leq\nu_{k}+n_{k}$, and it follows that
all moments are finite for these terms. For the term, $\nabla_{\nu_{k}}$,
we note that $\log\left(1+X_{k}^{\prime}X_{k}/\nu_{k}\right)$ has
the same distribution as $\log(1/U)$, where $U\sim\mathrm{Beta}(\frac{\nu_{k}}{2},\frac{n_{k}}{2})$,
for which all moments are finite,\footnote{$\mathbb{E}([\log(1/U)]^{m})=\psi^{(m-1)}(\nu_{k})-\psi^{(m-1)}(\nu_{k}+n_{k})$,
where $\psi^{(n)}(x)$ is the $n$-th derivatives of $\log\Gamma(x)$,
which are known as polygamma functions.} see Lemma \ref{lem:qHomogeneous} for expressions of the first two
moments. It can also be show that $\nabla_{\tilde{\Xi}}$ has finite
moment when $K=1$, however this need not be the case in general.
When $K\geq2$, there are interaction terms, such as $W_{k}X_{k}^{\prime}X_{l}$,
which lacks $W_{l}$ to dampen the effect of large realizations of
$X_{l}$. For the case, $K\geq2$, we show $\min_{k}\nu_{k}>1$ is
needed to guarantee the existences of $\mathbb{E}\nabla_{\tilde{\Xi}}$,
$\mathbb{E}\nabla_{\mu\tilde{\Xi}^{\prime}}$, and $\mathbb{E}\nabla_{\tilde{\Xi}\nu_{k}}$,
and $\min_{k}\nu_{k}>2$ is needed for $\mathbb{E}\nabla_{\tilde{\Xi}\tilde{\Xi}^{\prime}}$
to be well-defined. The latter is intuitive since it involves the
terms, $\frac{\nu_{k}}{\nu_{k}-2}$, $k=1,\ldots,K$.
\begin{thm}
\label{thm:ExpecScoreHess}Given the assumptions of Theorem \ref{thm:ScoreHess},
we have $\mathbb{E}[\nabla_{\mu}]=0$ and $\mathbb{E}[\nabla_{\nu_{k}}]=0$,
for $k=1,2,\ldots,K$. If, in addition, $\min_{k}\nu_{k}>1$, we have
$\mathbb{E}[\nabla_{\tilde{\Xi}}]=0$. 

The Fisher information, 
\[
\mathcal{I}=\mathbb{E}\left[\nabla_{\theta}\nabla_{\theta}^{\prime}\right]=\left[\begin{array}{ccc}
\mathcal{I}_{\mu} & \mathcal{I}_{\tilde{\Xi}\mu} & \mathcal{I}_{\nu\mu}\\
\mathcal{I}_{\mu\tilde{\Xi}} & \mathcal{I}_{\tilde{\Xi}} & \mathcal{I}_{\tilde{\Xi}\nu}\\
\mathcal{I}_{\mu\nu} & \mathcal{I}_{\nu\tilde{\Xi}} & \mathcal{I}_{\nu}
\end{array}\right],
\]
is such that 
\begin{align*}
\mathcal{I}_{\mu}=\mathbb{E}(\nabla_{\mu}\nabla_{\mu}^{\prime}) & =\sum_{k=1}^{K}\phi_{k}A^{\prime}J_{k}A,\qquad\phi_{k}\equiv\frac{\nu_{k}+n_{k}}{\nu_{k}+n_{k}+2},
\end{align*}
is well-defined for all $\boldsymbol{\nu}>0$. Moreover, if $\min_{k}\nu_{k}>1$,
then $\mathcal{I}_{\mu\tilde{\Xi}}=\mathcal{I}_{\tilde{\Xi}\mu}^{\prime}=0$,
$\mathcal{I}_{\mu\nu}=\mathcal{I}_{\mu\nu}^{\prime}=0$ (are well-defined
and equal zero). If $\min_{k}\nu_{k}>2$, then $\mathcal{I}_{\tilde{\Xi}}$
is well-defined and given by
\begin{align*}
\mathcal{I}_{\tilde{\Xi}} & =\mathbb{E}(\nabla_{\tilde{\Xi}}\nabla_{\tilde{\Xi}}^{\prime})=\left(A\otimes A^{\prime}\right)K_{n}+\Upsilon_{K},
\end{align*}
where $\Upsilon_{K}=\sum_{k=1}^{K}\Psi_{k}$ with
\begin{align*}
\Psi_{k} & =\phi_{k}J_{k}^{\bullet}\otimes\left(A^{\prime}J_{k}A\right)+\left(\phi_{k}-1\right)\left[\left(J_{k}A\otimes A^{\prime}J_{k}\right)K_{n}+{\rm vec}(A^{\prime}J_{k}){\rm vec}(A^{\prime}J_{k})^{\prime}\right],
\end{align*}
and $J_{k}^{\bullet}=\sum_{l\neq k}\frac{\nu_{l}}{\nu_{l}-2}J_{l}+J_{k}$
. Next, $\mathcal{I}_{\nu}$ is a diagonal matrix with diagonal elements
\[
\mathcal{I}_{\nu_{k}}=\mathbb{E}(\nabla_{\nu_{k}}^{2})=\tfrac{1}{4}\left(\psi^{\prime}\left(\tfrac{\nu_{k}}{2}\right)-\psi^{\prime}\left(\tfrac{\nu_{k}+n_{k}}{2}\right)\right)-\tfrac{n_{k}\left(\nu_{k}+n_{k}+4\right)}{2\nu_{k}\left(\nu_{k}+n_{k}+2\right)\left(\nu_{k}+n_{k}\right)}.
\]
Finally, $\mathcal{I}_{\tilde{\Xi}\nu}=\mathcal{I}_{\nu\tilde{\Xi}}^{\prime}$
is an $n^{2}\times K$ matrix with $k$-th column given by
\[
\mathcal{I}_{\tilde{\Xi}\nu_{k}}=\mathbb{E}\left[\nabla_{\nu_{k}}\nabla_{\tilde{\Xi}}\right]=\tfrac{\phi_{k}-1}{v_{k}+n_{k}}{\rm vec}\left(A^{\prime}J_{k}\right).
\]
When $\mathcal{I}_{\theta_{1}\theta_{2}}$ is well defined, $\theta_{1},\theta_{2}\in\{\mu,\tilde{\Xi},\boldsymbol{\nu}$\}
then so is the corresponding submatrix of $\mathcal{J}=\mathbb{E}\left(-\nabla_{\theta\theta^{\prime}}\right)$
and the information matrix identity, $\mathcal{I}_{\theta_{1}\theta_{2}}=\mathcal{J}_{\theta_{1}\theta_{2}}$
holds.
\end{thm}

\subsubsection{The Jacobian for the identified $\Xi$}

Recall that $\Xi=\tilde{\Xi}P$, where $P=\mathrm{diag}(P_{11},\ldots,P_{KK})$
with $P_{kk}=\tilde{\Xi}_{kk}^{\prime}(\tilde{\Xi}_{kk}\tilde{\Xi}_{kk}^{\prime})^{-\frac{1}{2}}\in\mathbb{R}^{n_{k}\times n_{k}}$.
From the results we established for the unstructured $\tilde{\Xi}$,
we obtain the results for the identified parametrization, $\Xi$,
by means of the Jacobian matrix of $\partial{\rm vec}(\Xi)$ with
respect to $\partial{\rm vec}(\tilde{\Xi})$. Below, $K_{n,m}\in\mathbb{R}^{mn\times mn}$
denotes the commutation matrix.
\begin{lem}
\label{lem:JacobianMatrix}The Jacobian is given by
\[
M_{\tilde{\Xi}}\equiv\frac{\partial{\rm vec}(\Xi)}{\partial{\rm vec}(\tilde{\Xi})^{\prime}}={\rm diag}\left(\Gamma_{11},\Gamma_{22},\ldots,\Gamma_{KK}\right),
\]
where $\Gamma_{kk}=K_{n_{k},n}\Pi^{(k)}K_{n,n_{k}}$ and $\Pi^{(k)}\in\mathbb{R}^{nn_{k}\times nn_{k}}$
is made of the submatrices
\[
\Pi_{ij}^{(k)}=\begin{cases}
0_{n_{i}n_{k}\times n_{j}n_{k}} & i\neq j,j\neq k,\\
I_{n_{i}}\otimes P_{kk}^{\prime} & i=j\neq k,\\
\left(\tilde{\Xi}_{ik}\otimes I_{n_{k}}\right)\frac{\partial{\rm vec}\left(P_{kk}^{\prime}\right)}{\partial{\rm vec}\left(\tilde{\Xi}_{kk}^{\prime}\right)^{\prime}} & i\neq j,j=k,\\
\left(I_{n_{k}}\otimes P_{kk}^{\prime}\right)+\left(\tilde{\Xi}_{kk}\otimes I_{n_{k}}\right)\frac{\partial{\rm vec}\left(P_{kk}^{\prime}\right)}{\partial{\rm vec}\left(\tilde{\Xi}_{kk}^{\prime}\right)^{\prime}} & i=j=k,
\end{cases}
\]
for $i,j=1,\ldots,K$, with 
\[
\frac{\partial{\rm vec}\left(P_{kk}^{\prime}\right)}{\partial{\rm vec}(\tilde{\Xi}_{kk}^{\prime})^{\prime}}=K_{n_{k}}\left(\Xi_{kk}^{-1}\otimes I_{n_{k}}\right)\left[I_{n_{k}^{2}}-\left(I_{n_{k}}\otimes\tilde{\Xi}_{kk}+\Xi_{kk}\otimes P_{kk}^{\prime}\right)^{-1}\left(I_{n_{k}^{2}}+K_{n_{k}}\right)\left(I_{n_{k}}\otimes\tilde{\Xi}_{kk}\right)\right].
\]
\end{lem}
Lemma \ref{lem:JacobianMatrix} is important for deriving the score
and hessian matrix with respect to the identified $\Xi$. More specific,
let $M_{\Xi}^{+}$ be the Moore-Penrose inverse of $M_{\tilde{\Xi}}$
matrix when evaluated at $\tilde{\Xi}=\Xi$, then we have
\begin{equation}
\nabla_{\Xi}^{\prime}=\nabla_{\tilde{\Xi}}^{\prime}M_{\Xi}^{+},\quad\nabla_{\Xi\Xi^{\prime}}=\ensuremath{M_{\Xi}^{+\prime}\nabla_{\tilde{\Xi}\tilde{\Xi}^{\prime}}M_{\Xi}^{+}}+\left(\nabla_{\tilde{\Xi}}^{\prime}\otimes I_{n}\right)\frac{\partial{\rm vec}(M_{\Xi}^{+\prime})}{\partial{\rm vec}(\Xi)^{\prime}},\label{eq:HessXi}
\end{equation}
and $\nabla_{\Xi\mu^{\prime}}=M_{\Xi}^{+\prime}\nabla_{\tilde{\Xi}\nu\mu^{\prime}}$,
$\nabla_{\Xi\nu^{\prime}}=M_{\Xi}^{+\prime}\nabla_{\tilde{\Xi}\nu^{\prime}}$.
This also facilitate the computation of the two information matrices,
given by
\[
\mathcal{I}_{\Xi}=M_{\Xi}^{+\prime}\mathcal{I}_{\tilde{\Xi}}M_{\Xi}^{+},\quad\mathcal{J}_{\Xi}=M_{\Xi}^{+\prime}\mathcal{J}_{\tilde{\Xi}}M_{\Xi}^{+}.
\]
Note that the expression for $\mathcal{J}_{\Xi}$ does not requires
$\frac{\partial{\rm vec}(M_{\Xi}^{+\prime})}{\partial{\rm vec}(\Xi)^{\prime}}$
to be evaluated because $\mathbb{E}[\nabla_{\tilde{\Xi}}]=0$. Also
note that $\mathcal{I}_{\tilde{\Xi}}$ has reduced rank whereas $\mathcal{I}_{\Xi}$
is nonsingular, because $\Xi$ is identified. 

\subsection{Asymptotic Properties of Maximum Likelihood Estimator}

Given a random sample of convolution-$t$ distributed random variables,
$Y_{1},\ldots,Y_{T}$, the asymptotic properties of maximum likelihood
estimators of $\theta=\{\mu,\Xi,\bm{\nu}\}$ are derived below. 
\begin{assumption}
\label{assu:Compact}$\Theta$ is compact and for all $\theta\in\Theta$
it holds that $\min_{k}\nu_{k}>0$ and $\max_{k}\nu_{k}<\infty$.
\end{assumption}
\begin{thm}
\label{thm:MLE-consistent-asN}Let $Y_{1},\ldots,Y_{T}\sim iid\ CT_{\boldsymbol{\nu}_{0},\boldsymbol{n}_{0}}(\mu,\Xi_{0})$.
Suppose that Assumptions \ref{assu:XiInvertible}- \ref{assu:Compact}
hold and that $\Xi$ is identified by Theorem \ref{Thm:XiExists}.
$(i)$The maximum likelihood estimator, $\hat{\theta}_{T}=\max_{\theta\in\Theta}\ell_{T}(\theta)$,
where $\ell_{T}$ is the log-likelihood function, is consistent. $(ii)$
Suppose further that $\min_{k}\nu_{k}>2$ and that $\theta=\{\mu_{0},\bm{\nu}_{0},\Xi_{0}\}$
is interior to $\Theta$, then 
\[
\sqrt{T}(\hat{\theta}_{T}-\theta)\overset{d}{\rightarrow}N(0,\mathcal{I}_{\theta}^{-1}).
\]
\end{thm}
When $K=1$, $\mathrm{CT}_{\boldsymbol{\nu},\boldsymbol{n}}(\mu,\Xi)$
simplifies to the multivariate $t$-distribution, for which many results
exist and \citet{ZhuGalbraith:2010} established results for a univariate
generalized $t$-distribution. Simulation-based evidence in the Supplemental
Material suggests that the convergence rates for $\mu$, $\boldsymbol{\nu}$,
and $\Xi_{kk}$ continues to be $\sqrt{T}$ in non-standard situations
where $\min_{k}\nu_{k}\leq2$, whereas $\Xi$ coefficient in the off-diagonal
blocks, $\Xi_{kl}$, $k\neq l$, have faster rates of convergence.
The asymptotic normality of $\hat{\boldsymbol{\nu}}$ appears to be
quite robust and also hold for small values of $\nu_{k}$, whereas
the other parameters have more spiky limit distributions when $\min_{k}\nu_{k}\leq2$,
that resemble Laplace distributions. This is particularly pronounced
for the coefficients in the off-diagonal blocks of $\Xi$. 

Note that the analytical expression of the Hessian matrix in Theorem
\ref{thm:ScoreHess} and (\ref{eq:HessXi}) facilitate the computation
of the sandwich-form of the asymptotic covariance matrix, $\frac{1}{T}\mathcal{\hat{J}}_{T}^{-1}\mathcal{\hat{I}}_{T}\mathcal{\hat{J}}_{T}^{-1}$,
where $\mathcal{\hat{J}}_{T}=\frac{1}{T}\sum_{t=1}^{T}\nabla_{\hat{\theta}\hat{\theta}^{\prime},t}$
and $\mathcal{\hat{I}}_{T}=\frac{1}{T}\sum_{t=1}^{T}\left(\nabla_{\hat{\theta},t}\nabla_{\hat{\theta},t}^{\prime}\right)$.
Note that $\hat{\theta}_{T}=\max_{\theta\in\Theta}\ell_{T}(\theta)$
implies $\frac{1}{T}\sum_{t=1}^{T}\nabla_{\tilde{\Xi},t}=0$ under
our assumptions, such that the term $\frac{\partial{\rm vec}(M_{\Xi}^{+\prime})}{\partial{\rm vec}(\Xi)^{\prime}}$
is not needed in the computation of $\mathcal{\hat{J}}_{T}$.

\subsection{Simulation Study\label{sec:Simulation-Study}}

We conduct a simulation study to verify the accuracy the asymptotic
results for the maximum likelihood estimators. We consider the following
trivariate system, given by
\[
\mu=\left[\begin{array}{c}
\cellcolor{black!10}0.1\\
\cellcolor{black!5}0.2\\
\cellcolor{black!5}0.3
\end{array}\right],\quad\Xi=\left[\begin{array}{ccc}
\cellcolor{black!10}0.6 & \cellcolor{black!5}0.3 & \cellcolor{black!5}0.1\\
\cellcolor{black!10}0.5 & \cellcolor{black!5}0.7 & \cellcolor{black!5}0.2\\
\cellcolor{black!10}0.4 & \cellcolor{black!5}0.2 & \cellcolor{black!5}0.8
\end{array}\right]\quad\nu=\left[\begin{array}{c}
\cellcolor{black!10}4\\
\cellcolor{black!5}8
\end{array}\right]
\]
where the first two elements belong to one group, i.e. $n_{1}=1$
and $n_{2}=2$. And the structure of $\Xi_{11}$ and $\Xi_{22}$ are
all symmetric and positive definite. Additional, because $n_{1}\neq n_{2}$,
then according to Theorem \ref{thm:Identify}, the parameters $\theta=\left(\mu^{\prime},\nu^{\prime},{\rm vec}\left(\Xi\right)^{\prime}\right)^{\prime}$
are identified. We conduct a simulation study with $T=500,1000,2000,4000$,
based on 50,000 Monte-Carlo simulations. We report the mean, standard
deviation of the estimated parameters, $\alpha_{L}^{0.025}$ and $\alpha_{R}^{0.025}$
(defined below). The asymptotic standard deviations are also included
for comparison, which are taken from the square root of the diagonal
elements of the Cramèr-Rao bound, i.e., $\mathcal{I}_{\theta}^{-1}/T$.
Note that we report the inverse of degree of freedoms, i.e. $1/\nu_{k}$,
because its empirical distribution is more close to normal especially
in small samples. 

To examine the normality of asymptotic distribution, we reports the
empirical left and right quantiles, defined as
\[
\alpha_{L}^{0.025}=\frac{1}{T}\sum_{t=1}^{T}\left[\frac{\hat{\theta}_{t}-\theta}{{\rm aStd}\left(\theta\right)}<-1.96\right],\quad\alpha_{R}^{0.025}=\frac{1}{T}\sum_{t=1}^{T}\left[\frac{\hat{\theta}_{t}-\theta}{{\rm aStd}\left(\theta\right)}>1.96\right]
\]
and we will have $\alpha_{L}^{0.025}\rightarrow0.025$ and $\alpha_{R}^{0.025}\rightarrow0.025$
if $\hat{\theta}_{t}$ is approximately normally distributed. Table
\ref{tab:MLESimuLarge} reports the simulation results in large and
small samples, respectively. We can find that all maximum likelihood
estimators $\hat{\theta}$ are consistent, and the empirical standard
deviations match the asymptotic standard deviations very well. The
empirical distributions of $\hat{\theta}$ converge to normal distribution
when the sample size increases. In the Table \ref{tab:MLESimuSmall}
in Appendix \ref{subsec:Simulations_SmallSample}, we also reports
the estimation results in small samples.
\begin{table}
\caption{Simulations with Large Samples}

\begin{centering}
\vspace{0.2cm}
\begin{footnotesize}
\begin{tabularx}{\textwidth}{p{1.0cm}Yp{-0.2cm}YYYYYp{-0.2cm}YYYYYY}
\toprule
\midrule
          & {True} &       & Mean   & Std   & aStd  &$\alpha_L^{0.025}$ &$\alpha_R^{0.025}$ &       & Mean   & Std   & aStd &$\alpha_L^{0.025}$ &$\alpha_R^{0.025}$\\
\midrule
\\[-0.3cm]
          &   &  &  \multicolumn{5}{c}{$T=500$}     &       & \multicolumn{5}{c}{$T=1,000$} \\
\\[-0.3cm]           
    $\mu_1$   & 0.1000   &       & 0.0998 & 0.0355 & 0.0353 & 0.0260 & 0.0245 &       & 0.1000 & 0.0250 & 0.0250 & 0.0264 & 0.0250 \\
\\[-0.3cm]
    $\mu_2$   & 0.2000   &       & 0.1998 & 0.0446 & 0.0444 & 0.0257 & 0.0252 &       & 0.2000 & 0.0314 & 0.0314 & 0.0256 & 0.0247 \\
\\[-0.3cm]
    $\mu_3$   & 0.3000   &       & 0.2997 & 0.0457 & 0.0456 & 0.0255 & 0.0245 &       & 0.3001 & 0.0325 & 0.0322 & 0.0258 & 0.0263 \\
\\[-0.3cm]
    \multirow{9}[0]{*}{$\rm{vec}(\Xi)$} & 0.6000   &       & 0.5950 & 0.0392 & 0.0381 & 0.0325 & 0.0256 &       & 0.5979 & 0.0273 & 0.0269 & 0.0275 & 0.0253 \\
\\[-0.3cm]
          & 0.5000   &       & 0.4958 & 0.0543 & 0.0500 & 0.0367 & 0.0335 &       & 0.4983 & 0.0369 & 0.0354 & 0.0292 & 0.0316 \\
\\[-0.3cm]
          & 0.4000   &       & 0.3963 & 0.0578 & 0.0520 & 0.0387 & 0.0368 &       & 0.3988 & 0.0389 & 0.0368 & 0.0315 & 0.0331 \\
\\[-0.3cm]
          & 0.3000   &       & 0.3003 & 0.0492 & 0.0451 & 0.0351 & 0.0367 &       & 0.3001 & 0.0335 & 0.0319 & 0.0307 & 0.0323 \\
\\[-0.3cm]
          & 0.7000   &       & 0.6985 & 0.0478 & 0.0452 & 0.0347 & 0.0294 &       & 0.6991 & 0.0330 & 0.0320 & 0.0309 & 0.0273 \\
\\[-0.3cm]
          & 0.2000   &       & 0.2005 & 0.0335 & 0.0314 & 0.0342 & 0.0315 &       & 0.2001 & 0.0229 & 0.0222 & 0.0295 & 0.0274 \\
\\[-0.3cm]
          & 0.1000   &       & 0.1011 & 0.0433 & 0.0391 & 0.0357 & 0.0383 &       & 0.1002 & 0.0290 & 0.0277 & 0.0309 & 0.0307 \\
\\[-0.3cm]
          & 0.2000   &       & 0.2005 & 0.0335 & 0.0314 & 0.0342 & 0.0315 &       & 0.2001 & 0.0229 & 0.0222 & 0.0295 & 0.0274 \\
\\[-0.3cm]
          & 0.8000   &       & 0.7979 & 0.0444 & 0.0425 & 0.0344 & 0.0261 &       & 0.7988 & 0.0307 & 0.0301 & 0.0289 & 0.0253 \\
\\[-0.3cm]
    $1/\nu_1$  & 0.2500 &       & 0.2497 & 0.0472 & 0.0463 & 0.0293 & 0.0253 &       & 0.2495 & 0.0330 & 0.0327 & 0.0274 & 0.0240 \\
\\[-0.3cm]
    $1/\nu_2$  & 0.1250 &       & 0.1233 & 0.0287 & 0.0280 & 0.0333 & 0.0238 &       & 0.1243 & 0.0201 & 0.0198 & 0.0298 & 0.0250 \\
\\[-0.3cm]
          &       &       &       &       &       &       &       &       &  \\
          &  &    & \multicolumn{5}{c}{$T=2,000$}     &       & \multicolumn{5}{c}{$T=4,000$} \\
\\[-0.2cm]
    $\mu_1$   & 0.1000   &       & 0.1000 & 0.0176 & 0.0177 & 0.0248 & 0.0253 &       & 0.1001 & 0.0125 & 0.0125 & 0.0253 & 0.0258 \\
\\[-0.3cm]          
    $\mu_2$   & 0.2000   &       & 0.2000 & 0.0222 & 0.0222 & 0.0244 & 0.0263 &       & 0.2000 & 0.0158 & 0.0157 & 0.0250 & 0.0260 \\
\\[-0.3cm]
    $\mu_3$   & 0.3000   &       & 0.3000 & 0.0228 & 0.0228 & 0.0252 & 0.0255 &       & 0.3000 & 0.0162 & 0.0161 & 0.0266 & 0.0250 \\
\\[-0.3cm]
    \multirow{9}[0]{*}{$\rm{vec}(\Xi)$} & 0.6000 &       & 0.5989 & 0.0191 & 0.0190 & 0.0271 & 0.0247 &       & 0.5994 & 0.0135 & 0.0135 & 0.0259 & 0.0253 \\
\\[-0.3cm]
          & 0.5000   &       & 0.4990 & 0.0255 & 0.0250 & 0.0278 & 0.0278 &       & 0.4995 & 0.0179 & 0.0177 & 0.0256 & 0.0263 \\
\\[-0.3cm]
          & 0.4000   &       & 0.3992 & 0.0266 & 0.0260 & 0.0279 & 0.0288 &       & 0.3994 & 0.0186 & 0.0184 & 0.0266 & 0.0262 \\
\\[-0.3cm]
          & 0.3000   &       & 0.3002 & 0.0231 & 0.0226 & 0.0278 & 0.0287 &       & 0.3002 & 0.0161 & 0.0159 & 0.0267 & 0.0267 \\
\\[-0.3cm]
          & 0.7000   &       & 0.6998 & 0.0231 & 0.0226 & 0.0279 & 0.0270 &       & 0.7000 & 0.0161 & 0.0160 & 0.0263 & 0.0258 \\
\\[-0.3cm]
          & 0.2000   &       & 0.2001 & 0.0159 & 0.0157 & 0.0273 & 0.0267 &       & 0.2001 & 0.0112 & 0.0111 & 0.0252 & 0.0254 \\
\\[-0.3cm]
          & 0.1000   &       & 0.1002 & 0.0200 & 0.0196 & 0.0275 & 0.0294 &       & 0.1002 & 0.0139 & 0.0138 & 0.0254 & 0.0264 \\
\\[-0.3cm]
          & 0.2000   &       & 0.2001 & 0.0159 & 0.0157 & 0.0273 & 0.0267 &       & 0.2001 & 0.0112 & 0.0111 & 0.0252 & 0.0254 \\
\\[-0.3cm]
          & 0.8000   &       & 0.7994 & 0.0215 & 0.0213 & 0.0270 & 0.0249 &       & 0.7998 & 0.0151 & 0.0150 & 0.0271 & 0.0246 \\
\\[-0.3cm]
    $1/\nu_1$  & 0.2500 &       & 0.2500 & 0.0232 & 0.0232 & 0.0259 & 0.0246 &       & 0.2498 & 0.0164 & 0.0164 & 0.0269 & 0.0244 \\
\\[-0.3cm]
    $1/\nu_2$  & 0.1250 &       & 0.1246 & 0.0141 & 0.0140 & 0.0268 & 0.0242 &       & 0.1248 & 0.0099 & 0.0099 & 0.0267 & 0.0237 \\
\\[0.0cm]
\\[-0.5cm]
\midrule
\bottomrule
\end{tabularx}
\end{footnotesize}
\par\end{centering}
{\small{}Note: The mean, standard deviation of the MLE estimated parameters
for convolution-$t$ distribution. The asymptotic standard deviations
are also included for comparison. we conduct a simulation study with
sample sizes $T=500,1000,2000,4000$, based on 50,000 Monte-Carlo
simulations. \label{tab:MLESimuLarge}}{\small\par}
\end{table}

\section{Standardized Distributions}

The conventional multivariate $t$-distribution, $X\sim t_{n,\nu}(\mu,\Sigma)$
with $\nu>2$ has $\mathrm{var}(X)=\tfrac{\nu}{\nu-2}\Sigma$. Some
theoretical results are more elegantly expressed in terms of standardized
variables, $\tilde{X}=\mu+\sqrt{\tfrac{\nu-2}{\nu}}(X-\mu)$, which
has $\mathbb{E}(\tilde{X})=\mu$, $\mathrm{var}(\tilde{X})=\Sigma$,
and density
\begin{equation}
f(\tilde{x})=\tfrac{\Gamma\left(\tfrac{\nu+n}{2}\right)}{\Gamma\left(\tfrac{\nu}{2}\right)}[(\nu-2)\pi]^{-\frac{n}{2}}|\Sigma|^{-\frac{1}{2}}\left[1+\tfrac{1}{\nu-2}\left(\tilde{x}-\mu\right){}^{\prime}\Sigma^{-1}\left(\tilde{x}-\mu\right)\right]^{-\frac{\nu+n}{2}},\qquad\nu>2.\label{eq:Std-t-density}
\end{equation}
We will refer to this distribution as the standardized $t$-distribution
and use the notation, $\tilde{X}\sim t_{n,\nu}^{\mathrm{std}}(\mu,\Sigma)$.
Obviously, if $Y_{1}$ is a linear combination of independent $t$-distributed
vectors, then it is also a linear combination of independent standardized
$t$-distributed vectors if all $\nu_{k}>2$. Note that the standardized
$t$-distribution disentangles the degrees of freedom parameter, $\nu$,
from the variance. This is convenient in dynamic volatility models,
such as the one we estimate in the empirical analysis. In our empirical
analysis we will use the multivariate convolution-$t$ distribution,
given by $Y=\mu+\Xi X$, where $X\in\mathbb{R}^{n}$ is composed of
$K$ independent standardized multivariate $t$-distributions with
$X_{k}\sim t_{n_{k},\nu_{k}}^{\mathrm{std}}(0,I_{k})$ and $n=\sum_{k}n_{k}$.
This redefines $\Xi$ slightly, because elements are scaled by $\sqrt{\tfrac{\nu_{k}}{\nu_{k}-2}}$.
We have the following notation for this standardized multivariate
convolution-$t$ distribution
\[
Y\sim\mathrm{CT}_{\boldsymbol{n},\boldsymbol{\nu}}^{\mathrm{std}}(\mu,\Xi),
\]
where $\boldsymbol{\nu}=(\nu_{1},\ldots,\nu_{K})^{\prime}$ with $\nu_{k}>2$
for all $k=1,\ldots,K$.

\section{Approximating Marginal Distributions of $\mathrm{CT}_{\boldsymbol{n},\boldsymbol{\nu}}^{{\rm std}}(\mu,\Xi)$\label{sec:Approximating}}

The marginal distribution of the univariate, $Y_{1}=\mu+\beta^{\prime}X$,
where $X\sim\mathrm{CT}_{\boldsymbol{n},\boldsymbol{\nu}}^{\mathrm{std}}(0,I_{n}),$
and $\min_{k}\nu_{k}>4$, then convolution of $t$-distributions may
be approximated by other distributions. In this section, we will approximate
the distribution of $Y_{1}$ with a standardized $t$-distribution,
which is either determined either by matching the first four moments
(Method of Moments) or determined by minimizing the Kullback-Leibler
divergence (KL divergence). For convolutions of two $t$-distributions,
the method of moments approach was used in \citet{Patil:1965}. Our
result in Theorem \ref{thm:Convo-t-kurtosis} makes it simple to generalize
the results in \citet{Patil:1965} to convolutions of multiple $t$-distributions.

Let $g(\cdot)$ denote the true density function of $Y_{1}$, as derived
in Theorem \ref{thm:Convo-t}, and let $f_{\mu,\sigma^{2},\nu}(y)$
denote the density of a standardized $t$-distribution, with mean
$\mu$, variance, $\sigma^{2}$, and $\nu$ degrees of freedom. The
first four moments of $Y_{1}$ were derived in Theorem \ref{thm:Convo-t-kurtosis},
and we will determine the $t$-distribution with the same first four
moments. The third centralized moment is zero for both distributions,
and we match the match the first, second, and fourth moments by selecting
$\mu$, $\sigma^{2}$, and $\nu$, and this is achieved with $\mu_{\star}=\mu$,
$\sigma_{\star}^{2}=\sum_{k=1}^{K}\beta_{k}^{\prime}\beta_{k}$, and
$\nu_{\star}=4+\frac{6}{\kappa_{Y_{1}}},$ such that the moment-matching
density is given by,
\begin{equation}
f^{\star}(y)=f_{\mu_{\star},\sigma_{\star}^{2},\nu_{\mathrm{\star}}}(y)=\tfrac{\Gamma\left(\tfrac{\nu_{\star}+1}{2}\right)}{\Gamma\left(\tfrac{\nu_{\star}}{2}\right)}\left[(\nu_{\star}-2)\pi\omega^{\prime}\omega\right]^{-\frac{1}{2}}\left[1+\frac{\left(y-\mu_{\star}\right)^{2}}{\left(\nu_{\star}-2\right)\omega^{\prime}\omega}\right]^{-\frac{\nu_{\star}+1}{2}}.\label{eq:Std-t-density-fit}
\end{equation}

Alternatively we can determine $\mu$, $\sigma^{2}$, and $\nu$ by
minimizing the Kullback-Leibler discrepancy between $f_{\mu,\sigma^{2},\nu}(y)$
and $g$. The Kullback-Leibler discrepancy is given by
\[
\ensuremath{{\rm KLIC}(g,f)=\int_{-\infty}^{+\infty}\log\left(\frac{f(y)}{g(y)}\right)g(y)\mathrm{d}y},
\]
and by solving
\[
(\mu_{\ast},\sigma_{\ast}^{2},\nu_{\ast})=\arg\max_{\mu,\sigma^{2},\nu}{\rm KLIC}(g,f_{\mu,\sigma^{2},\nu}),
\]
we obtain the best approximating density, $f^{\ast}(y)=f_{\mu_{\ast},\sigma_{\ast}^{2},\nu_{\ast}}(y)$,
in terms of the $\mathrm{KLIC}$.

In the left panels of Figure \ref{fig:KLICden}, we consider the case
where the true density is a convolution of two independent $t$-distributions,
with the same degrees of freedom, and in the right panels we consider
the case where the true density is a convolution of a normal distribution
and an independent $t$-distribution. In the upper panels, (a) and
(b), we report how well $f^{\star}$ (line) and $f^{\ast}$ (line)
approximates the true density for $\nu>4$, in units of logarithm
KLIC. In both cases it is evident that the approximating densities
differ from the true density. Thus, a convolution $t$-distribution
is not a $t$-distribution. The upper panels also show that the method
of moments density is different from the KLIC minimizing density,
especially for small values of $\nu$. And their differences decrease
when $\nu$ becomes larger. However, one should note that the difference
are small as the scale is in logarithmic units of KLIC.
\begin{figure}[tbph]
\begin{centering}
\includegraphics[width=1\textwidth]{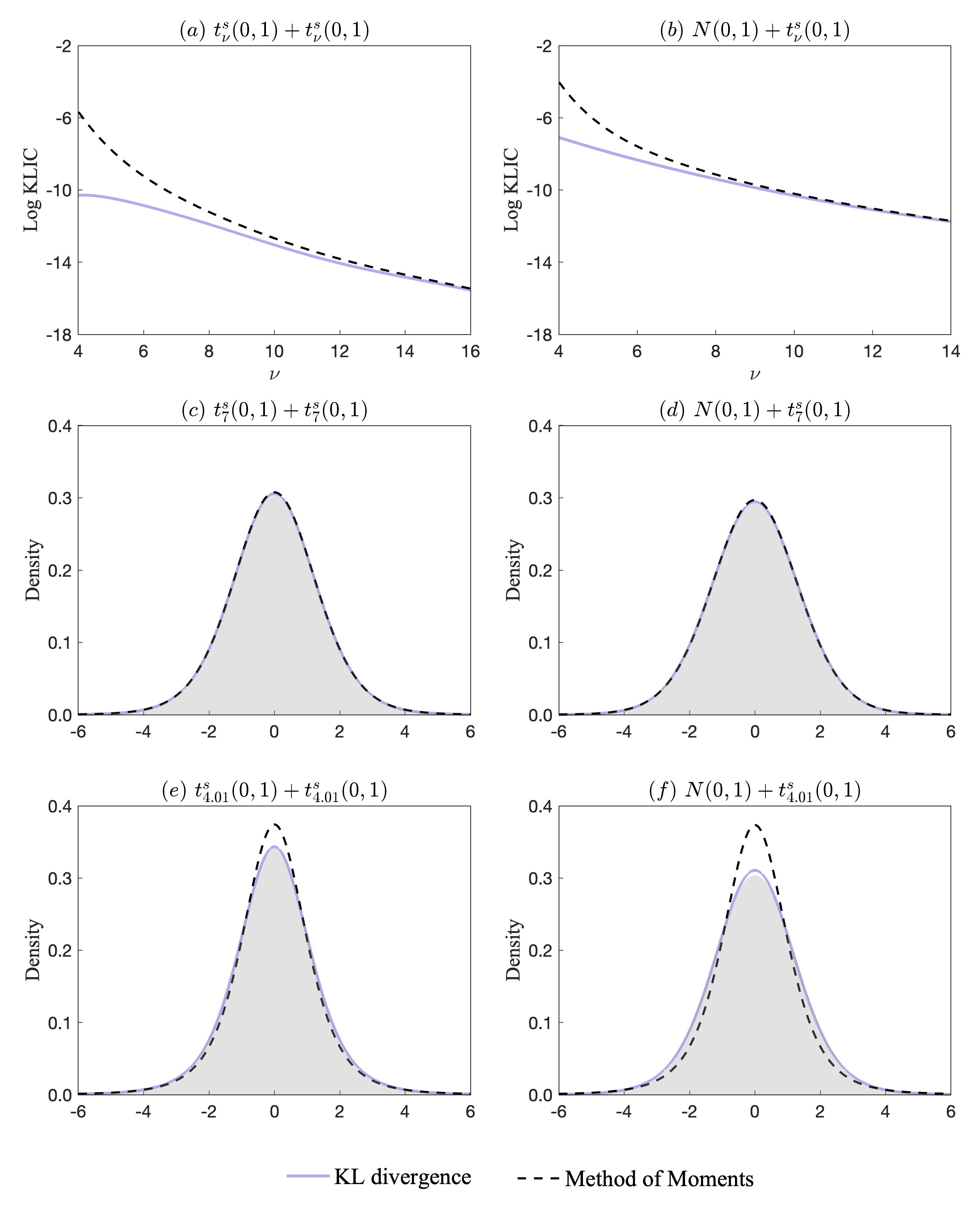}
\par\end{centering}
\caption{{\small{}Panel (a) and (b) measures the differences of $f^{\ast}$
(solid line) and $f^{\star}$ (dash line) relative to the true density
true density of $Y=X_{1}+X_{2}$ in terms of logarithm of the KLIC
for $\nu>4$. Panel (c) and (d) plots the density functions of $f^{\ast}$
and $f^{\star}$, as well as the true density of $Y=X_{1}+X_{2}$
represented with the shaded area.\label{fig:KLICden}}}
\end{figure}

In Panels (c) and (e) we plot the approximating densities (lines)
and the true density (shaded area) of $Y=X_{1}+X_{2}$ with $X_{1},X_{2}\sim t_{\nu}^{{\rm std}}(0,1)$
for $\nu=7$ and $\nu=4.01$. Similarly in Panels (d) and (f) we plot
the approximating densities (lines) and the true density (shaded area)
of $Y=Z+X$ with $Z\sim N(0,1)$ and $X\sim t_{\nu}^{{\rm std}}$
with $\nu=7$ and $\nu=4.01$, respectively. Panels (c) and (d) show
that the approximating densities (lines) match the true density very
well when $\nu\geq7$. However, when $\nu$ become smaller and approaches
$4$, Panel (e) and (f) shows that there exit significant differences
between the approximating density from method of moments and true
density. As for the $f^{\ast}$ (line) from minimizing KLIC, it still
works in the sum of two $t$-distribution with same degree of freedom
$\nu=4.01$, but it performs poor for the sum of a normal and $t$-distribution
with $\nu=4.01$.

We note that the method of moments method is relatively poor, when
the convolutions involve $t$-distributions with small degrees of
freedom.

\section{Empirical Analysis of Realized Volatilities\label{sec:Empirical-Analysis}}

Realized measures of volatility are typically found to be very persistent
and HAR model by \citet{Corsi:2009} (and variations thereof) are
often used to model realized measures, see e.g. \citet{Andersen2007},
\citet{CorsiFusariVecchia2013}, and \citet{Bollerslev2016}.

In this section, we model a vector of daily logarithmically transformed
realized kernel estimators, see \citet[2011]{BNHLS-RK:2008}\nocite{BNHLS-MRK:2011}.
We use a HAR specification to model the conditional mean of the transformed
realized measures and use convolution-$t$ distributions to model
the errors in the model. 

\subsection{Data}

The sample spans the period from January 1, 2001 to December 31, 2020,
which includes $T=4,972$ trading days and we include $n=10$ assets
in stock markets in our analysis. These consist of S\&P 500 index
(SPX) presenting market portfolio, GE, HOM, and MMM, from the Industrials
sector, JNJ, LLY, and MRK, from the Health Care sector, and AAPL,
IBM, and INTC, from the Information Technology sector. We compute
daily realized kernel estimator using intraday transaction data, where
the high-frequency data of SPX were downloaded from TickData.com and
the remaining nine individual stocks were obtained from the TAQ database.
These were cleaned following the methodology in \citet{BNHLS-RKpractice:2009},
and we used the Parzen kernel to compute the realized kernel estimates.

\subsection{HAR Model with Convolution-$t$ Innovations}

Let $Y_{i,t}=\log RK_{i,t}$ where $i=1,\ldots,n$, where $i$ indexes
the stocks and $t=1,\ldots,T$ the day. The righthand side variables
in the HAR model are defined by:
\[
Y_{t}^{(d)}=Y_{t-1},\quad Y_{t}^{(w)}=\frac{1}{4}\sum_{i=2}^{5}Y_{t-i},\quad\text{and}\quad Y_{t}^{(m)}=\frac{1}{17}\sum_{i=6}^{22}Y_{t-i},
\]
which correspond to daily, weekly, and monthly frequencies, respectively.
Thees are defined such that no lagged $Y_{t-i}$ is used twice, which
makes their coefficients easier to interpret. Our HAR model is given
by
\begin{equation}
Y_{t}=\underbrace{\xi+\beta_{d}Y_{t}^{(d)}+\beta_{w}Y_{t}^{(w)}+\beta_{m}Y_{t}^{(m)}}_{\mu_{t}}+V_{t},\qquad\text{with}\quad V_{t}\sim iid\ \mathrm{CT}_{\boldsymbol{n},\boldsymbol{\nu}}^{\mathrm{std}}(0,\Xi)\label{eq:HARmodel}
\end{equation}
where $\xi\in\mathbb{R}^{n}$, $\beta_{d}$, $\beta_{w}$, and $\beta_{m}$
are diagonal matrices. From $V_{t}\sim\mathrm{CT}_{\boldsymbol{n},\boldsymbol{\nu}}^{\mathrm{std}}(0,\Xi)$
it follows that the conditional covariance of $Y_{t}$ is given by
$\Sigma=\Xi\Xi^{\prime}$, and we can express each elements of $Y_{t}$
as an univariate HAR model
\begin{equation}
Y_{t,i}=\xi_{i}+\beta_{d,i}Y_{t,i}^{(d)}+\beta_{w,i}Y_{t,i}^{(w)}+\beta_{m,i}Y_{t,i}^{(m)}+V_{t,i},\quad i=1,2,\ldots,n,\label{eq:IndiHAR}
\end{equation}
where $\beta_{d,i}$ is the $i$-th diagonal element of $\beta_{d}$,
and $\beta_{w,i}$ and $\beta_{m,i}$ defined similarly. 

Let $U_{t}=\Xi^{-1}V_{t}$ be the standardized residuals, and let
$\mathcal{F}_{t}$ be the natural filtration, then the conditional
log-likelihood function $Y_{t}|\mathcal{F}_{t-1}$ is,
\begin{equation}
\ell(Y_{t})=-\frac{1}{2}\log|\Xi\Xi^{\prime}|+\ensuremath{\sum_{k=1}^{K}c_{k}-\frac{\nu_{k}+n_{k}}{2}\log\left(1+\frac{1}{\nu_{k}-2}U_{k,t}^{\prime}U_{k,t}\right).}\label{eq:LogL}
\end{equation}

The expression for the marginal density of convolution $t$-distributions
in Theorem \ref{thm:Convo-t} is particularly useful in our empirical
analysis, because it facilitates a factorization of the joint density
into marginal densities and the copula density by Sklar's theorem.
This leads to the following decomposition of the log-likelihood, 
\[
\ell(Y_{t})=\sum_{i=1}^{n}\ell(Y_{t,i})+\log\left(c\left(Y_{t}\right)\right),
\]
where $c(Y_{t})$ denotes the copula density. This factorization can
be used to compare different model specifications. For instance, if
one specification has a larger value for the log-likelihood, then
the factorization can be used to investigate if the gains are driven
by gains in the marginal distributions, by gains in the copula density,
or a combination of the two. It may reveal that some model features
important for the marginal distribution, whereas other features are
more specific to the copula.

\subsection{Estimation Method and Distributional Specifications}

Motivated by the diagonal-block structure of the information matrix
shown in Theorem \ref{thm:ScoreHess}, we adopt a two-sage estimation
method. In the first stage, we estimate the regression coefficients,
$\beta_{d,i}$, $\beta_{w,i}$, and $\beta_{m,i}$, by least squares
for each of the ten assets separately and stack their residuals into
the vector $\hat{V}_{t}$. In the second stage, we fit parametric
distributions to $\hat{V}_{t}$, by maximizing the corresponding log-likelihood
function, $\ell=\sum_{t=1}^{T}\ell_{t-1}\left(Y_{t}\right)$ in (\ref{eq:LogL}).

The multivariate convolution-$t$ distribution, $\mathrm{CT}_{\boldsymbol{n},\boldsymbol{\nu}}^{\mathrm{std}}(0,\Xi)$,
is characterized by the group assignments, $\bm{n}$, the vector of
degrees of freedoms, $\bm{\nu}$, and a scale-rotation matrix, $\Xi$.
In this paper, we consider three types of group assignments $\bm{n}$.
The first time is the single cluster case, $K=1$ (i.e. with $\bm{n}=n=10$),
which is simply a multivariate $t$-distribution (or Gaussian). The
second case has $K=4$, with $\bm{n}=\left(1,3,3,3\right)$, which
matches the sector classification of the 10 assets, and use ``Cluster-$t$''
to refer this structure. The third case, has $K=10$ independent components
(i.e. no cluster structure) and $\bm{n}=\left(1,1,\ldots,1\right)$.
We use ``Hetero-$t$'' as the label for this structure. 

For the scale-rotation matrix, $\Xi$, we consider both the just-identified
asymmetric structure and the over-identified symmetric structure.
We also combine the symmetric and asymmetric $\Xi$-matrices with
block structure, which has the label ``Block'', whereas $\Xi$-matrices
without an imposed block structure are labelled as ``Just-Identified''.
Imposing block structure on $\Xi$ is motivated by the approximate
block structure seen in the empirical conditional covariance matrix
of $Y_{t}$ (shown below). One should remember that, for the Gaussian
case with $\bm{\nu}\rightarrow\infty$ and the conventional multivariate-$t$
case with $K=1$, only the covariance matrix $\Sigma=\Xi\Xi^{\prime}$
is identified.

\subsection{Empirical Results}

We present the results for the HAR regression coefficients from first-stage
estimation in Table \ref{tab:BetaEstimation}. The HAR coefficients
are reported along with their corresponding standard errors in parentheses.
For the logarithmic realized kernel estimators, $Y_{i,t}$, $i=1,\ldots,10,$
we report the estimated values of the model-implied expected values,
$\mathbb{E}[Y_{t,i}]$, their persistence parameter, $\pi_{i}=\beta_{d,i}+\beta_{w,i}+\beta_{m,i}$,
and the standard deviation of the residuals, $\sigma_{u_{i}}$. There
are some variations in the levels of log-volatility, $\mathbb{E}[Y]$,
across assets, as one would expect. The estimated HAR regression coefficients
and conditional standard deviations, $\hat{\sigma}_{1},\ldots,\hat{\sigma}_{9}$,
are very similar for all stocks, whereas the SPX is somewhat different.
It has a higher weight on the first lag and less weight on distant
lags as defined by $\beta_{m}$. The average log-volatility is also
substantially smaller for SPX, implying that the volatilities for
the individual stocks are substantially larger.\footnote{The log-differences of between 0.66 and 1.32 translate to variances
that are between 92\% and 275\% larger or, alternatively, volatilities
that are 39\%-94\% larger.}

From the first-stage residuals, $\hat{V}_{t}$, $t=1,\ldots,T$, we
compute the empirical covariance matrix, $\hat{\Sigma}=\frac{1}{T}\sum_{t=1}^{T}\hat{V}_{t}\hat{V}_{t}^{\prime}$,
which we present in Table \ref{tab:CovEst}. The sector-based block
structure is illustrated with the shaded regions. Interestingly, the
covariance matrix has an approximate block structure that coincides
with the sector classification of the stocks. The within-sector covariances
are much higher than between-sector covariances, and we will explore
this structure in Section \ref{subsec:Estimations-with-Block}.

\begin{table}
\caption{HAR Coefficients and Key Statistics (OLS Estimation)}

\begin{centering}
\vspace{0.2cm}
\begin{footnotesize}
\begin{tabularx}{\textwidth}{p{0.7cm}Yp{0cm}YYYp{0cm}YYYp{0cm}YYYY}
\toprule
\midrule
          & Market &       & \multicolumn{3}{c}{Industrials} &       & \multicolumn{3}{c}{Health Care} &       & \multicolumn{3}{c}{Information Tech.} \\
\cmidrule{2-2}\cmidrule{4-6}\cmidrule{8-10}\cmidrule{12-14}
          & SPX   &       & GE    & HOM   & MMM   &       & JNJ   & LLY   & MRK   &       & AAPL  & IBM   & INTC \\
\midrule           
          &       &       &       &       &       &       &       &       &       &       &       &       &  \\
    $\xi$ & -0.432 &       & -0.329 & -0.402 & -0.544 &       & -0.666 & -0.623 & -0.671 &       & -0.371 & -0.405 & -0.536 \\
          & \textit{(0.075)} &       & \textit{(0.073)} & \textit{(0.078)} & \textit{(0.094)} &       & \textit{(0.108)} & \textit{(0.105)} & \textit{(0.105)} &       & \textit{(0.075)} & \textit{(0.079)} & \textit{(0.091)} \\
\\[-0.2cm]         
    $\beta_d$ & 0.557 &       & 0.423 & 0.449 & 0.415 &       & 0.405 & 0.389 & 0.433 &       & 0.446 & 0.433 & 0.417 \\
          & \textit{(0.013)} &       & \textit{(0.014)} & \textit{(0.014)} & \textit{(0.014)} &       & \textit{(0.014)} & \textit{(0.014)} & \textit{(0.014)} &       & \textit{(0.014)} & \textit{(0.014)} & \textit{(0.014)} \\
\\[-0.2cm]
   $\beta_w$ & 0.312 &       & 0.342 & 0.338 & 0.353 &       & 0.325 & 0.348 & 0.326 &       & 0.323 & 0.321 & 0.364 \\
          & \textit{(0.017)} &       & \textit{(0.019)} & \textit{(0.019)} & \textit{(0.019)} &       & \textit{(0.019)} & \textit{(0.020)} & \textit{(0.019)} &       & \textit{(0.019)} & \textit{(0.019)} & \textit{(0.019)} \\
\\[-0.2cm]
    $\beta_m$ & 0.088 &       & 0.198 & 0.169 & 0.174 &       & 0.197 & 0.195 & 0.168 &       & 0.189 & 0.199 & 0.159 \\
          & \textit{(0.014)} &       & \textit{(0.018)} & \textit{(0.017)} & \textit{(0.018)} &       & \textit{(0.019)} & \textit{(0.019)} & \textit{(0.018)} &       & \textit{(0.017)} & \textit{(0.018)} & \textit{(0.017)} \\
          &       &       &       &       &       &       &       &       &       &       &       &       &  \\
    $\mathbb{E}(Y)$  & -10.05 &       & -8.938 & -9.053 & -9.396 &       & -9.179 & -9.155 & -9.159 &       & -8.851 & -8.726 & -9.086 \\          
    $\pi$    & 0.957 &       & 0.963 & 0.956 & 0.942 &       & 0.927 & 0.932 & 0.927 &       & 0.958 & 0.954 & 0.941 \\
    $\sigma_v$ & 0.505 &       & 0.557 & 0.538 & 0.527 &       & 0.546 & 0.559 & 0.528 &       & 0.509 & 0.505 & 0.532 \\
\\[0.0cm]
\\[-0.5cm]
\midrule
\bottomrule
\end{tabularx}
\end{footnotesize}
\par\end{centering}
{\small{}Note: HAR regression coefficients for logarithmically transformed
of realized measures of volatility along with the corresponding standard
errors in parentheses. The mean, $\mathbb{E}\left(Y\right)$, persistence,
$\pi=\beta_{d}+\beta_{w}+\beta_{m}$, and the standard deviation of
residuals, $\sigma_{v}$, as implied by the estimated models are listed
at the bottom of the table.\label{tab:BetaEstimation}}{\small\par}
\end{table}

\begin{table}
\caption{Empirical Covariance Matrix}

\begin{centering}
\vspace{0.2cm}
\begin{footnotesize}
\begin{tabularx}{\textwidth}{p{1cm}YYYYYYYYYYYY}
\toprule
\midrule
 & Market       & \multicolumn{3}{c}{Industrials} &   \multicolumn{3}{c}{Health Care}      & \multicolumn{3}{c}{Information Tech.} \\
\cmidrule(l){2-2}    
\cmidrule(l){3-5}
\cmidrule(l){6-8}
\cmidrule(l){9-11}
 & SPX & GE    & HOM   & MMM   & JNJ   & LLY   & MRK  & AAPL  & IBM   & INTC \\
\midrule
\\[-0.25cm]
 SPX & \cellcolor{black!8} 0.255 & 0.135 & 0.146 & 0.143 & \cellcolor{black!8} 0.103 & \cellcolor{black!8} 0.113 & \cellcolor{black!8} 0.110 & 0.137 & 0.122 & 0.146 \\
\\[-0.25cm]
 GE & 0.135 & \cellcolor{black!8}0.310 & \cellcolor{black!8}0.143 & \cellcolor{black!8}0.122 & 0.083 & 0.091 & 0.094 & \cellcolor{black!8}0.102 & \cellcolor{black!8}0.100 & \cellcolor{black!8}0.111 \\
 &  & \cellcolor{black!8} & \cellcolor{black!8} & \cellcolor{black!8} &  &  &  & \cellcolor{black!8} & \cellcolor{black!8} & \cellcolor{black!8}
 \\[-0.25cm]
  HOM & 0.146 & \cellcolor{black!8}0.143 & \cellcolor{black!8}0.290 & \cellcolor{black!8}0.139 & 0.085 & 0.094 & 0.095 & \cellcolor{black!8}0.108 & \cellcolor{black!8}0.102 & \cellcolor{black!8}0.121 \\
 &  & \cellcolor{black!8} & \cellcolor{black!8} & \cellcolor{black!8} &  &  &  & \cellcolor{black!8} & \cellcolor{black!8} & \cellcolor{black!8}  
\\[-0.25cm]
   MMM &  0.143 & \cellcolor{black!8}0.122 & \cellcolor{black!8}0.139 & \cellcolor{black!8}0.278 & 0.098 & 0.100 & 0.096 & \cellcolor{black!8}0.110 & \cellcolor{black!8}0.099 & \cellcolor{black!8}0.111 \\
\\[-0.25cm]
 JNJ & \cellcolor{black!8} 0.103 & 0.083 & 0.085 & 0.098 & \cellcolor{black!8}0.298 & \cellcolor{black!8}0.131 & \cellcolor{black!8}0.118 & 0.085 & 0.072 & 0.095 \\
  & \cellcolor{black!8}  &  &  &  & \cellcolor{black!8} & \cellcolor{black!8} & \cellcolor{black!8} &  &  &  
\\[-0.25cm]
  LLY & \cellcolor{black!8} 0.113 & 0.091 & 0.094 & 0.100 & \cellcolor{black!8}0.131 & \cellcolor{black!8}0.313 &\cellcolor{black!8} 0.140 & 0.088 & 0.076 & 0.092 \\
 & \cellcolor{black!8}  &  &  &  & \cellcolor{black!8} & \cellcolor{black!8} & \cellcolor{black!8} &  &  &   
\\[-0.25cm]
   MRK &  \cellcolor{black!8} 0.110 & 0.094 & 0.095 & 0.096 & \cellcolor{black!8}0.118 & \cellcolor{black!8}0.140 & \cellcolor{black!8}0.279 & 0.087 & 0.077 & 0.086 \\
\\[-0.25cm]
 AAPL & 0.137 & \cellcolor{black!8}0.102 & \cellcolor{black!8}0.108 & \cellcolor{black!8}0.110 & 0.085 & 0.088 & 0.087 & \cellcolor{black!8}0.259 & \cellcolor{black!8}0.118 & \cellcolor{black!8}0.125 \\
 &  & \cellcolor{black!8} & \cellcolor{black!8} & \cellcolor{black!8} &  &  &  & \cellcolor{black!8} & \cellcolor{black!8} & \cellcolor{black!8} 
\\[-0.25cm]
 IBM &   0.122 & \cellcolor{black!8}0.100 & \cellcolor{black!8}0.102 & \cellcolor{black!8}0.099 & 0.072 & 0.076 & 0.077 & \cellcolor{black!8}0.118 & \cellcolor{black!8}0.255 & \cellcolor{black!8}0.118 \\
 &  & \cellcolor{black!8} & \cellcolor{black!8} & \cellcolor{black!8} &  &  &  & \cellcolor{black!8} & \cellcolor{black!8} & \cellcolor{black!8}  
\\[-0.25cm]
 INTC &   0.146 & \cellcolor{black!8}0.111 & \cellcolor{black!8}0.121 & \cellcolor{black!8}0.111 & 0.095 & 0.092 & 0.086 & \cellcolor{black!8}0.125 & \cellcolor{black!8}0.118 & \cellcolor{black!8}0.284  \\
\\[0.0cm]
\\[-0.5cm]
\midrule
\bottomrule
\end{tabularx}
\end{footnotesize}
\par\end{centering}
{\small{}Note: The full-sample empirical covariance matrix $\hat{\Sigma}=\frac{1}{T}\sum_{t=1}^{T}\hat{V}_{t}\hat{V}_{t}^{\prime}$
where the residuals $\hat{V}_{t}$ are from the first-stage estimation.
The sector-based block structure is illustrated with the shaded regions.\label{tab:CovEst}}{\small\par}
\end{table}

\subsubsection{Results for Just-Identified $\Xi$}

Next, we report the second stage estimation results for the just-identified
$\Xi$ matrix, as defined by Theorem \ref{thm:Identify}. We maximize
the log-likelihood for for the vectors of residuals, $\hat{V}_{t}$,
$t=1,\ldots,T$, for each of the six distributional specifications.
Table \ref{tab:DFest-free} reports the estimates of the degrees of
freedom vectors, $\boldsymbol{\nu}=(\nu_{1},\ldots,\nu_{K})^{\prime}$
while Table \ref{tab:Best-Free} presents the estimated $\Xi$ matrix.
In Table \ref{tab:DFest-free}, we also include the number of parameters
$p$, the joint log-likelihood function $\ell$, the marginal log-likelihood
$\ell_{m}$, the logarithmically transformed copula density, $\ell_{c}$,
which is determined residually by $\ell_{c}=\ell-\ell_{m}$. We also
report the Bayesian information criterion (BIC).

Several key observations emerge from Table \ref{tab:DFest-free}.
There is strong evidence against the Gaussian specification. The multivariate-$t$
specification leads to a substantially better log-likelihood, resulting
from large improvement in both the marginal distributions, $\ell_{m}$,
and the log-copula density, $\ell_{c}$. The convolution-$t$ specifications,
both the Cluster-$t$ and Hetero-$t$ specifications, lead to further
improvements over the multivariate $t$-distribution. Once again we
observe large improvements in the log-likelihood, with the largest
gains originating from the log-copula density. There is some heterogenous
values in the estimated degrees of freedom. For instance, the log-volatility
related to the Health Care sector shows the heaviest tails. Both the
Cluster-$t$ and the Hetero-$t$ distributions, suggest that a symmetric
$\Xi$-matrix is not supported by the data. It is too restrictive
to capture the complex nonlinear dependences, which the just-identified
structure can accommodate. This is supported by the large improvements
in the log-copula densities $\ell_{c}$, whereas marginal densities
$\ell_{m}$ are slightly lower with the just-identified $\Xi$. In
terms of the total log-likelihood the Hetero-$t$ specification is
inferior to the Cluster-$t$ structure, despite having more free parameters. 

The Hetero-$t$ specifications have ten degrees of freedom parameters
whereas the Cluster-$t$ specifications have four. While the six additional
parameters result in slightly better marginal densities, $\ell_{m}$,
the copula densities, $\ell_{c}$, are actually worse for the Hetero-$t$
specifications. Recall that the Cluster-$t$ specification is not
nested in the corresponding Hetero-$t$ specification, because the
former imply non-linear dependencies between variables in the same
cluster variable, $X_{k}$, whereas the Hetero-$t$ structure implies
that all elements of $X$ are independent. The empirical results strongly
indicate that the cluster structure captures important nonlinear dependences,
because the log-copula densities are larger for the Cluster-$t$ specifications.
Fifth, the empirical evidence also favors the (just identified) asymmetric
structure for $\Xi$, over the (over identified) symmetric structure.
Despite having 36 and 45 additional parameters, the BIC favors the
asymmetric structure for both types of convolution-$t$ distributions. 

Table \ref{tab:Best-Free} reports the estimated $\Xi$ matrix for
each of the convolution-$t$ specifications in Table \ref{tab:DFest-free}.
The results for just-identified asymmetric $\Xi$-matrices are shown
in Panels (a) and (b) and the results for over-identified symmetric
$\Xi$-matrices are shown in Panels (c) and (d). Each of the four
estimates of $\Xi$ have an approximate block structures, which we
have highlighted with the (lighter) shaded regions. A darker shade
is used to highlight the largest elements of $\hat{\Xi}$. The asymmetric
structures in Panels (a) and (b) are particularly interesting, because
they indicate a factor structure, where all assets load on a common
(market) factor, where as the individual nine stocks also have large
loadings on a single distinct variable. This factor structure is similar
to the traditional CAPM model, but it is the structure implied by
convolution-$t$ distributions, that enables us to identify this factor
structure. Second, we also notice that the loading coefficients are
very similar within each cluster, which motives the block-matrix structure
on $\Xi$, which we explore below. Panel (c) and (c) present the estimates
for over-identified symmetric $\Xi$ matrices. All coefficients are
estimated to positive with the diagonal of $\hat{\Xi}$ being the
dominant elements. For this specification, we can therefore largely
link the degrees of freedom parameter of an element of $X$ to the
corresponding elements of $Y$. The approximated block structure is
also very much evident for the symmetric $\Xi$-matrices.

\begin{table}
\caption{Maximum Likelihood Estimates of $\boldsymbol{\nu}$ with Just-Identified
$\Xi$ Matrix}

\begin{centering}
\vspace{0.2cm}
\begin{footnotesize}
\begin{tabularx}{\textwidth}{p{2.0cm}Yp{0.3cm}Yp{0.3cm}YYp{0.3cm}YY}
    \toprule
          & Gaussian &       & {Multi-$t$} &       & \multicolumn{2}{c}{Cluster-$t$} &       & \multicolumn{2}{c}{Hetro-$t$} \\
\cmidrule{6-7}\cmidrule{9-10}          &       &       &       &       & Sym   & Asym  &       & Sym   & Asym \\    
    \midrule
          &       &       &       &       &       &       &       &       &  \\
    $\nu_{(\text{Joint/Market})}$    & $\infty$      &       & 9.983 &       & 7.742 & 8.615 &       & \cellcolor{black!8}7.806 & \cellcolor{black!8}8.666 \\
          &       &       & \textit{(0.371)} &       & \textit{(0.847)} & \textit{(1.338)} &       & \cellcolor{black!8}\textit{(0.950)} & \cellcolor{black!8}\textit{(1.263)} \\
          &       &       &       &       &       &       &       & 7.393 & 6.584 \\
          &       &       &       &       &       &       &       & \textit{(1.122)} & \textit{(0.805)} \\          
    $\nu_{(\text{Industrials})}$    &       &       &       &       & 6.963 & 6.247 &       & 5.726 & 4.809 \\
          &       &       &       &       & \textit{(0.358)} & \textit{(0.375)} &       & \textit{(0.705)} & \textit{(0.385)} \\
          &       &       &       &       &       &       &       & 5.829 & 5.262 \\
          &       &       &       &       &       &       &       & \textit{(0.683)} & \textit{(0.402)} \\
          &       &       &       &       &       &       &       & \cellcolor{black!8}5.796 & \cellcolor{black!8}5.500 \\
          &       &       &       &       &       &       &       & \cellcolor{black!8}\textit{(0.676)} & \cellcolor{black!8}\textit{(0.523)} \\                    
    $\nu_{(\text{Health Care})}$    &       &       &       &       & 6.156 & 5.943 &       & \cellcolor{black!8}5.035 & \cellcolor{black!8}4.844 \\
          &       &       &       &       & \textit{(0.309)} & \textit{(0.364)} &       & \cellcolor{black!8}\textit{(0.516)} & \cellcolor{black!8}\textit{(0.349)} \\
          &       &       &       &       &       &       &       & \cellcolor{black!8}6.075 & \cellcolor{black!8}5.681 \\
          &       &       &       &       &       &       &       & \cellcolor{black!8}\textit{(0.663)} & \cellcolor{black!8}\textit{(0.497)} \\
          &       &       &       &       &       &       &       & 7.667 & 6.701 \\
          &       &       &       &       &       &       &       & \textit{(1.022)} & \textit{(0.742)} \\          
    $\nu_{(\text{Info Tech.})}$    &       &       &       &       & 8.988 & 7.749 &       & 8.177 & 6.813 \\
          &       &       &       &       & \textit{(0.593)} & \textit{(0.625)} &       & \textit{(1.166)} & \textit{(0.701)} \\
          &       &       &       &       &       &       &       & 8.279 & 7.568 \\
          &       &       &       &       &       &       &       & \textit{(1.390)} & \textit{(1.039)} \\
          &       &       &       &       &       &       &       &       &  \\
    $p$     & 55    &       & 56    &       & 59    & 95    &       & {65} & 110 \\
          &       &       &       &       &       &       &       &       &  \\
    $\ell$   & -30696 &       & -29467 &       & -29282 & \textbf{-29053} &       & -29346 & -29074 \\
    $\ell_m$ & -38837 &       & -38016 &       & -37970 & -37991 &       & \textbf{-37966} & -37973 \\
    $\ell_c$ & 8141  &       & 8549  &       & 8689  & \textbf{8938} &       & 8621  & 8898 \\
          &       &       &       &       &       &       &       &       &  \\
    BIC   & 61860 &       & 59411 &       & 59066 & \textbf{58915} &       & 59245 & 59084 \\
\\[0.0cm]
\\[-0.5cm]
\midrule
\bottomrule
\end{tabularx}
\end{footnotesize}
\par\end{centering}
{\small{}Note: Maximum likelihood estimates of the degrees of freedom
vectors, $\boldsymbol{\nu}=(\nu_{1},\ldots,\nu_{K})^{\prime}$, for
the residuals from the first stage estimation, where $\Xi$ is either
symmetric or asymmetric (just-identified). The joint log-likelihood,
$\ell$, the marginal log-likelihood, $\ell_{m}$, and the residual
copula log-likelihood, $\ell_{c}$, are given along with the Bayesian
Information Criteria, ${\rm BIC}=-2\ell+p\log T$, and the the number
of free parameters, $p$, for each specification.\label{tab:DFest-free}}{\small\par}
\end{table}

\begin{table}
\caption{Estimation of $\Xi$ Matrices with Just-Identified Structure}

\begin{centering}
\vspace{0.2cm}          
\begin{footnotesize}
\begin{tabularx}{\textwidth}{p{1cm}YYYYYYYYYYYY}
\toprule
\midrule
 & Market       & \multicolumn{3}{c}{Industrials} &   \multicolumn{3}{c}{Health Care}      & \multicolumn{3}{c}{Information Tech.} \\
\cmidrule(l){2-2}    
\cmidrule(l){3-5}
\cmidrule(l){6-8}
\cmidrule(l){9-11}
 & SPX & GE    & HOM   & MMM   & JNJ   & LLY   & MRK  & AAPL  & IBM   & INTC \\
\midrule
 \multicolumn{11}{c}{\textit{Panel A: Asymmetric $\Xi$ Matrix from Cluster-$t$}} \\
\\[-0.33cm]
 SPX &\cellcolor{black!22} 0.404 & -0.152 & -0.085 & -0.111 & \cellcolor{black!8}-0.048 & \cellcolor{black!8}-0.047 & \cellcolor{black!8}-0.039 & -0.133 & -0.139 & -0.082 \\
\\[-0.33cm]
 GE & \cellcolor{black!22} 0.369 & \cellcolor{black!22} 0.384 &\cellcolor{black!8} -0.022 &\cellcolor{black!8}  -0.041 & -0.030 & -0.032 & -0.048 & \cellcolor{black!8}-0.094 & \cellcolor{black!8}-0.089 & \cellcolor{black!8}-0.071 \\
 & \cellcolor{black!22} & \cellcolor{black!8} & \cellcolor{black!8} & \cellcolor{black!8} &  &  &  & \cellcolor{black!8} & \cellcolor{black!8} & \cellcolor{black!8}
 \\[-0.33cm]
 HOM & \cellcolor{black!22} 0.350 & \cellcolor{black!8}-0.022 & \cellcolor{black!22} 0.372 & \cellcolor{black!8}-0.005 & -0.057 & -0.038 & -0.024 & \cellcolor{black!8}-0.081 & \cellcolor{black!8}-0.083 & \cellcolor{black!8}-0.050 \\
 & \cellcolor{black!22} & \cellcolor{black!8} & \cellcolor{black!8} & \cellcolor{black!8} &  &  &  & \cellcolor{black!8} & \cellcolor{black!8} & \cellcolor{black!8}
\\[-0.33cm]
 MMM &  \cellcolor{black!22} 0.342 & \cellcolor{black!8}-0.041 & \cellcolor{black!8}-0.005 & \cellcolor{black!22} 0.360 & -0.050 & -0.048 & -0.002 & \cellcolor{black!8}-0.077 & \cellcolor{black!8}-0.092 & \cellcolor{black!8}-0.085 \\
\\[-0.33cm]
 JNJ & \cellcolor{black!22} 0.280 & -0.034 & 0.016 & 0.035 & \cellcolor{black!22} 0.460 & \cellcolor{black!8}0.039 & \cellcolor{black!8}0.037 & -0.026 & -0.042 & -0.026 \\
 & \cellcolor{black!22} &  &  &  & \cellcolor{black!8} & \cellcolor{black!8} & \cellcolor{black!8} &  & & 
\\[-0.33cm]
LLY & \cellcolor{black!22} 0.294 & -0.033 & 0.002 & 0.023 & \cellcolor{black!8}0.039 & \cellcolor{black!22} 0.451 & \cellcolor{black!8} 0.048 & -0.015 & -0.060 & -0.034 \\
 & \cellcolor{black!22} &  &  &  & \cellcolor{black!8} & \cellcolor{black!8} & \cellcolor{black!8} &  & &  
\\[-0.33cm]
 MRK &  \cellcolor{black!22} 0.282 & -0.005 & 0.000 & -0.023 & \cellcolor{black!8} 0.037 & \cellcolor{black!8} 0.048 & \cellcolor{black!22} 0.434 & -0.032 & -0.035 & -0.039 \\
\\[-0.33cm]
 AAPL & \cellcolor{black!22} 0.375 & \cellcolor{black!8}-0.056 & \cellcolor{black!8}-0.030 & \cellcolor{black!8}-0.030 & -0.032 & -0.051 & -0.026 & \cellcolor{black!22}  0.324 & \cellcolor{black!8} -0.036 & \cellcolor{black!8} -0.035 \\
 & \cellcolor{black!22} & \cellcolor{black!8} & \cellcolor{black!8} & \cellcolor{black!8} &  &  &  & \cellcolor{black!8} & \cellcolor{black!8} & \cellcolor{black!8}
\\[-0.33cm]
 IBM &   \cellcolor{black!22} 0.361 & \cellcolor{black!8}-0.042 & \cellcolor{black!8}-0.019 & \cellcolor{black!8}-0.021 & -0.037 & -0.025 & -0.037 & \cellcolor{black!8} -0.036 & \cellcolor{black!22} 0.337 & \cellcolor{black!8} -0.036 \\
 & \cellcolor{black!22} & \cellcolor{black!8} & \cellcolor{black!8} & \cellcolor{black!8} &  &  &  & \cellcolor{black!8} & \cellcolor{black!8} & \cellcolor{black!8}
\\[-0.33cm]
INTC &   \cellcolor{black!22} 0.377 & \cellcolor{black!8}-0.042 & \cellcolor{black!8}-0.019 & \cellcolor{black!8}-0.003 & -0.017 & -0.032 & -0.023 & \cellcolor{black!8} -0.035 &\cellcolor{black!8}  -0.036 & \cellcolor{black!22} 0.365  \\
\\[-0.2cm]
 \multicolumn{11}{c}{\textit{Panel B: Asymmetric $\Xi$ Matrix from Hetero-$t$}} \\
\\[-0.33cm]
 SPX & \cellcolor{black!22} 0.416 & -0.127 & -0.103 & -0.101 & \cellcolor{black!8}-0.050 & \cellcolor{black!8}-0.059 & \cellcolor{black!8}-0.061 & -0.123 & -0.114 & -0.095 \\ 
\\[-0.33cm]
 GE & \cellcolor{black!22} 0.376 & \cellcolor{black!22} 0.380 & \cellcolor{black!8} 0.043 & \cellcolor{black!8} -0.014 & -0.029 & -0.033 & -0.048 & \cellcolor{black!8} -0.072 & \cellcolor{black!8} -0.065 & \cellcolor{black!8} -0.067 \\ 
 & \cellcolor{black!22} & \cellcolor{black!8} & \cellcolor{black!8} & \cellcolor{black!8} &  &  &  & \cellcolor{black!8} & \cellcolor{black!8} & \cellcolor{black!8}
 \\[-0.33cm]
 HOM & \cellcolor{black!22} 0.350 & \cellcolor{black!8} -0.076 & \cellcolor{black!22} 0.377 & \cellcolor{black!8} 0.000 & -0.042 & -0.036 & -0.031 & \cellcolor{black!8} -0.073 & \cellcolor{black!8} -0.054 & \cellcolor{black!8} -0.068  \\
 & \cellcolor{black!22} & \cellcolor{black!8} & \cellcolor{black!8} & \cellcolor{black!8} &  &  &  & \cellcolor{black!8} & \cellcolor{black!8} & \cellcolor{black!8}
\\[-0.33cm]
 MMM &  \cellcolor{black!22} 0.338 & \cellcolor{black!8} -0.053 & \cellcolor{black!8} -0.008 &\cellcolor{black!22} 0.369 & -0.037 & -0.032 & -0.006 & \cellcolor{black!8} -0.068 & \cellcolor{black!8} -0.066 & \cellcolor{black!8} -0.093 \\
\\[-0.33cm]
 JNJ & \cellcolor{black!22} 0.281 & -0.032 & 0.001 & 0.019 & \cellcolor{black!22} 0.447 & \cellcolor{black!8} 0.045 & \cellcolor{black!8} 0.098 & -0.032 & -0.036 & -0.028 \\
 & \cellcolor{black!22} &  &  &  & \cellcolor{black!8} & \cellcolor{black!8} & \cellcolor{black!8} &  & &  
\\[-0.33cm]
LLY & \cellcolor{black!22}   0.297 & -0.035 & -0.006 & 0.003 & \cellcolor{black!8} 0.021 & \cellcolor{black!22} 0.455 & \cellcolor{black!8} 0.067 & -0.017 & -0.045 & -0.046 \\
 & \cellcolor{black!22} &  &  &  & \cellcolor{black!8} & \cellcolor{black!8} & \cellcolor{black!8} &  & & 
\\[-0.33cm]
 MRK &  \cellcolor{black!22} 0.285 & -0.010 & -0.004 & -0.027 & \cellcolor{black!8} -0.027 & \cellcolor{black!8} 0.032 & \cellcolor{black!22} 0.430 & -0.039 & -0.036 & -0.036 \\
\\[-0.33cm]
 AAPL & \cellcolor{black!22} 0.367 & \cellcolor{black!8} -0.055 & \cellcolor{black!8} -0.029 & \cellcolor{black!8} -0.016 & -0.022 & -0.048 & -0.024 & \cellcolor{black!22} 0.333 & \cellcolor{black!8} -0.058 & \cellcolor{black!8} -0.002 \\
 & \cellcolor{black!22} & \cellcolor{black!8} & \cellcolor{black!8} & \cellcolor{black!8} &  &  &  & \cellcolor{black!8} & \cellcolor{black!8} & \cellcolor{black!8}
\\[-0.33cm]
 IBM &   \cellcolor{black!22} 0.358 & \cellcolor{black!8} -0.037 & \cellcolor{black!8} -0.028 & \cellcolor{black!8} -0.011 & -0.031 & -0.031 & -0.034 & \cellcolor{black!8} 0.004 & \cellcolor{black!22} 0.345 & \cellcolor{black!8} 0.017 \\
 & \cellcolor{black!22} & \cellcolor{black!8} & \cellcolor{black!8} & \cellcolor{black!8} &  &  &  & \cellcolor{black!8} & \cellcolor{black!8} & \cellcolor{black!8}
\\[-0.33cm]
INTC &   \cellcolor{black!22} 0.368 & \cellcolor{black!8} -0.047 & \cellcolor{black!8} -0.004 & 0.009 & -0.013 & -0.023 & -0.033 & \cellcolor{black!8} -0.061 & \cellcolor{black!8} -0.075 & \cellcolor{black!22} 0.361  \\
\\[-0.2cm]
 \multicolumn{11}{c}{\textit{Panel C: Symmetric $\Xi$ Matrix from Cluster-$t$}} \\
\\[-0.33cm]
 SPX & \cellcolor{black!22} 0.427 & 0.098 & 0.111 & 0.107 & \cellcolor{black!8} 0.067 & \cellcolor{black!8} 0.076 & \cellcolor{black!8} 0.077 & 0.104 & 0.090 & 0.110 \\
\\[-0.33cm]
GE & 0.098 & \cellcolor{black!22}0.523 & \cellcolor{black!8}0.106 & \cellcolor{black!8}0.089 & 0.051 & 0.055 & 0.060 & \cellcolor{black!8}0.064 & \cellcolor{black!8}0.068 & \cellcolor{black!8}0.072 \\
 & & \cellcolor{black!8} & \cellcolor{black!8} & \cellcolor{black!8} &  &  &  & \cellcolor{black!8} & \cellcolor{black!8} & \cellcolor{black!8}
\\[-0.33cm]
HOM & 0.111 & \cellcolor{black!8}0.106 & \cellcolor{black!22}0.483 & \cellcolor{black!8}0.110 & 0.052 & 0.059 & 0.061 & \cellcolor{black!8}0.070 & \cellcolor{black!8}0.070 & \cellcolor{black!8}0.082 \\
&  & \cellcolor{black!8} & \cellcolor{black!8} & \cellcolor{black!8} &  &  &  & \cellcolor{black!8} & \cellcolor{black!8} & \cellcolor{black!8}  
\\[-0.33cm]
 MMM &  0.107 & \cellcolor{black!8}0.089 & \cellcolor{black!8}0.110 & \cellcolor{black!22}0.477 & 0.063 & 0.064 & 0.061 & \cellcolor{black!8}0.073 & \cellcolor{black!8}0.065 & \cellcolor{black!8}0.071 \\
\\[-0.33cm]
 JNJ  & \cellcolor{black!8} 0.067 & 0.051 & 0.052 & 0.063 & \cellcolor{black!22} 0.517 & \cellcolor{black!8} 0.100 & \cellcolor{black!8} 0.092 & 0.054 & 0.043 & 0.062 \\
 & \cellcolor{black!8} &  &  &  & \cellcolor{black!8} & \cellcolor{black!8} & \cellcolor{black!8} &  & & 
\\[-0.33cm]
 LLY & \cellcolor{black!8} 0.076 & 0.055 & 0.059 & 0.064 & \cellcolor{black!8} 0.100 & \cellcolor{black!8} \cellcolor{black!22} 0.516 & \cellcolor{black!8} 0.107 & 0.055 & 0.046 & 0.054 \\
 & \cellcolor{black!8} &  &  &  & \cellcolor{black!8} & \cellcolor{black!8} & \cellcolor{black!8} &  & &   
\\[-0.33cm]
 MRK &   \cellcolor{black!8} 0.077 & 0.060 & 0.061 & 0.061 & \cellcolor{black!8} 0.092 & \cellcolor{black!8} \cellcolor{black!8} 0.107 & \cellcolor{black!22} 0.492 & 0.055 & 0.048 & 0.051 \\
\\[-0.33cm]
 AAPL & 0.104 & \cellcolor{black!8}0.064 & \cellcolor{black!8}0.070 & \cellcolor{black!8}0.073 & 0.054 & 0.055 & 0.055 & \cellcolor{black!22}0.462 & \cellcolor{black!8}0.096 & \cellcolor{black!8}0.095 \\
 & & \cellcolor{black!8} & \cellcolor{black!8} & \cellcolor{black!8} &  &  &  & \cellcolor{black!8} & \cellcolor{black!8} & \cellcolor{black!8} 
\\[-0.33cm]
 IBM &   0.090 & \cellcolor{black!8}0.068 & \cellcolor{black!8}0.070 & \cellcolor{black!8}0.065 & 0.043 & 0.046 & 0.048 & \cellcolor{black!8}0.096 & \cellcolor{black!22}0.463 & \cellcolor{black!8}0.089 \\
 & & \cellcolor{black!8} & \cellcolor{black!8} & \cellcolor{black!8} &  &  &  & \cellcolor{black!8} & \cellcolor{black!8} & \cellcolor{black!8}  
\\[-0.33cm]
 INTC &   0.110 & \cellcolor{black!8}0.072 & \cellcolor{black!8}0.082 & \cellcolor{black!8}0.071 & 0.062 & 0.054 & 0.051 & \cellcolor{black!8}0.095 & \cellcolor{black!8}0.089 & \cellcolor{black!22}0.485 \\
 \\[-0.2cm]
 \multicolumn{11}{c}{\textit{Panel D: Symmetric $\Xi$ Matrix from Hetero-$t$}} \\
\\[-0.33cm]
 SPX & \cellcolor{black!22} 0.427 & 0.097 & 0.113 & 0.109 & \cellcolor{black!8} 0.068 & \cellcolor{black!8} 0.077 & \cellcolor{black!8} 0.077 & 0.105 & 0.090 & 0.111 \\
\\[-0.33cm]
GE &  0.097 & \cellcolor{black!22}0.515 & \cellcolor{black!8}0.105 & \cellcolor{black!8}0.087 & 0.052 & 0.056 & 0.061 & \cellcolor{black!8}0.065 & \cellcolor{black!8}0.068 & \cellcolor{black!8}0.072 \\
 & & \cellcolor{black!8} & \cellcolor{black!8} & \cellcolor{black!8} &  &  &  & \cellcolor{black!8} & \cellcolor{black!8} & \cellcolor{black!8}
\\[-0.33cm]
HOM &  0.113 & \cellcolor{black!8}0.105 & \cellcolor{black!22}0.489 & \cellcolor{black!8}0.108 & 0.053 & 0.061 & 0.063 & \cellcolor{black!8}0.072 & \cellcolor{black!8}0.071 & \cellcolor{black!8}0.084 \\
&  & \cellcolor{black!8} & \cellcolor{black!8} & \cellcolor{black!8} &  &  &  & \cellcolor{black!8} & \cellcolor{black!8} & \cellcolor{black!8}  
\\[-0.33cm]
 MMM &   0.109 & \cellcolor{black!8}0.087 & \cellcolor{black!8}0.108 & \cellcolor{black!22}0.478 & 0.064 & 0.067 & 0.063 & \cellcolor{black!8}0.074 & \cellcolor{black!8}0.067 & \cellcolor{black!8}0.072 \\
\\[-0.33cm]
 JNJ  & \cellcolor{black!8} 0.068 & 0.052 & 0.053 & 0.064 & \cellcolor{black!22} 0.514 & \cellcolor{black!8} 0.100 & \cellcolor{black!8} 0.090 & 0.055 & 0.045 & 0.063 \\
 & \cellcolor{black!8} &  &  &  & \cellcolor{black!8} & \cellcolor{black!8} & \cellcolor{black!8} &  & & 
\\[-0.33cm]
 LLY &  \cellcolor{black!8} 0.077 & 0.056 & 0.061 & 0.067 & \cellcolor{black!8} 0.100 & \cellcolor{black!22} 0.522 & \cellcolor{black!8} 0.108 & 0.056 & 0.048 & 0.056 \\
 & \cellcolor{black!8} &  &  &  & \cellcolor{black!8} & \cellcolor{black!8} & \cellcolor{black!8} &  & &   
\\[-0.33cm]
 MRK &   \cellcolor{black!8} 0.077 & 0.061 & 0.063 & 0.063 & \cellcolor{black!8} 0.090 & \cellcolor{black!8} 0.108 &  \cellcolor{black!22}0.487 &  0.056 &  0.050 & 0.052 \\
\\[-0.33cm]
 AAPL & 0.105 & \cellcolor{black!8}0.065 & \cellcolor{black!8}0.072 & \cellcolor{black!8}0.074 & 0.055 & 0.056 & 0.056 & \cellcolor{black!22}0.463 & \cellcolor{black!8}0.095 & \cellcolor{black!8}0.095 \\
 & & \cellcolor{black!8} & \cellcolor{black!8} & \cellcolor{black!8} &  &  &  & \cellcolor{black!8} & \cellcolor{black!8} & \cellcolor{black!8} 
\\[-0.33cm]
 IBM &    0.090 & \cellcolor{black!8}0.068 & \cellcolor{black!8}0.071 & \cellcolor{black!8}0.067 & 0.045 & 0.048 & 0.050 & \cellcolor{black!8}0.095 & \cellcolor{black!22}0.463 & \cellcolor{black!8}0.088 \\
 & & \cellcolor{black!8} & \cellcolor{black!8} & \cellcolor{black!8} &  &  &  & \cellcolor{black!8} & \cellcolor{black!8} & \cellcolor{black!8}  
\\[-0.33cm]
 INTC &   0.111 & \cellcolor{black!8}0.072 & \cellcolor{black!8}0.084 & \cellcolor{black!8}0.072 & 0.063 & 0.056 & 0.052 & \cellcolor{black!8}0.095 & \cellcolor{black!8}0.088 & \cellcolor{black!22}0.485 \\
\\[0.0cm]
\\[-0.5cm]
\midrule
\bottomrule
\end{tabularx}
\end{footnotesize}
\par\end{centering}
{\small{}Note: Estimated $\Xi$ matrix with just-identified structure
from Table \ref{tab:DFest-free} .\label{tab:Best-Free}}{\small\par}
\end{table}

\subsubsection{Estimations with Block $\Xi$ matrix\label{subsec:Estimations-with-Block}}

Next, we imposed a block structure on $\Xi$, which reduces the number
of free parameters in $\Xi$. By imposing block structure, the free
parameters in $\Xi$ would greatly decrease which only depends on
the number of $K$, which is at most $K\left(K+1\right)$ if number
of elements in each group is at least two.\footnote{When there are $\tilde{K}\leq K$ groups with only one element, this
number become $K\left(K+1\right)-\tilde{K}$. The reason for the distinction
between these two cases is that an $1\times1$ diagonal block has
only one coefficient.} If additional asymmetry is imposed, then this number becomes $K\left(K+3\right)/2$.

Table \ref{tab:DFest-block} presents the estimation results of degrees
of freedom under block $\Xi$ matrix. When compared with the results
without block structure in Table \ref{tab:DFest-free}, we find that:
first, the estimates of degrees of freedom are very similar to these
in Table \ref{tab:DFest-free}, and the findings about the log-likelihood
values maintain. Second, importantly, the BIC values are much smaller
than these with Just-Identified structure, which means the specification
with block structure is not too restrictive and is preferred according
to BIC. The estimated $\Xi$ matrix is provided in Table \ref{tab:Best-Block},
and we can find that the factor patterns are robust under block structure.

\begin{table}
\caption{Maximum Likelihood Estimates of $\boldsymbol{\nu}$ with Block $\Xi$
Matrix}

\begin{centering}
\vspace{0.2cm}
\begin{footnotesize}
\begin{tabularx}{\textwidth}{p{2.0cm}Yp{0.3cm}Yp{0.3cm}YYp{0.3cm}YY}
    \toprule
          & Gaussian &       & {Multi-$t$} &       & \multicolumn{2}{c}{Cluster-$t$} &       & \multicolumn{2}{c}{Hetro-$t$} \\
\cmidrule{6-7}\cmidrule{9-10}          &       &       &       &       & Sym   & Asym  &       & Sym   & Asym \\    
    \midrule
          &       &       &       &       &       &       &       &       &  \\
    $\nu_{(\text{Joint/Market})}$    & $\infty$      &       & 9.973 &       & 7.621 & 8.629 &       & \cellcolor{black!8}7.701 & \cellcolor{black!8}8.645 \\
          &       &       & \textit{(0.371)} &       & \textit{(0.834)} & \textit{(1.103)} &       & \cellcolor{black!8}\textit{(0.833)} & \cellcolor{black!8}\textit{(1.120)} \\
          &       &       &       &       &       &       &       & 8.134 & 7.272 \\
          &       &       &       &       &       &       &       & \textit{(0.714)} & \textit{(0.556)} \\
    $\nu_{(\text{Industrials})}$    &       &       &       &       & 6.983 & 6.283 &       & 5.335 & 4.390 \\
          &       &       &       &       & \textit{(0.359)} & \textit{(0.290)} &       & \textit{(0.346)} & \textit{(0.223)} \\
          &       &       &       &       &       &       &       & 5.027 & 4.609 \\
          &       &       &       &       &       &       &       & \textit{(0.317)} & \textit{(0.267)} \\
          &       &       &       &       &       &       &       & \cellcolor{black!8}6.012 & \cellcolor{black!8}6.039 \\
          &       &       &       &       &       &       &       & \cellcolor{black!8}\textit{(0.419)} & \cellcolor{black!8}\textit{(0.407)} \\      
    $\nu_{(\text{Health Care})}$    &       &       &       &       & 6.107 & 5.963 &       & \cellcolor{black!8}5.408 & \cellcolor{black!8}5.161 \\
          &       &       &       &       & \textit{(0.301)} & \textit{(0.279)} &       & \cellcolor{black!8}\textit{(0.328)} & \cellcolor{black!8}\textit{(0.294)} \\
          &       &       &       &       &       &       &       & \cellcolor{black!8}5.212 & \cellcolor{black!8}4.871 \\
          &       &       &       &       &       &       &       & \cellcolor{black!8}\textit{(0.367)} & \cellcolor{black!8}\textit{(0.300)} \\
          &       &       &       &       &       &       &       & 7.100 & 6.066 \\
          &       &       &       &       &       &       &       & \textit{(0.686)} & \textit{(0.483)} \\
    $\nu_{(\text{Info Tech.})}$    &       &       &       &       & 8.985 & 7.772 &       & 7.737 & 7.326 \\
          &       &       &       &       & \textit{(0.589)} & \textit{(0.443)} &       & \textit{(0.810)} & \textit{(0.646)} \\
          &       &       &       &       &       &       &       & 9.129 & 7.607 \\
          &       &       &       &       &       &       &       & \textit{(0.927)} & \textit{(0.674)} \\
          &       &       &       &       &       &       &       &       &  \\
    $p$   & 13    &       & 14    &       & 17    & 23    &       & {23} & 29 \\          &       &       &       &       &       &       &       &       &  \\
    $\ell$   & -30825 &       & -29589 &       & -29401 & \textbf{-29191} &       & -29450 & -29200 \\
    $\ell_m$ & -38860 &       & -38034 &       & -37987 & -38009 &       & \textbf{-37983} & -37985 \\
    $\ell_c$ & 8035  &       & 8445  &       & 8587  & \textbf{8818} &       & 8533  & 8784 \\
          &       &       &       &       &       &       &       &       &  \\
    BIC   & 61761 &       & 59297 &       & 58947 & \textbf{58578} &       & 59096 & 58647 \\
\\[0.0cm]
\\[-0.5cm]
\midrule
\bottomrule
\end{tabularx}
\end{footnotesize}
\par\end{centering}
{\small{}Note: The MLE estimates of the degrees of freedom vectors,
$\boldsymbol{\nu}=(\nu_{1},\ldots,\nu_{K})^{\prime}$ based on the
residuals $\hat{V}_{t}$ obtained from the first stage estimation
under the block $\Xi$ matrix. The joint log-likelihood, $\ell$,
the marginal log-likelihood, $\ell_{m}$, and the residual copula
log-likelihood, $\ell_{c}$, and Bayesian information criteria ${\rm BIC}=-2\ell+p\log T$,
are listed in the bottom four rows, where $p$ is the number of free
parameters under constrains. We don't report the estimates of $\Xi$
for the convolution-$t$ to conserve space.\label{tab:DFest-block}}{\small\par}
\end{table}

\begin{table}
\caption{Estimation of $\Xi$ Matrices with Block Structure}

\begin{centering}
\vspace{0.2cm}
\begin{footnotesize}
\begin{tabularx}{\textwidth}{p{1cm}YYYYYYYYYYYY}
\toprule
\midrule
 & Market       & \multicolumn{3}{c}{Industrials} &   \multicolumn{3}{c}{Health Care}      & \multicolumn{3}{c}{Information Tech.} \\
\cmidrule(l){2-2}    
\cmidrule(l){3-5}
\cmidrule(l){6-8}
\cmidrule(l){9-11}
 & SPX & GE    & HOM   & MMM   & JNJ   & LLY   & MRK  & AAPL  & IBM   & INTC \\
\midrule
 \multicolumn{11}{c}{\textit{Panel A: Asymmetric $\Xi$ Matrix from Cluster-$t$}} \\
\\[-0.33cm]
 SPX & \cellcolor{black!22}0.404& -0.111& -0.111& -0.111& \cellcolor{black!8}-0.044& \cellcolor{black!8}-0.044& \cellcolor{black!8}-0.044& -0.128& -0.128& -0.128  \\
\\[-0.33cm]
 GE & \cellcolor{black!22}0.352& \cellcolor{black!22}0.377& \cellcolor{black!8}-0.019& \cellcolor{black!8}-0.019& -0.036& -0.036& -0.036& \cellcolor{black!8}-0.078& \cellcolor{black!8}-0.078& \cellcolor{black!8}-0.078  \\
 & \cellcolor{black!22} & \cellcolor{black!8} & \cellcolor{black!8} & \cellcolor{black!8} &  &  &  & \cellcolor{black!8} & \cellcolor{black!8} & \cellcolor{black!8}
 \\[-0.33cm]
 HOM & \cellcolor{black!22}0.352& \cellcolor{black!8}-0.019& \cellcolor{black!22}0.377& \cellcolor{black!8}\cellcolor{black!8}-0.019& -0.036& -0.036& -0.036& \cellcolor{black!8}-0.078& \cellcolor{black!8}-0.078& \cellcolor{black!8}-0.078  \\
 & \cellcolor{black!22} & \cellcolor{black!8} & \cellcolor{black!8} & \cellcolor{black!8} &  &  &  & \cellcolor{black!8} & \cellcolor{black!8} & \cellcolor{black!8}
\\[-0.33cm]
 MMM &  \cellcolor{black!22}0.352& \cellcolor{black!8}-0.019& \cellcolor{black!8}-0.019& \cellcolor{black!22}0.377& -0.036& -0.036& -0.036& \cellcolor{black!8}-0.078& \cellcolor{black!8}-0.078& \cellcolor{black!8}-0.078  \\
\\[-0.33cm]
 JNJ & \cellcolor{black!22}0.285& 0.000& 0.000& 0.000& \cellcolor{black!22}0.451& \cellcolor{black!8}0.041& \cellcolor{black!8}0.041& -0.036& -0.036& -0.036  \\
 & \cellcolor{black!22} &  &  &  & \cellcolor{black!8} & \cellcolor{black!8} & \cellcolor{black!8} &  & & 
\\[-0.33cm]
LLY & \cellcolor{black!22}0.285& 0.000& 0.000& 0.000& \cellcolor{black!8}0.041& \cellcolor{black!22}0.451& \cellcolor{black!8}0.041& -0.036& -0.036& -0.036   \\
 & \cellcolor{black!22} &  &  &  & \cellcolor{black!8} & \cellcolor{black!8} & \cellcolor{black!8} &  & &  
\\[-0.33cm]
 MRK &   \cellcolor{black!22}0.285& 0.000& 0.000& 0.000& \cellcolor{black!8}0.041& \cellcolor{black!8}0.041& \cellcolor{black!22}0.451& -0.036& -0.036& -0.036  \\
\\[-0.33cm]
 AAPL & \cellcolor{black!22}0.374& \cellcolor{black!8}-0.034& \cellcolor{black!8}-0.034& \cellcolor{black!8}-0.034& -0.032& -0.032& -0.032& \cellcolor{black!22}0.338& \cellcolor{black!8}-0.041& \cellcolor{black!8}-0.041  \\
 & \cellcolor{black!22} & \cellcolor{black!8} & \cellcolor{black!8} & \cellcolor{black!8} &  &  &  & \cellcolor{black!8} & \cellcolor{black!8} & \cellcolor{black!8}
\\[-0.33cm]
 IBM &   \cellcolor{black!22}0.374& \cellcolor{black!8}-0.034& \cellcolor{black!8}-0.034& \cellcolor{black!8}-0.034& -0.032& -0.032& -0.032& \cellcolor{black!8}-0.041& \cellcolor{black!22}0.338& \cellcolor{black!8}-0.041 \\
 & \cellcolor{black!22} & \cellcolor{black!8} & \cellcolor{black!8} & \cellcolor{black!8} &  &  &  & \cellcolor{black!8} & \cellcolor{black!8} & \cellcolor{black!8}
\\[-0.33cm]
INTC &   \cellcolor{black!22}0.374& \cellcolor{black!8}-0.034& \cellcolor{black!8}-0.034& \cellcolor{black!8}-0.034& -0.032& -0.032& -0.032& \cellcolor{black!8}-0.041& \cellcolor{black!8}-0.041& \cellcolor{black!22}0.338  \\
\\[-0.2cm]
 \multicolumn{11}{c}{\textit{Panel B: Asymmetric $\Xi$ Matrix from Hetero-$t$}} \\
\\[-0.33cm]
 SPX & \cellcolor{black!22}0.416 & -0.106 & -0.106 & -0.106 & \cellcolor{black!8}-0.058 & \cellcolor{black!8}-0.058 & \cellcolor{black!8}-0.058 & -0.116 & -0.116 & -0.116 \\ 
\\[-0.33cm]
 GE & \cellcolor{black!22}0.353 & \cellcolor{black!22}0.388 & \cellcolor{black!8}-0.015 & \cellcolor{black!8}-0.015 & -0.032 & -0.032 & -0.032 & \cellcolor{black!8}-0.068 & \cellcolor{black!8}-0.068 & \cellcolor{black!8}-0.068 \\ 
 & \cellcolor{black!22} & \cellcolor{black!8} & \cellcolor{black!8} & \cellcolor{black!8} &  &  &  & \cellcolor{black!8} & \cellcolor{black!8} & \cellcolor{black!8}
 \\[-0.33cm]
 HOM & \cellcolor{black!22}0.353 & \cellcolor{black!8}-0.015 & \cellcolor{black!22}0.388 & \cellcolor{black!8}-0.015 & -0.032 & -0.032 & -0.032 & \cellcolor{black!8}-0.068 & \cellcolor{black!8}-0.068 & \cellcolor{black!8}-0.068  \\
 & \cellcolor{black!22} & \cellcolor{black!8} & \cellcolor{black!8} & \cellcolor{black!8} &  &  &  & \cellcolor{black!8} & \cellcolor{black!8} & \cellcolor{black!8}
\\[-0.33cm]
 MMM &  \cellcolor{black!22}0.353 & \cellcolor{black!8}-0.015 & \cellcolor{black!8}-0.015 & \cellcolor{black!22}0.388 & -0.032 & -0.032 & -0.032 & -0.068 & -0.068 & -0.068 \\
\\[-0.33cm]
 JNJ & \cellcolor{black!22}0.289 & -0.009 & -0.009 & -0.009 & \cellcolor{black!22}0.449 & \cellcolor{black!8}0.039 & \cellcolor{black!8}0.039 & -0.037 & -0.037 & -0.037 \\
 & \cellcolor{black!22} &  &  &  & \cellcolor{black!8} & \cellcolor{black!8} & \cellcolor{black!8} &  & & 
\\[-0.33cm]
LLY & \cellcolor{black!22}0.289 & -0.009 & -0.009 & -0.009 & \cellcolor{black!8}0.039 & \cellcolor{black!22}0.449 & \cellcolor{black!8}0.039 & -0.037 & -0.037 & -0.037 \\
 & \cellcolor{black!22} &  &  &  & \cellcolor{black!8} & \cellcolor{black!8} & \cellcolor{black!8} &  & &  
\\[-0.33cm]
 MRK &  \cellcolor{black!22}0.289 & -0.009 & -0.009 & -0.009 & \cellcolor{black!8}0.039 & \cellcolor{black!8}0.039 & \cellcolor{black!22}0.449 & -0.037 & -0.037 & -0.037 \\
\\[-0.33cm]
 AAPL & \cellcolor{black!22}0.369 & \cellcolor{black!8}-0.027 & \cellcolor{black!8}-0.027 & \cellcolor{black!8}-0.027 & -0.030 & -0.030 & -0.030 & \cellcolor{black!22}0.348 & \cellcolor{black!8}-0.032 & \cellcolor{black!8}-0.032 \\
 & \cellcolor{black!22} & \cellcolor{black!8} & \cellcolor{black!8} & \cellcolor{black!8} &  &  &  & \cellcolor{black!8} & \cellcolor{black!8} & \cellcolor{black!8}
\\[-0.33cm]
 IBM &    \cellcolor{black!22}0.369 & \cellcolor{black!8}-0.027 & \cellcolor{black!8}-0.027 & \cellcolor{black!8}-0.027 & -0.030 & -0.030 & -0.030 & \cellcolor{black!8}-0.032 & \cellcolor{black!22}0.348 & \cellcolor{black!8}-0.032 \\
 & \cellcolor{black!22} & \cellcolor{black!8} & \cellcolor{black!8} & \cellcolor{black!8} &  &  &  & \cellcolor{black!8} & \cellcolor{black!8} & \cellcolor{black!8}
\\[-0.33cm]
INTC &   \cellcolor{black!22}0.369 & -0.027 & -0.027 & -0.027 & -0.030 & -0.030 & -0.030 & \cellcolor{black!8}-0.032 & \cellcolor{black!8}-0.032 & \cellcolor{black!22}0.348 \\
\\[-0.2cm]
 \multicolumn{11}{c}{\textit{Panel C: Symmetric $\Xi$ Matrix from Cluster-$t$}} \\
\\[-0.33cm]
 SPX &  \cellcolor{black!22}0.428 & 0.106 & 0.106 & 0.106 & \cellcolor{black!8}0.073 & \cellcolor{black!8}0.073 & \cellcolor{black!8}0.073 & 0.101 & 0.101 & 0.101 \\
\\[-0.33cm]
GE & 0.106 & \cellcolor{black!22}0.495 & \cellcolor{black!8}0.101 & \cellcolor{black!8}0.101 & 0.059 & 0.059 & 0.059 & \cellcolor{black!8}0.071 & \cellcolor{black!8}0.071 & \cellcolor{black!8}0.071 \\
 & & \cellcolor{black!8} & \cellcolor{black!8} & \cellcolor{black!8} &  &  &  & \cellcolor{black!8} & \cellcolor{black!8} & \cellcolor{black!8}
\\[-0.33cm]
HOM & 0.106 & \cellcolor{black!8}0.101 & \cellcolor{black!22}0.495 & \cellcolor{black!8}0.101 & 0.059 & 0.059 & 0.059 & \cellcolor{black!8}0.071 & \cellcolor{black!8}0.071 & \cellcolor{black!8}0.071 \\
&  & \cellcolor{black!8} & \cellcolor{black!8} & \cellcolor{black!8} &  &  &  & \cellcolor{black!8} & \cellcolor{black!8} & \cellcolor{black!8}  
\\[-0.33cm]
 MMM  & 0.106 & \cellcolor{black!8}0.101 & \cellcolor{black!8}0.101 & \cellcolor{black!22}0.495 & 0.059 & 0.059 & 0.059 & \cellcolor{black!8}0.071 & \cellcolor{black!8}0.071 & \cellcolor{black!8}0.071 \\
\\[-0.33cm]
 JNJ  & \cellcolor{black!8}0.073 & 0.059 & 0.059 & 0.059 & \cellcolor{black!22}0.509 & \cellcolor{black!8}0.100 & \cellcolor{black!8}0.100 & 0.052 & 0.052 & 0.052 \\
 & \cellcolor{black!8} &  &  &  & \cellcolor{black!8} & \cellcolor{black!8} & \cellcolor{black!8} &  & & 
\\[-0.33cm]
 LLY & \cellcolor{black!8}0.073 & 0.059 & 0.059 & 0.059 & \cellcolor{black!8}0.100 & \cellcolor{black!22}0.509 & \cellcolor{black!8}0.100 & 0.052 & 0.052 & 0.052 \\
 & \cellcolor{black!8} &  &  &  & \cellcolor{black!8} & \cellcolor{black!8} & \cellcolor{black!8} &  & &   
\\[-0.33cm]
 MRK & \cellcolor{black!8}0.073 & 0.059 & 0.059 & 0.059 & \cellcolor{black!8}0.100 & \cellcolor{black!8}0.100 & \cellcolor{black!22}0.509 & 0.052 & 0.052 & 0.052 \\
\\[-0.33cm]
 AAPL & 0.101 & \cellcolor{black!8}0.071 & \cellcolor{black!8}0.071 & \cellcolor{black!8}0.071 & 0.052 & 0.052 & 0.052 & \cellcolor{black!22}0.471 & \cellcolor{black!8}0.093 & \cellcolor{black!8}0.093 \\
 & & \cellcolor{black!8} & \cellcolor{black!8} & \cellcolor{black!8} &  &  &  & \cellcolor{black!8} & \cellcolor{black!8} & \cellcolor{black!8} 
\\[-0.33cm]
 IBM & 0.101 & \cellcolor{black!8}0.071 & \cellcolor{black!8}0.071 & \cellcolor{black!8}0.071 & 0.052 & 0.052 & 0.052 & \cellcolor{black!8}0.093 & \cellcolor{black!22}0.471 & \cellcolor{black!8}0.093 \\
 & & \cellcolor{black!8} & \cellcolor{black!8} & \cellcolor{black!8} &  &  &  & \cellcolor{black!8} & \cellcolor{black!8} & \cellcolor{black!8}  
\\[-0.33cm]
 INTC  & 0.101 & \cellcolor{black!8}0.071 & \cellcolor{black!8}0.071 & \cellcolor{black!8}0.071 & 0.052 & 0.052 & 0.052 & \cellcolor{black!8}0.093 & \cellcolor{black!8}0.093 & \cellcolor{black!22}0.471 \\
 \\[-0.2cm]
 \multicolumn{11}{c}{\textit{Panel D: Symmetric $\Xi$ Matrix from Hetero-$t$}} \\
\\[-0.33cm]
 SPX & \cellcolor{black!22}0.427 & 0.107 & 0.107 & 0.107 & \cellcolor{black!8}0.074 & \cellcolor{black!8}0.074 & \cellcolor{black!8}0.074 & 0.102 & 0.102 & 0.102 \\
\\[-0.33cm]
GE &  0.107 & \cellcolor{black!22}0.500 & \cellcolor{black!8}0.101 & \cellcolor{black!8}0.101 & 0.061 & 0.061 & 0.061 & \cellcolor{black!8}0.072 & \cellcolor{black!8}0.072 & \cellcolor{black!8}0.072 \\
 & & \cellcolor{black!8} & \cellcolor{black!8} & \cellcolor{black!8} &  &  &  & \cellcolor{black!8} & \cellcolor{black!8} & \cellcolor{black!8}
\\[-0.33cm]
HOM &  0.107 & \cellcolor{black!8}0.101 & \cellcolor{black!22}0.500 & \cellcolor{black!8}0.101 & 0.061 & 0.061 & 0.061 & \cellcolor{black!8}0.072 & \cellcolor{black!8}0.072 & \cellcolor{black!8}0.072 \\
&  & \cellcolor{black!8} & \cellcolor{black!8} & \cellcolor{black!8} &  &  &  & \cellcolor{black!8} & \cellcolor{black!8} & \cellcolor{black!8}  
\\[-0.33cm]
 MMM &   0.107 & \cellcolor{black!8}0.101 & \cellcolor{black!8}0.101 & \cellcolor{black!22}0.500 & 0.061 & 0.061 & 0.061 & \cellcolor{black!8}0.072 & \cellcolor{black!8}0.072 & \cellcolor{black!8}0.072 \\
\\[-0.33cm]
 JNJ  & \cellcolor{black!8}0.074 & 0.061 & 0.061 & 0.061 & \cellcolor{black!22}0.509 & \cellcolor{black!8}0.100 & \cellcolor{black!8}0.100 & 0.054 & 0.054 & 0.054 \\
 & \cellcolor{black!8} &  &  &  & \cellcolor{black!8} & \cellcolor{black!8} & \cellcolor{black!8} &  & & 
\\[-0.33cm]
 LLY &  \cellcolor{black!8}0.074 & 0.061 & 0.061 & 0.061 & \cellcolor{black!8}0.100 & \cellcolor{black!22}0.509 & \cellcolor{black!8}0.100 & 0.054 & 0.054 & 0.054 \\
 & \cellcolor{black!8} &  &  &  & \cellcolor{black!8} & \cellcolor{black!8} & \cellcolor{black!8} &  & &   
\\[-0.33cm]
 MRK &   \cellcolor{black!8}0.074 & 0.061 & 0.061 & 0.061 & \cellcolor{black!8}0.100 & \cellcolor{black!8}0.100 & \cellcolor{black!22}0.509 & 0.054 & 0.054 & 0.054 \\
\\[-0.33cm]
 AAPL & 0.102 & \cellcolor{black!8}0.072 & \cellcolor{black!8}0.072 & \cellcolor{black!8}0.072 & 0.054 & 0.054 & 0.054 & \cellcolor{black!22}0.471 & \cellcolor{black!8}0.093 & \cellcolor{black!8}0.093 \\
 & & \cellcolor{black!8} & \cellcolor{black!8} & \cellcolor{black!8} &  &  &  & \cellcolor{black!8} & \cellcolor{black!8} & \cellcolor{black!8} 
\\[-0.33cm]
 IBM &    0.102 & \cellcolor{black!8}0.072 & \cellcolor{black!8}0.072 & \cellcolor{black!8}0.072 & 0.054 & 0.054 & 0.054 & \cellcolor{black!8}0.093 & \cellcolor{black!22}0.471 & \cellcolor{black!8}0.093 \\
 & & \cellcolor{black!8} & \cellcolor{black!8} & \cellcolor{black!8} &  &  &  & \cellcolor{black!8} & \cellcolor{black!8} & \cellcolor{black!8}  
\\[-0.33cm]
 INTC &  0.102 & \cellcolor{black!8}0.072 & \cellcolor{black!8}0.072 & \cellcolor{black!8}0.072 & 0.054 & 0.054 & 0.054 & \cellcolor{black!8}0.093 & \cellcolor{black!8}0.093 & \cellcolor{black!22}0.471 \\
\\[0.0cm]
\\[-0.5cm]
\midrule
\bottomrule
\end{tabularx}
\end{footnotesize}
\par\end{centering}
{\small{}Note: Estimated $\Xi$ matrix with block structure from Table
\ref{tab:DFest-block}.\label{tab:Best-Block}}{\small\par}
\end{table}

\section{Summary\label{sec:Summery}}

We have introduced the convolution-$t$ distributions, which is a
versatile class of multivariate distributions for modeling heavy-tailed
data with complex dependencies. We have characterized a several properties
of these distributions and detailed results for the marginal distributions
of convolution-$t$ distributions, such as expressions their densities
and cumulative distribution functions. We obtained simple expressions
for the first four moments and used these to develop an approximation
method for convolution-$t$ distributions. 

A attractive feature of this class of distributions is that their
log-likelihood function have simple expressions. This makes estimation
and inference relatively straight forward. We have analyzed the identification
problem, which motivates a particular parametrization, and we established
consistency and asymptotic normality of the maximum likelihood estimator
in Theorem \ref{thm:MLE-consistent-asN}, 

Our analysis of the dynamic properties of ten realized volatility
measures highlights the empirical relevance of the new class of distributions.
The empirical results obtained with convolution-$t$ distributions
provide a far more nuanced understanding of the nonlinear dependencies
in these time series. The convolution-$t$ distributions improve the
empirical fit, with improvements seen in both marginal distributions
and their interdependencies, where the latter is characterized by
the copula density. The specifications labelled Cluster-$t$ and Hetero-$t$
both offered substantial improvements in the empirical fit.

We identified interesting cluster structures that align with the sector
classifications of stocks. There were identified from non-linear dependencies
that are made possible by the convolution-$t$ distributions. We found
heterogenous degrees of freedom with stocks in the Health Care sector
exhibiting the heaviest tails, $\nu_{k}\approx5.96$, whereas stock
in the Information Technology sector were about about $\nu_{k}\approx7.77$.

Conventional heavy-tailed distributions have also been found to be
useful in multivariate GARCH models, see e.g. \citet{CrealKoopmanLucas:2011},
and the convolution-$t$ distributions might prove to be useful in
this context.

\bibliographystyle{apalike}
\bibliography{prh}
\newpage{}

\appendix
\setcounter{equation}{0}
\global\long\def\theequation{A.\arabic{equation}}%

\setcounter{lem}{0}
\global\long\def\thelem{A.\arabic{lem}}%

\section{Appendix of Proofs\label{sec:Proofs}}

Below we adopt som notations from \citet{CrealKoopmanLucasJBES:2012},
$A_{\otimes}\equiv A\otimes A$ and $A\oplus B=A\otimes B+B\otimes A$,
where $\otimes$ is the Kronecker product. The ${\rm vec}(A)$ operator
stacks the columns of $m\times n$ matrix $A$ consecutively into
the $mn\times1$ column vector.

For a linear combination of a multivariate Student's $t$-distribution
we have the following expression for its characteristic function.
\begin{lem}
\label{lem:MVt2Y1}Suppose that $X\sim t_{n,\nu}(\mu,\Sigma)$ and
$\beta\in\mathbb{R}^{n}$. Then the characteristic function $Y_{1}=\beta^{\prime}X$
is given by
\[
\varphi_{Y_{1}}(s)=e^{is\alpha}\phi_{\nu}(\omega s),
\]
where $\alpha=\beta^{\prime}\mu$ and $\omega=\sqrt{\beta^{\prime}\Sigma\beta}$.
\end{lem}
\noindent \textbf{Proof of Lemma \ref{lem:MVt2Y1}.} A multivariate
$t$-distribution, $X\sim t_{n,\nu}(\mu,\Sigma)$, has the representation,
$X=\mu+\tfrac{1}{\sqrt{W/\nu}}Z$, where $Z\sim N\left(0,\Sigma\right)$
and $W\sim\mathrm{Gamma}(\tfrac{\nu}{2},2)$ are independent. From
the properties of normal distribution, we have $\beta^{\prime}Z\sim N(0,\omega^{2})$
where $\omega^{2}=\beta^{\prime}\Sigma\beta$, such that 
\[
Y_{1}=\alpha+\omega\tfrac{\tilde{Z}}{\sqrt{W/\nu}}=\alpha+\omega U,
\]
where $\alpha=\beta^{\prime}\mu$ and $\tilde{Z}\sim N\left(0,1\right)$
is independent of $W$, such that $U=\tfrac{\tilde{Z}}{\sqrt{W/\nu}}\sim t_{\nu}\left(0,1\right)$.
The result for the characteristic function now follows from
\[
\varphi_{Y_{1}}(s)=\mathbb{E}[\exp(is(\alpha+\omega U))]=\exp(is\alpha)\mathbb{E}[\exp(\omega sU))]=\exp(is\alpha)\phi_{\nu}(\omega s).
\]
This completes the proof.

\hfill{}$\square$

\noindent \textbf{Proof of Theorem \ref{thm:Convo-t}.} By the same
arguments as in the proof of Lemma \ref{lem:MVt2Y1}, we have that
$\beta_{k}^{\prime}X_{k}=\alpha_{k}+\omega_{k}U_{k}$, where $\alpha_{k}=\beta_{k}^{\prime}\mu_{k}$,
$\omega_{k}=\sqrt{\beta_{k}^{\prime}\Sigma_{k}\beta_{k}}$, and $U_{1},\ldots,U_{k}$
are independent with $t_{k}\sim t_{\nu_{k}}(0,1)$. The characteristic
function of $Y_{1}=\sum_{k=1}^{K}\beta_{k}^{\prime}X_{k}$ is therefore
given by
\begin{align*}
\varphi_{Y_{1}}(t) & =e^{it\alpha_{\bullet}}\prod_{k=1}^{K}\phi_{\nu_{k}}(\omega_{k}t),
\end{align*}
where $\alpha_{\bullet}=\sum_{k=1}^{K}\alpha_{k}$. The expression
for the density, $f_{Y_{1}}(y)$, and cdf, $F_{Y_{1}}(y)$, follows
by the Gil-Pelaez inversion theorem.

\hfill{}$\square$

\noindent \textbf{Proof of Theorem \ref{thm:Convo-t-kurtosis}.}
Let $X_{k}$ follows independent scaled multivariate $t$-distributions
with $X_{k}\sim t_{n_{k},\nu_{k}}\left(0,I_{k}\right)$ with $\nu_{k}>4$
for $k=1,2,\ldots,K$, then from Theorem \ref{thm:Convo-t}, the excess
kurtosis of the random variable $Y_{1}=\mu+\sum_{k=1}^{K}\beta_{k}^{\prime}X_{k}$
can be expressed by
\[
Y_{1}=\mu+\omega^{\prime}U,
\]
where $\omega$ is a $K\times1$ vector with $\omega_{k}=\sqrt{\frac{\nu_{k}}{\nu_{k}-2}\beta_{k}^{\prime}\beta_{k}}$.
The $U$ is a $K\times1$ random vector with $U_{k}\sim t_{\nu_{k}}(0,\sqrt{\frac{\nu_{k}-2}{\nu_{k}}})$
having an independent univariate scaled Student's $t$-distribution
with unit variance ${\rm var}\left(U_{k}\right)=1$. So the expectation
of $Y_{1}$ is $\mu$ and variance is $\omega^{\prime}\omega$.

Since $\nu_{k}>4$ for $k=1,\ldots,K$ the excess kurtosis of $Y_{1}$
is given by
\[
\kappa_{Y_{1}}=\mathbb{E}\left[\left(\frac{Y_{1}-\alpha_{\bullet}}{\sqrt{\omega^{\prime}\omega}}\right)^{4}\right]-3=\frac{\mathbb{E}\left[\left(\omega^{\prime}U\right)^{4}\right]}{\left(\omega^{\prime}\omega\right)^{2}}-3=\frac{{\rm var}\left(U^{\prime}AU\right)}{\left(\omega^{\prime}\omega\right)^{2}}-3,
\]
where we defined $A\equiv\omega\omega^{\prime}$ in the last equality. 

The excess kurtosis of $U_{k}$ is $\kappa_{X_{k}}=\frac{6}{\nu_{k}-4}$,
and from \citet[p.11]{Seber2003} we have that the variance of a quadratic
forms of a vector, $U$, with independent elements, is
\begin{align*}
{\rm var}\left(U^{\prime}AU\right) & =\sum_{k=1}^{K}A_{kk}^{2}\mathbb{E}\left(U_{k}^{4}\right)+\left[\left[{\rm tr}(A)\right]^{2}-\sum_{k=1}^{K}3A_{kk}^{2}+2{\rm tr}\left(A^{2}\right)\right]\\
 & =\sum_{k=1}^{K}A_{kk}^{2}\kappa_{X_{k}}+\left[{\rm tr}(A)\right]^{2}+2{\rm tr}\left(A^{2}\right)\\
 & =\sum_{k=1}^{K}\omega_{k}^{4}\kappa_{X_{k}}+3(\omega^{\prime}\omega)^{2},
\end{align*}
where we used that $A_{kk}=\omega_{k}^{2}$, ${\rm tr}(A)=\omega^{\prime}\omega$,
and ${\rm tr}(A^{2})=\left(\omega^{\prime}\omega\right)^{2}$. Hence,
the excess kurtosis of $Y_{1}$ is given by
\[
\kappa_{Y_{1}}=\frac{{\rm var}\left(U^{\prime}AU\right)}{\left(\omega^{\prime}\omega\right)^{2}}-3=\sum_{k=1}^{K}\frac{\omega_{k}^{4}}{\left(\omega^{\prime}\omega\right)^{2}}\kappa_{X_{k}}.
\]
This completes the proof.

\hfill{}$\square$

\noindent \textbf{Proof of Proposition \ref{Prop:Y=00003DZ+X_density}.}
First note that for $s\geq0$ we have $\varphi_{Y}(s)=\varphi_{Z}(s)\varphi_{X}(s)=e^{-s^{2}/2}e^{-s}$,
such that $f_{Y}(y)=\frac{1}{\pi}\int_{0}^{\infty}\cos(sy)e^{-s^{2}/2-s}\mathrm{d}s$
follows by Euler's formula 
\[
\mathrm{Re}[e^{-isy}]=\mathrm{Re}[\cos(-sy)+i\sin(-sy)]=\cos(-sy)=\cos(sy).
\]
Next, we factorize the integral into two terms,
\begin{align*}
\frac{1}{\pi}\int_{0}^{\infty}\cos(sy)e^{-\frac{s^{2}}{2}-s}\mathrm{d}s & =\frac{1}{2\pi}\int_{0}^{\infty}\left(e^{isy}+e^{-isy}\right)e^{-\frac{s^{2}}{2}-s}\mathrm{d}s\\
 & =\underbrace{\frac{1}{2\pi}\int_{0}^{\infty}e^{isy}e^{-\frac{s^{2}}{2}-s}\mathrm{d}s}_{I_{1}}+\underbrace{\frac{1}{2\pi}\int_{0}^{\infty}e^{-isy}e^{-\frac{s^{2}}{2}-s}\mathrm{d}s}_{I_{2}}.
\end{align*}
Using $isy-\frac{s^{2}}{2}-s=\tfrac{1}{2}(1-iy)^{2}-\tfrac{1}{2}(s+1-iy)^{2}$
we rewrite the first term as
\begin{eqnarray*}
I_{1} & = & \frac{e^{\tfrac{1}{2}(1-iy)^{2}}}{2\pi}\int_{0}^{\infty}e^{-\tfrac{1}{2}(s+1-iy)^{2}}\mathrm{d}s=\frac{e^{\tfrac{1}{2}(1-iy)^{2}}}{2\pi}\int_{\frac{1-iy}{\sqrt{2}}}^{\infty}e^{-u^{2}}\sqrt{2}\mathrm{d}u\\
 & = & \frac{e^{\tfrac{1}{2}(1-iy)^{2}}}{2\sqrt{2\pi}}\frac{2}{\sqrt{\pi}}\int_{\frac{1-iy}{\sqrt{2}}}^{\infty}e^{-u^{2}}\mathrm{d}u=\frac{e^{\tfrac{1}{2}(1-iy)^{2}}}{2\sqrt{2\pi}}\left[1-{\rm erf}\left(\tfrac{1-iy}{\sqrt{2}}\right)\right],
\end{eqnarray*}
where we applied the substitution $u=(s+1-iy)/\sqrt{2}$. 

For the second term, we use the related identity, $-isy-\frac{s^{2}}{2}-s=\tfrac{1}{2}(1+iy)^{2}-\tfrac{1}{2}(s+1+iy))^{2}$,
and the substitution, $u=(s+1+iy)/\sqrt{2}$, to rewrite it as,
\[
I_{2}=\frac{e^{\tfrac{1}{2}(1+iy)^{2}}}{2\pi}\int_{0}^{\infty}e^{-\tfrac{1}{2}(s+1+iy)^{2}}\mathrm{d}s=\frac{e^{\tfrac{1}{2}(1+iy)^{2}}}{2\sqrt{2\pi}}\left[1-{\rm erf}\left(\tfrac{1+iy}{\sqrt{2}}\right)\right].
\]
Combining the results we arrive at
\begin{align*}
f_{Y}\left(y\right) & =\frac{1}{2\sqrt{2\pi}}\left[e^{\frac{(1-iy)^{2}}{2}}{\rm erfc}\left(\tfrac{1-iy}{\sqrt{2}}\right)+e^{\frac{(1+iy)^{2}}{2}}{\rm erfc}\left(\tfrac{1+iy}{\sqrt{2}}\right)\right]\\
 & =\frac{1}{\sqrt{2\pi}}{\rm Re}\left[e^{\frac{(1+\mathrm{i}y)^{2}}{2}}{\rm erfc}\left(\tfrac{1+iy}{\sqrt{2}}\right)\right].
\end{align*}
This completes the proof.

\hfill{}$\square$

\noindent \textbf{Proof of Theorem \ref{thm:Identify}.} We focus
on the case where $\nu_{k}<\infty$ for all $k$ and take the group
assignments $n_{1},n_{2},\ldots,n_{K}$ as given. Note that we can
express $Y=\mu+\Xi X$ as
\[
Y=\mu+\sum_{k=1}^{K}\Xi_{k}X_{k},\quad{\rm where}\ \Xi=\left[\Xi_{1},\ldots,\Xi_{K}\right],\ \Xi_{k}\in\mathbb{R}^{n\times n_{k}}.
\]
As $X_{k}\sim t_{n_{k},\nu_{k}}(0,I_{k})$, $k=1,\ldots,K$, are independent,
the characteristic function of $Y$ is
\begin{align*}
\varphi_{Y}(r) & =\exp\left(ir^{\prime}\mu\right)\prod_{k=1}^{K}\frac{K_{\frac{\nu_{k}}{2}}(\sqrt{\nu_{k}r^{\prime}\Omega_{k}r})\left(\nu_{k}r^{\prime}\Omega_{k}r\right)^{\frac{\nu_{k}}{4}}}{\Gamma\left(\frac{\nu_{k}}{2}\right)2^{\frac{\nu_{k}}{2}-1}},\quad{\rm where}\quad r\in\mathbb{R}^{n},
\end{align*}
where $\Omega_{k}=\Xi_{k}\Xi_{k}^{\prime}$ for $k=1,\ldots,K$. It
is evident that only $\mu$ and the pairs, $\{\nu_{k},\Omega_{k}\}$,
$k=1,\ldots,K$, matter for the distribution. Moreover, because the
multivariate-$t$ distribution is not a stable distribution, except
for $\nu=1$ (Cauchy) and $\nu=\infty$ (Gaussian), we can identify
$\{n_{k},\nu_{k},\Omega_{k}\}$, $k=1,\ldots,K$, along with $\mu$.

\hfill{}$\square$
\begin{lem}
\label{lem:B=00003DBP}If $P$ is orthonormal, $P^{\prime}P=I$, and
both $B$ and $BP$ are symmetric and positive definite, then $B=BP$.
\end{lem}
\begin{proof}
By the symmetry we have $P^{\prime}B=BP$ and
\[
B=PBP=P^{\prime}BP^{\prime}.
\]
So it follows that
\[
(B)^{2}=P^{\prime}BP^{\prime}PBP=P^{\prime}BBP=(BP)^{2},
\]
From the eigendecomposition of the two symmetric matrices, $B=Q\Lambda Q^{\prime}$
and $BP=N\varUpsilon N^{\prime}$, we now have 
\[
Q\Lambda Q^{\prime}Q\Lambda Q^{\prime}=Q\Lambda^{2}Q^{\prime}=N\varUpsilon^{2}N^{\prime},
\]
such that$\Lambda^{2}=Q^{\prime}N\varUpsilon^{2}N^{\prime}Q$, from
which we can conclude that $\Lambda$ and $\varUpsilon$ are the same
diagonal matrices, aside from a reordering of the elements. Hence
$Q^{\prime}N$ is a permutation matrix. 
\[
B=Q\Lambda Q^{\prime}=Q[Q^{\prime}N\varUpsilon N^{\prime}Q]Q^{\prime}=Q[Q^{\prime}BPQ]Q^{\prime}=BP.
\]
\[
B=BP\Leftrightarrow B(I-P)=0,
\]
which implies $P=I$ because $B$ is pd.
\end{proof}
\noindent \textbf{Proof of Theorem \ref{Thm:XiExists}.} Since $\Omega_{k}$
is psd with rank $n_{k}$, we can express it as $\Omega_{k}=\tilde{\Xi}_{k}\tilde{\Xi}_{k}^{\prime}$
for some matrix $\tilde{\Xi}_{k}\in\mathbb{R}^{n\times n_{k}}$ with
full column rank, $n_{k}$. Next, let $P_{kk}\in\mathbb{R}^{n_{k}\times n_{k}}$
be orthonormal, $k=1,\ldots,K$, and define
\begin{equation}
\Xi=\tilde{\Xi}P,\qquad\text{where}\qquad P=\left(\begin{array}{cccc}
P_{11} & 0 & \cdots & 0\\
0 & P_{22} & \ddots & \vdots\\
\vdots & \ddots & \ddots & 0\\
0 & \cdots & 0 & P_{KK}
\end{array}\right)\in\mathbb{R}^{n\times n}.\label{eq:BigP}
\end{equation}
It follow that $\Xi_{k}\Xi_{k}^{\prime}=\tilde{\Xi}_{k}P_{kk}P_{kk}^{\prime}\tilde{\Xi}_{k}^{\prime}=\Omega_{k}$.
If $\tilde{\Xi}_{kk}$ is invertible, we can set 
\[
P_{kk}=\tilde{\Xi}_{kk}^{\prime}(\tilde{\Xi}_{kk}\tilde{\Xi}_{kk}^{\prime})^{-\frac{1}{2}},
\]
and it is easy to verify that $P_{kk}^{\prime}P_{kk}=I_{n_{k}}$.
Moreover, it follows that $\Xi_{kk}=\tilde{\Xi}_{kk}P_{kk}=(\tilde{\Xi}_{kk}\tilde{\Xi}_{kk}^{\prime})^{\frac{1}{2}}$
is a symmetric positive definite matrix. More generally, if $\tilde{\Xi}_{kk}$
is not invertible, but has (reduced) rank $r$, then we the singular
value decomposition gives us,
\[
\tilde{\Xi}_{kk}=V\Lambda_{k}U^{\prime},\qquad\text{for some}\quad V,U\in\mathbb{R}^{n_{k}\times r},\quad\text{with}\quad V^{\prime}V=U^{\prime}U=I_{r},
\]
where $\Lambda_{k}$ is an $r\times r$ diagonal matrix with the positive
eigenvalues of $\tilde{\Xi}_{kk}$ along the diagonal. For this case,
we define the auxiliary full-rank matrix,
\[
\check{\Xi}_{kk}=\tilde{\Xi}_{kk}+V_{\bot}U_{\bot}^{\prime},
\]
where $V_{\bot}\in\mathbb{R}^{n_{k}\times(n_{k}-r)}$ is orthogonal
to $V$, $V_{\bot}^{\prime}V=0$ and normalized, $V_{\bot}^{\prime}V_{\bot}=I_{n_{k}-r}$
, and $U_{\bot}$ is defined similarly. We now set 
\begin{eqnarray*}
P_{kk} & = & \check{\Xi}_{kk}^{\prime}(\check{\Xi}_{kk}\check{\Xi}_{kk}^{\prime})^{-\frac{1}{2}},\\
 & = & (V\Lambda_{k}U^{\prime}+V_{\bot}U_{\bot}^{\prime})^{\prime}[(V\Lambda_{k}U^{\prime}+V_{\bot}U_{\bot}^{\prime})(V\Lambda_{k}U^{\prime}+V_{\bot}U_{\bot}^{\prime})^{\prime}]^{-1/2}\\
 & = & (V\Lambda_{k}U^{\prime}+V_{\bot}U_{\bot}^{\prime})^{\prime}[(V\Lambda_{k}U^{\prime}U\Lambda_{k}V^{\prime}+V_{\bot}V_{\bot}^{\prime}]^{-1/2}\\
 & = & (U\Lambda_{k}V^{\prime}+U_{\bot}V_{\bot}^{\prime})(V\Lambda_{k}^{-1}V^{\prime}+V_{\bot}V_{\bot}^{\prime})\\
 & = & UV^{\prime}+U_{\bot}V_{\bot}^{\prime},
\end{eqnarray*}
and find that
\[
\Xi_{kk}=\tilde{\Xi}_{kk}P_{kk}=V\Lambda_{k}U^{\prime}(UV^{\prime}+U_{\bot}V_{\bot}^{\prime})=V\Lambda_{k}V^{\prime}=V\Lambda_{k}^{1/2}U^{\prime}U\Lambda_{k}^{1/2}V^{\prime}=(\tilde{\Xi}_{kk}\tilde{\Xi}_{kk}^{\prime})^{\frac{1}{2}},
\]
is symmetric and positive semidefinite with the same rank as $\tilde{\Xi}_{kk}$.
This proves the existence of a $\Xi$ matrix with symmetric positive
semidefinite diagonal blocks, $\Xi_{11},\ldots,\Xi_{KK}$. If $\Xi_{kk}$
is also invertible (implying it is symmetric and positive definite)
then it follows that $\Xi_{kk}$ is unique by Lemma \ref{lem:B=00003DBP}.
To see this, consider $\Xi P$ where $P$ has the structure in (\ref{eq:BigP}),
its $k$-th diagonal block is given by $\Xi_{kk}P_{kk}$ and for it
to be symmetric and pd, it follows Lemma \ref{lem:B=00003DBP} that
$P_{kk}=I_{n_{k}}$. Thus, if $\Xi_{11},\ldots,\Xi_{KK}$ are all
invertible, then $P$ must be $I_{n}$, which implies that $\Xi$
is unique. 

\hfill{}$\square$

\noindent \textbf{Proof of Theorem \ref{thm:ScoreHess} (Score vectors).}
We have $X_{k}=e_{k}^{\prime}X=e_{k}^{\prime}A\left(Y-\mu\right),$where
$X_{k}\sim t_{\nu_{k}}\left(0,I_{n_{k}}\right)$, $k=1,\ldots,K$
are independent, and the log-likelihood function is therefore given
by
\begin{align*}
\ell\left(Y\right) & =-\tfrac{1}{2}\log|\tilde{\Xi}\tilde{\Xi}^{\prime}|+\sum_{k=1}^{K}c_{k}-\tfrac{\nu_{k}+n_{k}}{2}\log\left(1+\tfrac{1}{\nu_{k}}X_{k}^{\prime}X_{k}\right),
\end{align*}
where $c_{k}=\log\Gamma\left(\frac{\nu_{k}+n_{k}}{2}\right)-\log\Gamma\left(\frac{\nu_{k}}{2}\right)-\frac{n_{k}}{2}\log\left(\nu_{k}\pi\right)$.
We find that
\begin{align*}
\frac{\partial\left(X_{k}^{\prime}X_{k}\right)}{\partial{\rm vec}(\tilde{\Xi})^{\prime}} & =\frac{\partial\left(X_{k}^{\prime}X_{k}\right)}{\partial X_{k}^{\prime}}\frac{\partial{\rm vec}(e_{k}^{\prime}A\left(Y-\mu\right))}{\partial{\rm vec}(A)^{\prime}}\frac{\partial{\rm vec}(A)}{\partial{\rm vec}(\tilde{\Xi})^{\prime}}\\
 & =-2X_{k}^{\prime}\left(\left(Y-\mu\right)^{\prime}\otimes e_{k}^{\prime}\right)\left(A^{\prime}\otimes A\right)\\
 & =-2X_{k}^{\prime}\left(X^{\prime}\otimes e_{k}^{\prime}A\right)\\
 & =-2{\rm vec}\left(A^{\prime}e_{k}X_{k}X^{\prime}\right)^{\prime},
\end{align*}
and using the formulae
\[
\frac{\partial{\rm vec}\left(\tilde{\Xi}\tilde{\Xi}^{\prime}\right)}{\partial{\rm vec}(\tilde{\Xi})^{\prime}}=\left(I_{n^{2}}+K_{n}\right)\left(\Xi^{\prime}\otimes I_{n}\right),\quad\text{and}\quad\frac{\partial\log|M|}{\partial{\rm vec}(M)^{\prime}}={\rm vec}\left(M^{-1}\right)^{\prime},
\]
for any symmetric matrix $M$ with positive determinant, we obtain
\begin{align*}
-\tfrac{1}{2}\frac{\partial\log|\tilde{\Xi}\tilde{\Xi}^{\prime}|}{\partial{\rm vec}(\tilde{\Xi})^{\prime}} & =-\tfrac{1}{2}{\rm vec}\left(AA^{\prime}\right)^{\prime}\left(I_{n^{2}}+K_{n}\right)\left(\tilde{\Xi}^{\prime}\otimes I_{n}\right)=-{\rm vec}\left(A^{\prime}\right)^{\prime}.
\end{align*}
Define $W_{k}=\left(v_{k}+n_{k}\right)/\left(v_{k}+X_{k}^{\prime}X_{k}\right)$,
then we have the form of the score given by
\begin{align*}
\nabla_{\tilde{\Xi}}^{\prime}=\frac{\partial\ell}{\partial{\rm vec}(\tilde{\Xi})^{\prime}} & =\sum_{k=1}^{K}W_{k}{\rm vec}\left(A^{\prime}e_{k}X_{k}X^{\prime}\right)^{\prime}-{\rm vec}\left(A^{\prime}\right)^{\prime}.
\end{align*}

The expressions for $\nabla_{\nu_{k}}=\frac{\partial\ell}{\partial\nu_{k}}$
and $\nabla_{\mu}^{\prime}=\frac{\partial\ell}{\partial\mu^{\prime}}$
follows directly.

\hfill{}$\square$

\noindent \textbf{Proof of Theorem \ref{thm:ScoreHess} (Hessian
matrix).} First we note that 
\[
\frac{\partial X_{k}}{\partial\mu^{\prime}}=-e_{k}^{\prime}A,\quad\frac{\partial X_{k}}{\partial{\rm vec}(\tilde{\Xi})^{\prime}}=-X^{\prime}\otimes e_{k}^{\prime}A,\quad\frac{\partial X_{k}}{\partial\nu_{l}}=0,
\]
and
\[
\frac{\partial W_{k}}{\partial X_{k}^{\prime}}=\frac{\partial W_{k}}{\partial\left(X_{k}^{\prime}X_{k}\right)}\frac{\partial\left(X_{k}^{\prime}X_{k}\right)}{\partial X_{k}^{\prime}}=-\tfrac{\nu_{k}+n_{k}}{\left(\nu_{k}+Z_{k}^{\prime}Z_{k}\right)^{2}}2X_{k}^{\prime}=-\tfrac{2}{\nu_{k}+n_{k}}W_{k}^{2}X_{k}^{\prime},
\]
such that
\begin{align*}
\frac{\partial W_{k}}{\partial\mu^{\prime}} & =\tfrac{2}{\nu_{k}+n_{k}}W_{k}^{2}X_{k}^{\prime}e_{k}^{\prime}A,\\
\frac{\partial W_{k}}{\partial{\rm vec}(\tilde{\Xi})^{\prime}} & =\tfrac{2}{\nu_{k}+n_{k}}W_{k}^{2}{\rm vec}\left(A^{\prime}e_{k}X_{k}X^{\prime}\right)^{\prime},\\
\frac{\partial W_{k}}{\partial\nu_{k}} & =\tfrac{X_{k}^{\prime}X_{k}-n_{k}}{\left(v_{k}+X_{k}^{\prime}X_{k}\right)^{2}}=\tfrac{1}{\nu_{k}+n_{k}}\left(W_{k}-W_{k}^{2}\right).
\end{align*}

We proceed to derive the expressions for the three terms, $\nabla_{\mu\mu^{\prime}}$,
$\nabla_{\mu\Xi^{\prime}}$, and $\nabla_{\mu\nu_{k}}$, derived from
$\nabla_{\mu}=\sum_{k=1}^{K}W_{k}A^{\prime}e_{k}X_{k}$. The first
is given by
\begin{align*}
\nabla_{\mu\mu^{\prime}} & =\sum_{k=1}^{K}A^{\prime}e_{k}X_{k}\frac{\partial W_{k}}{\partial\mu^{\prime}}+W_{k}A^{\prime}e_{k}\frac{\partial Z_{k}}{\partial\mu^{\prime}}\\
 & =\sum_{k=1}^{K}\tfrac{2W_{k}^{2}}{\nu_{k}+n_{k}}A^{\prime}e_{k}X_{k}X_{k}^{\prime}e_{k}^{\prime}A-W_{k}A^{\prime}J_{k}A.
\end{align*}
For the next term we make use of an alternative expression for $\nabla_{\mu}$
given by:
\[
\nabla_{\mu}=\sum_{k=1}^{K}W_{k}A^{\prime}e_{k}X_{k}=\sum_{k=1}^{K}W_{k}\left(X_{k}^{\prime}e_{k}^{\prime}\otimes I_{n}\right){\rm vec}\left(A^{\prime}\right),
\]
and find that 
\begin{align*}
\nabla_{\mu\tilde{\Xi}^{\prime}} & =\sum_{k=1}^{K}W_{k}\left(X_{k}^{\prime}e_{k}^{\prime}\otimes I_{n}\right)\frac{\partial{\rm vec}\left(A^{\prime}\right)}{\partial{\rm vec}(\tilde{\Xi})^{\prime}}+A^{\prime}e_{k}X_{k}\frac{\partial W_{k}}{\partial{\rm vec}(\tilde{\Xi})^{\prime}}+W_{k}A^{\prime}e_{k}\frac{\partial X_{k}}{\partial{\rm vec}(\tilde{\Xi})^{\prime}}\\
 & =\sum_{k=1}^{K}W_{k}\left(X_{k}^{\prime}e_{k}^{\prime}\otimes I_{n}\right)K_{n}\ensuremath{\left(A^{\prime}\otimes A\right)}+\tfrac{2W_{k}^{2}}{\nu_{k}+n_{k}}A^{\prime}e_{k}X_{k}{\rm vec}\left(A^{\prime}e_{k}X_{k}X^{\prime}\right)^{\prime}-W_{k}A^{\prime}e_{k}\left(X^{\prime}\otimes e_{k}^{\prime}A\right)\\
 & =\sum_{k=1}^{K}W_{k}\left(AX_{k}^{\prime}e_{k}^{\prime}\otimes A^{\prime}\right)K_{n}+\tfrac{2W_{k}^{2}}{\nu_{k}+n_{k}}A^{\prime}e_{k}X_{k}{\rm vec}\left(A^{\prime}e_{k}X_{k}X^{\prime}\right)^{\prime}-W_{k}A^{\prime}\left(X^{\prime}\otimes J_{k}A\right).
\end{align*}
For the third term we have
\[
\nabla_{\mu\nu_{k}}=\sum_{k=1}^{K}A^{\prime}e_{k}X_{k}\frac{\partial W_{k}}{\partial\nu_{k}}=\sum_{k=1}^{K}\tfrac{1}{\nu_{k}+n_{k}}A^{\prime}e_{k}X_{k}\left(W_{k}-W_{k}^{2}\right).
\]

Next we obtain the expressions for $\nabla_{\tilde{\Xi}\nu_{k}}$
and $\nabla_{\tilde{\Xi}\tilde{\Xi}^{\prime}}$, starting from $\nabla_{\tilde{\Xi}}=\sum_{k=1}^{K}W_{k}{\rm vec}\left(A^{\prime}e_{k}X_{k}X^{\prime}\right)-{\rm vec}\left(A^{\prime}\right)$.
The first expression is simply
\[
\nabla_{\tilde{\Xi}\nu_{k}}=\tfrac{1}{\nu_{k}+n_{k}}\left(W_{k}-W_{k}^{2}\right){\rm vec}\left(A^{\prime}e_{k}X_{k}X^{\prime}\right),
\]
and the second takes the form,
\[
\nabla_{\tilde{\Xi}\tilde{\Xi}^{\prime}}=\sum_{k=1}^{K}{\rm vec}\left(A^{\prime}e_{k}X_{k}X^{\prime}\right)\frac{\partial W_{k}}{\partial{\rm vec}(\tilde{\Xi})^{\prime}}+W_{k}\frac{\partial{\rm vec}\left(A^{\prime}e_{k}X_{k}X^{\prime}\right)}{\partial{\rm vec}(\tilde{\Xi})^{\prime}}-\frac{\partial{\rm vec}\left(A^{\prime}\right)}{\partial{\rm vec}(\tilde{\Xi})^{\prime}}
\]
where
\[
\frac{\partial{\rm vec}\left(A^{\prime}\right)}{\partial{\rm vec}(\tilde{\Xi})^{\prime}}=K_{n}\ensuremath{\left(A^{\prime}\otimes A\right)}=\ensuremath{\left(A\otimes A^{\prime}\right)}K_{n}=\left(I_{n}\otimes A^{\prime}\right)K_{n}\left(I_{n}\otimes A\right),
\]
\begin{align*}
{\rm vec}\left(A^{\prime}e_{k}X_{k}X^{\prime}\right)\frac{\partial W_{k}}{\partial{\rm vec}(\tilde{\Xi})^{\prime}} & =\tfrac{2W_{k}^{2}}{\nu_{k}+n_{k}}{\rm vec}\left(A^{\prime}e_{k}X_{k}X^{\prime}\right){\rm vec}\left(A^{\prime}e_{k}X_{k}X^{\prime}\right)^{\prime},
\end{align*}
and
\begin{equation}
\frac{\partial{\rm vec}\left(A^{\prime}e_{k}X_{k}X^{\prime}\right)}{\partial{\rm vec}(\tilde{\Xi})^{\prime}}=\ensuremath{\left(XX_{k}^{\prime}\otimes I_{n}\right)\frac{\partial{\rm vec}\left(A^{\prime}e_{k}\right)}{\partial{\rm vec}\left(\Xi\right)^{\prime}}+\left(I_{n}\otimes A^{\prime}e_{k}\right)\frac{\partial{\rm vec}\left(X_{k}X^{\prime}\right)}{\partial{\rm vec}\left(\Xi\right)^{\prime}}}.\label{eq:daex-dXi}
\end{equation}
Next,
\[
\frac{\partial{\rm vec}\left(A^{\prime}e_{k}\right)}{\partial{\rm vec}(\tilde{\Xi})^{\prime}}=\left(e_{k}^{\prime}\otimes I_{n}\right)\frac{\partial{\rm vec}\left(A^{\prime}\right)}{\partial{\rm vec}(\tilde{\Xi})^{\prime}}=\left(e_{k}^{\prime}\otimes I_{n}\right)\ensuremath{\left(A\otimes A^{\prime}\right)}K_{n}=\left(e_{k}^{\prime}A\otimes A^{\prime}\right)K_{n},
\]
shows that the first term in (\ref{eq:daex-dXi}) can be expressed
as
\[
\left(XX_{k}^{\prime}\otimes I_{n}\right)\frac{\partial{\rm vec}\left(A^{\prime}e_{k}\right)}{\partial{\rm vec}(\tilde{\Xi})^{\prime}}=\left(XX_{k}^{\prime}\otimes I_{n}\right)\left(e_{k}^{\prime}A\otimes A^{\prime}\right)K_{n}=\left(XX_{k}^{\prime}e_{k}^{\prime}A\otimes A^{\prime}\right)K_{n},
\]
and for the second term in (\ref{eq:daex-dXi}) we have
\[
\frac{\partial{\rm vec}\left(X_{k}X^{\prime}\right)}{\partial{\rm vec}(\tilde{\Xi})^{\prime}}=\left(X\otimes I_{n_{k}}\right)\frac{\partial{\rm vec}\left(X_{k}\right)}{\partial{\rm vec}(\tilde{\Xi})^{\prime}}+\left(I_{n}\otimes X_{k}\right)\frac{\partial{\rm vec}\left(X^{\prime}\right)}{\partial{\rm vec}(\tilde{\Xi})^{\prime}},
\]
where
\[
\frac{\partial{\rm vec}\left(X_{k}\right)}{\partial{\rm vec}(\tilde{\Xi})^{\prime}}=-X^{\prime}\otimes e_{k}^{\prime}A,
\]
and
\begin{align*}
\frac{\partial{\rm vec}\left(X^{\prime}\right)}{\partial{\rm vec}(\tilde{\Xi})^{\prime}} & =\frac{\partial{\rm vec}\left(\ensuremath{\sum_{k=1}^{K}e_{k}X_{k}}\right)}{\partial{\rm vec}(\tilde{\Xi})^{\prime}}=\sum_{k=1}^{K}\frac{\partial{\rm vec}\left(e_{k}X_{k}\right)}{\partial{\rm vec}(\tilde{\Xi})^{\prime}}\\
 & =\sum_{k=1}^{K}\frac{\partial\left(e_{k}X_{k}\right)}{\partial{\rm vec}(\tilde{\Xi})^{\prime}}=\sum_{k=1}^{K}e_{k}\frac{\partial{\rm vec}\left(\ensuremath{X_{k}}\right)}{\partial{\rm vec}(\tilde{\Xi})^{\prime}}\\
 & =-\sum_{k=1}^{K}e_{k}\left(X^{\prime}\otimes e_{k}^{\prime}A\right)=-\sum_{k=1}^{K}\left(X^{\prime}\otimes J_{k}A\right)\\
 & =-\left(X^{\prime}\otimes A\right).
\end{align*}
So we have
\begin{align*}
\frac{\partial{\rm vec}\left(X_{k}X^{\prime}\right)}{\partial{\rm vec}(\tilde{\Xi})^{\prime}} & =-\left(X\otimes I_{n_{k}}\right)\left(X^{\prime}\otimes e_{k}^{\prime}A\right)-\left(I_{n}\otimes X_{k}\right)\left(X^{\prime}\otimes A\right)\\
 & =-\left(XX^{\prime}\otimes e_{k}^{\prime}A\right)-\left(I_{n}\otimes X_{k}\right)\left(A\otimes X^{\prime}\right)K_{n}\\
 & =-\left(XX^{\prime}\otimes e_{k}^{\prime}A\right)-\left(A\otimes X_{k}X^{\prime}\right)K_{n}
\end{align*}
where we use that $\ensuremath{\left(Y\otimes x^{\prime}\right)K_{sm}=x^{\prime}\otimes Y}$
when $x$ is a vector.

Combined, we have
\begin{align*}
\frac{\partial{\rm vec}\left(A^{\prime}e_{k}X_{k}X^{\prime}\right)}{\partial{\rm vec}(\tilde{\Xi})^{\prime}} & =\ensuremath{\left(XX_{k}^{\prime}e_{k}^{\prime}A\otimes A^{\prime}\right)K_{n}-\left(I_{n}\otimes A^{\prime}e_{k}\right)\left(\left(XX^{\prime}\otimes e_{k}^{\prime}A\right)+\left(A\otimes X_{k}X^{\prime}\right)K_{n}\right)}\\
 & =\left(XX_{k}^{\prime}e_{k}^{\prime}A\otimes A^{\prime}\right)K_{n}-\left(XX^{\prime}\otimes A^{\prime}J_{k}A\right)-\left(A\otimes A^{\prime}e_{k}X_{k}X^{\prime}\right)K_{n},
\end{align*}
such that 
\begin{align*}
\nabla_{\tilde{\Xi}\tilde{\Xi}^{\prime}} & =-\ensuremath{\left(A\otimes A^{\prime}\right)}K_{n}+\sum_{k=1}^{K}\tfrac{2W_{k}^{2}}{\nu_{k}+n_{k}}{\rm vec}\left(A^{\prime}e_{k}X_{k}X^{\prime}\right){\rm vec}\left(A^{\prime}e_{k}X_{k}X^{\prime}\right)^{\prime}\\
 & \quad+\sum_{k=1}^{K}W_{k}\left[\left(XX_{k}^{\prime}e_{k}^{\prime}A\otimes A^{\prime}\right)K_{n}-\left(XX^{\prime}\otimes A^{\prime}J_{k}A\right)-\left(A\otimes A^{\prime}e_{k}X_{k}X^{\prime}\right)K_{n}\right].
\end{align*}

Finally, we derive $\nabla_{\nu\nu^{\prime}}$ starting from 
\[
\nabla_{\nu_{k}}=\tfrac{1}{2}\left[\psi\left(\tfrac{v_{k}+n_{k}}{2}\right)-\psi\left(\tfrac{v_{k}}{2}\right)+1-W_{k}-\log\left(1+\tfrac{X_{k}^{\prime}X_{k}}{\nu_{k}}\right)\right].
\]
For $k\neq l$ we have $\nabla_{\nu_{k}\nu_{l}}=0$, and for $k=l$
we have
\begin{align*}
\nabla_{\nu_{k}\nu_{k}} & =\tfrac{1}{2}\left[\tfrac{1}{2}\psi^{\prime}\left(\tfrac{v_{k}+n_{k}}{2}\right)-\tfrac{1}{2}\psi^{\prime}\left(\tfrac{v_{k}}{2}\right)+\tfrac{1}{\nu_{k}+n_{k}}\left(W_{k}^{2}-W_{k}\right)+\tfrac{1}{\nu_{k}}\left[1-\tfrac{\nu_{k}}{\nu_{k}+X_{k}^{\prime}X_{k}}\right]\right]\\
 & =\tfrac{1}{4}\psi^{\prime}\left(\tfrac{v_{k}+n_{k}}{2}\right)-\tfrac{1}{4}\psi^{\prime}\left(\tfrac{v_{k}}{2}\right)+\tfrac{1}{2\nu_{k}}+\tfrac{1}{2}\tfrac{1}{\nu_{k}+n_{k}}W_{k}^{2}-\tfrac{1}{\left(\nu_{k}+n_{k}\right)}W_{k}.
\end{align*}

\hfill{}$\square$

Next, we derive some expectations for quantities that involve a multivariate
$t$-distribution.
\begin{lem}
\label{lem:qHomogeneous}Suppose that $X\sim t_{\nu}(0,I_{n})$ and
let $W=\frac{\nu+n}{\nu+X^{\prime}X}$ and
\[
\zeta_{p,q}=\left(\tfrac{\nu+n}{\nu}\right)^{\frac{p}{2}}\left(\tfrac{\nu}{2}\right)^{\frac{q}{2}}\frac{\Gamma(\tfrac{\nu+n}{2})}{\Gamma(\tfrac{\nu}{2})}\frac{\Gamma(\tfrac{\nu+p-q}{2})}{\Gamma(\tfrac{\nu+p+n}{2})}\in\mathbb{R}.
\]

(i) For any integrable function $g$ and any $p>-\nu$, it holds that
\[
\mathbb{E}\left[W^{\frac{p}{2}}g(X)\right]=\zeta_{p,0}\mathbb{E}\left[g(Y)\right],\quad Y\sim t_{\nu+p}\left(0,\tfrac{v}{v+p}I_{n}\right).
\]

(ii) Moreover, if $g$ is homogeneous of degree $q<\nu+p$, then
\[
\mathbb{E}\left[W^{\frac{p}{2}}g(X)\right]=\zeta_{p,q}\mathbb{E}\left[g(Z)\right],\quad Z\sim N\left(0,I_{n}\right).
\]

(iii) The following expectations hold
\begin{align*}
\mathbb{E}\left[W\log\left(1+\tfrac{X^{\prime}X}{\nu}\right)\right] & =\psi\left(\tfrac{\nu+n}{2}+1\right)-\psi\left(\tfrac{\nu}{2}+1\right),\\
\mathbb{E}\left[W\log\left(1+\tfrac{Q}{\nu}\right){\rm vec}\left(XX^{\prime}\right)\right] & =\left[\psi\left(\tfrac{\nu+n}{2}+1\right)-\psi\left(\tfrac{\nu}{2}\right)\right]{\rm vec}\left(I_{n}\right),\\
\mathbb{E}\left[\log\left(1+\tfrac{X^{\prime}X}{\nu}\right)\right] & =\psi\left(\tfrac{\nu+n}{2}\right)-\psi\left(\tfrac{\nu}{2}\right),\\
\mathbb{E}\left[\log^{2}\left(1+\tfrac{X^{\prime}X}{\nu}\right)\right] & =\psi^{\prime}\left(\tfrac{\nu}{2}\right)-\psi^{\prime}\left(\tfrac{\nu+n}{2}\right)+\left[\psi\left(\tfrac{\nu+n}{2}\right)-\psi\left(\tfrac{\nu}{2}\right)\right]^{2},
\end{align*}
where $\psi\left(\cdot\right)$ and $\psi^{\prime}\left(\cdot\right)$
are the digamma and trigamma functions, respectively. 
\end{lem}
Note that $p$ in part (iii) may be negative, since $-\nu<0$. If
$p/2$ is a positive integer, then
\begin{eqnarray*}
\zeta_{p,q} & = & \left(\tfrac{\nu+n}{\nu}\right)^{\frac{p}{2}}\left(\tfrac{\nu}{2}\right)^{\frac{q}{2}}\prod_{k=0}^{\tfrac{p}{2}-1}\frac{\nu+q+2k}{\nu+n+2k}.
\end{eqnarray*}
where we used $\Gamma(x+1)=x\Gamma(x)$, repeatedly, and the following
simplifications follow: 
\[
\begin{array}{rclcrcl}
\zeta_{0,2} & = & \tfrac{\nu}{\nu-2}, & \qquad\qquad & \zeta_{0,4} & = & \frac{\nu^{2}}{\left(\nu-2\right)\left(\nu-4\right)},\\
\zeta_{2,0} & = & 1, & \qquad\qquad & \zeta_{2,2} & = & 1,\\
\zeta_{2,1} & = & \left(\tfrac{\nu+1}{\nu}\right)\left(\tfrac{\nu}{2}\right)^{\frac{1}{2}}, &  & \zeta_{4,2} & = & \frac{\nu+n}{\nu+n+2},\\
\zeta_{4,0} & = & \frac{(\nu+2)(\nu+n)}{\nu(\nu+n+2)}, &  & \zeta_{4,4} & = & \frac{\nu+n}{\nu+n+2}.
\end{array}
\]

\noindent\textbf{Proof of Lemma }\ref{lem:qHomogeneous}. Let $\kappa_{\nu,n}=\Gamma(\tfrac{\nu+n}{2})/\Gamma(\tfrac{\nu}{2})$,
and the density for $X\sim t_{\nu}(0,I_{n})$ is 
\[
f_{x}(x)=\kappa_{\nu,n}[\nu\pi]^{-\frac{n}{2}}\left(1+\tfrac{x^{\prime}x}{\nu}\right)^{-\frac{\nu+n}{2}},
\]
whereas the density for $Y\sim t_{\nu+p}(0,\ensuremath{\frac{v}{v+p}}I_{n})$
is
\begin{eqnarray*}
f_{y}(y) & = & \kappa_{\nu+p,n}[(\nu+p)\pi]^{-\frac{n}{2}}\left(\tfrac{\nu}{\nu+p}\right)^{-\frac{n}{2}}\left(1+\tfrac{1}{\nu+p}x^{\prime}\left[\tfrac{\nu}{\nu+p}I_{n}\right]^{-1}x\right)^{-\frac{\nu+p+n}{2}}\\
 & = & \kappa_{\nu+p,n}[\nu\pi]^{-\frac{n}{2}}\left(1+\tfrac{x^{\prime}x}{\nu}\right)^{-\frac{\nu+p+n}{2}}.
\end{eqnarray*}
The expected value we seek is given by
\begin{eqnarray*}
\mathbb{E}\left[W^{\frac{p}{2}}g(X)\right] & = & \int\left(\tfrac{\nu+n}{\nu+x^{\prime}x}\right)^{\frac{p}{2}}g(x)\kappa_{\nu,n}[\nu\pi]^{-\frac{n}{2}}\left(1+\tfrac{x^{\prime}x}{\nu}\right)^{-\frac{\nu+n}{2}}\ensuremath{\mathrm{d}x}\\
 & = & \left(\tfrac{\nu+n}{\nu}\right)^{\frac{p}{2}}\int g(x)\kappa_{\nu,n}[\nu\pi]^{-\frac{n}{2}}\left(1+\tfrac{x^{\prime}x}{\nu}\right)^{-\frac{\nu+p+n}{2}}\ensuremath{\mathrm{d}x}\\
 & = & \left(\tfrac{\nu+n}{\nu}\right)^{\frac{p}{2}}\frac{\kappa_{\nu,n}}{\kappa_{\nu+p,n}}\int g(x)f_{y}(x)\ensuremath{\mathrm{d}x,}\\
 & = & \zeta_{p,0}\mathbb{E}\left[g(Y)\right],\quad Y\sim t_{\nu+p}(0,\tfrac{v}{v+p}I_{n}),
\end{eqnarray*}
and the results for part (i) follows, since
\[
\zeta_{p,0}=\left(\tfrac{\nu+n}{\nu}\right)^{\frac{p}{2}}\frac{\Gamma(\tfrac{\nu+n}{2})/\Gamma(\tfrac{\nu}{2})}{\Gamma(\tfrac{\nu+p+n}{2})/\Gamma(\tfrac{\nu+p}{2})}=\left(\tfrac{\nu+n}{\nu}\right)^{\frac{p}{2}}\frac{\kappa_{\nu,n}}{\kappa_{\nu+p,n}}.
\]
This completes the proof of (i).

To prove (ii) we use that $Y\sim t_{\nu+p}\left(0,\ensuremath{\frac{v}{v+p}}I_{n}\right)$
can be expressed as $Y=Z/\sqrt{\xi/\nu}$ where $Z\sim N(0,I_{n})$
and $\xi$ is an independent $\chi^{2}$-distributed random variable
with $\nu+p$ degrees of freedom. Hence, $Y=\psi Z$, with $\psi=1/\sqrt{\xi/\nu}$,
such that
\[
\psi^{q}=\left(\tfrac{\nu}{\xi}\right)^{\frac{q}{2}}.
\]
 Now using part (i) and that $g$ is homogeneous, we find
\begin{eqnarray*}
\mathbb{E}\left[W^{\frac{p}{2}}g(X)\right] & = & \zeta_{p,0}\mathbb{E}\left[g(Y)\right]=\zeta_{p,0}\mathbb{E}\left[\psi^{q}g(Z)\right]\\
 & = & \zeta_{p,0}\nu^{\frac{q}{2}}\mathbb{E}[\xi^{-\frac{q}{2}}]\mathbb{E}\left[g(Z)\right],\\
 & = & \zeta_{p,0}\nu^{\frac{q}{2}}\frac{\Gamma(\tfrac{\nu+p-q}{2})}{\Gamma(\tfrac{\nu+p}{2})}\left(\tfrac{1}{2}\right)^{\frac{q}{2}}\mathbb{E}\left[g(Z)\right]\\
 & = & \left(\tfrac{\nu+n}{\nu}\right)^{\frac{p}{2}}\frac{\Gamma(\tfrac{\nu+n}{2})/\Gamma(\tfrac{\nu}{2})}{\Gamma(\tfrac{\nu+p+n}{2})/\Gamma(\tfrac{\nu+p}{2})}\nu^{\frac{q}{2}}\frac{\Gamma(\tfrac{\nu+p-q}{2})}{\Gamma(\tfrac{\nu+p}{2})}\left(\tfrac{1}{2}\right)^{\frac{q}{2}}\mathbb{E}\left[g(Z)\right]\\
 & = & \zeta_{p,q}\mathbb{E}\left[g(Z)\right],
\end{eqnarray*}
where we used that $\xi\sim\chi_{\nu+p}^{2}$ and $Z$ are independent,
and that 
\[
\ensuremath{\mathbb{E}\left(\xi^{-\frac{q}{2}}\right)=\frac{\Gamma(\tfrac{\nu+p-q}{2})}{\Gamma(\tfrac{\nu+p}{2})}\left(\tfrac{1}{2}\right)^{\frac{q}{2}}},\qquad\text{for}\quad q<\nu+p.
\]
This completes the proof of (ii).

For the first expression in (iii), we find that for $X\sim t_{\nu}\left(0,I_{n}\right)$,
we have $X^{\prime}X/n\sim F_{n,v}$, and $S=(1+X^{\prime}X/\nu)^{-1}\sim{\rm Beta}\left(\tfrac{\nu}{2},\tfrac{n}{2}\right)$.
So we have
\begin{align*}
\mathbb{E}\left(W\log\left(1+\tfrac{X^{\prime}X}{\nu}\right)\right) & =\tfrac{\nu+n}{\nu}\mathbb{E}\left(S\log\left(S\right)\right)=\tfrac{\nu+n}{\nu}\frac{1}{\mathrm{B}(\tfrac{\nu}{2},\tfrac{n}{2})}\int_{0}^{1}\log\left(x\right)x^{\frac{\nu}{2}}(1-x)^{\frac{n}{2}-1}\mathrm{d}x.
\end{align*}
By using the following integral from \citet[p. 540]{GradshteynRyzhik2007},
\begin{equation}
\ensuremath{\int_{0}^{1}x^{\mu-1}\left(1-x^{r}\right)^{s-1}\log x\ {\rm d}x=\tfrac{1}{r^{2}}\mathrm{B}\left(\tfrac{\mu}{r},s\right)\left[\psi\left(\tfrac{\mu}{r}\right)-\psi\left(\tfrac{\mu}{r}+s\right)\right]},\label{eq:Int1}
\end{equation}
where $\psi\left(\cdot\right)$ is the digamma function, and let $r=1,s=\tfrac{n}{2},\mu=\frac{\nu}{2}+1$,
we have
\[
\mathbb{E}\left(W\log\left(1+\tfrac{X^{\prime}X}{\nu}\right)\right)=\tfrac{\nu+n}{\nu}\tfrac{\mathrm{B}\left(\frac{\nu}{2}+1,\frac{n}{2}\right)}{\mathrm{B}(\tfrac{\nu}{2},\tfrac{n}{2})}\left[\psi\left(\tfrac{\nu}{2}+1\right)-\psi\left(\tfrac{\nu+n}{2}+1\right)\right]=\psi\left(\tfrac{\nu}{2}+1\right)-\psi\left(\tfrac{\nu+n}{2}+1\right).
\]

For the first expression in (iii), it is well known that for an elliptic
distribution $X\sim EL\left(0,I_{n}\right)$, the two random variables,
$Q\equiv X^{\prime}X$ and $U\equiv Q^{-\frac{1}{2}}X$, are independent,
and $U$ has a uniform distribution on a unit sphere, see e.g. \citet[section 2]{Mitchell1989}.
This result applies to the multivariate $t$-distribution, $X\sim t_{\nu}\left(0,I_{n}\right)$,
such that
\begin{align*}
\mathbb{E}\left[W\log\left(1+\tfrac{Q}{\nu}\right){\rm vec}\left(XX^{\prime}\right)\right] & =\mathbb{E}\left[QW\log\left(1+\tfrac{Q}{\nu}\right){\rm vec}\left(XX^{\prime}/Q\right)\right]\\
 & =\mathbb{E}\left[QW\log\left(1+\tfrac{Q}{\nu}\right){\rm vec}\left(UU^{\prime}\right)\right]\\
 & =\mathbb{E}\left[QW\log\left(1+\tfrac{Q}{\nu}\right)\right]\mathbb{E}\left[{\rm vec}\left(UU^{\prime}\right)\right]\\
 & =\tfrac{\nu+n}{n}\mathbb{E}\left[\tfrac{Q/v}{1+Q/v}\log\left(1+\tfrac{Q}{\nu}\right)\right]{\rm vec}\left(I_{n}\right),
\end{align*}
where the last equality uses that $\mathbb{E}\left[{\rm vec}\left(UU^{\prime}\right)\right]=\tfrac{1}{n}{\rm vec}\left(I_{n}\right)$.
Then we have
\begin{align*}
\mathbb{E}\left[\tfrac{Q/v}{1+Q/v}\log\left(1+\tfrac{Q}{\nu}\right)\right] & =-\mathbb{E}\left[\left(1-S\right)\log\left(S\right)\right]\\
 & =\frac{-1}{\mathrm{B}(\tfrac{\nu}{2},\tfrac{n}{2})}\int_{0}^{1}\log\left(x\right)x^{\frac{\nu}{2}-1}(1-x)^{\frac{n}{2}}dx\\
 & =\frac{\mathrm{B}\left(\frac{\nu}{2},\tfrac{n}{2}+1\right)}{\mathrm{B}(\tfrac{\nu}{2},\tfrac{n}{2})}\left[\psi\left(\tfrac{\nu+n}{2}+1\right)-\psi\left(\tfrac{\nu}{2}\right)\right]\\
 & =\tfrac{n}{\nu+n}\left[\psi\left(\tfrac{\nu+n}{2}+1\right)-\psi\left(\tfrac{\nu}{2}\right)\right],
\end{align*}
where we let $r=1,s=\tfrac{n}{2}+1,\mu=\frac{\nu}{2}$ in the integral
(\ref{eq:Int1}). This leads to the second expression in (iii)
\[
\mathbb{E}\left[W\log\left(1+\tfrac{X^{\prime}X}{\nu}\right){\rm vec}\left(XX^{\prime}\right)\right]=\left[\psi\left(\tfrac{\nu+n}{2}+1\right)-\psi\left(\tfrac{\nu}{2}\right)\right]{\rm vec}\left(I_{n}\right).
\]
If we set $r=1,s=\tfrac{n}{2},\mu=\tfrac{\nu}{2}$ in (\ref{eq:Int1})
then we obtain the third expression in (iii) 
\[
\mathbb{E}\left[\log\left(1+\tfrac{Q}{\nu}\right)\right]=-\mathbb{E}\left[\log\left(S\right)\right]=\frac{-1}{\mathrm{B}(\tfrac{\nu}{2},\tfrac{n}{2})}\int_{0}^{1}\log\left(x\right)x^{\frac{\nu}{2}-1}(1-x)^{\frac{n}{2}-1}dx=\psi\left(\tfrac{\nu+n}{2}\right)-\psi\left(\tfrac{\nu}{2}\right).
\]

Finally, for the fourth expression in (iii), we use 
\[
\ensuremath{\ensuremath{\int_{0}^{1}(\log)^{2}x^{p-1}\left(1-x^{r}\right)^{q-1}dx=\tfrac{\mathrm{B}\left(\tfrac{p}{r},q\right)}{r^{3}}\left\{ \psi^{\prime}\left(\tfrac{p}{r}\right)-\psi^{\prime}\left(\tfrac{p}{r}+q\right)+\left[\psi\left(\tfrac{p}{r}\right)-\psi\left(\tfrac{p}{r}+q\right)\right]^{2}\right\} }},
\]
see \citet[p. 543]{GradshteynRyzhik2007}, and set $p=\frac{\nu}{2},r=1,q=\tfrac{n}{2}$,
whereby we obtain
\begin{align*}
\mathbb{E}\left[\log^{2}\left(1+\tfrac{Q}{\nu}\right)\right] & =\mathbb{E}\left[\log^{2}\left(S\right)\right]\\
 & =\frac{1}{\mathrm{B}(\tfrac{\nu}{2},\tfrac{n}{2})}\int_{0}^{1}\log^{2}\left(x\right)x^{\frac{\nu}{2}-1}(1-x)^{\frac{n}{2}-1}dx\\
 & =\frac{1}{\mathrm{B}(\tfrac{\nu}{2},\tfrac{n}{2})}\mathrm{B}\left(\tfrac{\nu}{2},\tfrac{n}{2}\right)\left\{ \psi^{\prime}\left(\tfrac{\nu}{2}\right)-\psi^{\prime}\left(\tfrac{\nu+n}{2}\right)+\left[\psi\left(\tfrac{\nu}{2}\right)-\psi\left(\tfrac{\nu+n}{2}\right)\right]^{2}\right\} \\
 & =\psi^{\prime}\left(\tfrac{\nu}{2}\right)-\psi^{\prime}\left(\tfrac{\nu+n}{2}\right)+\left[\psi\left(\tfrac{\nu+n}{2}\right)-\psi\left(\tfrac{\nu}{2}\right)\right]^{2}.
\end{align*}
This completes the proof.

\hfill{}$\square$

\noindent \textbf{Proof of Theorem \ref{thm:ExpecScoreHess} (Expected
Scores).} We have $\mathbb{E}\left(\nabla_{\mu}\right)=\sum_{k=1}^{K}A^{\prime}e_{k}\mathbb{E}\left(W_{k}X_{k}\right)=0$,
\begin{align*}
\mathbb{E}\left(\nabla_{\nu_{k}}\right) & =\tfrac{1}{2}\mathbb{E}\left[\psi\left(\tfrac{v_{k}+n_{k}}{2}\right)-\psi\left(\tfrac{v_{k}}{2}\right)+1-W_{k}-\log\left(1+\tfrac{X_{k}^{\prime}X_{k}}{\nu_{k}}\right)\right]\\
 & =\tfrac{1}{2}\left[\psi\left(\tfrac{v_{k}+n_{k}}{2}\right)-\psi\left(\tfrac{v_{k}}{2}\right)+1-1-\left[\psi\left(\tfrac{v_{k}+n_{k}}{2}\right)-\psi\left(\tfrac{v_{k}}{2}\right)\right]\right]=0.
\end{align*}
Next, observe that and $\mathbb{E}\left(W_{k}X_{k}X_{k}^{\prime}\right)=I_{n_{k}}$
and for $k\neq l$ we have $\mathbb{E}\left(W_{k}X_{k}X_{l}^{\prime}\right)=0$,
which we used $\nu_{l}>1$. We therefore have $\mathbb{E}\left(W_{k}X_{k}X^{\prime}\right)=e_{k}^{\prime}$
and it follows that 
\begin{eqnarray*}
\mathbb{E}\left(\nabla_{\tilde{\Xi}}\right) & = & \sum_{k=1}^{K}{\rm vec}\left(A^{\prime}e_{k}\mathbb{E}\left(W_{k}X_{k}X^{\prime}\right)\right)-{\rm vec}\left(A^{\prime}\right)\\
 & = & \sum_{k=1}^{K}{\rm vec}\left(A^{\prime}e_{k}e_{k}^{\prime}\right)-{\rm vec}\left(A^{\prime}\right)=0,
\end{eqnarray*}
since $\sum_{k=1}^{K}e_{k}e_{k}^{\prime}=I_{n}$.

\hfill{}$\square$

\noindent \textbf{Proof of Theorem \ref{thm:ExpecScoreHess} (Information
matrix).} We have 
\begin{align*}
\mathbb{E}\left(\nabla_{\mu}\nabla_{\mu}^{\prime}\right) & =\sum_{k=1}^{K}\sum_{l=1}^{K}A^{\prime}e_{l}\mathbb{E}\left(W_{l}W_{k}X_{l}X_{k}^{\prime}\right)e_{k}^{\prime}A\\
 & =\sum_{k=1}^{K}A^{\prime}e_{k}\mathbb{E}\left(W_{k}^{2}X_{k}X_{k}^{\prime}\right)e_{k}^{\prime}A+\sum_{k=1}^{K}\sum_{l\neq k}^{K}A^{\prime}e_{l}\mathbb{E}\left(W_{l}X_{l}\right)\mathbb{E}\left(W_{k}X_{k}^{\prime}\right)e_{k}^{\prime}A\\
 & =\sum_{k=1}^{K}A^{\prime}e_{k}\mathbb{E}\left(W_{k}^{2}X_{k}X_{k}^{\prime}\right)e_{k}^{\prime}A\\
 & =\sum_{k=1}^{K}\phi_{k}A^{\prime}e_{k}e_{k}^{\prime}A=\sum_{k=1}^{K}\phi_{k}A^{\prime}J_{k}A,
\end{align*}
where we used Lemma \ref{lem:qHomogeneous} $\mathbb{E}\left(W_{k}^{2}X_{k}X_{k}^{\prime}\right)=\zeta_{4,2}I_{n_{k}}$,
where $\zeta_{4,2}=\frac{\nu_{k}+n_{k}}{\nu_{k}+n_{k}+2}=\phi_{k}$.
Next we have,
\[
\mathbb{E}\left(\nabla_{\mu}\nabla_{\tilde{\Xi}}^{\prime}\right)=\mathbb{E}\left(\nabla_{\mu}\nabla_{\nu}^{\prime}\right)=0,
\]
because $\nabla_{\tilde{\Xi}}^{\prime}$ and $\nabla_{\nu}^{\prime}$
are both even functions of $X$, whereas $\nabla_{\mu}$ is an odd
function.

\subsubsection*{The term: $\mathbb{E}\left(\nabla_{\tilde{\Xi}}\nabla_{\tilde{\Xi}}^{\prime}\right)$ }

It is convenient to define $M=\left(I\otimes A^{\prime}\right)$ and
$\nabla_{\tilde{\Xi}}^{s}=\sum_{k=1}^{K}W_{k}{\rm vec}\left(e_{k}X_{k}X^{\prime}\right)-{\rm vec}\left(I_{n}\right)$
(whose distribution does not depend on $A$), such that 

\[
\nabla_{\tilde{\Xi}}=\sum_{k=1}^{K}W_{k}{\rm vec}\left(A^{\prime}e_{k}X_{k}X^{\prime}\right)-{\rm vec}\left(A^{\prime}\right)=M\nabla_{\tilde{\Xi}}^{s}.
\]
It follows that $\mathcal{I}_{\tilde{\Xi}}=M\mathbb{E}\left(\nabla_{\tilde{\Xi}}^{s}\nabla_{\tilde{\Xi}}^{s\prime}\right)M^{\prime}$,
where
\begin{align*}
\mathbb{E}\left(\nabla_{\tilde{\Xi}}^{s}\nabla_{\tilde{\Xi}}^{s\prime}\right) & =\mathbb{E}\left[\sum_{k=1}^{K}\sum_{l=1}^{K}W_{k}W_{l}{\rm vec}\left(e_{k}X_{k}X^{\prime}\right){\rm vec}\left(e_{l}X_{l}X^{\prime}\right)^{\prime}-{\rm vec}\left(I_{n}\right){\rm vec}\left(I_{n}\right)^{\prime}\right].
\end{align*}

Next we derive the expected value of 
\[
\varphi\left(k,l\right)\equiv W_{k}W_{l}{\rm vec}\left(e_{k}X_{k}X^{\prime}\right){\rm vec}\left(e_{l}X_{l}X^{\prime}\right)^{\prime}.
\]

For $k=l$ we have
\begin{align*}
\mathbb{E}\left[\varphi\left(k,k\right)\right]= & \mathbb{E}\left[W_{k}^{2}\sum_{p=1}^{K}{\rm vec}\left(e_{k}X_{k}X_{p}^{\prime}e_{p}^{\prime}\right)\sum_{q=1}^{K}{\rm vec}\left(e_{k}X_{k}X_{q}^{\prime}e_{q}^{\prime}\right)^{\prime}\right]\\
= & \mathbb{E}\left[W_{k}^{2}\sum_{p\neq k}{\rm vec}\left(e_{k}X_{k}X_{p}^{\prime}e_{p}^{\prime}\right){\rm vec}\left(e_{k}X_{k}X_{p}^{\prime}e_{p}^{\prime}\right)^{\prime}+W_{k}^{2}{\rm vec}\left(e_{k}X_{k}X_{k}^{\prime}e_{k}^{\prime}\right){\rm vec}\left(e_{k}X_{k}X_{k}^{\prime}e_{k}^{\prime}\right)^{\prime}\right].
\end{align*}
From Lemma \ref{lem:qHomogeneous} we have the expression for $\phi_{k}=\zeta_{42}=\zeta_{44}$
and $\zeta_{0,2}=\frac{\nu}{\nu-2}$, we have
\begin{align*}
 & \sum_{p\neq k}\mathbb{E}\left[W_{k}^{2}{\rm vec}\left(e_{k}X_{k}X_{p}^{\prime}e_{p}^{\prime}\right){\rm vec}\left(e_{k}X_{k}X_{p}^{\prime}e_{p}^{\prime}\right)^{\prime}\right]\\
= & \sum_{p\neq k}\left(e_{p}\otimes e_{k}\right)\mathbb{E}\left[W_{k}^{2}{\rm vec}\left(X_{k}X_{p}^{\prime}\right){\rm vec}\left(X_{k}X_{p}^{\prime}\right)^{\prime}\right]\left(e_{p}^{\prime}\otimes e_{k}^{\prime}\right)\\
= & \phi_{k}\sum_{p\neq k}\tfrac{\nu_{p}}{\nu_{p}-2}\left(e_{p}\otimes e_{k}\right)\left(e_{p}^{\prime}\otimes e_{k}^{\prime}\right)\\
= & \phi_{k}\sum_{p\neq k}\tfrac{\nu_{p}}{\nu_{p}-2}\left(J_{p}\otimes J_{k}\right),
\end{align*}
where we applied Lemma \ref{lem:qHomogeneous} (ii) to the conditional
expectation, $\mathbb{E}[W_{k}^{2}g(X_{k})|X_{p}]=\zeta_{4,2}\mathbb{E}[g(Z)|X_{p}]$
with $g(X_{k})={\rm vec}\left(X_{k}X_{p}^{\prime}\right){\rm vec}\left(X_{k}X_{p}^{\prime}\right)^{\prime}$,
using that $X_{k}$ and $X_{p}$ are independent, followed by taking
unconditional expectation
\[
\mathbb{E}[g(Z)]=\mathbb{E}[{\rm vec}\left(ZX_{p}^{\prime}\right){\rm vec}\left(ZX_{p}^{\prime}\right)^{\prime}]=\tfrac{\nu_{p}}{\nu_{p}-2}I_{n_{k}n_{p}}.
\]
Similarly, using Lemma \ref{lem:qHomogeneous} and $\zeta_{4,4}=\phi_{k}$
we have 
\begin{align*}
 & \mathbb{E}\left[W_{k}^{2}{\rm vec}\left(e_{k}X_{k}X_{k}^{\prime}e_{k}^{\prime}\right){\rm vec}\left(e_{k}X_{k}X_{k}^{\prime}e_{k}^{\prime}\right)^{\prime}\right]\\
= & \left(e_{k}\otimes e_{k}\right)\mathbb{E}\left[\omega_{k}^{2}{\rm vec}\left(X_{k}X_{k}^{\prime}\right){\rm vec}\left(X_{k}X_{k}^{\prime}\right)^{\prime}\right]\left(e_{k}^{\prime}\otimes e_{k}^{\prime}\right)\\
= & \phi_{k}\left(e_{k}\otimes e_{k}\right)\left[I_{n_{k}^{2}}+K_{n_{k}}+{\rm vec}(I_{n_{k}}){\rm vec}(I_{n_{k}})^{\prime}\right]\left(e_{k}^{\prime}\otimes e_{k}^{\prime}\right)\\
= & \phi_{k}\left[J_{k\otimes}\left(K_{n}+I_{n^{2}}\right)+{\rm vec}(J_{k}){\rm vec}(J_{k})^{\prime}\right].
\end{align*}
Combining the results, we have
\begin{align}
\mathbb{E}\left[\varphi\left(k,k\right)\right] & =\phi_{k}\sum_{p\neq k}\tfrac{\nu_{p}}{\nu_{p}-2}J_{p}\otimes J_{k}+\phi_{k}\left[J_{k\otimes}\left(K_{n}+I_{n^{2}}\right)+{\rm vec}(J_{k}){\rm vec}(J_{k})^{\prime}\right]\nonumber \\
 & =\phi_{k}J_{k}^{\bullet}\otimes J_{k}-\phi_{k}J_{k\otimes}+\phi_{k}\left[J_{k\otimes}\left(K_{n}+I_{n^{2}}\right)+{\rm vec}(J_{k}){\rm vec}(J_{k})^{\prime}\right]\nonumber \\
 & =\phi_{k}J_{k}^{\bullet}\otimes J_{k}+\phi_{k}\left[J_{k\otimes}K_{n}+{\rm vec}(J_{k}){\rm vec}(J_{k})^{\prime}\right],\label{eq:phikk}
\end{align}
where $J_{k}^{\bullet}=\sum_{p\neq k}\frac{\nu_{p}}{\nu_{p}-2}J_{p}+J_{k}$.

Next, for $k\neq l$, we have
\begin{align*}
\mathbb{E}\left[\varphi\left(k,l\right)\right]= & \mathbb{E}\left[W_{k}W_{l}\sum_{p=1}^{K}{\rm vec}\left(e_{k}X_{k}X_{p}^{\prime}e_{p}^{\prime}\right)\sum_{q=1}^{K}{\rm vec}\left(e_{l}X_{l}X_{q}^{\prime}e_{q}^{\prime}\right)^{\prime}\right]\\
= & \mathbb{E}\left[W_{k}W_{l}{\rm vec}\left(e_{k}X_{k}X_{k}^{\prime}e_{k}^{\prime}\right){\rm vec}\left(e_{l}X_{l}X_{l}^{\prime}e_{l}^{\prime}\right)^{\prime}+W_{k}W_{l}{\rm vec}\left(e_{k}X_{k}X_{l}^{\prime}e_{l}^{\prime}\right){\rm vec}\left(e_{l}X_{l}X_{k}^{\prime}e_{k}^{\prime}\right)^{\prime}\right].
\end{align*}
From Lemma \ref{lem:qHomogeneous} with $p=q=2$ we find that
\begin{align*}
 & \mathbb{E}\left[W_{k}W_{l}{\rm vec}\left(e_{k}X_{k}X_{k}^{\prime}e_{k}^{\prime}\right){\rm vec}\left(e_{l}X_{l}X_{q}^{\prime}e_{q}^{\prime}\right)^{\prime}\right]\\
= & {\rm vec}\left(e_{k}\mathbb{E}\left[W_{k}X_{k}X_{k}^{\prime}\right]e_{k}^{\prime}\right){\rm vec}\left(e_{l}\mathbb{E}\left[W_{l}X_{l}X_{l}^{\prime}\right]e_{l}^{\prime}\right)^{\prime}\\
= & {\rm vec}\left(J_{k}\right){\rm vec}\left(J_{l}\right)^{\prime},
\end{align*}
and we also have
\begin{align*}
 & \mathbb{E}\left[W_{k}W_{l}{\rm vec}\left(e_{k}X_{k}X_{l}^{\prime}e_{l}^{\prime}\right){\rm vec}\left(e_{l}X_{l}X_{k}^{\prime}e_{k}^{\prime}\right)^{\prime}\right]\\
= & \mathbb{E}\left[W_{k}W_{l}{\rm vec}\left(e_{k}X_{k}X_{l}^{\prime}e_{l}^{\prime}\right){\rm vec}\left(e_{k}X_{k}X_{l}^{\prime}e_{l}^{\prime}\right)^{\prime}K_{n}\right]\\
= & \left(e_{l}\otimes e_{k}\right)\mathbb{E}\left[{\rm vec}\left(Z_{k}Z_{l}^{\prime}\right){\rm vec}\left(Z_{k}Z_{l}^{\prime}\right)^{\prime}\right]\left(e_{l}^{\prime}\otimes e_{k}^{\prime}\right)K_{n}\\
= & \left(e_{l}\otimes e_{k}\right)\left(e_{l}^{\prime}\otimes e_{k}^{\prime}\right)K_{n}\\
= & \left(J_{l}\otimes J_{k}\right)K_{n},
\end{align*}
where $Z_{k}\sim N(0,I_{n_{k}})$. Finally, we have
\[
\mathbb{E}\left[\varphi\left(k,l\right)\right]={\rm vec}\left(J_{k}\right){\rm vec}\left(J_{l}\right)^{\prime}+\left(J_{l}\otimes J_{k}\right)K_{n}.
\]

We are finally ready to derive the expression for $\mathbb{E}\left(\nabla_{\tilde{\Xi}}^{s}\nabla_{\tilde{\Xi}}^{s\prime}\right)$,
which takes the form
\begin{align*}
\mathbb{E}\left(\nabla_{\tilde{\Xi}}^{s}\nabla_{\tilde{\Xi}}^{s\prime}\right) & =\mathbb{E}\left[\sum_{k=1}^{K}\sum_{l=1}^{K}W_{k}W_{l}{\rm vec}\left(e_{k}X_{k}X^{\prime}\right){\rm vec}\left(e_{l}V_{l}X^{\prime}\right)^{\prime}\right]-{\rm vec}\left(I_{n}\right){\rm vec}\left(I_{n}\right)^{\prime}\\
 & =\sum_{k=1}^{K}\sum_{l=1}^{K}\left[{\rm vec}\left(J_{k}\right){\rm vec}\left(J_{l}\right)^{\prime}+\left(J_{l}\otimes J_{k}\right)K_{n}\right]-{\rm vec}\left(I_{n}\right){\rm vec}\left(I_{n}\right)^{\prime}\\
 & \quad+\sum_{k=1}^{K}\left[\mathbb{E}\left[\varphi\left(k,k\right)\right]-{\rm vec}\left(J_{k}\right){\rm vec}\left(J_{k}\right)^{\prime}-J_{k\otimes}K_{n}\right]\\
 & =K_{n}+\tilde{\Upsilon}_{K},
\end{align*}
where $\tilde{\Upsilon}_{K}=\sum_{k}\tilde{\Psi}_{k}$ with 
\begin{align*}
\tilde{\Psi}_{k} & =\phi_{k}J_{k}^{\bullet}\otimes J_{k}+\phi_{k}\left[J_{k\otimes}K_{n}+{\rm vec}(J_{k}){\rm vec}(J_{k})^{\prime}\right]\ensuremath{-{\rm vec}\left(J_{k}\right){\rm vec}\left(J_{k}\right)^{\prime}-J_{k\otimes}K_{n}}\\
 & =\phi_{k}J_{k}^{\bullet}\otimes J_{k}+\left(\phi_{k}-1\right)\left(J_{k\otimes}K_{n}+{\rm vec}(J_{k}){\rm vec}(J_{k})^{\prime}\right).
\end{align*}
So, the final formula for the information matrix is given by
\begin{align*}
\mathcal{I}_{\tilde{\Xi}} & =M\mathbb{E}\left(\nabla_{\tilde{\Xi}}^{s}\nabla_{\tilde{\Xi}}^{s\prime}\right)M^{\prime}=M\left(K_{n}+\tilde{\Upsilon}_{K}\right)M^{\prime}\\
 & =\left(A\otimes A^{\prime}\right)K_{n}+\Upsilon_{K},
\end{align*}
where $\Upsilon_{K}=\sum_{k=1}^{K}\Psi_{k}$ with
\begin{align*}
\Psi_{k}=M\tilde{\Psi}_{k}M^{\prime} & =\phi_{k}J_{k}^{\bullet}\otimes A^{\prime}J_{k}A+\left(\phi_{k}-1\right)\left[\left(J_{k}A\otimes A^{\prime}J_{k}\right)K_{n}+{\rm vec}(A^{\prime}J_{k}){\rm vec}(A^{\prime}J_{k})^{\prime}\right].
\end{align*}

\subsubsection*{Cross-Covariance: $\mathbb{E}\left(\nabla_{\nu_{k}}\nabla_{\nu_{l}}\right)$}

Because $X_{k}$ and $X_{l}$ are independent for $k\neq l$, we have
$\mathbb{E}\left(\nabla_{\nu_{k}}\nabla_{\nu_{l}}\right)=0$ for $k\neq l$.
For $k=l$, we first note that 
\begin{align*}
\mathbb{E}\left[W_{k}^{2}\right] & =\tfrac{\left(\nu_{k}+n_{k}\right)\left(\nu_{k}+2\right)}{\nu_{k}\left(\nu_{k}+n_{k}+2\right)},\\
\mathbb{E}\left[W_{k}\log\left(1+\tfrac{X_{k}^{\prime}X_{k}}{\nu_{k}}\right)\right] & =\psi\left(\tfrac{\nu+n}{2}+1\right)-\psi\left(\tfrac{\nu}{2}+1\right),\\
\mathbb{E}\left[\log^{2}\left(1+\tfrac{X_{k}^{\prime}X_{k}}{\nu_{k}}\right)\right] & =\psi^{\prime}\left(\tfrac{\nu_{k}}{2}\right)-\psi^{\prime}\left(\tfrac{\nu_{k}+n_{k}}{2}\right)+\left[\psi\left(\tfrac{\nu_{k}+n_{k}}{2}\right)-\psi\left(\tfrac{\nu_{k}}{2}\right)\right]^{2},
\end{align*}
using Lemma \ref{lem:qHomogeneous} (ii) and (iii) and, after some
algebra, it follows that 
\begin{eqnarray*}
\mathbb{E}\left(\nabla_{\nu_{k}}^{2}\right) & = & \tfrac{1}{4}\mathbb{E}\left[W_{k}^{2}+\log^{2}\left(1+\tfrac{X_{k}^{\prime}X_{k}}{\nu_{k}}\right)+2W_{k}\log\left(1+\tfrac{X_{k}^{\prime}X_{k}}{\nu_{k}}\right)\right]-\tfrac{1}{4}\left[\psi\left(\tfrac{v_{k}+n_{k}}{2}\right)-\psi\left(\tfrac{v_{k}}{2}\right)+1\right]^{2}.\\
 & = & \tfrac{1}{4}\left[\psi^{\prime}\left(\tfrac{\nu_{k}}{2}\right)-\psi^{\prime}\left(\tfrac{\nu_{k}+n_{k}}{2}\right)\right]-\tfrac{n_{k}\left(\nu_{k}+n_{k}+4\right)}{2\nu_{k}\left(\nu_{k}+n_{k}+2\right)\left(\nu_{k}+n_{k}\right)}.
\end{eqnarray*}

\subsubsection*{Cross-Covariance: $\mathbb{E}\left(\nabla_{\nu}\nabla_{\tilde{\Xi}}^{\prime}\right)$}

We first compute
\begin{align}
\mathbb{E}\left[\nabla_{\nu_{k}}\nabla_{\tilde{\Xi}}^{s}\right] & =\tfrac{1}{2}\left[\psi\left(\tfrac{v_{k}+n_{k}}{2}\right)-\psi\left(\tfrac{v_{k}}{2}\right)+1-W_{k}-\log\left(1+\tfrac{X_{k}^{\prime}X_{k}}{\nu_{k}}\right)\right]\left[\sum_{l=1}^{K}W_{l}{\rm vec}\left(e_{l}X_{l}X^{\prime}\right)-{\rm vec}\left(I_{n}\right)\right]\nonumber \\
 & =-\tfrac{1}{2}\mathbb{E}\left[\left[W_{k}+\log\left(1+\tfrac{X_{k}^{\prime}X_{k}}{\nu_{k}}\right)\right]\sum_{l=1}^{K}W_{l}{\rm vec}\left(e_{l}X_{l}X^{\prime}\right)\right]+\tfrac{1}{2}\left(\psi\left(\tfrac{v_{k}+n_{k}}{2}\right)-\psi\left(\tfrac{v_{k}}{2}\right)+1\right){\rm vec}\left(I_{n}\right),\label{eq:somecrossterm}
\end{align}
For $k=l$, we have
\begin{align*}
\mathbb{E}\left[W_{k}^{2}{\rm vec}\left(e_{k}X_{k}X\right)\right]= & \mathbb{E}\left[W_{k}^{2}\sum_{p=1}^{K}{\rm vec}\left(e_{k}X_{k}X_{p}^{\prime}e_{p}^{\prime}\right)\right]\\
= & \mathbb{E}\left[W_{k}^{2}{\rm vec}\left(e_{k}X_{k}X_{k}^{\prime}e_{k}^{\prime}\right)\right]\\
= & e_{k\otimes}\mathbb{E}\left[W_{k}^{2}{\rm vec}\left(X_{k}X_{k}^{\prime}\right)\right]\\
= & e_{k\otimes}\phi_{k}{\rm vec}\left(I_{n}\right)=\phi_{k}{\rm vec}\left(J_{k}\right),
\end{align*}
and for $k\neq l$ we have
\begin{align*}
\mathbb{E}\left[W_{k}W_{l}{\rm vec}\left(e_{l}X_{l}X^{\prime}\right)\right]= & \mathbb{E}\left[W_{k}W_{l}\sum_{p=1}^{K}{\rm vec}\left(e_{l}X_{l}X_{p}^{\prime}e_{p}^{\prime}\right)\right]\\
= & \mathbb{E}\left[W_{k}W_{l}{\rm vec}\left(e_{l}X_{l}X_{k}^{\prime}e_{k}^{\prime}\right)\right]+\mathbb{E}\left[W_{k}W_{l}{\rm vec}\left(e_{l}X_{g}X_{l}^{\prime}e_{l}^{\prime}\right)\right]\\
= & \mathbb{E}\left[{\rm vec}\left(e_{l}\left(W_{l}X_{l}\right)\left(W_{k}X_{k}^{\prime}\right)e_{k}^{\prime}\right)\right]+\mathbb{E}\left[W_{l}{\rm vec}\left(e_{l}X_{l}X_{l}^{\prime}e_{l}^{\prime}\right)\right]\\
= & \mathbb{E}\left[W_{l}{\rm vec}\left(e_{l}X_{l}X_{l}^{\prime}e_{l}^{\prime}\right)\right]\\
= & {\rm vec}(J_{l}),
\end{align*}
where we used $\zeta_{2,2}=1$. Thus, we have
\begin{align*}
\mathbb{E}\left[W_{k}\sum_{l=1}^{K}W_{l}{\rm vec}\left(e_{l}X_{l}X^{\prime}\right)\right] & =\mathbb{E}\left[W_{k}^{2}{\rm vec}\left(e_{k}X_{k}X^{\prime}\right)\right]+\mathbb{E}\left[W_{k}\sum_{l\neq k}^{K}W_{l}{\rm vec}\left(e_{l}X_{l}X^{\prime}\right)\right]\\
 & =\phi_{k}{\rm vec}\left(J_{k}\right)+\sum_{l\neq k}^{K}{\rm vec}\left(J_{l}\right)\\
 & =\left(\phi_{k}-1\right){\rm vec}\left(J_{k}\right)+{\rm vec}\left(I_{n}\right).
\end{align*}
For the remaining term in (\ref{eq:somecrossterm}), we have, for
$k=l$, 
\begin{align*}
\mathbb{E}\left[\log\left(1+\tfrac{X_{k}^{\prime}X_{k}}{\nu_{k}}\right)W_{k}{\rm vec}\left(e_{k}X_{k}X^{\prime}\right)\right]= & \mathbb{E}\left[\log\left(1+\tfrac{X_{k}^{\prime}X_{k}}{\nu_{k}}\right)W_{k}{\rm vec}\left(e_{k}X_{k}X_{k}^{\prime}e_{k}^{\prime}\right)\right]\\
= & e_{k\otimes}\mathbb{E}\left[\log\left(1+\tfrac{X_{k}^{\prime}X_{k}}{\nu_{k}}\right)W_{k}{\rm vec}\left(X_{k}X_{k}^{\prime}\right)\right]\\
= & \left[\psi\left(\tfrac{v_{k}+n_{k}}{2}+1\right)-\psi\left(\tfrac{v_{k}}{2}\right)\right]{\rm vec}\left(J_{k}\right),
\end{align*}
where we used Lemma \ref{lem:qHomogeneous} (iii). For $k\neq l$,
we have
\begin{align*}
 & \mathbb{E}\left[\log\left(1+\tfrac{X_{k}^{\prime}X_{k}}{\nu_{k}}\right)W_{l}{\rm vec}\left(e_{l}X_{g}X^{\prime}\right)\right]\\
= & \mathbb{E}\left[\log\left(1+\tfrac{X_{k}^{\prime}X_{k}}{\nu_{k}}\right)W_{l}\sum_{p=1}^{K}{\rm vec}\left(e_{l}X_{l}X_{p}^{\prime}e_{p}^{\prime}\right)\right]\\
= & \mathbb{E}\left[\log\left(1+\tfrac{X_{k}^{\prime}X_{k}}{\nu_{k}}\right)W_{l}{\rm vec}\left(e_{l}X_{l}X_{l}^{\prime}e_{l}^{\prime}\right)\right]+\mathbb{E}\left[\log\left(1+\tfrac{X_{k}^{\prime}X_{k}}{\nu_{k}}\right)W_{l}{\rm vec}\left(e_{l}X_{l}X_{k}^{\prime}e_{k}^{\prime}\right)\right]\\
= & \mathbb{E}\left[\log\left(1+\tfrac{X_{k}^{\prime}X_{k}}{\nu_{k}}\right)\right]\times\mathbb{E}\left[W_{l}{\rm vec}\left(e_{l}X_{l}X_{l}^{\prime}e_{l}^{\prime}\right)\right]\\
= & \left[\psi\left(\tfrac{v_{k}+n_{k}}{2}\right)-\psi\left(\tfrac{v_{k}}{2}\right)\right]{\rm vec}\left(J_{l}\right),
\end{align*}
and combined we have
\begin{align*}
 & \mathbb{E}\left[\log\left(1+\tfrac{X_{k}^{\prime}X_{k}}{\nu_{k}}\right)\sum_{l=1}^{K}W_{l}{\rm vec}\left(e_{l}X_{l}X^{\prime}\right)\right]\\
= & \left[\psi\left(\tfrac{v_{k}+n_{k}}{2}+1\right)-\psi\left(\tfrac{v_{k}}{2}\right)\right]{\rm vec}\left(J_{k}\right)+\left[\psi\left(\tfrac{v_{k}+n_{k}}{2}\right)-\psi\left(\tfrac{v_{k}}{2}\right)\right]\sum_{l\neq k}{\rm vec}\left(J_{l}\right)\\
= & \left[\psi\left(\tfrac{v_{k}+n_{k}}{2}+1\right)-\psi\left(\tfrac{v_{k}}{2}\right)\right]{\rm vec}\left(J_{k}\right)+\left[\psi\left(\tfrac{v_{k}+n_{k}}{2}\right)-\psi\left(\tfrac{v_{k}}{2}\right)\right]\left[{\rm vec}\left(I_{n}\right)-{\rm vec}\left(J_{k}\right)\right]\\
= & \left[\psi\left(\tfrac{v_{k}+n_{k}}{2}\right)-\psi\left(\tfrac{v_{k}}{2}\right)\right]{\rm vec}\left(I_{n}\right)+{\rm vec}\left(J_{k}\right)\left[\psi\left(\tfrac{v_{k}+n_{k}}{2}+1\right)-\psi\left(\tfrac{v_{k}}{2}\right)-\psi\left(\tfrac{v_{k}+n_{k}}{2}\right)+\psi\left(\tfrac{v_{k}}{2}\right)\right]\\
= & \left[\psi\left(\tfrac{v_{k}+n_{k}}{2}\right)-\psi\left(\tfrac{v_{k}}{2}\right)\right]{\rm vec}\left(I_{n}\right)+\tfrac{2}{v_{k}+n_{k}}{\rm vec}\left(J_{k}\right).
\end{align*}
where we used $\ensuremath{\psi(z+1)=\psi(z)+\frac{1}{z}}$ in the
last identity. Thus, we have
\begin{align*}
 & \mathbb{E}\left[\left[W_{k}+\log\left(1+\tfrac{X_{k}^{\prime}X_{k}}{\nu_{k}}\right)\right]\sum_{l=1}^{K}W_{l}{\rm vec}\left(e_{l}X_{l}X^{\prime}\right)\right]\\
= & \left(\phi_{k}-1\right){\rm vec}\left(J_{k}\right)+{\rm vec}\left(I_{n}\right)+\left[\psi\left(\tfrac{v_{k}+n_{k}}{2}\right)-\psi\left(\tfrac{v_{k}}{2}\right)\right]{\rm vec}\left(I_{n}\right)+\tfrac{2}{v_{k}+n_{k}}{\rm vec}\left(J_{k}\right)\\
= & \left(\phi_{k}-1+\tfrac{2}{v_{k}+n_{k}}\right){\rm vec}\left(J_{k}\right)+\left[\psi\left(\tfrac{v_{k}+n_{k}}{2}\right)-\psi\left(\tfrac{v_{k}}{2}\right)+1\right]{\rm vec}\left(I_{n}\right),
\end{align*}
where the last term (scaled by $-\frac{1}{2}$) is cancelled by the
last term in (\ref{eq:somecrossterm}), and we have
\begin{align*}
\mathbb{E}\left[\nabla_{\nu_{k}}\nabla_{\tilde{\Xi}}^{s}\right] & =-\tfrac{1}{2}\left(\phi_{k}-1+\tfrac{2}{v_{k}+n_{k}}\right){\rm vec}\left(J_{k}\right)\\
 & =-\tfrac{1}{2}\left(\tfrac{4}{\left(v_{k}+n_{k}\right)\left(v_{k}+n_{k}+2\right)}\right){\rm vec}\left(J_{k}\right)\\
 & =-\tfrac{2}{\left(v_{k}+n_{k}\right)\left(v_{k}+n_{k}+2\right)}{\rm vec}\left(J_{k}\right).
\end{align*}
It follows that
\begin{align*}
\mathbb{E}\left[\nabla_{\nu_{k}}\nabla_{\tilde{\Xi}}\right] & =-\tfrac{2}{\left(v_{k}+n_{k}\right)\left(v_{k}+n_{k}+2\right)}M{\rm vec}\left(J_{k}\right)=-\tfrac{2}{\left(v_{k}+n_{k}\right)\left(v_{k}+n_{k}+2\right)}{\rm vec}\left(A^{\prime}J_{k}\right).
\end{align*}
$\square$

\noindent \textbf{Proof of Theorem \ref{thm:ExpecScoreHess} (Expected
Hessian matrix).}

\subsubsection*{The terms: $\mathbb{E}(\nabla_{\mu\mu^{\prime}})$, $\mathbb{E}(\nabla_{\mu\tilde{\Xi}^{\prime}})$,
$\mathbb{E}(\nabla_{\mu\nu_{k}})$}

\begin{align*}
\mathbb{E}[\nabla_{\mu\mu^{\prime}}] & =\sum_{k=1}^{K}\tfrac{2}{\nu_{k}+n_{k}}A^{\prime}e_{k}\mathbb{E}[W_{k}^{2}X_{k}X_{k}^{\prime}]e_{k}^{\prime}A-\mathbb{E}[W_{k}]A^{\prime}J_{k}A\\
 & =\sum_{k=1}^{K}\left(\tfrac{2}{\nu_{k}+n_{k}}\phi_{k}-1\right)A^{\prime}e_{k}e_{k}^{\prime}A-A^{\prime}J_{k}A\\
 & =-\sum_{k=1}^{K}\phi_{k}A^{\prime}J_{k}A.
\end{align*}
The expectations of $\nabla_{\mu\tilde{\Xi}^{\prime}}$ and $\nabla_{\mu\nu_{k}}$
are both zeros, because they are odd functions of $X$.

\subsubsection*{The term: $\mathbb{E}(\nabla_{\Xi\Xi^{\prime}})$}

Let $\nabla_{\tilde{\Xi}\tilde{\Xi}^{\prime}}=\left(I_{n}\otimes A^{\prime}\right)\nabla_{\tilde{\Xi}\tilde{\Xi}^{\prime}}^{s}\left(I_{n}\otimes A\right)$,
so we have
\begin{align*}
\nabla_{\tilde{\Xi}\tilde{\Xi}^{\prime}}^{s} & =-K_{n}+\sum_{k=1}^{K}\tfrac{2W_{k}^{2}}{\nu_{k}+n_{k}}{\rm vec}\left(e_{k}X_{k}X^{\prime}\right){\rm vec}\left(e_{k}X_{k}X^{\prime}\right)^{\prime}\\
 & \quad+\sum_{k=1}^{K}W_{k}\left[\left(XX_{k}^{\prime}e_{k}^{\prime}\otimes I_{n}\right)K_{n}-\left(XX^{\prime}\otimes J_{k}\right)-\left(I_{n}\otimes e_{k}X_{k}X^{\prime}\right)K_{n}\right].
\end{align*}
From (\ref{eq:phikk}) we have
\[
\mathbb{E}\left[{\rm vec}\left(e_{k}X_{k}X^{\prime}\right){\rm vec}\left(e_{k}X_{k}X^{\prime}\right)^{\prime}\right]=\phi_{k}\left[J_{k}^{\bullet}\otimes J_{k}+J_{k\otimes}K_{n}+{\rm vec}(J_{k}){\rm vec}(J_{k})^{\prime}\right],
\]
where $J_{k}^{\bullet}=\sum_{p\neq k}\xi_{p}J_{p}+J_{k}$. Next,
\begin{align*}
\mathbb{E}\left[W_{k}\left(XX_{k}^{\prime}e_{k}^{\prime}\otimes I_{n}\right)K_{n}\right] & =\mathbb{E}[\sum_{p}W_{k}\left(e_{p}X_{p}X_{k}^{\prime}e_{k}^{\prime}\otimes I_{n}\right)K_{n}]\\
 & =\mathbb{E}[W_{k}\left(e_{k}X_{k}X_{k}^{\prime}e_{k}^{\prime}\otimes I_{n}\right)K_{n}]\\
 & =\left(e_{k}e_{k}^{\prime}\otimes I_{n}\right)K_{n},\qquad(\text{using }\mathbb{E}[W_{k}X_{k}X_{k}^{\prime}]=I_{n_{k}})\\
 & =\left(J_{k}\otimes I_{n}\right)K_{n},
\end{align*}
similarly
\[
\mathbb{E}\left[W_{k}\left(I_{n}\otimes e_{k}X_{k}X^{\prime}\right)K_{n}\right]K_{n}=\left(I_{n}\otimes J_{k}\right)K_{n},
\]
and
\begin{align*}
\mathbb{E}\left[W_{k}XX^{\prime}\right] & =\mathbb{E}\left[\sum_{l}\sum_{p}W_{k}e_{l}X_{l}X_{p}^{\prime}e_{p}^{\prime}\right]\\
 & =\mathbb{E}\left[\sum_{l\neq k}W_{k}e_{l}X_{l}X_{l}^{\prime}e_{l}^{\prime}+W_{k}e_{k}X_{k}X_{k}^{\prime}e_{k}^{\prime}\right]\\
 & =\sum_{l\neq k}e_{l}\mathbb{E}\left[X_{l}X_{l}^{\prime}\right]e_{l}^{\prime}+e_{k}\mathbb{E}\left[W_{k}X_{k}X_{k}^{\prime}\right]e_{k}^{\prime}\\
 & =J_{k}+\sum_{l\neq k}\tfrac{\nu_{l}}{\nu_{l}-2}J_{l}=J_{k}^{\bullet},
\end{align*}
leads to $\mathbb{E}\left[W_{k}\left(XX^{\prime}\otimes J_{k}\right)\right]=\left(J_{k}^{\bullet}\otimes J_{k}\right)$.
Combining the terms we have 
\begin{align*}
\mathbb{E}\left(\nabla_{\tilde{\Xi}\tilde{\Xi}^{\prime}}^{s}\right) & =-K_{n}+\sum_{k=1}^{K}\varphi_{k},
\end{align*}
where
\begin{align*}
\varphi_{k}= & \tfrac{2}{\nu_{k}+n_{k}}\phi_{k}\left[J_{k}^{\bullet}\otimes J_{k}+J_{k\otimes}K_{n}+{\rm vec}(J_{k}){\rm vec}(J_{k})^{\prime}\right]\\
 & +\left(J_{k}\otimes I_{n}\right)K_{n}-\left(J_{k}^{\bullet}\otimes J_{k}\right)-\left(I_{n}\otimes J_{k}\right)K_{n}\\
= & \left(\tfrac{2}{\nu_{k}+n_{k}+2}-1\right)J_{k}^{\bullet}\otimes J_{k}+\tfrac{2}{\nu_{k}+n_{k}+2}\left[J_{k\otimes}K_{n}+{\rm vec}(J_{k}){\rm vec}(J_{k})^{\prime}\right]+\left[\left(J_{k}\otimes I_{n}\right)-\left(I_{n}\otimes J_{k}\right)\right]K_{n}\\
= & -\phi_{k}J_{k}^{\bullet}\otimes J_{k}+\left(1-\phi_{k}\right)\left[J_{k\otimes}K_{n}+{\rm vec}(J_{k}){\rm vec}(J_{k})^{\prime}\right]+\left(J_{k}\otimes I_{n}\right)K_{n}-\left(I_{n}\otimes J_{k}\right)K_{n},
\end{align*}
such that
\begin{align*}
\mathbb{E}\left(\nabla_{\tilde{\Xi}\tilde{\Xi}^{\prime}}^{s}\right) & =-K_{n}+\sum_{k=1}^{K}\varphi_{k}=-K_{n}-\sum_{k=1}^{K}\phi_{k}J_{k}^{\bullet}\otimes J_{k}+\left(\phi_{k}-1\right)\left[J_{k\otimes}K_{n}+{\rm vec}(J_{k}){\rm vec}(J_{k})^{\prime}\right],
\end{align*}
where we used $\sum_{k=1}^{K}\left(J_{k}\otimes I_{n}\right)K_{n}-\left(I_{n}\otimes J_{k}\right)K_{n}=0$.
Finally, we have shown that
\begin{align*}
\mathbb{E}\left(\nabla_{\tilde{\Xi}\tilde{\Xi}^{\prime}}\right) & =\left(I_{n}\otimes A^{\prime}\right)\mathbb{E}\left(\nabla_{\Xi\Xi^{\prime}}^{s}\right)\left(I_{n}\otimes A\right)\\
 & =-K_{n}-\sum_{k=1}^{K}\phi_{k}J_{k}^{\bullet}\otimes A^{\prime}J_{k}A-\left(\phi_{k}-1\right)\left[\left(J_{k}A\otimes A^{\prime}J_{k}\right)K_{n}+{\rm vec}\left(A^{\prime}J_{k}\right){\rm vec}\left(A^{\prime}J_{k}\right)^{\prime}\right],
\end{align*}
which proves that $\mathbb{E}\left(\nabla_{\tilde{\Xi}\tilde{\Xi}^{\prime}}\right)=-\mathcal{I}_{\tilde{\Xi}}$.

\subsubsection*{The term: $\mathbb{E}(\nabla_{\tilde{\Xi}\nu^{\prime}})$}

We have
\[
\mathbb{E}[\nabla_{\tilde{\Xi}\nu_{k}}]=\tfrac{1}{\nu_{k}+n_{k}}\left[\mathbb{E}\left[W_{k}{\rm vec}\left(A^{\prime}e_{k}X_{k}X^{\prime}\right)\right]-\mathbb{E}\left[W_{k}^{2}{\rm vec}\left(A^{\prime}e_{k}X_{k}X^{\prime}\right)\right]\right],
\]
where
\begin{align*}
\mathbb{E}\left[W_{k}{\rm vec}\left(A^{\prime}e_{k}X_{k}X^{\prime}\right)\right] & =\sum_{l=1}^{K}\mathbb{E}\left[W_{k}{\rm vec}\left(A^{\prime}e_{k}X_{k}X_{l}^{\prime}e_{l}^{\prime}\right)\right]\\
 & =\mathbb{E}\left[W_{k}{\rm vec}\left(A^{\prime}e_{k}X_{k}X_{k}^{\prime}e_{k}^{\prime}\right)\right]\\
 & ={\rm vec}\left(A^{\prime}J_{k}\right),
\end{align*}
and
\begin{align*}
\mathbb{E}\left[W_{k}^{2}{\rm vec}\left(A^{\prime}e_{k}X_{k}X^{\prime}\right)\right] & =\sum_{l=1}^{K}\mathbb{E}\left[W_{k}^{2}{\rm vec}\left(A^{\prime}e_{k}X_{k}X_{l}^{\prime}e_{l}^{\prime}\right)\right]\\
 & =\mathbb{E}\left[W_{k}^{2}{\rm vec}\left(A^{\prime}e_{k}X_{k}X_{k}^{\prime}e_{k}^{\prime}\right)\right]\\
 & =\phi_{k}{\rm vec}\left(A^{\prime}J_{k}\right).
\end{align*}
So, we have
\begin{align*}
\mathbb{E}[\nabla_{\tilde{\Xi}\nu_{k}}] & =\tfrac{1}{\nu_{k}+n_{k}}\left[{\rm vec}\left(A^{\prime}J_{k}\right)-\phi_{k}{\rm vec}\left(A^{\prime}J_{k}\right)\right]\\
 & =\tfrac{1}{\left(\nu_{k}+n_{k}+2\right)\left(\nu_{k}+n_{k}\right)}{\rm vec}\left(A^{\prime}J_{k}\right),
\end{align*}
which equals $-\mathcal{I}_{\tilde{\Xi}\nu_{k}}$.

\subsubsection*{The term: $\mathbb{E}(\nabla_{\nu\nu^{\prime}})$}

Using $\mathbb{E}\left[W_{k}^{2}\right]=\ensuremath{\tfrac{\left(\nu_{k}+2\right)(\nu_{k}+n_{k})}{\nu_{k}(\nu_{k}+n_{k}+2)}}$
and $\mathbb{E}\left[W_{k}\right]=1$, we have
\begin{eqnarray*}
\nabla_{\nu_{k}\nu_{k}} & = & \tfrac{1}{4}\psi^{\prime}\left(\tfrac{v_{k}+n_{k}}{2}\right)-\tfrac{1}{4}\psi^{\prime}\left(\tfrac{v_{k}}{2}\right)+\tfrac{1}{2\nu_{k}}+\tfrac{1}{2}\tfrac{1}{\nu_{k}+n_{k}}W_{k}^{2}-\tfrac{1}{\left(\nu_{k}+n_{k}\right)}W_{k}\\
 & = & \tfrac{1}{2}\left[\tfrac{1}{2}\psi^{\prime}\left(\tfrac{v_{k}+n_{k}}{2}\right)-\tfrac{1}{2}\psi^{\prime}\left(\tfrac{v_{k}}{2}\right)+\tfrac{1}{\nu_{k}+n_{k}}\left(W_{k}^{2}-W_{k}\right)+\tfrac{1}{\nu_{k}}\left[1-\tfrac{\nu_{k}}{\nu_{k}+X_{k}^{\prime}X_{k}}\right]\right],\\
 & = & \tfrac{1}{4}\psi^{\prime}\left(\tfrac{v_{k}+n_{k}}{2}\right)-\tfrac{1}{4}\psi^{\prime}\left(\tfrac{v_{k}}{2}\right)+\tfrac{1}{2\nu_{k}}+\tfrac{1}{2}\tfrac{1}{\nu_{k}+n_{k}}\ensuremath{\tfrac{\left(\nu_{k}+2\right)(\nu_{k}+n_{k})}{\nu_{k}(\nu_{k}+n_{k}+2)}}-\tfrac{1}{\left(\nu_{k}+n_{k}\right)}\\
 & = & \tfrac{1}{4}\psi^{\prime}\left(\tfrac{v_{k}+n_{k}}{2}\right)-\tfrac{1}{4}\psi^{\prime}\left(\tfrac{v_{k}}{2}\right)+\tfrac{1}{2}\left[\tfrac{1}{\nu_{k}}+\ensuremath{\tfrac{\nu_{k}+2}{\nu_{k}(\nu_{k}+n_{k}+2)}}-\tfrac{2}{\nu_{k}+n_{k}}\right]\\
 & = & \tfrac{1}{4}\psi^{\prime}\left(\tfrac{v_{k}+n_{k}}{2}\right)-\tfrac{1}{4}\psi^{\prime}\left(\tfrac{v_{k}}{2}\right)+\tfrac{1}{2}\left[\ensuremath{\tfrac{n_{k}\left(\nu_{k}+n_{k}+4\right)}{\nu_{k}(\nu_{k}+n_{k}+2)\left(\nu_{k}+n_{k}\right)}}\right],
\end{eqnarray*}
which equals $-\mathcal{I}_{\nu_{k}}$.

\hfill{}$\square$

\noindent \textbf{Proof of Lemma \ref{lem:JacobianMatrix}.} First,
we can express the $\Xi$ by using the cluster/group structure $\tilde{\Xi}=\left[\tilde{\Xi}_{1},\tilde{\Xi}_{2},\ldots,\tilde{\Xi}_{K}\right]$
with $\tilde{\Xi}_{k}\in\mathbb{R}^{n\times n_{k}}$, and $\tilde{\Xi}_{k}=\left[\tilde{\Xi}_{1k}^{\prime},\tilde{\Xi}_{2k}^{\prime},\ldots,\tilde{\Xi}_{Kk}^{\prime}\right]^{\prime}$
with $\tilde{\Xi}_{lk}\in\mathbb{R}^{n_{l}\times n_{k}}$. Then, we
have
\[
\frac{\partial{\rm vec}\left(\Xi\right)}{\partial{\rm vec}\left(\tilde{\Xi}\right)^{\prime}}=\left[\begin{array}{cccc}
\Gamma_{11} & \Gamma_{12} & \ldots & \Gamma_{1K}\\
\Gamma_{21} & \Gamma_{22} &  & \Gamma_{2K}\\
\vdots & \vdots & \ddots & \vdots\\
\Gamma_{K1} & \Gamma_{K2} &  & \Gamma_{KK}
\end{array}\right],\quad{\rm where}\quad\Gamma_{kl}=\frac{\partial{\rm vec}\left(\Xi_{k}\right)}{\partial{\rm vec}\left(\tilde{\Xi}_{l}\right)^{\prime}}\in\mathbb{R}^{nn_{k}\times nn_{k}}.
\]
Because $\Xi_{k}=\tilde{\Xi}_{k}P_{kk}$ and $P_{kk}$ only depends
on $\tilde{\Xi}_{kk}$, we have $\Gamma_{kl}=0$ if $k\neq l$. Then,
we have
\[
\Gamma_{kk}=\frac{\partial{\rm vec}\left(\Xi_{k}\right)}{\partial{\rm vec}\left(\Xi_{k}^{\prime}\right)^{\prime}}\frac{\partial{\rm vec}\left(\Xi_{k}^{\prime}\right)}{\partial{\rm vec}\left(\tilde{\Xi}_{k}^{\prime}\right)^{\prime}}\frac{\partial{\rm vec}\left(\tilde{\Xi}_{k}^{\prime}\right)}{\partial{\rm vec}\left(\tilde{\Xi}_{k}\right)^{\prime}}=K_{n_{k},n}\frac{\partial{\rm vec}\left(\Xi_{k}^{\prime}\right)}{\partial{\rm vec}\left(\tilde{\Xi}_{k}^{\prime}\right)^{\prime}}K_{n,n_{k}}
\]
and we have
\[
\frac{\partial{\rm vec}\left(\Xi_{k}^{\prime}\right)}{\partial{\rm vec}\left(\tilde{\Xi}_{k}^{\prime}\right)^{\prime}}=\left[\begin{array}{cccc}
\Pi_{11} & \Pi_{12} & \ldots & \Pi_{1K}\\
\Pi_{21} & \Pi_{22} &  & \Pi_{2K}\\
\vdots & \vdots & \ddots & \vdots\\
\Pi_{K1} & \Pi_{K2} & \cdots & \Pi_{KK}
\end{array}\right],
\]
and we have
\[
\Pi_{ij}=\frac{\partial{\rm vec}\left(\Xi_{ik}^{\prime}\right)}{\partial{\rm vec}\left(\tilde{\Xi}_{jk}^{\prime}\right)^{\prime}}=\frac{\partial{\rm vec}\left(P_{kk}^{\prime}\tilde{\Xi}_{ik}^{\prime}\right)}{\partial{\rm vec}\left(\tilde{\Xi}_{jk}^{\prime}\right)^{\prime}}=\begin{cases}
\bm{0} & i\neq j,j\neq k\\
I_{n_{i}}\otimes P_{kk}^{\prime} & i=j\neq k\\
\left(\tilde{\Xi}_{ik}\otimes I_{n_{k}}\right)\frac{\partial{\rm vec}\left(P_{kk}^{\prime}\right)}{\partial{\rm vec}\left(\tilde{\Xi}_{kk}^{\prime}\right)^{\prime}} & i\neq j,j=k\\
\left(I_{n_{i}}\otimes P_{kk}^{\prime}\right)+\left(\tilde{\Xi}_{ik}\otimes I_{n_{k}}\right)\frac{\partial{\rm vec}\left(P_{kk}^{\prime}\right)}{\partial{\rm vec}\left(\tilde{\Xi}_{kk}^{\prime}\right)^{\prime}} & i=j=k.
\end{cases}
\]
Because we have
\begin{align*}
\mathrm{d}{\rm vec}\left(P_{kk}^{\prime}\tilde{\Xi}_{ik}^{\prime}\right) & =\left(\tilde{\Xi}_{ik}\otimes I_{n_{k}}\right)d{\rm vec}\left(P_{kk}^{\prime}\right)+\left(I_{n_{i}}\otimes P_{kk}^{\prime}\right)d{\rm vec}\left(\tilde{\Xi}_{ik}^{\prime}\right),
\end{align*}
and 
\begin{align*}
 & \frac{\partial{\rm vec}\left(P_{kk}\right)}{\partial{\rm vec}\left(\tilde{\Xi}_{kk}\right)^{\prime}}=\frac{\partial{\rm vec}\left(\tilde{\Xi}_{kk}^{\prime}(\tilde{\Xi}_{kk}\tilde{\Xi}_{kk}^{\prime})^{-\frac{1}{2}}\right)}{\partial{\rm vec}\left(\tilde{\Xi}_{kk}\right)^{\prime}}\\
= & \ensuremath{\left((\tilde{\Xi}_{kk}\tilde{\Xi}_{kk}^{\prime})^{-\frac{1}{2}}\otimes I_{n_{k}}\right)\frac{\partial{\rm vec}\left(\tilde{\Xi}_{kk}^{\prime}\right)}{\partial{\rm vec}\left(\tilde{\Xi}_{kk}\right)^{\prime}}+\left(I_{n_{k}}\otimes\tilde{\Xi}_{kk}^{\prime}\right)\frac{\partial{\rm vec}\left((\tilde{\Xi}_{kk}\tilde{\Xi}_{kk}^{\prime})^{-\frac{1}{2}}\right)}{\partial{\rm vec}\left(\tilde{\Xi}_{kk}\right)^{\prime}}}\\
= & \left((\tilde{\Xi}_{kk}\tilde{\Xi}_{kk}^{\prime})^{-\frac{1}{2}}\otimes I_{n_{k}}\right)K_{n_{k}}+\left(I_{n_{k}}\otimes\tilde{\Xi}_{kk}^{\prime}\right)\frac{\partial{\rm vec}\left((\tilde{\Xi}_{kk}\tilde{\Xi}_{kk}^{\prime})^{-\frac{1}{2}}\right)}{\partial{\rm vec}\left(\tilde{\Xi}_{kk}\right)^{\prime}}\\
= & \left(I_{n_{k}}\otimes(\tilde{\Xi}_{kk}\tilde{\Xi}_{kk}^{\prime})^{-\frac{1}{2}}\right)K_{n_{k}}+\left(I_{n_{k}}\otimes\tilde{\Xi}_{kk}^{\prime}\right)\frac{\partial{\rm vec}\left((\tilde{\Xi}_{kk}\tilde{\Xi}_{kk}^{\prime})^{-\frac{1}{2}}\right)}{\partial{\rm vec}\left(\left(\tilde{\Xi}_{kk}\tilde{\Xi}_{kk}^{\prime}\right)^{-1}\right)^{\prime}}\frac{\partial{\rm vec}\left(\left(\tilde{\Xi}_{kk}\tilde{\Xi}_{kk}^{\prime}\right)^{-1}\right)^{\prime}}{\partial{\rm vec}\left(\tilde{\Xi}_{kk}\tilde{\Xi}_{kk}^{\prime}\right)^{\prime}}\frac{\partial{\rm vec}\left(\tilde{\Xi}_{kk}\tilde{\Xi}_{kk}^{\prime}\right)^{\prime}}{\partial{\rm vec}\left(\tilde{\Xi}_{kk}\right)^{\prime}}\\
= & \left((\tilde{\Xi}_{kk}\tilde{\Xi}_{kk}^{\prime})^{-\frac{1}{2}}\otimes I_{n_{k}}\right)K_{n_{k}}-\left\{ \left(I_{n_{k}}\otimes\tilde{\Xi}_{kk}^{\prime}\right)\ensuremath{\left((\tilde{\Xi}_{kk}\tilde{\Xi}_{kk}^{\prime})^{-\frac{1}{2}}\oplus(\tilde{\Xi}_{kk}\tilde{\Xi}_{kk}^{\prime})^{-\frac{1}{2}}\right)^{-1}}\right.\\
 & \quad\left.\times\left[\left(\tilde{\Xi}_{kk}\tilde{\Xi}_{kk}^{\prime}\right)^{-1}\otimes\left(\tilde{\Xi}_{kk}\tilde{\Xi}_{kk}^{\prime}\right)^{-1}\right]\left(I_{n_{k}^{2}}+K_{n_{k}}\right)\left(\tilde{\Xi}_{kk}\otimes I_{n_{k}}\right)\right\} .
\end{align*}
Next, from 
\begin{align*}
 & \ensuremath{\left((\tilde{\Xi}_{kk}\tilde{\Xi}_{kk}^{\prime})^{-\frac{1}{2}}\oplus(\tilde{\Xi}_{kk}\tilde{\Xi}_{kk}^{\prime})^{-\frac{1}{2}}\right)^{-1}}\left[\left(\tilde{\Xi}_{kk}\tilde{\Xi}_{kk}^{\prime}\right)^{-1}\otimes\left(\tilde{\Xi}_{kk}\tilde{\Xi}_{kk}^{\prime}\right)^{-1}\right]\\
= & \left[\left(\tilde{\Xi}_{kk}\tilde{\Xi}_{kk}^{\prime}\otimes\tilde{\Xi}_{kk}\tilde{\Xi}_{kk}^{\prime}\right)\left[(\tilde{\Xi}_{kk}\tilde{\Xi}_{kk}^{\prime})^{-\frac{1}{2}}\oplus(\tilde{\Xi}_{kk}\tilde{\Xi}_{kk}^{\prime})^{-\frac{1}{2}}\right]\right]^{-1}\\
= & \left[\left(\tilde{\Xi}_{kk}\tilde{\Xi}_{kk}^{\prime}\otimes\tilde{\Xi}_{kk}\tilde{\Xi}_{kk}^{\prime}\right)\left[(\tilde{\Xi}_{kk}\tilde{\Xi}_{kk}^{\prime})^{-\frac{1}{2}}\otimes I+I\otimes(\tilde{\Xi}_{kk}\tilde{\Xi}_{kk}^{\prime})^{-\frac{1}{2}}\right]\right]^{-1}\\
= & \left[(\tilde{\Xi}_{kk}\tilde{\Xi}_{kk}^{\prime})^{\frac{1}{2}}\otimes\tilde{\Xi}_{kk}\tilde{\Xi}_{kk}^{\prime}+\tilde{\Xi}_{kk}\tilde{\Xi}_{kk}^{\prime}\otimes(\tilde{\Xi}_{kk}\tilde{\Xi}_{kk}^{\prime})^{\frac{1}{2}}\right]^{-1}\\
= & \left[(\tilde{\Xi}_{kk}\tilde{\Xi}_{kk}^{\prime})^{\frac{1}{2}}\oplus\tilde{\Xi}_{kk}\tilde{\Xi}_{kk}^{\prime}\right]^{-1},
\end{align*}
we find
\begin{align*}
 & \left(I_{n_{k}}\otimes\tilde{\Xi}_{kk}^{\prime}\right)\left[(\tilde{\Xi}_{kk}\tilde{\Xi}_{kk}^{\prime})^{\frac{1}{2}}\oplus\tilde{\Xi}_{kk}\tilde{\Xi}_{kk}^{\prime}\right]^{-1}\\
= & \left[\left((\tilde{\Xi}_{kk}\tilde{\Xi}_{kk}^{\prime})^{\frac{1}{2}}\otimes\tilde{\Xi}_{kk}\tilde{\Xi}_{kk}^{\prime}+\tilde{\Xi}_{kk}\tilde{\Xi}_{kk}^{\prime}\otimes(\tilde{\Xi}_{kk}\tilde{\Xi}_{kk}^{\prime})^{\frac{1}{2}}\right)\left(I_{n_{k}}\otimes\tilde{\Xi}_{kk}^{\prime-1}\right)\right]^{-1}\\
= & \left[\left((\tilde{\Xi}_{kk}\tilde{\Xi}_{kk}^{\prime})^{\frac{1}{2}}\otimes\tilde{\Xi}_{kk}+\tilde{\Xi}_{kk}\tilde{\Xi}_{kk}^{\prime}\otimes(\tilde{\Xi}_{kk}\tilde{\Xi}_{kk}^{\prime})^{\frac{1}{2}}\tilde{\Xi}_{kk}^{\prime-1}\right)\right]^{-1}\\
= & \left[(\tilde{\Xi}_{kk}\tilde{\Xi}_{kk}^{\prime})^{\frac{1}{2}}\otimes\tilde{\Xi}_{kk}+\tilde{\Xi}_{kk}\tilde{\Xi}_{kk}^{\prime}\otimes P_{kk}^{-1}\right]^{-1}\\
= & \left[P_{kk}^{-1}\tilde{\Xi}_{kk}^{\prime}\otimes\tilde{\Xi}_{kk}+\tilde{\Xi}_{kk}\tilde{\Xi}_{kk}^{\prime}\otimes P_{kk}^{-1}\right]^{-1}\\
= & \left[I_{n_{k}}\otimes\tilde{\Xi}_{kk}+\tilde{\Xi}_{kk}P_{kk}\otimes P_{kk}^{-1}\left(P_{kk}^{-1}\tilde{\Xi}_{kk}^{\prime}\otimes I_{n_{k}}\right)\right]^{-1}\\
= & \left(\tilde{\Xi}_{kk}^{\prime-1}P_{kk}\otimes I_{n_{k}}\right)\left[I_{n_{k}}\otimes\tilde{\Xi}_{kk}+\tilde{\Xi}_{kk}P_{kk}\otimes P_{kk}^{-1}\right]^{-1}.
\end{align*}
Thus, we have
\begin{align*}
\frac{\partial{\rm vec}\left(P_{kk}\right)}{\partial{\rm vec}\left(\tilde{\Xi}_{kk}\right)^{\prime}} & =\left((\tilde{\Xi}_{kk}\tilde{\Xi}_{kk}^{\prime})^{-\frac{1}{2}}\otimes I_{n_{k}}\right)K_{n_{k}}-\left[P_{kk}^{-1}\tilde{\Xi}_{kk}^{\prime}\otimes\tilde{\Xi}_{kk}+\tilde{\Xi}_{kk}\tilde{\Xi}_{kk}^{\prime}\otimes P_{kk}^{-1}\right]^{-1}\left(I_{n_{k}^{2}}+K_{n_{k}}\right)\left(\tilde{\Xi}_{kk}\otimes I_{n_{k}}\right)\\
 & =\left(\tilde{\Xi}_{kk}^{\prime-1}P_{kk}\otimes I_{n_{k}}\right)K_{n_{k}}-\left[P_{kk}^{-1}\tilde{\Xi}_{kk}^{\prime}\otimes\tilde{\Xi}_{kk}+\tilde{\Xi}_{kk}\tilde{\Xi}_{kk}^{\prime}\otimes P_{kk}^{-1}\right]^{-1}\left(I_{n_{k}^{2}}+K_{n_{k}}\right)\left(\tilde{\Xi}_{kk}\otimes I_{n_{k}}\right)\\
 & =\left(\tilde{\Xi}_{kk}^{\prime-1}P_{kk}\otimes I_{n_{k}}\right)\left[K_{n_{k}}-\left[I_{n_{k}}\otimes\tilde{\Xi}_{kk}+\tilde{\Xi}_{kk}P_{kk}\otimes P_{kk}^{-1}\right]^{-1}\left(I_{n_{k}^{2}}+K_{n_{k}}\right)\left(\tilde{\Xi}_{kk}\otimes I_{n_{k}}\right)\right].
\end{align*}
Finally, we have
\begin{align*}
\frac{\partial{\rm vec}\left(P_{kk}^{\prime}\right)}{\partial{\rm vec}\left(\tilde{\Xi}_{kk}^{\prime}\right)^{\prime}} & =\frac{\partial{\rm vec}\left(P_{kk}^{\prime}\right)}{\partial{\rm vec}\left(P_{kk}\right)^{\prime}}\frac{\partial{\rm vec}\left(P_{kk}\right)}{\partial{\rm vec}\left(\tilde{\Xi}_{kk}\right)^{\prime}}\frac{\partial{\rm vec}\left(\tilde{\Xi}_{kk}\right)}{\partial{\rm vec}\left(\tilde{\Xi}_{kk}^{\prime}\right)^{\prime}}=K_{n_{k}}\frac{\partial{\rm vec}\left(P_{kk}\right)}{\partial{\rm vec}\left(\tilde{\Xi}_{kk}\right)^{\prime}}K_{n_{k}}\\
 & =K_{n_{k}}\left(\tilde{\Xi}_{kk}^{\prime-1}P_{kk}\otimes I_{n_{k}}\right)\\
 & \quad\times\left[I_{n_{k}^{2}}-\left[I_{n_{k}}\otimes\tilde{\Xi}_{kk}+\tilde{\Xi}_{kk}P_{kk}\otimes P_{kk}^{-1}\right]^{-1}\left(I_{n_{k}^{2}}+K_{n_{k}}\right)\left(I_{n_{k}}\otimes\tilde{\Xi}_{kk}\right)\right].
\end{align*}
$\square$

\noindent \textbf{Proof of Theorem \ref{thm:MLE-consistent-asN}.}
Part (i). The consistency of the MLE, $\hat{\theta}_{T}$, follows
by verifying the four conditions of \citet[theorem 2.5]{NeweyMcFadden:1994}.
(i) Our identification assumptions ensures that $\theta\neq\theta_{0}$
implies $f\left(y_{i}|\theta\right)\neq f\left(y_{i}|\theta_{0}\right)$.
(ii) The compactness condition is Assumption \ref{assu:Compact}.
(iii) the log-likelihood function 
\begin{equation}
\log f(Y|\theta)=-\tfrac{1}{2}\log|\Xi\Xi^{\prime}|+\sum_{k=1}^{K}c_{k}-\tfrac{\nu_{k}+n_{k}}{2}\log\left[1+\tfrac{1}{\nu_{k}}(Y-\mu)^{\prime}A^{\prime}e_{k}e_{k}^{\prime}A(Y-\mu)\right],\label{eq:LLequ}
\end{equation}
where $A=\Xi^{-1}$ is continuous for $\theta\in\Theta$ with probability
one. What remains to verified is the dominance condition: (iv) $\mathbb{E}_{\theta_{0}}\left[\sup_{\theta\in\Theta}|\log f(y|\theta)|\right]<\infty$.

First note that
\[
\left|\log f(Y|\theta)\right|\leq\left|\tfrac{1}{2}\log|\Xi\Xi^{\prime}|\right|+\sum_{k=1}^{K}|c_{k}|+R_{k}(\theta)\leq C_{0}+\sum_{k=1}^{K}R_{k}(\theta),
\]
where $C_{0}=\sup_{\theta\in\Theta}\left|\tfrac{1}{2}\log|\Xi\Xi^{\prime}|\right|+\sum_{k=1}^{K}|c_{k}|<\infty$
(since $\log|\Xi\Xi^{\prime}|$ is continuous and $\Theta$ is compact)
and
\begin{eqnarray*}
R_{k}(\theta) & = & \tfrac{\nu_{k}+n_{k}}{2}\log\left[1+\tfrac{1}{\nu_{k}}(Y-\mu)^{\prime}A^{\prime}e_{k}e_{k}^{\prime}A(Y-\mu)\right].\\
 & = & \tfrac{\nu_{k}+n_{k}}{2}\log\left[1+\tfrac{1}{\nu_{k}}(X-\xi)^{\prime}\Psi(X-\xi)\right]\\
 & \leq & \tfrac{\nu_{k}+n_{k}}{2}\log\left[1+\tfrac{\rho(\Psi)}{\nu_{k}}(X-\xi)^{\prime}(X-\xi)\right],
\end{eqnarray*}
where we substituted, $Y=\mu_{0}+\Xi_{0}X$, such that
\[
(Y-\mu)=(\mu_{0}+\Xi_{0}X-\mu)=\Xi_{0}(X-A_{0}(\mu-\mu_{0}))=\Xi_{0}(X-\xi),
\]
with $\xi=A_{0}(\mu-\mu_{0})$, $A_{0}=\Xi_{0}^{-1}$, $\Psi_{k}=\Xi_{0}^{\prime}A^{\prime}e_{k}e_{k}^{\prime}A\Xi_{0}$,
and $\rho(\Psi_{k})$ denoting the largest eigenvalue of $\Psi_{k}$. 

Next we add the positive term $\tfrac{\rho(\Psi_{k})}{\nu_{k}}(X+\xi)^{\prime}(X+\xi)$
\begin{eqnarray*}
R_{k}(\theta) & \leq & \tfrac{\nu_{k}+n_{k}}{2}\log\left[1+\tfrac{\rho(\Psi_{k})}{\nu_{k}}(X-\xi)^{\prime}(X-\xi)+\tfrac{\rho(\Psi_{k})}{\nu_{k}}(X+\xi)^{\prime}(X+\xi)\right]\\
 & = & \tfrac{\nu_{k}+n_{k}}{2}\log\left[1+2\tfrac{\rho(\Psi_{k})}{\nu_{k}}X^{\prime}X+2\tfrac{\rho(\Psi_{k})}{\nu_{k}}\xi^{\prime}\xi\right]\\
 & = & \tfrac{\nu_{k}+n_{k}}{2}\log\left[1+2\tfrac{\rho(\Psi_{k})}{\nu_{k}}\xi^{\prime}\xi+\sum_{l=1}^{K}2\tfrac{\rho(\Psi_{k})}{\nu_{k}}X_{l}^{\prime}X_{l}\right]\\
 & \leq & \tfrac{\nu_{k}+n_{k}}{2}\log\left[1+2\tfrac{\rho(\Psi_{k})}{\nu_{k}}\xi^{\prime}\xi\right]+\tfrac{\nu_{k}+n_{k}}{2}\sum_{l=1}^{K}\log\left[1+2\tfrac{\rho(\Psi_{k})}{\nu_{k}}X_{l}^{\prime}X_{l}\right],
\end{eqnarray*}
where we used $\log(1+a_{1}+\cdots+a_{M})\leq\sum_{j=1}^{M}\log(1+a_{j})$
for $a_{1},\ldots,a_{M}\geq0$. The first term is bounded by
\[
C_{1,k}=\sup_{\theta\in\Theta}\tfrac{\nu_{k}+n_{k}}{2}\log\left[1+2\tfrac{\rho(\Psi_{k})}{\nu_{k}}\xi^{\prime}\xi\right]<\infty,
\]
since $\lim_{\nu\rightarrow\infty}\nu\log(1+a/\nu)\rightarrow a<\infty$.
Next, using $\log(1+ab)\leq\log(1+a)(1+b)]$ for $a,b>0$ we find
\begin{eqnarray*}
\log\left[1+2\tfrac{\rho(\Psi_{k})}{\nu_{k}}X_{l}^{\prime}X_{l}\right] & = & \log\left[1+2\tfrac{\rho(\Psi_{k})\nu_{l,0}}{\nu_{k}}\tfrac{X_{l}^{\prime}X_{l}}{\nu_{l,0}}\right]\\
 & \leq & \log\left(1+2\tfrac{\rho(\Psi_{k})\nu_{l,0}}{\nu_{k}}\right)+\log\left(1+\tfrac{X_{l}^{\prime}X_{l}}{\nu_{l,0}}\right),
\end{eqnarray*}
where 
\[
C_{2,k}\equiv\sup_{\theta\in\Theta}\tfrac{\nu_{k}+n_{k}}{2}\log\left(1+2\tfrac{\rho(\Psi_{k})\nu_{l,0}}{\nu_{k}}\right)<\infty,
\]
and since 
\[
\mathbb{E}_{\theta_{0}}\left[\log\left(1+\tfrac{X_{l}^{\prime}X_{l}}{\nu_{l,0}}\right)\right]=\psi\left(\tfrac{\nu_{l,0}+n_{l}}{2}\right)-\psi\left(\tfrac{\nu_{l,0}}{2}\right),
\]
it follows that
\[
\mathbb{E}_{\theta_{0}}\left[\sup_{\theta\in\Theta}\left|\log f(Y|\theta)\right|\right]\leq C_{0}+\sum_{k=1}^{K}C_{1,k}+C_{2,k}+\tfrac{\nu_{k}+n_{k}}{2}\sum_{l=1}^{K}\left[\psi\left(\tfrac{\nu_{l,0}+n_{l}}{2}\right)-\psi\left(\tfrac{\nu_{l,0}}{2}\right)\right],
\]
which proves that $\mathbb{E}_{\theta_{0}}\left[\sup_{\theta\in\Theta}R_{k}(\theta)\right]<\infty$.

Next we prove part (ii) concerning the limit distribution. Given the
four conditions of \citet[theorem 2.5]{NeweyMcFadden:1994} in the
proof of consistency, we proceed to establish the asymptotic normality
by verifying the conditions of \citet[theorem 2.4]{Amemiya(85)}.
These are as follows: (i) $\frac{\partial^{2}f(y\mid\theta)}{\partial\theta\partial\theta^{\prime}}$
exists and is continuous in an open, convex neighborhood $\mathscr{N}$
of $\theta_{0}$ (ii) $\int\frac{\partial f(y\mid\theta)}{\partial\theta}dy=0$,
(iii) $\int\frac{\partial^{2}f\left(y|\theta\right)}{\partial\theta\partial\theta^{\prime}}dy=0$,
(iv) $\ensuremath{{\rm plim}\frac{1}{T}\sum_{t=1}^{T}\frac{\partial^{2}\log f(y_{t}\mid\theta)}{\partial\theta\partial\theta^{\prime}}=\mathbb{E}_{\theta_{0}}\left[\frac{\partial^{2}\log f\left(y|\theta\right)}{\partial\theta\partial\theta^{\prime}}\right]}$
uniformly in $\theta$ in an open neighborhood of $\theta_{0}$. Then,
we have $\ensuremath{\sqrt{T}\left(\hat{\theta}_{T}-\theta_{0}\right)\rightarrow N\left\{ 0,\mathcal{J}^{-1}\right\} }$,
where $\mathcal{J}=-\mathbb{E}_{\theta_{0}}\left[\left.\frac{\partial^{2}\log f\left(y|\theta\right)}{\partial\theta\partial\theta^{\prime}}\right|_{\theta=\theta_{0}}\right]$.

Note that because $\mathscr{N}\subset\Theta$ where $\Theta$ is the
parameter sets, it is sufficient to verify the Conditions (i), (ii),
(iii), and (iv) hold for $\forall\theta\in\Theta$ and $\theta_{0}$
is interior to $\Theta$, where $\Theta$ is compact by our assumption
in the proof of consistency. Before giving a formal proof, we first
define the random vector $Z=Z(Y|\theta)$, $Z_{k}=Z_{k}(Y|\theta)$,
and random variable $S_{k}=S_{k}\left(Y|\theta\right)$, given by
\[
Z=A\left(Y-\mu\right),\quad Z_{k}=e_{k}^{\prime}Z,\quad S_{k}=\tfrac{\nu_{k}+n_{k}}{\nu_{k}+Z_{k}^{\prime}Z_{k}}
\]
where $A=\Xi^{-1}$ and $Y=\mu_{0}+\Xi_{0}X$. When $\mu=\mu_{0}$
and $\Xi=\Xi_{0}$, we have $Z=X$, $Z_{k}=X_{k}$, and $S_{k}=W_{k}=\tfrac{\nu_{k}+n_{k}}{\nu_{k}+X_{k}^{\prime}X_{k}}$.

\noindent \textbf{Condition (i)}: That $\theta_{0}$ is interior
to $\Theta$ is assumed. We prove that $\frac{\partial^{2}f(y\mid\theta)}{\partial\theta\partial\theta^{\prime}}$
exists and is continuous for $\forall\theta\in\Theta$ directly by
deriving the expression for $\frac{\partial^{2}\log f\left(Y|\theta\right)}{\partial\theta\partial\theta^{\prime}}$.
The form of such hessian matrix can be easily adapted from e Theorem
\ref{thm:ScoreHess} by replacing $X$, $X_{k}$, and $W_{k}$ with
$Z$, $Z_{k}$, and $S_{k}$. 

\noindent \textbf{Condition (ii)}: It is sufficient to show that
for $\forall\theta\in\Theta$, we have
\begin{align*}
\int\frac{\partial f(y|\theta)}{\partial\theta}\mathrm{d}y & =\int\nabla_{\theta}f(Y|\theta)\mathrm{d}Y=\mathbb{E}_{\theta}\left[\nabla_{\theta}\right]=0,\quad\nabla_{\theta}=\frac{\partial\log f\left(Y|\theta\right)}{\partial\theta}
\end{align*}
Adapted from Theorem \ref{thm:ScoreHess}, we can easily obtain that
$\mathbb{E}_{\theta}\left[\nabla_{\theta}\right]=0$ when $\min_{k}\nu_{k}>1$.
Note that $\mathbb{E}\left[\nabla_{\Xi}\right]=\mathbb{E}\left[M_{\Xi}^{+\prime}\nabla_{\Xi}\right]=M_{\Xi}^{+\prime}\mathbb{E}\left[\nabla_{\Xi}\right]=0$.

\noindent \textbf{Condition (iii)}: It is sufficient to show that
for $\forall\ \theta\in\Theta$, we have 
\begin{align*}
\int\frac{\partial^{2}f\left(y|\theta\right)}{\partial\theta\partial\theta^{\prime}}dy & =\int\left[\nabla_{\theta}\nabla_{\theta}^{\prime}+\nabla_{\theta\theta^{\prime}}\right]f(y\mid\theta)\mathrm{d}y=\underbrace{\mathbb{E}_{\theta}\left[\nabla_{\theta}\nabla_{\theta}^{\prime}\right]}_{\mathcal{I}^{\theta}}+\underbrace{\mathbb{E}_{\theta}\left[\nabla_{\theta\theta^{\prime}}\right]}_{-\mathcal{J^{\theta}}}=0.
\end{align*}
This requires to verify the existences both of the Fisher information
matrix $\mathcal{I}^{\theta}=\mathbb{E}_{\theta}\left[\nabla_{\theta}\nabla_{\theta}^{\prime}\right]$
and the expectation of Hessian matrix $\mathcal{J^{\theta}}=-\mathbb{E}_{\theta}\left[\nabla_{\theta\theta^{\prime}}\right]$,
and show that $\mathcal{I^{\theta}}=\mathcal{J^{\theta}}$. Such results
can be easily verified by adapting from Theorem \ref{thm:ScoreHess},
which requires $\min_{k}\nu_{k}>2$. Note that for $\Xi$, instead
of $\tilde{\Xi}$, we have
\begin{align*}
\mathcal{I}_{\Xi} & =\mathbb{E}\left[\nabla_{\Xi}\nabla_{\Xi}^{\prime}\right]=M_{\Xi}^{+\prime}\mathbb{E}\left[\nabla_{\tilde{\Xi}}\nabla_{\tilde{\Xi}}^{\prime}\right]M_{\Xi}^{+}=M_{\Xi}^{+\prime}\mathcal{I}_{\tilde{\Xi}}M_{\Xi}^{+}\\
\mathcal{J}_{\Xi} & =\mathbb{E}\left[\nabla_{\Xi\Xi^{\prime}}\right]=M_{\Xi}^{+\prime}\mathbb{E}\left[\nabla_{\tilde{\Xi}\tilde{\Xi}^{\prime}}\right]M_{\Xi}^{+\prime}-\left(\mathbb{E}(\nabla_{\tilde{\Xi}}^{\prime})\otimes I_{n}\right)\frac{\partial{\rm vec}(M_{\Xi}^{+\prime})}{\partial{\rm vec}(\Xi)^{\prime}}=M_{\Xi}^{+\prime}\mathcal{J}_{\tilde{\Xi}}M_{\Xi}^{+},
\end{align*}
and from $\mathcal{I}_{\tilde{\Xi}}=\mathcal{J}_{\tilde{\Xi}}$, we
have $\mathcal{I}_{\Xi}=\mathcal{J}_{\Xi}$.

\noindent \textbf{Condition (iv)}: If there exist a $g(y)$ function
with $\|\frac{\partial^{2}\log f\left(y|\theta\right)}{\partial\theta\partial\theta^{\prime}}\|\leqslant g(y)$
for all $\theta\in\Theta$ and $\mathbb{E}_{\theta_{0}}\left[g(z)\right]<\infty$,
then we have $\ensuremath{{\rm plim}\ \frac{1}{T}\sum_{t=1}^{T}\frac{\partial^{2}\log f(y_{t}\mid\theta)}{\partial\theta\partial\theta^{\prime}}=\mathbb{E}_{\theta_{0}}\left[\frac{\partial^{2}\log f\left(y|\theta\right)}{\partial\theta\partial\theta^{\prime}}\right]}$
uniformly in $\theta\in\Theta$, see e.g. \citet[lemma 2.4]{NeweyMcFadden:1994}.
This section, we determine the dominating function, $g(y)$.

First, adapted from Theorem \ref{thm:ScoreHess}, by replacing $X$,
$X_{k}$, and $W_{k}$ with $Z$, $Z_{k}$, and $S_{k}$, the expression
of Hessian matrix of $\log f\left(Y|\theta\right)$ with respective
to $\theta$ are given by
\begin{align*}
\nabla_{\mu\mu^{\prime}}^{\theta} & =\sum_{k=1}^{K}\tfrac{2S_{k}^{2}}{\nu_{k}+n_{k}}A^{\prime}e_{k}Z_{k}Z_{k}^{\prime}e_{k}^{\prime}A-S_{k}A^{\prime}J_{k}A\\
\nabla_{\mu\tilde{\Xi}^{\prime}}^{\theta} & =\sum_{k=1}^{K}S_{k}\left(AZ_{k}^{\prime}e_{k}^{\prime}\otimes A^{\prime}\right)K_{n}+\tfrac{2S_{k}^{2}}{\nu_{k}+n_{k}}A^{\prime}e_{k}Z_{k}{\rm vec}\left(A^{\prime}e_{k}Z_{k}Z^{\prime}\right)^{\prime}-S_{k}A^{\prime}\left(Z^{\prime}\otimes J_{k}A\right)\\
\nabla_{\mu\nu_{k}}^{\theta} & =\sum_{k=1}^{K}\tfrac{1}{\nu_{k}+n_{k}}A^{\prime}e_{k}Z_{k}\left(S_{k}-S_{k}^{2}\right)\\
\nabla_{\tilde{\Xi}\tilde{\Xi}^{\prime}}^{\theta} & =-\ensuremath{\left(A\otimes A^{\prime}\right)}K_{n}+\sum_{k=1}^{K}\tfrac{2S_{k}^{2}}{\nu_{k}+n_{k}}{\rm vec}\left(A^{\prime}e_{k}Z_{k}Z^{\prime}\right){\rm vec}\left(A^{\prime}e_{k}Z_{k}Z^{\prime}\right)^{\prime}\\
 & \quad+\sum_{k=1}^{K}S_{k}\left[\left(ZZ_{k}^{\prime}e_{k}^{\prime}A\otimes A^{\prime}\right)K_{n}-\left(ZZ^{\prime}\otimes A^{\prime}J_{k}A\right)-\left(A\otimes A^{\prime}e_{k}Z_{k}Z^{\prime}\right)K_{n}\right]\\
\nabla_{\tilde{\Xi}\nu_{k}}^{\theta} & =\tfrac{1}{\nu_{k}+n_{k}}\left(S_{k}-S_{k}^{2}\right){\rm vec}\left(A^{\prime}e_{k}Z_{k}Z^{\prime}\right)\\
\nabla_{\nu_{k}\nu_{k}}^{\theta} & =\tfrac{1}{4}\psi^{\prime}\left(\tfrac{v_{k}+n_{k}}{2}\right)-\tfrac{1}{4}\psi^{\prime}\left(\tfrac{v_{k}}{2}\right)+\tfrac{1}{2\nu_{k}}+\tfrac{1}{2}\tfrac{1}{\nu_{k}+n_{k}}S_{k}^{2}-\tfrac{1}{\left(\nu_{k}+n_{k}\right)}S_{k}\\
\nabla_{\nu_{k}\nu_{l}}^{\theta} & =0,\quad{\rm for}\quad k\neq l,
\end{align*}
where the $\theta$ in the superscript means that the Score and Hessian
matrix is evaluated at $\theta$. 

Note that we have $\nabla_{\mu\Xi^{\prime}}^{\theta}=\nabla_{\mu\tilde{\Xi}^{\prime}}^{\theta}M_{\Xi}^{+}$,
and $\nabla_{\Xi\Xi^{\prime}}^{\theta}=M_{\Xi}^{+\prime}\nabla_{\tilde{\Xi}\tilde{\Xi}^{\prime}}M_{\Xi}^{+\prime}+\left(\nabla_{\tilde{\Xi}}^{\prime}\otimes I_{n}\right)\frac{\partial{\rm vec}(M_{\Xi}^{+\prime})}{\partial{\rm vec}(\Xi)^{\prime}}$,
and $\nabla_{\Xi\nu_{k}}^{\theta}=M_{\Xi}^{+\prime}\nabla_{\tilde{\Xi}\nu_{k}}^{\theta}$,
and from the following inequalities
\begin{align*}
\left\Vert \nabla_{\mu\Xi^{\prime}}^{\theta}\right\Vert  & \leq\left\Vert \nabla_{\mu\tilde{\Xi}^{\prime}}^{\theta}\right\Vert \left\Vert M_{\Xi}^{+}\right\Vert ,\quad\left\Vert \nabla_{\Xi\nu_{k}}^{\theta}\right\Vert =\left\Vert M_{\Xi}^{+\prime}\right\Vert \left\Vert \nabla_{\tilde{\Xi}\nu_{k}}^{\theta}\right\Vert \\
\left\Vert \nabla_{\Xi\Xi^{\prime}}^{\theta}\right\Vert  & \leq\left\Vert M_{\Xi}^{+}\right\Vert ^{2}\left\Vert \nabla_{\tilde{\Xi}\tilde{\Xi}^{\prime}}\right\Vert +n\left\Vert \nabla_{\tilde{\Xi}}\right\Vert \left\Vert \tfrac{\partial{\rm vec}(M_{\Xi}^{+\prime})}{\partial{\rm vec}(\Xi)^{\prime}}\right\Vert ,
\end{align*}
and because $\left\Vert \nabla_{\tilde{\Xi}}\right\Vert \leq\sum_{k=1}^{K}\left\Vert A^{\prime}e_{k}\right\Vert \left\Vert S_{k}Z_{k}Z^{\prime}\right\Vert +\left\Vert A^{\prime}\right\Vert $,
and from Equation (\ref{eq:SkZkZ}), we have 
\[
\mathbb{E}_{\theta_{0}}\sup_{\theta\in\Theta}\left|S_{k}Z_{k,i}Z_{l,j}\right|<\infty,\quad{\rm if\ all}\ \nu_{k}>1
\]
It is enough to determine the dominating function, $g(y)$ under $\theta=\{\mu,\tilde{\Xi},\bm{\nu}\}$,
instead of $\theta=\{\mu,\Xi,\bm{\nu}\}$.

Next we establish bounds for each of the terms. First, we have
\begin{align*}
\left\Vert \nabla_{\mu\mu^{\prime}}^{\theta}\right\Vert  & =\left\Vert \sum_{k=1}^{K}\tfrac{2S_{k}^{2}}{\nu_{k}+n_{k}}A^{\prime}e_{k}Z_{k}Z_{k}^{\prime}e_{k}^{\prime}A-S_{k}A^{\prime}J_{k}A\right\Vert \\
 & \leq\sum_{k=1}^{K}\tfrac{2}{\nu_{k}+n_{k}}\left\Vert A^{\prime}e_{k}\right\Vert ^{2}\left\Vert S_{k}^{2}Z_{k}Z_{k}^{\prime}\right\Vert +S_{k}\left\Vert A^{\prime}J_{k}A\right\Vert \\
 & \leq\sum_{k=1}^{K}\tfrac{2}{\nu_{k}+n_{k}}\left\Vert A^{\prime}e_{k}\right\Vert ^{2}\left\Vert S_{k}^{2}Z_{k}Z_{k}^{\prime}\right\Vert +\tfrac{\nu_{k}+n_{k}}{\nu_{k}}\left\Vert A^{\prime}J_{k}A\right\Vert 
\end{align*}
which we use the fact that $S_{k}=\frac{\nu_{k}+n_{k}}{\nu_{k}+Z_{k}^{\prime}Z_{k}}\leq\frac{\nu_{k}+n_{k}}{\nu_{k}}$.
As for the $S_{k}^{2}Z_{k}Z_{k}^{\prime}$, it involves the following
terms
\begin{align*}
\left|S_{k}^{2}Z_{k,i}Z_{k,j}\right| & =\tfrac{\left(\nu_{k}+n_{k}\right)^{2}}{\left(\nu_{k}+\sum_{p}Z_{k,p}^{2}\right)^{2}}\left|Z_{k,i}Z_{k,j}\right|\\
 & \leq\tfrac{\left(\nu_{k}+n_{k}\right)^{2}}{\nu_{k}\left(\nu_{k}+\sum_{p}Z_{k,p}^{2}\right)}\left|Z_{k,i}Z_{k,j}\right|,\quad{\rm use}\ \tfrac{1}{\nu_{k}+\sum_{p}Z_{k,p}^{2}}<\tfrac{1}{\nu_{k}}\\
 & \leq\tfrac{\left(\nu_{k}+n_{k}\right)^{2}}{\nu_{k}\left(\nu_{k}+\sum_{p}Z_{k,p}^{2}\right)}\left|Z_{k}\right|_{\max}^{2},\quad\left|Z_{k}\right|_{\max}=\max_{i=1,\ldots,n_{k}}\ \left|Z_{k,i}\right|\\
 & \leq\tfrac{\left(\nu_{k}+n_{k}\right)^{2}}{\nu_{k}\left(\nu_{k}+\left|Z_{k}\right|_{\max}^{2}\right)}\left|Z_{k}\right|_{\max}^{2}\\
 & \leq\tfrac{1}{\nu_{k}}\left(\nu_{k}+n_{k}\right)^{2}
\end{align*}
So the term $\left\Vert A^{\prime}e_{k}\left(S_{k}^{2}Z_{k}Z_{k}^{\prime}\right)e_{k}^{\prime}A\right\Vert $
is bounded. This leads to $\left\Vert \nabla_{\mu\mu^{\prime}}^{\theta}\right\Vert $
is bounded. 

Next, we have
\begin{align*}
\left\Vert \nabla_{\mu\tilde{\Xi^{\prime}}}^{\theta}\right\Vert  & =\left\Vert \sum_{k=1}^{K}S_{k}\left(AZ_{k}^{\prime}e_{k}^{\prime}\otimes A^{\prime}\right)K_{n}+\tfrac{2S_{k}^{2}}{\nu_{k}+n_{k}}A^{\prime}e_{k}Z_{k}{\rm vec}\left(A^{\prime}e_{k}Z_{k}Z^{\prime}\right)^{\prime}-S_{k}A^{\prime}\left(Z^{\prime}\otimes J_{k}A\right)\right\Vert \\
 & \leq\sum_{k=1}^{K}\left\Vert S_{k}\left(AZ_{k}^{\prime}e_{k}^{\prime}\otimes A^{\prime}\right)K_{n}\right\Vert +\tfrac{2S_{k}^{2}}{\nu_{k}+n_{k}}\left\Vert A^{\prime}e_{k}Z_{k}{\rm vec}\left(A^{\prime}e_{k}Z_{k}Z^{\prime}\right)^{\prime}\right\Vert +\left\Vert S_{k}A^{\prime}\left(Z^{\prime}\otimes J_{k}A\right)\right\Vert ,
\end{align*}
where the first term involves 
\begin{align}
\left|S_{k}Z_{k,i}\right| & =\tfrac{\nu_{k}+n_{k}}{\nu_{k}+\sum_{p}Z_{k,p}^{2}}\left|Z_{k,i}\right|\nonumber \\
 & \leq\tfrac{\nu_{k}+n_{k}}{\nu_{k}+Z_{k,i}^{2}}\left|Z_{k,i}\right|\nonumber \\
 & \leq\tfrac{\nu_{k}+n_{k}}{2\sqrt{\nu_{k}}\left|Z_{k,i}\right|}\left|Z_{k,i}\right|\nonumber \\
 & =\tfrac{\nu_{k}+n_{k}}{2\sqrt{\nu_{k}}},\label{eq:SkZk}
\end{align}
the second term involves
\begin{align*}
\left|S_{k}^{2}Z_{k,p}Z_{k,i}Z_{l,j}\right| & =S_{k}^{2}\left|Z_{k,p}Z_{k,i}\right|\left|Z_{l,j}\right|\\
 & \leq\left(\nu_{k}+n_{k}\right)^{2}\left|Z_{l,j}\right|\\
 & =\left(\nu_{k}+n_{k}\right)^{2}\left|\tilde{e}_{l,i}Z\right|,\quad\tilde{e}_{l,i}Z=Z_{l,j}\\
 & =\left(\nu_{k}+n_{k}\right)^{2}\left|\tilde{e}_{l,i}A\left(Y-\mu\right)\right|\\
 & =\left(\nu_{k}+n_{k}\right)^{2}\left|\tilde{e}_{l,i}A\Xi_{0}\left(X-\xi\right)\right|,\quad\xi=A_{0}(\mu-\mu_{0})\\
 & \leq\left(\nu_{k}+n_{k}\right)^{2}\left[\left|\tilde{e}_{l,i}A\Xi_{0}X\right|+\left|\tilde{e}_{l,i}A\Xi_{0}\xi\right|\right]\\
 & \leq C_{1}\left(\theta\right)+\sum_{k=1}^{K}C_{2}\left(\theta\right)\left|X_{k}\right|,\quad C_{1}\left(\theta\right),C_{2}\left(\theta\right)>0\\
 & \leq C_{1}^{*}+\sum_{k=1}^{K}C_{2}^{*}\left|X_{k}\right|,\quad C_{i}^{*}=\sup_{\theta\in\Theta}C_{i}\left(\theta\right)\\
\sup_{\theta\in\Theta}\left|S_{k}^{2}Z_{k,p}Z_{k,i}Z_{l,j}\right| & \leq C_{1}^{*}+\sum_{k=1}^{K}C_{2}^{*}\left|X_{k}\right|\\
\mathbb{E}_{\theta_{0}}\sup_{\theta\in\Theta}\left|S_{k}^{2}Z_{k,p}Z_{k,i}Z_{l,j}\right| & \leq C_{1}^{*}+\sum_{k=1}^{K}C_{2}^{*}\mathbb{E}_{\theta_{0}}\left[\left|X_{k}\right|\right]<\infty,\quad{\rm if\ all}\ \nu_{k}>1,
\end{align*}
and the third term and final term involves
\begin{align*}
\left|S_{k}Z_{l,i}\right| & \leq\tfrac{\nu_{k}+n_{k}}{\nu_{k}}\left|Z_{l,i}\right|\leq\tilde{C}_{1}+\sum_{k=1}^{K}\tilde{C}_{2}\left|X_{k}\right|,\\
\mathbb{E}_{\theta_{0}}\sup_{\theta\in\Theta}\left|S_{k}^{2}Z_{k,p}Z_{k,i}Z_{l,j}\right| & <\tilde{C}_{1}+\sum_{k=1}^{K}\tilde{C}_{2}\mathbb{E}_{\theta_{0}}\left[\left|X_{k}\right|\right]<\infty,\quad{\rm if\ all}\ \nu_{k}>1.
\end{align*}
So, we have shown that $\mathbb{E}_{\theta_{0}}\left\Vert \sup_{\theta\in\Theta}\nabla_{\mu\Xi^{\prime}}^{\theta}\right\Vert <\infty$
for all $\nu_{k}>1$.

Next, we have
\begin{align*}
\left\Vert \nabla_{\mu\nu_{k}}^{\theta}\right\Vert  & =\left\Vert \sum_{k=1}^{K}\tfrac{1}{\nu_{k}+n_{k}}A^{\prime}e_{k}Z_{k}\left(S_{k}-S_{k}^{2}\right)\right\Vert \\
 & \leq\sum_{k=1}^{K}\tfrac{1}{\nu_{k}+n_{k}}\left\Vert A^{\prime}e_{k}S_{k}Z_{k}\right\Vert +\tfrac{1}{\nu_{k}+n_{k}}\left\Vert A^{\prime}e_{k}S_{k}^{2}Z_{k}\right\Vert .
\end{align*}
Again, the first term involves
\[
\left|S_{k}Z_{k}\right|=\tfrac{\nu_{k}+n_{k}}{\nu_{k}+Z_{k}^{\prime}Z_{k}}\left|Z_{k}\right|\leq\tfrac{\nu_{k}+n_{k}}{2\sqrt{\nu_{k}}\left|Z_{k}\right|}\left|Z_{k}\right|=\tfrac{\nu_{k}+n_{k}}{2\sqrt{\nu_{k}}}
\]
The second term involves
\[
\left|S_{k}^{2}Z_{k}\right|=\tfrac{\left(\nu_{k}+n_{k}\right)^{2}}{\left(\nu_{k}+Z_{k}^{\prime}Z_{k}\right)^{2}}\left|Z_{k}\right|\leq\tfrac{1}{\nu_{k}}\tfrac{\left(\nu_{k}+n_{k}\right)^{2}}{\left(\nu_{k}+Z_{k}^{\prime}Z_{k}\right)}\left|Z_{k}\right|\leq\tfrac{1}{\nu_{k}}\tfrac{\nu_{k}+n_{k}}{2\sqrt{\nu_{k}}\left|Z_{k}\right|}\left|Z_{k}\right|=\tfrac{1}{\nu_{k}}\tfrac{\nu_{k}+n_{k}}{2\sqrt{\nu_{k}}},
\]
which shows that $\left\Vert \nabla_{\mu\nu_{k}}^{\theta}\right\Vert $
is bounded for all $\nu_{k}>0$.

Next, we have
\begin{align*}
\nabla_{\tilde{\Xi}\tilde{\Xi^{\prime}}}^{\theta} & =-\ensuremath{\left(A\otimes A^{\prime}\right)}K_{n}+\sum_{k=1}^{K}\tfrac{2S_{k}^{2}}{\nu_{k}+n_{k}}{\rm vec}\left(A^{\prime}e_{k}Z_{k}Z^{\prime}\right){\rm vec}\left(A^{\prime}e_{k}Z_{k}Z^{\prime}\right)^{\prime}\\
 & \quad+\sum_{k=1}^{K}S_{k}\left[\left(ZZ_{k}^{\prime}e_{k}^{\prime}A\otimes A^{\prime}\right)K_{n}-\left(ZZ^{\prime}\otimes A^{\prime}J_{k}A\right)-\left(A\otimes A^{\prime}e_{k}Z_{k}Z^{\prime}\right)K_{n}\right].
\end{align*}
The first term involves
\begin{align*}
\left|S_{k}^{2}Z_{k,i}Z_{k,j}Z_{l,p}Z_{m,q}\right| & =\tfrac{\left(\nu_{k}+n_{k}\right)^{2}}{\left(\nu_{k}+Z_{k}^{\prime}Z_{k}\right)^{2}}\left|Z_{k,i}Z_{k,j}Z_{l,p}Z_{m,q}\right|\\
 & \leq\tfrac{\left(\nu_{k}+n_{k}\right)^{2}}{\nu_{k}}\left|Z_{l,p}Z_{m,q}\right|\\
 & =\tfrac{\left(\nu_{k}+n_{k}\right)^{2}}{\nu_{k}}\left|\tilde{e}_{l,p}ZZ^{\prime}\tilde{e}_{m,q}^{\prime}\right|\\
 & =\tfrac{\left(\nu_{k}+n_{k}\right)^{2}}{\nu_{k}}\left|\tilde{e}_{l,p}A\Xi_{0}\left(X-\xi\right)\left(X-\xi\right)^{\prime}\Xi_{0}^{\prime}A^{\prime}\tilde{e}_{m,q}^{\prime}\right|\\
 & =\tfrac{\left(\nu_{k}+n_{k}\right)^{2}}{\nu_{k}}\left|\tilde{e}_{l,p}A\Xi_{0}\left(XX^{\prime}-\xi X^{\prime}-X\xi^{\prime}+\xi\xi^{\prime}\right)\Xi_{0}^{\prime}A^{\prime}\tilde{e}_{m,q}^{\prime}\right|,\quad\Psi_{l,p}=\tilde{e}_{l,p}A\Xi_{0}\\
 & \leq C_{0}+C_{1}\left|\Psi_{l,p}XX^{\prime}\Psi_{l,p}^{\prime}\right|+C_{2}\left|\Psi_{l,p}\xi X^{\prime}\Psi_{l,p}^{\prime}\right|+C_{3}\left|\Psi_{l,p}X\xi^{\prime}\Psi_{l,p}^{\prime}\right|\\
 & \leq C_{0}^{*}+C_{1}^{*}\sum_{i}\sum_{j}\left|X_{i}X_{j}\right|+C_{1}^{2}\sum_{i}\left|X_{i}\right|.
\end{align*}
Thus, if $\nu_{k}>2$ for all $k$, we have
\[
\mathbb{E}_{\theta_{0}}\sup_{\theta\in\Theta}\left|S_{k}^{2}Z_{k,i}Z_{k,j}Z_{l,p}Z_{m,q}\right|\leq C_{0}^{*}+C_{1}^{*}\sum_{i}\sum_{j}\mathbb{E}_{\theta_{0}}\left|X_{i}X_{j}\right|+C_{1}^{2}\sum_{i}\mathbb{E}_{\theta_{0}}\left|X_{i}\right|<\infty.
\]
The second term involves
\begin{align}
\left|S_{k}Z_{k,i}Z_{l,j}\right| & =\left|S_{k}Z_{k,i}\right|\left|Z_{l,j}\right|\leq\tfrac{\nu_{k}+n_{k}}{2\sqrt{\nu_{k}}}\left|Z_{l,j}\right|\nonumber \\
 & \leq C_{1}^{*}+\sum_{k=1}^{K}C_{2}^{*}\left|X_{k}\right|,\nonumber \\
\mathbb{E}_{\theta_{0}}\sup_{\theta\in\Theta}\left|S_{k}Z_{k,i}Z_{l,j}\right| & \leq C_{1}^{*}+\sum_{k=1}^{K}C_{2}^{*}\mathbb{E}_{\theta_{0}}\left|X_{k}\right|<\infty,\quad\quad{\rm if}\ \min_{k}\nu_{k}>1.\label{eq:SkZkZ}
\end{align}
The third term, $S_{k}\left(ZZ^{\prime}\otimes A^{\prime}J_{k}A\right)$,
involves
\begin{align*}
\left|S_{k}Z_{p,i}Z_{q,j}\right| & \leq\tfrac{\nu_{k}+n_{k}}{\nu_{k}}\left|Z_{p,i}Z_{q,j}\right|\\
 & \leq C_{0}^{*}+C_{1}^{*}\sum_{i}\sum_{j}\left|X_{i}X_{j}\right|+C_{1}^{2}\sum_{i}\left|X_{i}\right|,\\
\mathbb{E}_{\theta_{0}}\sup_{\theta\in\Theta}\left|S_{k}Z_{p,i}Z_{q,j}\right| & \leq C_{0}^{*}+C_{1}^{*}\sum_{i}\sum_{j}\mathbb{E}_{\theta_{0}}\left|X_{i}X_{j}\right|+C_{1}^{2}\sum_{i}\mathbb{E}_{\theta_{0}}\left|X_{i}\right|\\
 & <\infty,\quad{\rm if}\ \min_{k}\nu_{k}>2
\end{align*}
Finally, the fourth term, $S_{k}\left(A\otimes A^{\prime}e_{k}Z_{k}Z^{\prime}\right)K_{n}$,
involves $\left|S_{k}Z_{k,i}Z_{l,j}\right|$, which satisfies $\mathbb{E}_{\theta_{0}}\sup_{\theta\in\Theta}\left|S_{k}Z_{k,i}Z_{l,j}\right|<\infty$,
if $\min_{k}\nu_{k}>1$.

Next, we have
\begin{align*}
\left\Vert \nabla_{\tilde{\Xi}\nu_{k}}^{\theta}\right\Vert  & =\tfrac{1}{\nu_{k}+n_{k}}\left\Vert \left(S_{k}-S_{k}^{2}\right){\rm vec}\left(A^{\prime}e_{k}Z_{k}Z^{\prime}\right)\right\Vert \\
 & \leq\tfrac{1}{\nu_{k}+n_{k}}\left\Vert S_{k}{\rm vec}\left(A^{\prime}e_{k}Z_{k}Z^{\prime}\right)\right\Vert +\tfrac{1}{\nu_{k}+n_{k}}\left\Vert S_{k}^{2}{\rm vec}\left(A^{\prime}e_{k}Z_{k}Z^{\prime}\right)\right\Vert ,
\end{align*}
which involves $\left|S_{k}Z_{k,i}Z_{l,k}\right|$ and $\left|S_{k}^{2}Z_{k,i}Z_{l,k}\right|$.
We need $\nu_{k}>1$ to ensure $\mathbb{E}_{\theta_{0}}\sup_{\theta\in\Theta}\left\Vert \nabla_{\tilde{\Xi}\nu_{k}}^{\theta}\right\Vert <\infty$
as shown earlier.

Nest, we have
\begin{align*}
\left|\nabla_{\nu_{k}\nu_{k}}^{\theta}\right| & =\left|\tfrac{1}{2}\left[\psi^{\prime}\left(\tfrac{v_{k}+n_{k}}{2}\right)-\tfrac{1}{2}\psi^{\prime}\left(\tfrac{v_{k}}{2}\right)\right]+\tfrac{1}{\nu_{k}+n_{k}}\left(S_{k}^{2}-S_{k}\right)+\tfrac{1}{\nu_{k}}-\tfrac{1}{\nu_{k}+Z_{k}^{\prime}Z_{k}}\right|\\
 & \leq\tfrac{1}{2}\psi^{\prime}\left(\tfrac{v_{k}+n_{k}}{2}\right)+\tfrac{1}{2}\psi^{\prime}\left(\tfrac{v_{k}}{2}\right)+\tfrac{1}{\nu_{k}+n_{k}}S_{k}^{2}+\tfrac{1}{\nu_{k}+n_{k}}S_{k}+\tfrac{1}{\nu_{k}}+\tfrac{1}{\nu_{k}+Z_{k}^{\prime}Z_{k}}\\
 & \leq\tfrac{1}{2}\psi^{\prime}\left(\tfrac{v_{k}+n_{k}}{2}\right)+\tfrac{1}{2}\psi^{\prime}\left(\tfrac{v_{k}}{2}\right)+\tfrac{1}{\nu_{k}+n_{k}}\tfrac{\left(\nu_{k}+n_{k}\right)^{2}}{\nu_{k}^{2}}+\tfrac{1}{\nu_{k}+n_{k}}\tfrac{\nu_{k}+n_{k}}{\nu_{k}}+\tfrac{1}{\nu_{k}}+\tfrac{1}{\nu_{k}}\\
 & <\infty.
\end{align*}
This verifies condition (iv) and completes the proof.

\hfill{}$\square$

\setcounter{table}{0}
\global\long\def\thetable{S.\arabic{table}}%

\setcounter{figure}{0}
\global\long\def\thefigure{S.\arabic{figure}}%

\part*{Supplemental Material for `Convolution-$t$ Distributions'}

\section*{General Moments of Convolution-$t$ Distributions\label{subsec:Moments-of-Convol-t}}

The $r$-th moment of $Y_{1}=\mu+\beta^{\prime}X$ exists if and only
if $r<\nu_{\min}\equiv\min_{k}(\nu_{k})$, and all odd moments, less
than $\nu_{\min}$, are zero because a convolution of symmetric distributions
is symmetric. I.e. $\mathbb{E}\left(Y_{1}-\mu\right)^{r}=0$ for $r=2m-1$
for some $m\in\mathbb{N}$, provided $r<\nu_{\min}$. We derived the
2nd and 4th centralized moments, $\mathbb{E}(Y_{1}-\mu)^{2}$ and
$\mathbb{E}(Y_{1}-\mu)^{4}$, in Theorem \ref{thm:Convo-t-kurtosis}. 

The $r$-th moment of $Y_{1}=\mu+\beta^{\prime}X$ exists if and only
if $r<\nu_{\min}\equiv\min_{k}(\nu_{k})$, and all odd moments, less
than $\nu_{\min}$, are zero because a convolution of symmetric distributions
is symmetric. General moments, including fractional moments, $r>0$,
can be obtained with the method in \citet[Theorem 11.4.4]{Kawata1972},
which states,
\[
\mathbb{E}|Y_{1}|^{r}=C_{p,r}\int_{0}^{\infty}s^{-(1+r)}\left[-{\rm Re}[\varphi_{Y_{1}}(s)]+\sum_{k=0}^{p}\frac{s^{2k}}{(2k)!}\varphi_{Y_{1}}^{(2k)}(0)\right]\mathrm{d}s,
\]
where $p=\lfloor r/2\rfloor$ and $C_{p,r}$ is the positive constant
given by
\[
\ensuremath{C_{p,r}=\int_{0}^{\infty}\left[L_{p}(x)/x^{1+r}\right]\mathrm{d}x},\quad\ensuremath{L_{p}(x)=-\cos x+\sum_{k=0}^{p}(-1)^{k}\frac{x^{2k}}{(2k)!}}.
\]
For small moments, $0<r<2$, this expression simplifies to 
\[
\mathbb{E}|Y_{1}|^{r}=\frac{r(1-r)}{\Gamma(2-r)}\frac{1}{\sin\left[(1-r)\pi/2\right]}\int_{0}^{\infty}\frac{1-{\rm Re}[\varphi_{Y_{1}}(s)]}{s^{1+r}}\mathrm{d}s.
\]
The coefficient of the integral on the right-hand side is interpreted
to be $2/\pi$ when $r=1$.

And because $\varphi_{Y_{1}}(s)=e^{is\mu}\prod_{k=1}^{K}\phi_{\nu_{k}}(\omega_{k}s)$
with 
\[
\phi_{\nu}(s)=\frac{K_{\frac{\nu}{2}}(\sqrt{\nu}|s|)(\sqrt{\nu}|s|)^{\frac{1}{2}\nu}}{\Gamma\left(\frac{\nu}{2}\right)2^{\frac{\nu}{2}-1}},\qquad\text{for}\quad s\in\mathbb{R},
\]
The derivatives of $\varphi_{Y_{1}}\left(s\right)$ with respect to
$s$ uses the following formula
\[
\ensuremath{\frac{\partial K_{\nu}(z)}{\partial z}=\frac{\nu}{z}K_{\nu}(z)-K_{\nu+1}(z)}.
\]

\section*{Simulation Results for Non-Standard Situations}

\subsection*{Simulations with Small Samples\label{subsec:Simulations_SmallSample}}

In Table \ref{tab:MLESimuSmall}, we conduct a simulation study with
$T=50,100,200,400$, based on 50,000 Monte-Carlo simulations. Distributional
specification is the same as in Section \ref{sec:Simulation-Study}.

\begin{table}
\caption{Simulations with Small Samples}

\begin{centering}
\vspace{0.2cm}
\begin{footnotesize}
\begin{tabularx}{\textwidth}{p{1.0cm}Yp{-0.2cm}YYYYYp{-0.2cm}YYYYYY}
\toprule
\midrule
          & {True} &       & Mean   & Std   & aStd  &$\alpha_L^{0.025}$ &$\alpha_R^{0.025}$ &       & Mean   & Std   & aStd &$\alpha_L^{0.025}$ &$\alpha_R^{0.025}$\\
\midrule
\\[-0.3cm]
          &   &  &  \multicolumn{5}{c}{$T=50$}     &       & \multicolumn{5}{c}{$T=100$} \\
\\[-0.3cm]
    $\mu_1$   & 0.1000 &       & 0.1006 & 0.1170 & 0.1117 & 0.0309 & 0.0315 &       & 0.1006 & 0.0811 & 0.0790 & 0.0284 & 0.0293 \\
\\[-0.3cm]    
    $\mu_2$   & 0.2000 &       & 0.2011 & 0.1457 & 0.1404 & 0.0290 & 0.0302 &       & 0.2006 & 0.1013 & 0.0993 & 0.0270 & 0.0284 \\
\\[-0.3cm]
    $\mu_3$   & 0.3000 &       & 0.3007 & 0.1486 & 0.1442 & 0.0283 & 0.0287 &       & 0.3004 & 0.1035 & 0.1020 & 0.0263 & 0.0263 \\
\\[-0.3cm]
    \multirow{9}[0]{*}{$\rm{vec}(\Xi)$} & 0.6000 &       & 0.5119 & 0.1434 & 0.1203 & 0.1352 & 0.0109 &       & 0.5562 & 0.1057 & 0.0851 & 0.0953 & 0.0185 \\
\\[-0.3cm]
          & 0.5000 &       & 0.4265 & 0.2029 & 0.1581 & 0.1072 & 0.0295 &       & 0.4630 & 0.1562 & 0.1118 & 0.0908 & 0.0429 \\
\\[-0.3cm]
          & 0.4000 &       & 0.3398 & 0.2261 & 0.1645 & 0.1031 & 0.0424 &       & 0.3708 & 0.1735 & 0.1163 & 0.0897 & 0.0572 \\
\\[-0.3cm]
          & 0.3000 &       & 0.3126 & 0.1770 & 0.1426 & 0.0530 & 0.0660 &       & 0.3074 & 0.1368 & 0.1009 & 0.0608 & 0.0750 \\
\\[-0.3cm]
          & 0.7000 &       & 0.6720 & 0.1650 & 0.1429 & 0.0687 & 0.0216 &       & 0.6880 & 0.1256 & 0.1010 & 0.0652 & 0.0358 \\
\\[-0.3cm]
          & 0.2000 &       & 0.2069 & 0.1180 & 0.0993 & 0.0509 & 0.0469 &       & 0.2040 & 0.0900 & 0.0702 & 0.0550 & 0.0585 \\
\\[-0.3cm]
          & 0.1000 &       & 0.1188 & 0.1564 & 0.1237 & 0.0462 & 0.0760 &       & 0.1093 & 0.1217 & 0.0874 & 0.0579 & 0.0794 \\
\\[-0.3cm]
          & 0.2000 &       & 0.2069 & 0.1180 & 0.0993 & 0.0509 & 0.0469 &       & 0.2040 & 0.0900 & 0.0702 & 0.0550 & 0.0585 \\
\\[-0.3cm]
          & 0.8000 &       & 0.7532 & 0.1581 & 0.1345 & 0.0829 & 0.0153 &       & 0.7774 & 0.1176 & 0.0951 & 0.0718 & 0.0243 \\
\\[-0.3cm]
    $1/\nu_1$  & 0.2500 &       & 0.2702 & 0.1699 & 0.1464 & 0.0000 & 0.0595 &       & 0.2531 & 0.1147 & 0.1035 & 0.0543 & 0.0368 \\
\\[-0.3cm]
    $1/\nu_2$  & 0.1250 &       & 0.1232 & 0.0998 & 0.0884 & 0.0000 & 0.0508 &       & 0.1195 & 0.0698 & 0.0625 & 0.0752 & 0.0350 \\
\\[-0.3cm]
          &       &       &       &       &       &       &       &       &  \\
          &  &    & \multicolumn{5}{c}{$T=200$}     &       & \multicolumn{5}{c}{$T=400$} \\
\\[-0.2cm]
     $\mu_1$   & 0.1000 &       & 0.1000 & 0.0567 & 0.0559 & 0.0279 & 0.0268 &       & 0.1000 & 0.0397 & 0.0395 & 0.0237 & 0.0245 \\
\\[-0.3cm]          
     $\mu_2$   & 0.2000 &       & 0.1997 & 0.0711 & 0.0702 & 0.0270 & 0.0256 &       & 0.2001 & 0.0498 & 0.0496 & 0.0248 & 0.0261 \\
\\[-0.3cm]
     $\mu_3$   & 0.3000 &       & 0.3001 & 0.0725 & 0.0721 & 0.0262 & 0.0259 &       & 0.3000 & 0.0512 & 0.0510 & 0.0252 & 0.0249 \\
\\[-0.3cm]
    \multirow{9}[0]{*}{$\rm{vec}(\Xi)$} & 0.6000 &       & 0.5824 & 0.0698 & 0.0602 & 0.0577 & 0.0246 &       & 0.5937 & 0.0449 & 0.0425 & 0.0351 & 0.0271 \\
\\[-0.3cm]
          & 0.5000 &       & 0.4855 & 0.1034 & 0.0790 & 0.0621 & 0.0445 &       & 0.4949 & 0.0636 & 0.0559 & 0.0411 & 0.0379 \\
\\[-0.3cm]
          & 0.4000 &       & 0.3882 & 0.1119 & 0.0822 & 0.0616 & 0.0506 &       & 0.3960 & 0.0672 & 0.0581 & 0.0415 & 0.0407 \\
\\[-0.3cm]
          & 0.3000 &       & 0.3017 & 0.0923 & 0.0713 & 0.0527 & 0.0555 &       & 0.3002 & 0.0574 & 0.0504 & 0.0395 & 0.0398 \\
\\[-0.3cm]
          & 0.7000 &       & 0.6945 & 0.0861 & 0.0714 & 0.0530 & 0.0364 &       & 0.6976 & 0.0554 & 0.0505 & 0.0393 & 0.0305 \\
\\[-0.3cm]
          & 0.2000 &       & 0.2013 & 0.0601 & 0.0496 & 0.0464 & 0.0458 &       & 0.2003 & 0.0386 & 0.0351 & 0.0366 & 0.0336 \\
\\[-0.3cm]
          & 0.1000 &       & 0.1036 & 0.0803 & 0.0618 & 0.0492 & 0.0580 &       & 0.1010 & 0.0499 & 0.0437 & 0.0388 & 0.0413 \\
\\[-0.3cm]
          & 0.2000 &       & 0.2013 & 0.0601 & 0.0496 & 0.0464 & 0.0458 &       & 0.2003 & 0.0386 & 0.0351 & 0.0366 & 0.0336 \\
\\[-0.3cm]
          & 0.8000 &       & 0.7915 & 0.0792 & 0.0672 & 0.0513 & 0.0293 &       & 0.7965 & 0.0509 & 0.0475 & 0.0385 & 0.0263 \\
\\[-0.3cm]
    $1/\nu_1$  & 0.2500 &       & 0.2504 & 0.0781 & 0.0732 & 0.0356 & 0.0307 &       & 0.2498 & 0.0533 & 0.0518 & 0.0306 & 0.0264 \\
\\[-0.3cm]
    $1/\nu_2$  & 0.1250 &       & 0.1210 & 0.0476 & 0.0442 & 0.0462 & 0.0261 &       & 0.1230 & 0.0324 & 0.0313 & 0.0345 & 0.0251 \\
\\[0.0cm]
\\[-0.5cm]
\midrule
\bottomrule
\end{tabularx}
\end{footnotesize}
\par\end{centering}
{\small{}Note: The mean, standard deviation of the MLE estimated parameters
for convolution-$t$ distribution. The asymptotic standard deviations
are also included for comparison. we conduct a simulation study with
sample sizes $T=50,100,200,400$, based on 50,000 Monte-Carlo simulations.
\label{tab:MLESimuSmall}}{\small\par}
\end{table}

\subsection*{Simulations when $\min_{k}\nu_{k}\protect\leq2$\label{subsec:Simulations_smallV}}

We conduct a simulation study to investigate the asymptotic properties
of the maximum likelihood estimators when $\min_{k}\nu_{k}\leq2$.
To provide a straightforward investigation for the effects of different
values of degree of freedom, we focus on the case when $\Xi$ is an
identity matrix, $\mu$ are zeros, and $\nu_{1}=\nu_{2}=\nu^{*}$.
The trivariate system is given by
\[
\mu=\left[\begin{array}{c}
\cellcolor{black!10}0\\
\cellcolor{black!5}0\\
\cellcolor{black!5}0
\end{array}\right],\quad\Xi=\left[\begin{array}{ccc}
\cellcolor{black!10}1 & \cellcolor{black!5}0 & \cellcolor{black!5}0\\
\cellcolor{black!10}0 & \cellcolor{black!5}1 & \cellcolor{black!5}0\\
\cellcolor{black!10}0 & \cellcolor{black!5}0 & \cellcolor{black!5}1
\end{array}\right]\quad\nu=\left[\begin{array}{c}
\cellcolor{black!10}\nu^{*}\\
\cellcolor{black!5}\nu^{*}
\end{array}\right]
\]
where the first two elements belong to one group, i.e. $n_{1}=1$
and $n_{2}=2$. We conduct a simulation study for several values of
$\nu^{*}=2,3/2,1,1/2$. To examine the converge rate of different
parameters, for each value of degree of freedom, we adopt $T=500,1000,2000,4000$,
based on 50,000 Monte-Carlo simulations. We report the mean, standard
deviation of the estimated parameters ${\rm Std}(\hat{\theta})$,
and two tails parameters, $\alpha_{L}^{0.025}$ and $\alpha_{R}^{0.025}$,
defined below
\[
\alpha_{L}^{0.025}=\frac{1}{T}\sum_{t=1}^{T}\left[\frac{\hat{\theta}_{t}-\theta}{{\rm Std}(\hat{\theta})}<-1.96\right],\quad\alpha_{R}^{0.025}=\frac{1}{T}\sum_{t=1}^{T}\left[\frac{\hat{\theta}_{t}-\theta}{{\rm Std}(\hat{\theta})}>1.96\right]
\]
and we will have $\alpha_{L}^{0.025}\rightarrow0.025$ and $\alpha_{R}^{0.025}\rightarrow0.025$
if $\hat{\theta}_{t}$ is approximately normally distributed. We also
include the ratio of standard deviations $R_{\sigma}$, defined as
\[
R_{\sigma}\left(T\right)=\frac{{\rm Std}(\hat{\theta}|\frac{T}{2})}{{\rm Std}(\hat{\theta}|T)},\quad{\rm for}\quad T=500,1000,2000,4000.
\]
If the converge rate is $\sqrt{T}$, we should have $R_{\sigma}\left(T\right)\rightarrow\sqrt{2}\approx1.4142$
when $T\rightarrow\infty$. Note that we also conduct a simulation
with $T=250$ to provide the ${\rm Std}(\hat{\theta})$ for $R_{\sigma}\left(500\right)$.

Tables \ref{tab:MLEv2}-\ref{tab:MLEv05} present the simulation results
for $\nu^{*}=2,3/2,1,1/2$, accordingly. Interestingly, the estimates
of the inverse degree of freedom parameter, appear to be well-approximated
by a Gaussian distribution in all these designs. Figures \ref{fig:nu2_N4000}-\ref{fig:nu05_N4000}
{\small{}present histograms of the elements of $\hat{\Xi}$ for the
sample size $T=4,000$ for} $\nu^{*}=2,3/2,1,1/2$, accordingly. These
designs have degrees of freedom that are smaller than the range covered
by Theorem \ref{thm:MLE-consistent-asN}, and the histograms reveal
a non-Gaussian limit distribution for all $\hat{\Xi}$-elements, which
is most pronounced for the elements in the off-diagonal blocks.

\begin{table}
\caption{Simulations with $\nu=2$}

\begin{centering}
\vspace{0.2cm}
\begin{footnotesize}
\begin{tabularx}{\textwidth}{p{0.9cm}Yp{-0.6cm}YYYYYp{-0.6cm}YYYYYY}
\toprule
\midrule
          & True &       & Mean   & Std   & $\alpha_L^{0.025}$  &$\alpha_R^{0.025}$ & $R_\sigma$ &       & Mean   & Std   & $\alpha_L^{0.025}$  &$\alpha_R^{0.025}$ & $R_\sigma$ \\
\midrule
\\[-0.3cm]
          &   &  &  \multicolumn{5}{c}{$T=500$}     &       & \multicolumn{5}{c}{$T=1,000$} \\
\\[-0.3cm]
    $\mu_1$   & 0     &       & 0.0002 & 0.0579 & 0.0250 & 0.0256 & 1.4332 &       & 0.0001 & 0.0407 & 0.0246 & 0.0246 & 1.4220 \\
    \\[-0.3cm]
    $\mu_2$   & 0     &       & 0.0003 & 0.0548 & 0.0250 & 0.0242 & 1.4184 &       & 0.0002 & 0.0387 & 0.0253 & 0.0255 & 1.4152 \\
    \\[-0.3cm]
    $\mu_3$   & 0     &       & 0.0006 & 0.0550 & 0.0231 & 0.0261 & 1.4125 &       & 0.0001 & 0.0390 & 0.0253 & 0.0250 & 1.4089 \\
\\[-0.3cm]
    \multirow{9}[0]{*}{$\rm{vec}(\Xi)$} & 1     &       & 0.9971 & 0.0635 & 0.0236 & 0.0272 & 1.4172 &       & 0.9986 & 0.0443 & 0.0232 & 0.0270 & 1.4324 \\
\\[-0.3cm]
          & \cellcolor{black!8}0     &  \cellcolor{black!8}     &\cellcolor{black!8} 0.0000 & \cellcolor{black!8}0.0233 &\cellcolor{black!8} 0.0296 &\cellcolor{black!8} 0.0291 & \cellcolor{black!8}1.5733 & \cellcolor{black!8}      &\cellcolor{black!8} 0.0000 &\cellcolor{black!8} 0.0154 &\cellcolor{black!8} 0.0290 &\cellcolor{black!8} 0.0290 &\cellcolor{black!8} 1.5108 \\
\\[-0.3cm]
          & \cellcolor{black!8}0     &   \cellcolor{black!8}    & \cellcolor{black!8}0.0002 & \cellcolor{black!8}0.0235 & \cellcolor{black!8}0.0281 & \cellcolor{black!8}0.0290 & \cellcolor{black!8}1.5541 &  \cellcolor{black!8}     &\cellcolor{black!8} 0.0001 &\cellcolor{black!8} 0.0154 &\cellcolor{black!8} 0.0280 & \cellcolor{black!8}0.0293 &\cellcolor{black!8} 1.5256 \\
\\[-0.3cm]
          & \cellcolor{black!8}0     &   \cellcolor{black!8}    & \cellcolor{black!8}0.0001 & \cellcolor{black!8} 0.0255 & \cellcolor{black!8}0.0296 & \cellcolor{black!8}0.0284 & \cellcolor{black!8}1.5660 &  \cellcolor{black!8}     & \cellcolor{black!8}0.0000 & \cellcolor{black!8}0.0166 & \cellcolor{black!8}0.0280 & \cellcolor{black!8}0.0279 & \cellcolor{black!8}1.5324 \\
\\[-0.3cm]
          & 1     &       & 0.9979 & 0.0526 & 0.0230 & 0.0273 & 1.4247 &       & 0.9990 & 0.0371 & 0.0237 & 0.0255 & 1.4179 \\
\\[-0.3cm]
          & 0     &       & 0.0002 & 0.0274 & 0.0252 & 0.0250 & 1.4121 &       & 0.0000 & 0.0194 & 0.0247 & 0.0249 & 1.4113 \\
\\[-0.3cm]
          & \cellcolor{black!8} 0     &  \cellcolor{black!8}     & \cellcolor{black!8}0.0000 & \cellcolor{black!8}0.0253 &\cellcolor{black!8} 0.0279 &\cellcolor{black!8} 0.0288 & \cellcolor{black!8}1.5845 & \cellcolor{black!8}      &\cellcolor{black!8} 0.0000 &\cellcolor{black!8} 0.0166 &\cellcolor{black!8} 0.0281 & \cellcolor{black!8}0.0276 & \cellcolor{black!8}1.5259 \\
\\[-0.3cm]
          & 0     &       & 0.0002 & 0.0274 & 0.0252 & 0.0250 & 1.4121 &       & 0.0000 & 0.0194 & 0.0247 & 0.0249 & 1.4113 \\
\\[-0.3cm]
          & 1     &       & 0.9974 & 0.0527 & 0.0244 & 0.0264 & 1.4205 &       & 0.9989 & 0.0373 & 0.0235 & 0.0270 & 1.4121 \\
\\[-0.3cm]
    $1/\nu_1$  & 0.5   &       & 0.4997 & 0.0574 & 0.0256 & 0.0253 & 1.4402 &       & 0.4998 & 0.0404 & 0.0254 & 0.0239 & 1.4202 \\
\\[-0.3cm]
    $1/\nu_2$  & 0.5   &       & 0.5001 & 0.0452 & 0.0239 & 0.0265 & 1.4277 &       & 0.4995 & 0.0318 & 0.0247 & 0.0243 & 1.4205 \\
\\[-0.3cm]
          &       &       &       &       &       &       &       &       &  \\
          &  &    & \multicolumn{5}{c}{$T=2,000$}     &       & \multicolumn{5}{c}{$T=4,000$} \\
\\[-0.2cm]
    $\mu_1$   & 0     &       & 0.0000 & 0.0290 & 0.0249 & 0.0251 & 1.4019 &       & 0.0001 & 0.0203 & 0.0246 & 0.0251 & 1.4308 \\
\\[-0.3cm]          
    $\mu_2$   & 0     &       & 0.0001 & 0.0275 & 0.0241 & 0.0251 & 1.4098 &       & 0.0000 & 0.0194 & 0.0252 & 0.0243 & 1.4196 \\
\\[-0.3cm]
    $\mu_3$   & 0     &       & 0.0002 & 0.0275 & 0.0244 & 0.0259 & 1.4190 &       & 0.0001 & 0.0194 & 0.0240 & 0.0251 & 1.4188 \\
\\[-0.3cm]
    \multirow{9}[0]{*}{$\rm{vec}(\Xi)$} & 1     &       & 0.9992 & 0.0316 & 0.0255 & 0.0249 & 1.4036 &       & 0.9997 & 0.0222 & 0.0246 & 0.0264 & 1.4228 \\
\\[-0.3cm]
          & \cellcolor{black!8}0     &   \cellcolor{black!8}    & \cellcolor{black!8}0.0000 & \cellcolor{black!8}0.0101 & \cellcolor{black!8}0.0282 & \cellcolor{black!8}0.0279 &\cellcolor{black!8} 1.5207 & \cellcolor{black!8}      & \cellcolor{black!8}0.0000 &\cellcolor{black!8} 0.0068 & \cellcolor{black!8}0.0287 & \cellcolor{black!8}0.0290 & \cellcolor{black!8}1.4903 \\
\\[-0.3cm]
          & \cellcolor{black!8}0     &   \cellcolor{black!8}    & \cellcolor{black!8}0.0001 &\cellcolor{black!8} 0.0102 & \cellcolor{black!8}0.0278 &\cellcolor{black!8} 0.0288 & \cellcolor{black!8}1.5194 &   \cellcolor{black!8}    &\cellcolor{black!8} 0.0001 &\cellcolor{black!8} 0.0068 & \cellcolor{black!8}0.0292 & \cellcolor{black!8}0.0283 & \cellcolor{black!8}1.4906 \\
\\[-0.3cm]
          & \cellcolor{black!8}0     &   \cellcolor{black!8}    &\cellcolor{black!8} 0.0000 & \cellcolor{black!8} 0.0109 & \cellcolor{black!8} 0.0283 & \cellcolor{black!8}0.0276 & \cellcolor{black!8}1.5258 &  \cellcolor{black!8}     &\cellcolor{black!8} 0.0000 &\cellcolor{black!8} 0.0074 & \cellcolor{black!8}0.0276 & \cellcolor{black!8}0.0275 &\cellcolor{black!8} 1.4805 \\
\\[-0.3cm]
          & 1     &       & 0.9994 & 0.0262 & 0.0254 & 0.0247 & 1.4175 &       & 0.9997 & 0.0185 & 0.0246 & 0.0258 & 1.4158 \\
\\[-0.3cm]
          & 0     &       & 0.0001 & 0.0137 & 0.0254 & 0.0257 & 1.4156 &       & 0.0000 & 0.0097 & 0.0250 & 0.0254 & 1.4177 \\
\\[-0.3cm]
          & \cellcolor{black!8}0     & \cellcolor{black!8}      & \cellcolor{black!8}0.0001 & \cellcolor{black!8}0.0109 & \cellcolor{black!8}0.0278 &\cellcolor{black!8} 0.0268 & \cellcolor{black!8}1.5125 & \cellcolor{black!8}      & \cellcolor{black!8}0.0000 & \cellcolor{black!8}0.0073 & \cellcolor{black!8}0.0274 & \cellcolor{black!8}0.0279 &\cellcolor{black!8} 1.4951 \\
\\[-0.3cm]
          & 0     &       & 0.0001 & 0.0137 & 0.0254 & 0.0257 & 1.4156 &       & 0.0000 & 0.0097 & 0.0250 & 0.0254 & 1.4177 \\
\\[-0.3cm]
          & 1     &       & 0.9995 & 0.0262 & 0.0241 & 0.0263 & 1.4239 &       & 0.9998 & 0.0186 & 0.0239 & 0.0251 & 1.4087 \\
\\[-0.3cm]
    $1/\nu_1$  & 0.5   &       & 0.4999 & 0.0285 & 0.0253 & 0.0253 & 1.4172 &       & 0.5000 & 0.0201 & 0.0245 & 0.0247 & 1.4176 \\
\\[-0.3cm]
   $1/\nu_2$  & 0.5   &       & 0.4999 & 0.0223 & 0.0255 & 0.0246 & 1.4241 &       & 0.4999 & 0.0158 & 0.0255 & 0.0251 & 1.4127 \\
\\[0.0cm]
\\[-0.5cm]
\midrule
\bottomrule
\end{tabularx}
\end{footnotesize}
\par\end{centering}
{\small{}Note: The mean, standard deviation of the MLE estimated parameters
for convolution-$t$ distribution when $\nu_{1}=\nu_{2}=2$. We conduct
a simulation study with sample sizes $T=500,1000,2000,4000$, based
on 50,000 Monte-Carlo simulations. This design has $\nu_{1}=\nu_{2}=2$
that is outside the range covered by the results in Theorem \ref{thm:MLE-consistent-asN}.
The shaded regions highlights the coefficients in the off-diagonal
blocks of $\Xi$, which appear to converge at a faster rate. \label{tab:MLEv2}}{\small\par}
\end{table}

\begin{table}
\caption{Simulations with  $\nu=3/2$}

\begin{centering}
\vspace{0.2cm}
\begin{footnotesize}
\begin{tabularx}{\textwidth}{p{0.9cm}Yp{-0.6cm}YYYYYp{-0.6cm}YYYYYY}
\toprule
\midrule
          & True &       & Mean   & Std   & $\alpha_L^{0.025}$  &$\alpha_R^{0.025}$ & $R_\sigma$ &       & Mean   & Std   & $\alpha_L^{0.025}$  &$\alpha_R^{0.025}$ & $R_\sigma$ \\
\midrule
\\[-0.3cm]
          &   &  &  \multicolumn{5}{c}{$T=500$}     &       & \multicolumn{5}{c}{$T=1,000$} \\
\\[-0.3cm]
    $\mu_1$   & 0     &       & 0.0002 & 0.0604 & 0.0246 & 0.0256 & 1.4275 &       & 0.0002 & 0.0423 & 0.0246 & 0.0255 & 1.4279 \\
\\[-0.3cm]    
    $\mu_2$   & 0     &       & 0.0003 & 0.0564 & 0.0246 & 0.0248 & 1.4171 &       & 0.0003 & 0.0398 & 0.0259 & 0.0247 & 1.4167 \\
\\[-0.3cm]
    $\mu_3$   & 0     &       & 0.0005 & 0.0565 & 0.0243 & 0.0243 & 1.4226 &       & 0.0001 & 0.0397 & 0.0254 & 0.0256 & 1.4230 \\
\\[-0.3cm]
    \multirow{9}[0]{*}{$\rm{vec}(\Xi)$} & 1     &       & 0.9964 & 0.0679 & 0.0233 & 0.0260 & 1.4210 &       & 0.9982 & 0.0474 & 0.0242 & 0.0260 & 1.4336 \\
\\[-0.3cm]
          & \cellcolor{black!8}0 & \cellcolor{black!8} & \cellcolor{black!8}0.0000 & \cellcolor{black!8}0.0142 & \cellcolor{black!8}0.0321 & \cellcolor{black!8}0.0320 & \cellcolor{black!8}1.6632 & \cellcolor{black!8} & \cellcolor{black!8}0.0001 & \cellcolor{black!8}0.0086 & \cellcolor{black!8}0.0324 & \cellcolor{black!8}0.0310 & \cellcolor{black!8}1.6536 \\
\\[-0.3cm]
          & \cellcolor{black!8}0 & \cellcolor{black!8} & \cellcolor{black!8}0.0000 & \cellcolor{black!8}0.0142 & \cellcolor{black!8}0.0320 & \cellcolor{black!8}0.0317 & \cellcolor{black!8}1.6913 & \cellcolor{black!8} & \cellcolor{black!8}0.0000 & \cellcolor{black!8}0.0086 & \cellcolor{black!8}0.0314 & \cellcolor{black!8}0.0307 & \cellcolor{black!8}1.6453 \\
\\[-0.3cm]
          & \cellcolor{black!8}0 & \cellcolor{black!8} & \cellcolor{black!8}0.0000 & \cellcolor{black!8}0.0160 & \cellcolor{black!8}0.0298 & \cellcolor{black!8}0.0306 & \cellcolor{black!8}1.6909 & \cellcolor{black!8} & \cellcolor{black!8}0.0001 & \cellcolor{black!8}0.0099 & \cellcolor{black!8}0.0292 & \cellcolor{black!8}0.0314 & \cellcolor{black!8}1.6173 \\
\\[-0.3cm]
          & 1     &       & 0.9971 & 0.0559 & 0.0240 & 0.0255 & 1.4228 &       & 0.9986 & 0.0397 & 0.0238 & 0.0256 & 1.4066 \\
\\[-0.3cm]
          & 0     &       & 0.0000 & 0.0280 & 0.0246 & 0.0258 & 1.4208 &       & 0.0001 & 0.0199 & 0.0242 & 0.0252 & 1.4077 \\
\\[-0.3cm]
          & \cellcolor{black!8}0 & \cellcolor{black!8} & \cellcolor{black!8}0.0001 & \cellcolor{black!8}0.0162 & \cellcolor{black!8}0.0293 & \cellcolor{black!8}0.0315 & \cellcolor{black!8}1.6665 & \cellcolor{black!8} & \cellcolor{black!8}0.0000 & \cellcolor{black!8}0.0098 & \cellcolor{black!8}0.0299 & \cellcolor{black!8}0.0299 & \cellcolor{black!8}1.6539 \\
\\[-0.3cm]
          & 0     &       & 0.0000 & 0.0280 & 0.0246 & 0.0258 & 1.4208 &       & 0.0001 & 0.0199 & 0.0242 & 0.0252 & 1.4077 \\
\\[-0.3cm]
          & 1     &       & 0.9970 & 0.0557 & 0.0234 & 0.0259 & 1.4266 &       & 0.9987 & 0.0396 & 0.0236 & 0.0261 & 1.4069 \\
\\[-0.3cm]
    $1/\nu_1$  & 0.6667 &       & 0.6669 & 0.0650 & 0.0245 & 0.0262 & 1.4305 &       & 0.6669 & 0.0457 & 0.0242 & 0.0266 & 1.4219 \\
\\[-0.3cm]
    $1/\nu_2$  & 0.6667 &       & 0.6674 & 0.0527 & 0.0228 & 0.0268 & 1.4173 &       & 0.6667 & 0.0372 & 0.0243 & 0.0260 & 1.4167 \\
\\[-0.3cm]
          &       &       &       &       &       &       &       &       &  \\
          &  &    & \multicolumn{5}{c}{$T=2,000$}     &       & \multicolumn{5}{c}{$T=4,000$} \\
\\[-0.2cm]
    $\mu_1$   & 0     &       & 0.0000 & 0.0300 & 0.0250 & 0.0255 & 1.4115 &       & 0.0001 & 0.0212 & 0.0250 & 0.0257 & 1.4120 \\
\\[-0.3cm]          
    $\mu_2$   & 0     &       & 0.0001 & 0.0280 & 0.0246 & 0.0259 & 1.4215 &       & 0.0001 & 0.0198 & 0.0245 & 0.0245 & 1.4152 \\
\\[-0.3cm]
    $\mu_3$   & 0     &       & 0.0000 & 0.0279 & 0.0247 & 0.0240 & 1.4219 &       & 0.0000 & 0.0198 & 0.0244 & 0.0253 & 1.4073 \\
\\[-0.3cm]
    \multirow{9}[0]{*}{$\rm{vec}(\Xi)$} & 1     &       & 0.9990 & 0.0335 & 0.0247 & 0.0252 & 1.4136 &       & 0.9994 & 0.0236 & 0.0258 & 0.0239 & 1.4205 \\
\\[-0.3cm]
          & \cellcolor{black!8}0 & \cellcolor{black!8} & \cellcolor{black!8}0.0000 & \cellcolor{black!8}0.0053 & \cellcolor{black!8}0.0310 & \cellcolor{black!8}0.0312 & \cellcolor{black!8}1.6105 & \cellcolor{black!8} & \cellcolor{black!8}0.0000 & \cellcolor{black!8}0.0033 & \cellcolor{black!8}0.0319 & \cellcolor{black!8}0.0320 & \cellcolor{black!8}1.6321 \\
\\[-0.3cm]
          & \cellcolor{black!8}0 & \cellcolor{black!8} & \cellcolor{black!8}0.0000 & \cellcolor{black!8}0.0053 & \cellcolor{black!8}0.0322 & \cellcolor{black!8}0.0307 & \cellcolor{black!8}1.6211 & \cellcolor{black!8} & \cellcolor{black!8}0.0000 & \cellcolor{black!8}0.0033 & \cellcolor{black!8}0.0325 & \cellcolor{black!8}0.0313 & \cellcolor{black!8}1.6264 \\
\\[-0.3cm]
          & \cellcolor{black!8}0 & \cellcolor{black!8} & \cellcolor{black!8}0.0000 & \cellcolor{black!8}0.0060 & \cellcolor{black!8}0.0305 & \cellcolor{black!8}0.0291 & \cellcolor{black!8}1.6457 & \cellcolor{black!8} & \cellcolor{black!8}0.0000 & \cellcolor{black!8}0.0037 & \cellcolor{black!8}0.0310 & \cellcolor{black!8}0.0297 & \cellcolor{black!8}1.6102 \\
\\[-0.3cm]
          & 1     &       & 0.9992 & 0.0279 & 0.0250 & 0.0238 & 1.4268 &       & 0.9997 & 0.0198 & 0.0241 & 0.0252 & 1.4080 \\
\\[-0.3cm]
          & 0     &       & 0.0000 & 0.0140 & 0.0250 & 0.0241 & 1.4205 &       & 0.0000 & 0.0100 & 0.0254 & 0.0256 & 1.4084 \\
\\[-0.3cm]
          & \cellcolor{black!8}0 & \cellcolor{black!8} & \cellcolor{black!8}0.0000 & \cellcolor{black!8}0.0060 & \cellcolor{black!8}0.0311 & \cellcolor{black!8}0.0289 & \cellcolor{black!8}1.6245 & \cellcolor{black!8} & \cellcolor{black!8}0.0000 & \cellcolor{black!8}0.0037 & \cellcolor{black!8}0.0306 & \cellcolor{black!8}0.0313 & \cellcolor{black!8}1.6164 \\
\\[-0.3cm]
          & 0     &       & 0.0000 & 0.0140 & 0.0250 & 0.0241 & 1.4205 &       & 0.0000 & 0.0100 & 0.0254 & 0.0256 & 1.4084 \\
\\[-0.3cm]
          & 1     &       & 0.9992 & 0.0280 & 0.0250 & 0.0249 & 1.4160 &       & 0.9997 & 0.0197 & 0.0239 & 0.0256 & 1.4166 \\
\\[-0.3cm]
  $1/\nu_1$  & 0.6667 &       & 0.6669 & 0.0321 & 0.0235 & 0.0257 & 1.4226 &       & 0.6667 & 0.0227 & 0.0260 & 0.0242 & 1.4135 \\
\\[-0.3cm]
   $1/\nu_2$  & 0.6667 &       & 0.6666 & 0.0262 & 0.0239 & 0.0250 & 1.4183 &       & 0.6668 & 0.0185 & 0.0241 & 0.0261 & 1.4209 \\
\\[0.0cm]
\\[-0.5cm]
\midrule
\bottomrule
\end{tabularx}
\end{footnotesize}
\par\end{centering}
{\small{}Note: The mean, standard deviation of the MLE estimated parameters
for convolution-$t$ distribution $\nu_{1}=\nu_{2}=3/2$. We conduct
a simulation study with sample sizes $T=500,1000,2000,4000$, based
on 50,000 Monte-Carlo simulations. This design has $\nu_{1}=\nu_{2}=3/2$
that is outside the range covered by the results in Theorem \ref{thm:MLE-consistent-asN}.
The shaded regions highlights the coefficients in the off-diagonal
blocks of $\Xi$, which appear to converge at a faster rate. \label{tab:MLEv15}}{\small\par}
\end{table}

\begin{table}
\caption{Simulations with $\nu=1$}

\begin{centering}
\vspace{0.2cm}
\begin{footnotesize}
\begin{tabularx}{\textwidth}{p{0.9cm}Yp{-0.6cm}YYYYYp{-0.6cm}YYYYYY}
\toprule
\midrule
          & True &       & Mean   & Std   & $\alpha_L^{0.025}$  &$\alpha_R^{0.025}$ & $R_\sigma$ &       & Mean   & Std   & $\alpha_L^{0.025}$  &$\alpha_R^{0.025}$ & $R_\sigma$ \\
\midrule
\\[-0.3cm]
          &   &  &  \multicolumn{5}{c}{$T=500$}     &       & \multicolumn{5}{c}{$T=1,000$} \\
\\[-0.3cm]
    $\mu_1$   & 0     &       & 0.0012 & 0.0571 & 0.0251 & 0.0273 & 1.4230 &       & 0.0006 & 0.0406 & 0.0254 & 0.0281 & 1.4068 \\
\\[-0.3cm]    
    $\mu_2$   & 0     &       & 0.0009 & 0.0530 & 0.0253 & 0.0275 & 1.4164 &       & 0.0006 & 0.0377 & 0.0243 & 0.0279 & 1.4067 \\
\\[-0.3cm]
    $\mu_3$   & 0     &       & 0.0015 & 0.0531 & 0.0240 & 0.0287 & 1.4141 &       & 0.0006 & 0.0375 & 0.0256 & 0.0266 & 1.4168 \\
\\[-0.3cm]
    \multirow{9}[0]{*}{$\rm{vec}(\Xi)$} & 1     &       & 0.9970 & 0.0662 & 0.0236 & 0.0291 & 1.4115 &       & 0.9981 & 0.0467 & 0.0256 & 0.0275 & 1.4180 \\
\\[-0.3cm]
          & \cellcolor{black!8}0 & \cellcolor{black!8} & \cellcolor{black!8}0.0000 & \cellcolor{black!8}0.0046 & \cellcolor{black!8}0.0311 & \cellcolor{black!8}0.0316 & \cellcolor{black!8}2.0418 & \cellcolor{black!8} & \cellcolor{black!8}0.0000 & \cellcolor{black!8}0.0023 & \cellcolor{black!8}0.0322 & \cellcolor{black!8}0.0310 & \cellcolor{black!8}2.0194 \\
\\[-0.3cm]
          & \cellcolor{black!8}0 & \cellcolor{black!8} & \cellcolor{black!8}0.0000 & \cellcolor{black!8}0.0046 & \cellcolor{black!8}0.0312 & \cellcolor{black!8}0.0297 & \cellcolor{black!8}1.9968 & \cellcolor{black!8} & \cellcolor{black!8}0.0000 & \cellcolor{black!8}0.0023 & \cellcolor{black!8}0.0308 & \cellcolor{black!8}0.0308 & \cellcolor{black!8}2.0126 \\
\\[-0.3cm]
          & \cellcolor{black!8}0 & \cellcolor{black!8} & \cellcolor{black!8}0.0000 & \cellcolor{black!8}0.0058 & \cellcolor{black!8}0.0308 & \cellcolor{black!8}0.0317 & \cellcolor{black!8}2.0385 & \cellcolor{black!8} & \cellcolor{black!8}0.0000 & \cellcolor{black!8}0.0030 & \cellcolor{black!8}0.0320 & \cellcolor{black!8}0.0319 & \cellcolor{black!8}1.9756 \\
\\[-0.3cm]
          & 1     &       & 0.9982 & 0.0542 & 0.0243 & 0.0281 & 1.4208 &       & 0.9995 & 0.0388 & 0.0244 & 0.0280 & 1.3956 \\
\\[-0.3cm]
          & 0     &       & 0.0004 & 0.0286 & 0.0254 & 0.0257 & 1.4143 &       & 0.0001 & 0.0201 & 0.0249 & 0.0258 & 1.4202 \\
\\[-0.3cm]
          & \cellcolor{black!8}0 & \cellcolor{black!8} & \cellcolor{black!8}0.0000 & \cellcolor{black!8}0.0059 & \cellcolor{black!8}0.0313 & \cellcolor{black!8}0.0314 & \cellcolor{black!8}2.0388 & \cellcolor{black!8} & \cellcolor{black!8}0.0000 & \cellcolor{black!8}0.0029 & \cellcolor{black!8}0.0314 & \cellcolor{black!8}0.0312 & \cellcolor{black!8}1.9978 \\
\\[-0.3cm]
          & 0     &       & 0.0004 & 0.0286 & 0.0254 & 0.0257 & 1.4143 &       & 0.0001 & 0.0201 & 0.0249 & 0.0258 & 1.4202 \\
\\[-0.3cm]
          & 1     &       & 0.9975 & 0.0546 & 0.0244 & 0.0281 & 1.4134 &       & 0.9993 & 0.0388 & 0.0238 & 0.0282 & 1.4048 \\
\\[-0.3cm]
    $1/\nu_1$  & 1     &       & 1.0016 & 0.0690 & 0.0245 & 0.0294 & 1.4058 &       & 1.0008 & 0.0492 & 0.0262 & 0.0275 & 1.4023 \\
\\[-0.3cm]
    $1/\nu_2$  & 1     &       & 1.0017 & 0.0612 & 0.0232 & 0.0290 & 1.4062 &       & 1.0006 & 0.0431 & 0.0245 & 0.0293 & 1.4188 \\
\\[-0.3cm]
          &       &       &       &       &       &       &       &       &  \\
          &  &    & \multicolumn{5}{c}{$T=2,000$}     &       & \multicolumn{5}{c}{$T=4,000$} \\
\\[-0.2cm]
    $\mu_1$   & 0     &       & 0.0005 & 0.0287 & 0.0246 & 0.0283 & 1.4114 &       & 0.0003 & 0.0204 & 0.0257 & 0.0279 & 1.4078 \\
\\[-0.3cm]          
    $\mu_2$   & 0     &       & 0.0003 & 0.0266 & 0.0260 & 0.0269 & 1.4149 &       & 0.0001 & 0.0189 & 0.0262 & 0.0259 & 1.4100 \\
\\[-0.3cm]
    $\mu_3$   & 0     &       & 0.0003 & 0.0266 & 0.0259 & 0.0272 & 1.4115 &       & 0.0002 & 0.0188 & 0.0259 & 0.0267 & 1.4125 \\
\\[-0.3cm]
    \multirow{9}[0]{*}{$\rm{vec}(\Xi)$} & 1     &       & 0.9990 & 0.0331 & 0.0253 & 0.0281 & 1.4108 &       & 0.9996 & 0.0234 & 0.0268 & 0.0274 & 1.4113 \\
\\[-0.3cm]
          & \cellcolor{black!8}0 & \cellcolor{black!8} & \cellcolor{black!8}0.0000 & \cellcolor{black!8}0.0011 & \cellcolor{black!8}0.0308 & \cellcolor{black!8}0.0311 & \cellcolor{black!8}1.9808 & \cellcolor{black!8} & \cellcolor{black!8}0.0000 & \cellcolor{black!8}0.0006 & \cellcolor{black!8}0.0301 & \cellcolor{black!8}0.0319 & \cellcolor{black!8}2.0367 \\
\\[-0.3cm]
          & \cellcolor{black!8}0 & \cellcolor{black!8} & \cellcolor{black!8}0.0000 & \cellcolor{black!8}0.0011 & \cellcolor{black!8}0.0304 & \cellcolor{black!8}0.0312 & \cellcolor{black!8}2.0205 & \cellcolor{black!8} & \cellcolor{black!8}0.0000 & \cellcolor{black!8}0.0006 & \cellcolor{black!8}0.0309 & \cellcolor{black!8}0.0313 & \cellcolor{black!8}2.0302 \\
\\[-0.3cm]
          & \cellcolor{black!8}0 & \cellcolor{black!8} & \cellcolor{black!8}0.0000 & \cellcolor{black!8}0.0015 & \cellcolor{black!8}0.0320 & \cellcolor{black!8}0.0304 & \cellcolor{black!8}2.0041 & \cellcolor{black!8} & \cellcolor{black!8}0.0000 & \cellcolor{black!8}0.0007 & \cellcolor{black!8}0.0315 & \cellcolor{black!8}0.0311 & \cellcolor{black!8}2.0121 \\
\\[-0.3cm]
          & 1     &       & 0.9999 & 0.0275 & 0.0237 & 0.0286 & 1.4132 &       & 0.9997 & 0.0196 & 0.0257 & 0.0274 & 1.3995 \\
\\[-0.3cm]
          & 0     &       & 0.0001 & 0.0142 & 0.0250 & 0.0259 & 1.4159 &       & 0.0001 & 0.0100 & 0.0243 & 0.0268 & 1.4223 \\
\\[-0.3cm]
          & \cellcolor{black!8}0 & \cellcolor{black!8} & \cellcolor{black!8}0.0000 & \cellcolor{black!8}0.0015 & \cellcolor{black!8}0.0306 & \cellcolor{black!8}0.0314 & \cellcolor{black!8}2.0211 & \cellcolor{black!8} & \cellcolor{black!8}0.0000 & \cellcolor{black!8}0.0007 & \cellcolor{black!8}0.0303 & \cellcolor{black!8}0.0308 & \cellcolor{black!8}1.9923 \\
\\[-0.3cm]
          & 0     &       & 0.0001 & 0.0142 & 0.0250 & 0.0259 & 1.4159 &       & 0.0001 & 0.0100 & 0.0243 & 0.0268 & 1.4223 \\
\\[-0.3cm]
          & 1     &       & 0.9997 & 0.0276 & 0.0262 & 0.0278 & 1.4079 &       & 0.9999 & 0.0197 & 0.0252 & 0.0278 & 1.3992 \\
\\[-0.3cm]
    $1/\nu_1$  & 1     &       & 1.0003 & 0.0347 & 0.0263 & 0.0282 & 1.4183 &       & 1.0004 & 0.0246 & 0.0259 & 0.0286 & 1.4120 \\
\\[-0.3cm]
    $1/\nu_2$  & 1     &       & 1.0004 & 0.0306 & 0.0254 & 0.0282 & 1.4095 &       & 1.0002 & 0.0217 & 0.0251 & 0.0278 & 1.4110 \\
\\[0.0cm]
\\[-0.5cm]
\midrule
\bottomrule
\end{tabularx}
\end{footnotesize}
\par\end{centering}
{\small{}Note: The mean, standard deviation of the MLE estimated parameters
for convolution-$t$ distribution $\nu_{1}=\nu_{2}=1$. We conduct
a simulation study with sample sizes $T=500,1000,2000,4000$, based
on 50,000 Monte-Carlo simulations. This design has $\nu_{1}=\nu_{2}=1$
that is outside the range covered by the results in Theorem \ref{thm:MLE-consistent-asN}.
The shaded regions highlights the coefficients in the off-diagonal
blocks of $\Xi$, which appear to converge at a faster rate.\label{tab:MLEv1}}{\small\par}
\end{table}

\begin{table}
\caption{Simulations with $\nu=1/2$}

\begin{centering}
\vspace{0.2cm}
\begin{footnotesize}
\begin{tabularx}{\textwidth}{p{0.8cm}Yp{-0.6cm}YYYYYp{-0.6cm}YYYYYY}
\toprule
\midrule
          & True &       & Mean   & Std   & $\alpha_L^{0.025}$  &$\alpha_R^{0.025}$ & $R_\sigma$ &       & Mean   & Std   & $\alpha_L^{0.025}$  &$\alpha_R^{0.025}$ & $R_\sigma$ \\
\midrule
\\[-0.3cm]
          &   &  &  \multicolumn{5}{c}{$T=500$}     &       & \multicolumn{5}{c}{$T=1,000$} \\
\\[-0.3cm]
    $\mu_1$   & 0     &       & 0.0008 & 0.0519 & 0.0320 & 0.0287 & 1.4181 &       & 0.0006 & 0.0362 & 0.0311 & 0.0292 & 1.4329 \\
\\[-0.3cm]    
    $\mu_2$   & 0     &       & 0.0008 & 0.0469 & 0.0303 & 0.0303 & 1.4207 &       & 0.0001 & 0.0332 & 0.0315 & 0.0287 & 1.4112 \\
\\[-0.3cm]
    $\mu_3$   & 0     &       & 0.0004 & 0.0473 & 0.0310 & 0.0286 & 1.4094 &       & 0.0002 & 0.0333 & 0.0325 & 0.0282 & 1.4213 \\
\\[-0.3cm]
    \multirow{9}[0]{*}{$\rm{vec}(\Xi)$} & 1     &       & 0.9962 & 0.0646 & 0.0333 & 0.0286 & 1.4107 &       & 0.9977 & 0.0459 & 0.0336 & 0.0273 & 1.4083 \\
\\[-0.3cm]
          & \cellcolor{black!8}0 & \cellcolor{black!8} & \cellcolor{black!8}0.0000 & \cellcolor{black!8} 1.0261$^\dagger$ & \cellcolor{black!8}0.0180 & \cellcolor{black!8}0.0188 & \cellcolor{black!8}3.9108 & \cellcolor{black!8} & \cellcolor{black!8}0.0000 & \cellcolor{black!8}0.2622$^\dagger$ & \cellcolor{black!8}0.0183 & \cellcolor{black!8}0.0178 & \cellcolor{black!8}3.9136 \\
\\[-0.3cm]
          & \cellcolor{black!8}0 & \cellcolor{black!8} & \cellcolor{black!8}0.0000 & \cellcolor{black!8}1.0403$^\dagger$ & \cellcolor{black!8}0.0178 & \cellcolor{black!8}0.0190 & \cellcolor{black!8}3.7635 & \cellcolor{black!8} & \cellcolor{black!8}0.0000 & \cellcolor{black!8}0.2644$^\dagger$ & \cellcolor{black!8}0.0184 & \cellcolor{black!8}0.0179 & \cellcolor{black!8}3.9347 \\
\\[-0.3cm]
          & \cellcolor{black!8}0 & \cellcolor{black!8} & \cellcolor{black!8}0.0000 & \cellcolor{black!8}1.5685$^\dagger$ & \cellcolor{black!8}0.0217 & \cellcolor{black!8}0.0217 & \cellcolor{black!8}4.1445 & \cellcolor{black!8} & \cellcolor{black!8}0.0000 & \cellcolor{black!8}0.4179$^\dagger$ & \cellcolor{black!8}0.0208 & \cellcolor{black!8}0.0213 & \cellcolor{black!8}3.7537 \\
\\[-0.3cm]
          & 1     &       & 0.9989 & 0.0506 & 0.0306 & 0.0297 & 1.4193 &       & 0.9995 & 0.0363 & 0.0311 & 0.0312 & 1.3929 \\
\\[-0.3cm]
          & 0     &       & 0.0004 & 0.0271 & 0.0269 & 0.0277 & 1.4190 &       & 0.0002 & 0.0190 & 0.0279 & 0.0283 & 1.4270 \\
\\[-0.3cm]
          & \cellcolor{black!8}0 & \cellcolor{black!8} & \cellcolor{black!8}0.0000 & \cellcolor{black!8}1.5737$^\dagger$ & \cellcolor{black!8}0.0211 & \cellcolor{black!8}0.0215 & \cellcolor{black!8}4.0704 & \cellcolor{black!8} & \cellcolor{black!8}0.0000 & \cellcolor{black!8}0.4232$^\dagger$ & \cellcolor{black!8}0.0207 & \cellcolor{black!8}0.0211 & \cellcolor{black!8}3.7182 \\
\\[-0.3cm]
          & 0     &       & 0.0004 & 0.0271 & 0.0269 & 0.0277 & 1.4190 &       & 0.0002 & 0.0190 & 0.0279 & 0.0283 & 1.4270 \\
\\[-0.3cm]
          & 1     &       & 0.9989 & 0.0509 & 0.0299 & 0.0312 & 1.4046 &       & 0.9993 & 0.0360 & 0.0315 & 0.0289 & 1.4131 \\
\\[-0.3cm]
    $1/\nu_1$  & 2     &       & 2.0026 & 0.1201 & 0.0221 & 0.0283 & 1.4040 &       & 2.0012 & 0.0841 & 0.0226 & 0.0267 & 1.4293 \\
\\[-0.3cm]
    $1/\nu_2$  & 2     &       & 2.0018 & 0.1089 & 0.0219 & 0.0286 & 1.4223 &       & 2.0010 & 0.0768 & 0.0236 & 0.0262 & 1.4173 \\
\\[-0.3cm]
          &       &       &       &       &       &       &       &       &  \\
          &  &    & \multicolumn{5}{c}{$T=2,000$}     &       & \multicolumn{5}{c}{$T=4,000$} \\
\\[-0.2cm]
    $\mu_1$   & 0     &       & 0.0001 & 0.0263 & 0.0332 & 0.0288 & 1.3797 &       & 0.0001 & 0.0190 & 0.0318 & 0.0292 & 1.3828 \\
\\[-0.3cm]          
    $\mu_2$   & 0     &       & 0.0002 & 0.0239 & 0.0339 & 0.0271 & 1.3910 &       & 0.0000 & 0.0173 & 0.0314 & 0.0281 & 1.3815 \\
\\[-0.3cm]
    $\mu_3$   & 0     &       & 0.0000 & 0.0237 & 0.0313 & 0.0280 & 1.4020 &       & 0.0001 & 0.0173 & 0.0308 & 0.0288 & 1.3673 \\
\\[-0.3cm]
    \multirow{9}[0]{*}{$\rm{vec}(\Xi)$} & 1     &       & 0.9989 & 0.0325 & 0.0353 & 0.0290 & 1.4099 &       & 0.9993 & 0.0240 & 0.0350 & 0.0282 & 1.3547 \\
\\[-0.3cm]
          & \cellcolor{black!8}0 & \cellcolor{black!8} & \cellcolor{black!8}0.0000 & \cellcolor{black!8}0.0658$^\dagger$ & \cellcolor{black!8}0.0176 & \cellcolor{black!8}0.0185 & \cellcolor{black!8}3.9835 & \cellcolor{black!8} & \cellcolor{black!8}0.0000 & \cellcolor{black!8}0.0160$^\dagger$ & \cellcolor{black!8}0.0173 & \cellcolor{black!8}0.0175 & \cellcolor{black!8}4.1040 \\
\\[-0.3cm]
          & \cellcolor{black!8}0 & \cellcolor{black!8} & \cellcolor{black!8}0.0000 & \cellcolor{black!8}0.0675$^\dagger$ & \cellcolor{black!8}0.0171 & \cellcolor{black!8}0.0179 & \cellcolor{black!8}3.9150 & \cellcolor{black!8} & \cellcolor{black!8}0.0000 & \cellcolor{black!8}0.0153$^\dagger$ & \cellcolor{black!8}0.0188 & \cellcolor{black!8}0.0188 & \cellcolor{black!8}4.4151 \\
\\[-0.3cm]
          & \cellcolor{black!8}0 & \cellcolor{black!8} & \cellcolor{black!8}0.0000 & \cellcolor{black!8}0.1058$^\dagger$ & \cellcolor{black!8}0.0195 & \cellcolor{black!8}0.0215 & \cellcolor{black!8}3.9487 & \cellcolor{black!8} & \cellcolor{black!8}0.0000 & \cellcolor{black!8}0.0257$^\dagger$ & \cellcolor{black!8}0.0213 & \cellcolor{black!8}0.0218 & \cellcolor{black!8}4.1253 \\
\\[-0.3cm]
          & 1     &       & 0.9997 & 0.0259 & 0.0313 & 0.0296 & 1.4041 &       & 0.9998 & 0.0190 & 0.0314 & 0.0302 & 1.3608 \\
\\[-0.3cm]
          & 0     &       & 0.0001 & 0.0136 & 0.0282 & 0.0277 & 1.3949 &       & 0.0001 & 0.0097 & 0.0280 & 0.0280 & 1.3985 \\
\\[-0.3cm]
          & \cellcolor{black!8}0 & \cellcolor{black!8} & \cellcolor{black!8}0.0000 & \cellcolor{black!8}0.1071$^\dagger$ & \cellcolor{black!8}0.0200 & \cellcolor{black!8}0.0216 & \cellcolor{black!8}3.9518 & \cellcolor{black!8} & \cellcolor{black!8}0.0000 & \cellcolor{black!8}0.0262$^\dagger$ & \cellcolor{black!8}0.0207 & \cellcolor{black!8}0.0210 & \cellcolor{black!8}4.0892 \\
\\[-0.3cm]
          & 0     &       & 0.0001 & 0.0136 & 0.0282 & 0.0277 & 1.3949 &       & 0.0001 & 0.0097 & 0.0280 & 0.0280 & 1.3985 \\
\\[-0.3cm]
          & 1     &       & 0.9997 & 0.0259 & 0.0315 & 0.0303 & 1.3908 &       & 0.9997 & 0.0190 & 0.0320 & 0.0287 & 1.3608 \\
\\[-0.3cm]
    $1/\nu_1$  & 2     &       & 2.0006 & 0.0596 & 0.0242 & 0.0257 & 1.4109 &       & 2.0004 & 0.0423 & 0.0240 & 0.0264 & 1.4086 \\
\\[-0.3cm]
    $1/\nu_2$  & 2     &       & 2.0008 & 0.0545 & 0.0219 & 0.0271 & 1.4092 &       & 2.0005 & 0.0387 & 0.0236 & 0.0261 & 1.4106 \\
\\[0.0cm]
\\[-0.5cm]
\midrule
\bottomrule
\end{tabularx}
\end{footnotesize}
\par\end{centering}
{\small{}Note: The mean, standard deviation of the MLE estimated parameters
for convolution-$t$ distribution $\nu_{1}=\nu_{2}=1/2$. We conduct
a simulation study with sample sizes $T=500,1000,2000,4000$, based
on 50,000 Monte-Carlo simulations. The unit of the values with $\dagger$
in the superscript is $10^{-4}$. This design has $\nu_{1}=\nu_{2}=1/2$
that is outside the range covered by the results in Theorem \ref{thm:MLE-consistent-asN}.
The shaded regions highlights the coefficients in the off-diagonal
blocks of $\Xi$, which appear to converge at a faster rate.\label{tab:MLEv05}}{\small\par}
\end{table}

\begin{figure}[ph]
\begin{centering}
\includegraphics[width=1\textwidth]{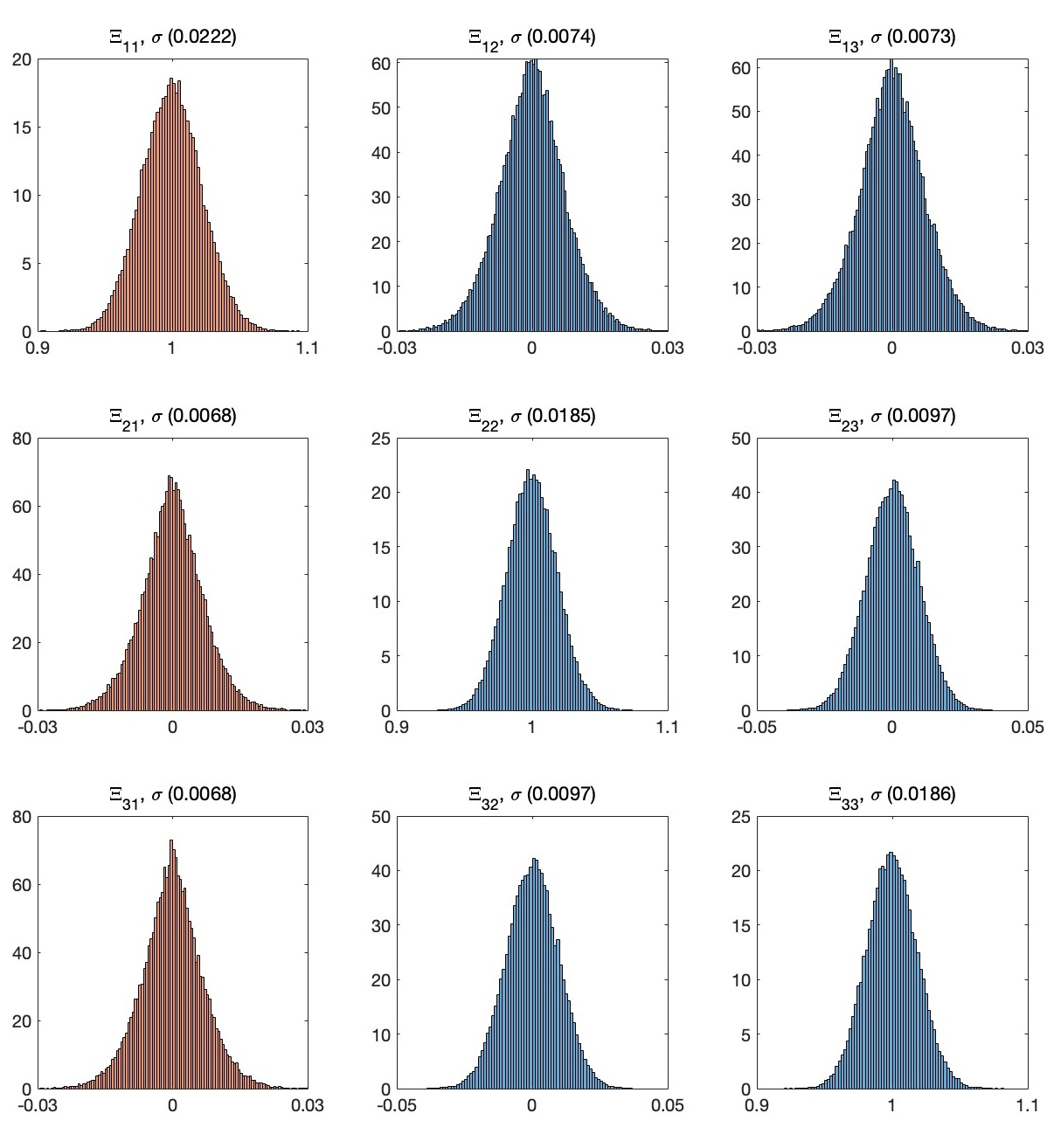}
\par\end{centering}
\caption{{\small{}Histograms for the elements of $\hat{\Xi}$ when $\nu=2$
and $T=4,000$. We also report the standard deviations in each subtitles.\label{fig:nu2_N4000}}}
\end{figure}

\begin{figure}[ph]
\begin{centering}
\includegraphics[width=1\textwidth]{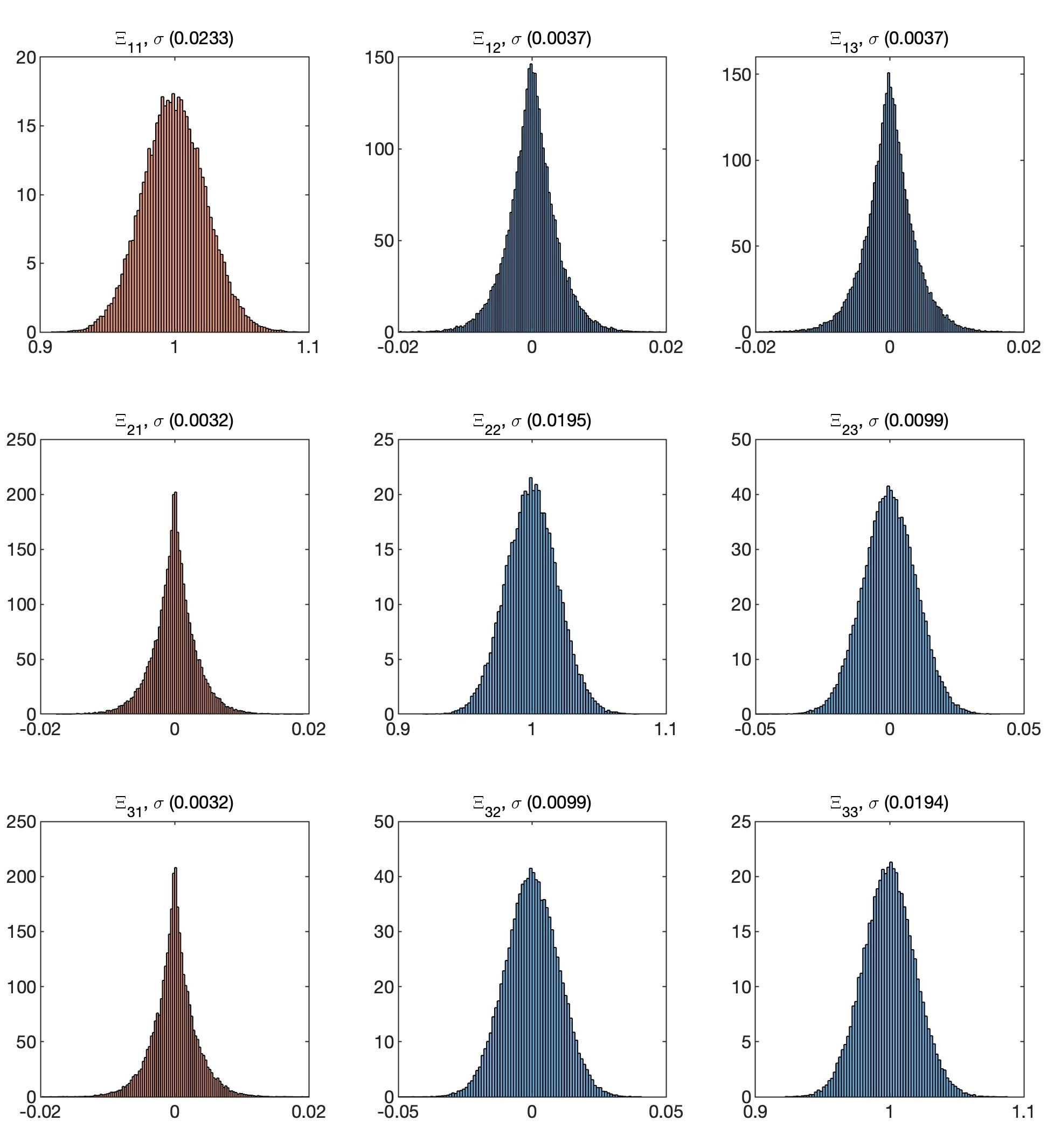}
\par\end{centering}
\caption{{\small{}Histograms for the elements of $\hat{\Xi}$ when $\nu=3/2$
and $T=4,000$. We also report the standard deviations in each subtitles.\label{fig:nu15_N4000}}}
\end{figure}

\begin{figure}[ph]
\begin{centering}
\includegraphics[width=1\textwidth]{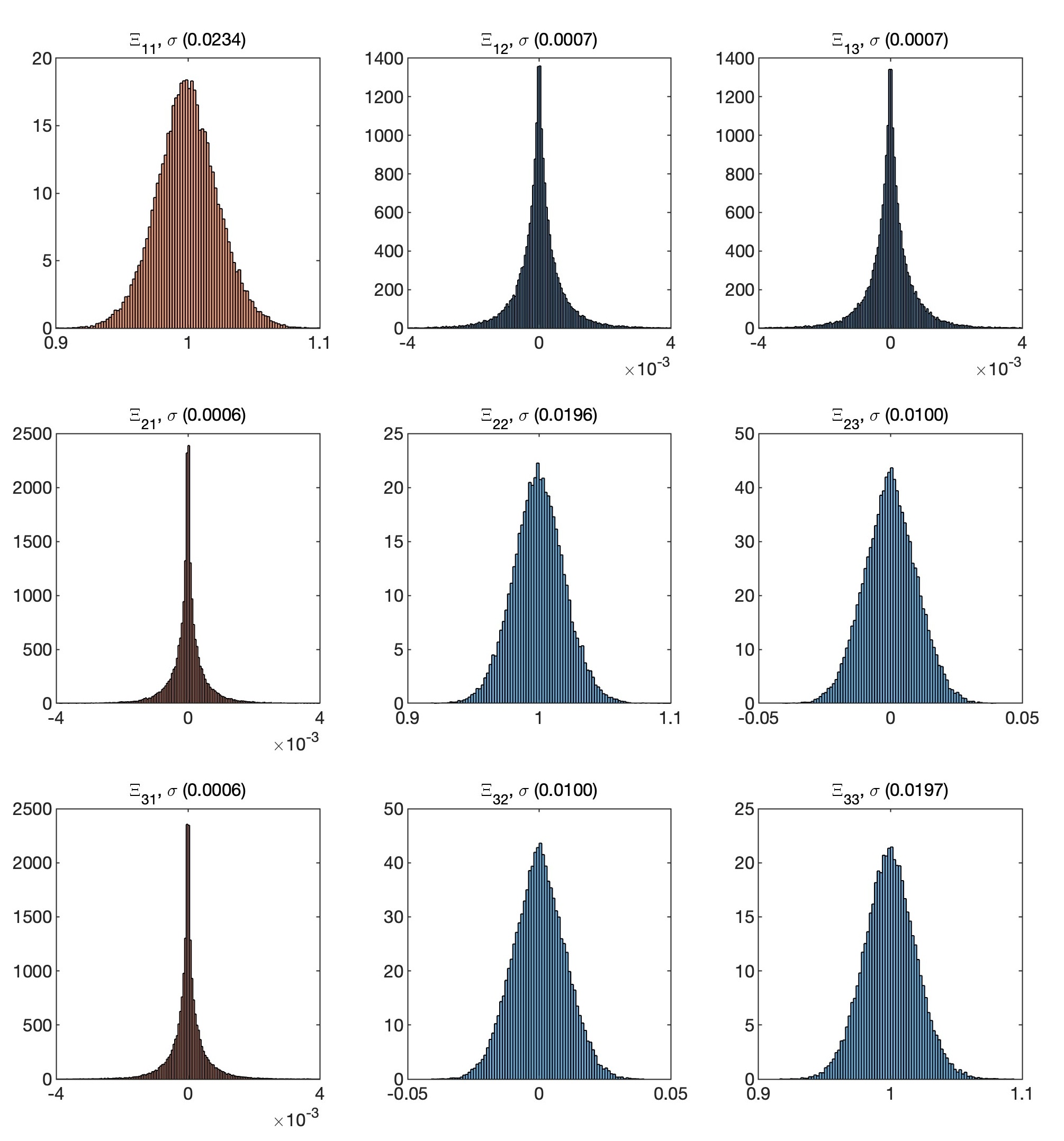}
\par\end{centering}
\caption{{\small{}Histograms for the elements of $\hat{\Xi}$ when $\nu=1$
and $T=4,000$. We also report the standard deviations in each subtitles.\label{fig:nu1_N4000}}}
\end{figure}

\begin{figure}[ph]
\begin{centering}
\includegraphics[width=1\textwidth]{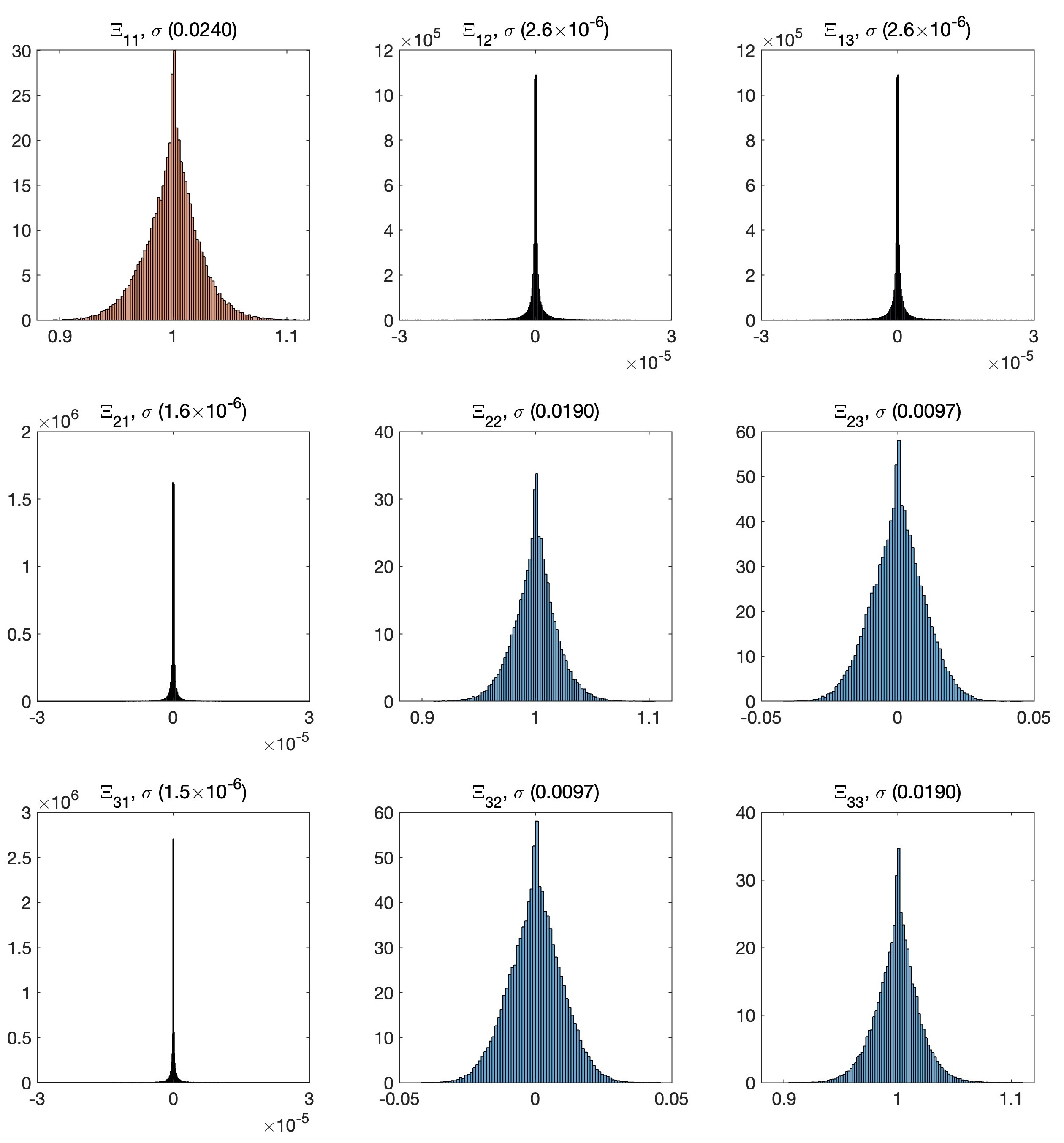}
\par\end{centering}
\caption{{\small{}Histograms for the elements of $\hat{\Xi}$ when $\nu=1/2$
and $T=4,000$. We also report the standard deviations in each subtitles.
\label{fig:nu05_N4000}}}
\end{figure}

\end{document}